%% file: arxiv.tex
\documentclass[aos,preprint]{imsart}

\RequirePackage{amsthm,amsmath,amsfonts,amssymb}
\RequirePackage[authoryear]{natbib} 
\usepackage{xcolor}
\usepackage{xr}
\usepackage{xr-hyper}

\usepackage[pdftex,bookmarks,colorlinks,breaklinks]{hyperref}  
\definecolor{dullmagenta}{rgb}{0.4,0,0.4}   
\definecolor{darkblue}{rgb}{0,0,0.4}
\definecolor{coquelicot}{rgb}{0.20, 0.12, 0.72}
\definecolor{navyblue}{rgb}{0,0,0.5}
\definecolor{newblue}{RGB}{46, 77, 167}
\hypersetup{linkcolor=blue,citecolor=blue,filecolor=blue,urlcolor=blue} 

\usepackage{pifont}
\newcommand{\cmark}{\text{\ding{51}}}
\newcommand{\xmark}{\text{\ding{55}}}

\newcommand\independent{\protect\mathpalette{\protect\independenT}{\perp}}
\def\independenT#1#2{\mathrel{\rlap{$#1#2$}\mkern2mu{#1#2}}}

\usepackage{epstopdf}

\usepackage{multirow}
\usepackage{subcaption}
\usepackage{mathtools}
\usepackage{bbm}
\usepackage{graphicx}  
\usepackage{flafter}  

\usepackage{bm}  

\usepackage{colortbl}
 \usepackage{arydshln}
\usepackage{booktabs}
\usepackage{algorithm}
\usepackage{algorithmic}

\usepackage{footnote}
\makesavenoteenv{tabular}
\makesavenoteenv{table}

\definecolor{coquelicot}{rgb}{0.90, 0.42, 0.72}
\definecolor{burntorange}{rgb}{0.8, 0.33, 0.0}

\definecolor{burntblue}{RGB}{0, 114, 206}
\def\tcr{\textcolor{red}}

\definecolor{burntorange}{rgb}{0.8, 0.33, 0.0}

 \definecolor{burntblue}{RGB}{0, 114, 206}

\input{Definitions.tex}

\startlocaldefs
\theoremstyle{plain}
\newtheorem{theorem}{Theorem}[section]
\newtheorem{corollary}[theorem]{Corollary}
\newtheorem{lemma}[theorem]{Lemma}

\theoremstyle{remark}
\newtheorem{assumption}{Assumption}
\newtheorem{remark}{Remark}

\endlocaldefs

\begin{document}

\begin{frontmatter}

\title{The Decaying Missing-at-Random Framework: Model Doubly Robust Causal Inference with Partially Labeled Data}

\begin{aug}
\author[A]{\fnms{Yuqian} \snm{Zhang}\ead[label=e1]{yuqianzhang@ruc.edu.cn}\thanksref{T1}},
\author[B]{\fnms{Abhishek } \snm{Chakrabortty}\ead[label=e2]{abhishek@stat.tamu.edu}\thanksref{T1}}
\and
\author[C]{\fnms{Jelena} \snm{Bradic}\ead[label=e3]{jbradic@ucsd.edu}\thanksref{T1}}

\thankstext{T1}{All authors contributed equally in this work.} %

\address[A]{Institute of Statistics and Big Data, Renmin University of China \printead{e1}}
\address[B]{Department of Statistics, Texas A\&M University \printead{e2}}
\address[C]{Department of Mathematics and Halicioglu Data Science Institute, University of California San Diego \printead{e3}}
\end{aug}

\begin{abstract}
{In {modern large-scale observational studies,} data collection constraints often result in partially labeled datasets, posing challenges for reliable causal inference, {especially due to potential labeling bias and relatively small size of the labeled data.} This paper introduces a \emph{decaying missing-at-random} (decaying MAR) framework {and associated approaches} for {doubly robust} causal inference {on treatment effects in such semi-supervised (SS) settings. This} simultaneously addresses selection bias {in the labeling mechanism} and the extreme imbalance between labeled and unlabeled groups, bridging the gap between {the} standard SS and missing data literatures, {while throughout allowing for confounded treatment assignment and high-dimensional confounders under appropriate sparsity conditions.} To ensure robust causal conclusions, we propose a \emph{bias-reduced SS (BRSS) estimator} for the average treatment effect, {a type of `model doubly robust' estimator appropriate for such settings,} establishing asymptotic normality {at the appropriate rate} under decaying labeling propensity scores, provided that at least one nuisance model is correctly specified. Our approach also relaxes sparsity conditions beyond those required in existing methods, including standard supervised approaches. Recognizing the asymmetry between labeling and treatment mechanisms, we further introduce a \emph{de-coupled BRSS (DC-BRSS)} estimator, which integrates inverse probability weighting (IPW) with bias-reducing techniques in nuisance estimation. This refinement further weakens model specification and sparsity requirements. Numerical experiments confirm the effectiveness and adaptability of our estimators in addressing labeling bias and model misspecification.}
\end{abstract}

\begin{keyword}
\kwd{Causal inference}
\kwd{extreme imbalance}
\kwd{generalizability}
\kwd{high dimensions}
\kwd{debiasing}
\kwd{semi-supervised learning}
\kwd{model double robustness}
\end{keyword}

\end{frontmatter}


\section{Introduction}\label{sec:intro}

{Semi-supervised (SS) learning has a long history in machine learning applications such as speech recognition and image classification \citep{zhu2005semi, chapelle2006semi}. In these problems, collecting {the outcome variable (or label)} \( Y \in \mathbb{R} \) is often costly and labor-intensive, whereas covariate information is abundant and readily available. This setting naturally leads to SS datasets, which consist of both labeled and unlabeled samples, with the latter typically far outnumbering the former. More recently, the growing availability of large-scale biomedical datasets, such as electronic health records (EHR) and integrative genomic studies, has further expanded the role of SS learning in causal inference. In these settings, rich baseline information {on covariates (or potential confounders)} is accessible for patients, but outcome data often require additional assessments and evaluations, leading to extreme missingness. Consequently, SS learning has gained increasing attention in causal inference research \citep{cheng2021robust, hou2021efficient, zhang2022high, chakrabortty2022general, kallus2024role}, where the primary objective is to estimate the effect of a \emph{binary treatment \( T \in \{0,1\} \)} on the outcome \( Y \) in the presence of substantial outcome missingness due to practical constraints.

SS causal inference can also be viewed as a data integration problem, where \emph{generalizability} and \emph{transportability} of causal conclusions have garnered growing interest; see, e.g., \cite{dahabreh2019generalizing, lesko2017generalizing, shi2023data, degtiar2023review, colnet2024causal, athey2020combining}. For example, in economic policy, combining randomized trial data on job training programs with large administrative datasets can lead to more precise estimates of program impacts. In public health, integrating survey outcomes with large-scale administrative records enables a more comprehensive evaluation of campaign effects on behavioral changes. Similarly, incorporating vast unlabeled datasets in integrative genomics can strengthen the discovery of gene-disease associations.}

{To formulate the causal question, we adopt} the potential outcome framework \citep{rubin1974estimating,imbens2015causal} and consider potential outcomes $Y(0)$ and $Y(1)$, {the versions of $Y$ that `would have been observed'} corresponding to treatment assignments $T=0$ and $T=1$, respectively. The {\it observable outcome} is denoted by $Y \equiv Y(T)$, and under the consistency assumption, $Y=TY(1)+(1-T)Y(0)$. The {\it average treatment effect} (ATE), {our causal estimand of interest,} characterizes the average causal effect of $T$ on $Y$ and is defined as:
\begin{align}\label{def:ATE}
\mu_0:=\theta_1-\theta_0,\;\;\text{where}\;\;\theta_1:=\E\{Y(1)\}\;\;\text{and}\;\;\theta_0:=\E\{Y(0)\}.
\end{align}
Estimating the ATE in {\it observational} studies, where the treatment $T$ and outcome $Y$ are influenced by a shared set of confounders $\bX\in\R^p$, presents challenges due to confounding.

The complexity deepens further in SS settings, where a labeled dataset \(\Lsc = (\bX_i, T_i, Y_i)_{i=1}^n\) is accompanied by a {\it substantially larger} unlabeled dataset \(\Usc = (\bX_i, T_i)_{i=n+1}^N\). {Hence, the observable outcome $Y \equiv Y(T)$ is missing {\it itself}, apart from the natural missingnness of one of $\{Y(1), Y(0)\}$ that is governed by $T$. Moreover, this work} also addresses challenges arising from {\it high-dimensional} covariates, accommodating scenarios where \( p \gg n \). Classical SS approaches assume that the labeled and unlabeled datasets follow the same distribution, relying on the assumption that {the outcomes $Y$ in $\Usc$} are missing completely at random (MCAR) \citep{cheng2021robust, zhang2022high, hou2021efficient, chakrabortty2022general}. This assumption facilitates efficient use of both labeled and unlabeled data, {with the primary goal being efficiency improvement (over the supervised approach on $\Lsc$)}. However, in practice, MCAR is often violated {-- the observational setting may naturally cause the labeling mechanism to be biased as well}. This paper aims to address these limitations by tackling selection bias {in labeling}, where the missingness of the outcome $Y$, denoted {via} the \emph{labeling indicator} $R\in \{0,1\}$, is itself observational and may depend on {\it both} $(T,\bX)$.
Moreover, in typical SS situations where the unlabeled dataset is significantly {larger} than the labeled dataset ($N \gg n$), the probability of observing $Y$ {\it decreases uniformly} as $N$ increases, violating the standard positivity assumption {in missing data literatures} \citep{bang2005doubly,tsiatis2007semiparametric}. This necessitates a framework that simultaneously addresses labeling bias and the extreme imbalance between labeled and unlabeled groups.

\subsection{SS causal inference with labeling bias: the `decaying' MAR framework}\label{sec:setup}

We formulate the problem as a {two-layered (but non-standard)} missing data problem {-- one naturally through $T$ and the other (more challenging one) through $R$}. Unlike traditional SS settings, our framework treats the labeling indicator \( R \) as \emph{random}. We allow the labeled fraction \( n/N>0 \) to approach zero arbitrarily closely and investigate the theoretical properties when both \( \P(R=1)\to 0 \) and \( N\to\infty \), {leading to {\it non-standard asymptotics}}. The labeling probability \( \P(R=1) \), the marginal distribution of \( R \equiv R_{(N)} \), and the conditional distribution of \( R\mid (T,\bX) \) are all considered as {\it decaying} functions of \( N \), with the limit taken as \( N\to\infty \). {Formally, we assume that the observed samples \( (\bZ_{N,i})_{i=1}^N = (\bZ_i)_{i=1}^N = (\bX_i, T_i, R_i, R_i Y_i)_{i=1}^N \) are drawn from the \( N \)-th row of a triangular array of vectors \( (\bZ_{N,i})_{1 \leq i \leq N, N \geq 1} \), where the elements within the same row are independent and identically distributed (i.i.d.). For simplicity, we omit the dependency on \( N \) when possible to streamline the notation.}

{Since the ATE parameter is given by the difference between two counterfactual means, $\mu_0 = \E\{Y(1)\} - \E\{Y(0)\} = \theta_1 - \theta_0$, {\it we focus on the estimation of $\theta_1$ throughout this work, {without loss of generality}}. Analogous results apply to $\theta_0$ and, consequently, to the estimation of $\mu_0 = \theta_1 - \theta_0$. We define the corresponding outcome regression (OR) models, the propensity score (PS) models for treatment $T$ and labeling $R$, as well as the PS model for a `product' indicator as follows. For $\bx$ in $\Xsc\subseteq\R^p$, the support of $\bX$, we define:
\begin{align}
&\textbf{OR model:}\quad m(\bx):=\E\{Y(1)\mid \bX=\bx\},\label{eq:key-defs0}\\
&\textbf{T-PS model:}\quad\pi(\bx):=\P(T=1\mid\bX=\bx),\label{eq:key-defs1}\\
&\textbf{R-PS model:}\quad p_N(\bx):=\P(R=1\mid T=1,\bX=\bx),\label{eq:key-defs2}\\
&\textbf{Product PS model:}\quad\gamma_N(\bx) ~:=~ \P(\Gamma = 1\mid\bX=\bx), \;\;\; \mbox{and} \;\;\; \gamma_{N} := \P(\Gamma = 1), \label{eq:key-defs3}\\
&\mbox{where} \;\;\; \Gamma:=TR ~~\mbox{is the {\it product indicator}.} \nonumber
\end{align}

Since the product PS function can be expressed as the product of the treatment and labeling PS functions, \( \gamma_N(\bx) = \pi(\bx)p_N(\bx) \), the complexities and sparsity structures of these functions are inherently {\it intertwined}. This interdependence necessitates the development of methods that capitalize on these relationships to improve robustness and efficiency; see Section \ref{subsec:SP} and Remark \ref{remark:MDR1} for a more detailed discussion and methodological development.}

To identify the parameter of interest, we assume the following basic conditions.

\begin{assumption}[Basic assumptions]\label{cond:basic}

(a) We assume the `no unmeasured confounding' (NUC) and overlap assumptions {\it for the treatment} $T$, so that for some constant $c_0\in(0,1/2)$:
\begin{align*}
& T \ind \{Y(0), Y(1)\} \mid \bX \;\; \leadsto \mbox{(\textbf{T-NUC})} \;\;\; \mbox{and} \;\;\; c_0 < \pi(\bx) < 1-c_0 \;\; \forall \; \bx \in \Xsc \;\; \leadsto \mbox{(\textbf{T-overlap})}.
\end{align*}

(b) We further assume a missing at random (MAR) condition and a \emph{`decaying overlap'} condition (DOC) {\it for the  labeling indicator} $R \equiv R_{(N)}$ as follows:
\begin{align*}
R \ind Y \mid (T,\bX) \;\; \leadsto \mbox{(\textbf{R-MAR}),} \;\;\,\mbox{and}\;\;\,  {p_N(\bx) > 0 \;\, \forall \, \bx\in\Xsc, \; \mbox{for each fixed $N$}\;\; \leadsto \mbox{(\textbf{R-DOC}).}}
\end{align*}
\end{assumption}

The T-NUC assumption (`ignorability'), and the T-overlap condition are commonly assumed \citep{rosenbaum1983central,imbens2015causal}. The R-MAR condition was used recently in \cite{wei2022doubly}, but the authors didn't consider the full SS setting of $N \gg n$. \cite{kallus2024role} discuss it but only develop theory under a relaxed case of $R\independent T\mid\bX$ with an {inappropriate} assumption $\P(R=1)=0$\footnote{{\cite{kallus2024role} required the assumptions \( \P(R = 1) = 0 \) and \( R \independent T \mid \bX \) in their initial preprints but relaxed them in the published version, following the release of the initial preprint of this manuscript.}}, which leads to $n=\sum_{i=1}^NR_i=0$ almost surely, i.e., a degenerate unsupervised setting. For this, the R-DOC assumption plays a critical role and is a non-standard condition. It is weaker than the traditional (strict) positivity condition, which requires the existence of a constant $c_0>0$ independent of $N$ such that $p_N(\bx)>c_0$ \citep{khan2010irregular,tsiatis2007semiparametric}. {In comparison, the R-DOC assumption accommodates scenarios where the labeling PS may decay to zero, i.e., \( p_N(\bx) \rightarrow 0 \), potentially uniformly over \( \bx \in \mathcal X \). By allowing \( p_N(\bx) \) to depend on both \( \bx \) and \( N \) without imposing uniform lower bounds, this assumption offers a flexible framework suitable for typical SS settings with a small labeling fraction \( n/N \), while also accommodating potential heterogeneity between the labeled and unlabeled groups.} When both the R-MAR and R-DOC conditions hold, we define this scenario as the \emph{`decaying' MAR setting}.

\paragraph*{Challenges and objectives}

{In this setup, we face challenges from the following aspects:
\begin{itemize}
\item \emph{A uniformly decaying labeling PS function.} Unlike standard SS approaches that assume a constant labeling PS function, we allow for non-degenerate scenarios. This generalization introduces additional challenges in establishing results that adapt to a decaying PS.

\pagebreak 
\item \emph{High-dimensional covariates.} In observational studies, the true treatment and labeling mechanisms are unknown. To satisfy the T-NUC and R-MAR conditions in Assumption \ref{cond:basic}, it is essential to collect a comprehensive set of potential confounders to reduce bias from unmeasured confounding. This requirement is even more critical in SS causal inference, where confounders affecting either \( T \) or \( R \) must be accounted for. Consequently, it is typically necessary to handle the more challenging high-dimensional scenario.

\par\smallskip
\item \emph{Bias introduced by nuisance estimation.} ATE estimation is highly sensitive to the initial estimation of the nuisance functions \eqref{eq:key-defs0}-\eqref{eq:key-defs3}. Misspecification of these nuisance models is common and can introduce substantial bias. Moreover, the bias of ATE estimation is also influenced by the sparsity structures of the nuisance models, particularly when regularized estimates are employed in high-dimensional settings. It is crucial to develop methods that are robust {-- not just in consistency but in ensuring valid
    {\it inference} via asymptotic normality (at the right rate) --} and perform well without imposing strict model or sparsity constraints.
\end{itemize}

Our objective is to develop \emph{robust} methods for SS causal inference, ensuring reliable causal conclusions. We aim to address the challenges outlined above by establishing {\it consistent and asymptotically normal} (CAN) {estimators} that: (a) adapt to decaying PS and high-dimensional settings, (b) remain valid under model misspecification, and (c) relax sparsity requirements.}

\subsection{Our contributions}\label{subsec:contribution}
We summarize our main contributions below.

\paragraph*{Establishing a unified framework for SS causal inference problems under labeling bias}
{We introduce a decaying MAR setting for causal inference problems, redefining the non-MCAR SS setup. This approach tackles a previously unexplored intersection by addressing the often-overlooked labeling bias in the SS literature and challenging the conventional positivity condition in the missing data domain. Our work establishes a unified framework that generalizes several existing strands, including causal inference in supervised settings, SS inference under the MCAR condition, missing data problems with non-decaying labeling probabilities, and mean estimation in a non-causal decaying MAR context; see comparisons in Table \ref{table:settings} and Section \ref{sec:compare-S2}. Additionally, our research advances the literature on the generalizability of randomized controlled trials (RCTs) by integrating RCTs with unlabeled external data. In our framework, \( R \) denotes whether an individual participates in the RCT, and the absence of a strict positivity condition enables our method to address scenarios where the external data size exceeds that of the RCT—a prevalent yet unresolved challenge in causal generalizability.

Once the causal decaying MAR setting is established, we recognize that the SS causal inference problem can be framed as a two-occasion missing data problem (see identification results in Lemma \ref{lemma:rep}). Therefore, achieving CAN estimation for the ATE becomes a natural extension of existing doubly robust (DR) methods \citep{chernozhukov2018double, zhang2023double}, {except now} within the causal decaying MAR framework. See Theorem \ref{thm:gen-rate-DR}, where we derive CAN results with adaptive rates that account for decay in the PS, under the {\it condition} that the product nuisance estimation error is small enough. This leads to the property known as {\it\underline{rate DR}} {(`rate double robust')}; see Definition 2 of \cite{smucler2019unifying} for its supervised counterpart.

\begin{table}[t]
\caption{Comparison of the missing outcome settings. `Labeling bias' allows $R$ to be dependent on $(Y,\bX)$. `Causal setup $+$ missing $Y$' allows for the observed outcome $Y=Y(T)$ to be possibly missing with $R,T\not\equiv1$.
}\label{table:settings}
\resizebox{\columnwidth}{!}{%
\begin{tabular}{| p{9.2cm}| c | c | c |}
\hline
\multirow{2}{*}{Settings} & \multirow{1}{*}{Labeling } & \multirow{1}{*}{Decaying } & Causal $+$ \\
 & Bias& PS& missing $Y$\\
\hline
\multirow{3}{*}{\parbox{9.2cm}{\cite{kawakita2013semi,azriel2022semi,chakrabortty2018efficient,zhang2019semi,tony2020semisupervised,chan2020semi,xue2021semi,chakrabortty2022semi}}} & \multirow{3}{*}{$\xmark$} & \multirow{3}{*}{$\cmark$} & \multirow{3}{*}{$\xmark$}\\
&&&\\
&&&\\
\hline
\multirow{2}{*}{\parbox{9.2cm}{\cite{cheng2021robust,hou2021efficient,zhang2022high,chakrabortty2022general}}} & \multirow{2}{*}{$\xmark$} & \multirow{2}{*}{$\cmark$} & \multirow{2}{*}{$\cmark$}\\
&&&\\
\hline
\multirow{3}{*}{\parbox{9.2cm}{\cite{rubin1976inference,robins1994estimation,robins1995semiparametric,
bang2005doubly,tsiatis2007semiparametric,kang2007demystifying,graham2011efficiency,chakrabortty2019high}}} & \multirow{3}{*}{$\cmark$} & \multirow{3}{*}{$\xmark$} & \multirow{3}{*}{$\xmark$}\\
&&&\\
&&&\\
\hline
\multirow{2}{*}{\parbox{9.2cm}{\cite{dahabreh2019generalizing,lesko2017generalizing,kallus2024role,wei2022doubly}}} & \multirow{2}{*}{$\cmark$} & \multirow{2}{*}{$\xmark$} & \multirow{2}{*}{$\cmark$}\\
&&&\\
\hline
\cite{zhang2023double} & $\cmark$ & $\cmark$ & $\xmark$\\
\hline
The proposed (causal) decaying MAR setting & $\cmark$ & $\cmark$ & $\cmark$\\
\hline
\end{tabular}
}
\end{table}

\paragraph*{Robust inference {for the ATE} through bias-reducing approaches}

While achieving rate DR {estimators} through DR representations and cross-fitting techniques has become a standard approach \citep{chernozhukov2018double, zhang2023double}, this requires consistent estimation for {\it both} the OR and product PS models {and} with {\it fast enough} convergence rates {-- any violations of these can only possibly ensure consistency but {\it not} the CAN property (nor rate) necessary for valid inference}. Such requirements are often too restrictive in practice, necessitating additional efforts to mitigate the impact of model misspecification and relatively slow nuisance estimation rates, especially in high-dimensional or non-parametric settings. To improve robustness {and guarantee CAN properties (therefore, valid inference)}, we propose two novel classes of DR methods, each based on distinct PS representations.

The first approach jointly models the product PS. To mitigate the bias introduced by model misspecification, we adopt the loss functions from \cite{tan2020model, smucler2019unifying, avagyan2021high, bradic2019sparsity} and extend the method to the decaying MAR setting, allowing for a decaying labeling PS. This method is referred to as the \emph{bias-reduced SS (BRSS)} estimator; see Section \ref{subsec:nuisance}. CAN results are established as long as either the OR model or the product PS model is correctly specified, a property also known as {\it\underline{model DR}}; see Definition 3 of \cite{smucler2019unifying}. Furthermore, using a specialized asymmetric cross-fitting technique, we further demonstrate that the proposed method is not only {\it model DR}, but also {\it\underline{sparsity DR}} (see a weaker version in Theorem 1 of \cite{bradic2019sparsity}). This means that CAN results can be achieved without the typical product-rate condition, as long as one nuisance model exhibits an ultra-sparse structure, and the sparsity condition for the other model is relatively mild; see Theorem \ref{cor:para}. In contrast to the aforementioned works that focus on supervised settings, our framework supports a decaying PS, making it applicable to SS problems where \( N \gg n \). {It is important to note that due} to the extreme missingness, the {\it effective sample sizes} for both nuisance estimation and final ATE estimation {\it deviate from standard scenarios} {(see below Theorem \ref{thm:gen-rate-DR} for its explicit characterization).} Our work provides an initial attempt to understand the necessary sparsity conditions for achieving CAN results in this non-standard setting. Moreover, even in canonical supervised settings with \( N = n \), we impose weaker sparsity conditions than those required by these approaches, as shown in Table \ref{table:sparsity} and Section \ref{sec:compare-S2}.

While directly modeling the product PS and treating the SS causal inference problem entirely as a two-occasion missing data problem is a feasible approach, it does not fully leverage the asymmetry between labeling and treatment mechanisms. Under the causal decaying MAR setup, we have rich information for the treatment PS due to the fully observed pairs \((\bX_i, T_i)\) with balanced treatment groups. In contrast, when dealing with a large number of unlabeled samples, the labeling groups \(R_i\) are highly imbalanced, making labeling PS estimation more difficult. Therefore, our second approach, the \emph{de-coupled BRSS (DC-BRSS)} estimator, de-couples the PS estimation. We estimate the treatment PS fully non-parametrically and introduce a novel integration of IPW and bias-reducing techniques in nuisance estimation. This approach further extends the range of both {\it sparsity DR} and {\it model DR} properties, as detailed in Theorem \ref{t4-SP} and Remark \ref{remark:MDR1}. {Finally, it is worth emphasizing that} our work tackles {technical} challenges {\it beyond} traditional high-dimensional analyses, such as extreme label imbalance, sequential nuisance estimation, and establishing CAN results under weaker model and sparsity conditions; see Remarks \ref{remark:challenge-ATE} and \ref{remark:challenge-nuisance}. The results, including those on nuisance estimation, are of independent interest and offer broader insights beyond the immediate scope of this study.}

\paragraph*{Organization}
{The {rest of the} paper is structured as follows. Section \ref{sec:psetup} introduces {some preliminaries, including} basic identifications and an initial rate DR estimator. Section \ref{sec:MT} presents the proposed bias-reduced estimators, followed by the main theoretical results in Section \ref{sec:theory} and a comparative analysis with existing approaches in Section \ref{sec:compare-S2}.} Numerical results based on simulated and pseudo-random datasets are provided in Sections \ref{sec:sim} and \ref{sec:data_analysis}, respectively, {followed by a concluding discussion in Section \ref{sec:disc}. Additional discussions, theoretical results, and proofs of all results in the main paper are included in the \hyperref[supp_mat]{Supplement} (Sections \ref{subsec:misT}-\ref{sec:proof_sec:gen}).} 

\paragraph*{Notation}
Throughout {the paper}, {we generically denote positive constants independent of \( N \) as lowercase or uppercase letters} \( c \) and \( C \) {which may vary across instances}.
The notations \( \P(\cdot) \) and \( \E(\cdot) \) represent {the probability and expectation under the joint distribution of $\mathbb{Z}:= \{\bX,T,R,Y(1),Y(0)\}$.} For any subset \( \Asc \) {of data}, \( \P_{\Asc} \) and \( \E_{\Asc}(\cdot) \) represent {probability and expectation under the joint distribution of $\Asc$.} \( \P_{\bX} \) and  {\( \E_{\bX}\{f(\bX)\} := \int f(\bx) d\P_{\bX}(\bx) \) respectively} denote the marginal distribution of \( \bX \) and the expectation of any function \( f \) {of $\bX$ (where $f(\cdot) \equiv \fhat(\cdot)$ itself maybe a random function $\ind$ of its random argument $\bX \sim \P_{\bX}$)}. For any \( r>0 \) {and any (non-random) function $g$ of $\mathbb{Z}$, we let \( \|g(\cdot)\|_{\P,r} := \{\E|g(\mathbb{Z})|^r\}^{1/r} \),} and for any vector \( \bz \in \mathbb{R}^d \) {(for any $d$), we let} \( \|\bz\|_r := (\sum_{j=1}^d|\bz_j|^r)^{1/r} \), \( \|\bz\|_0 := |\{j: \bz_j \neq 0\}| \), and \( \|\bz\|_\infty := \max_j|\bz_j| \). For any square matrix \( \bSigma\), \( \|\bSigma\|_r := \sup_{\|\bv\| \neq 0} \|\bSigma \bv\|_r / \|\bv\|_r \). For any $a,b\in\R$, let $a\land b=\min\{a,b\}$ and $a\lor b=\max\{a,b\}$. Lastly, \( \be_j \in \mathbb{R}^d \) denotes the \( j \)-th column of the identity matrix \( \mathbbm{I}_d \). {Finally, all key statistical notations used in the paper are compactly summarized in Table \ref{table:notations} of the \hyperref[supp_mat]{Supplement}.} 

\section{Preliminaries: identifications and a rate DR estimator}\label{sec:psetup}

We first establish some {basic identification results for $\theta_1$} under the causal decaying MAR framework. Next, {based on the identification}, we introduce a {natural} \emph{rate DR} SS estimator {of $\theta_1$} as a baseline procedure and demonstrate its rate DR property within a non-standard asymptotic regime.

\subsection{Identification of the parameter(s)}\label{sec:identification}
Given the definitions in \eqref{eq:key-defs0}-\eqref{eq:key-defs2} and Assumption \ref{cond:basic}, we introduce an augmented inverse probability weighting (AIPW) representation, also known as the doubly robust (DR) form, incorporating the OR model in \eqref{eq:key-defs0} and the product PS model in \eqref{eq:key-defs3}.

\begin{lemma}[Identification of $\theta_1$]\label{lemma:rep}
Let Assumption \ref{cond:basic} hold. Then, for any arbitrary functions $m^*(\cdot)$ and $\gamma_N^*(\cdot)$, as long as either $m^*(\cdot) = m(\cdot)$ or $\gamma_N^*(\cdot) = \gamma_N(\cdot)$ holds but not necessarily both, we have:
\[\E\left\{\psi^*_{N} (\bZ)\right\} =0,\;\;\mbox{with}\;\;\psi^*_{N}(\bZ)=m^*(\bX) - \theta_1+\frac{\Gamma}{\gamma_N^*(\bX)} \left\{Y - m^*(\bX)\right\}.\]
\end{lemma}

{The representation above clarifies that estimating \(\theta_1 = \E\{Y(1)\}\) can be framed as a mean estimation problem under a MAR labeling, where the \emph{effective labeling indicator} is \(\Gamma = TR\). In this context, the potential outcome \(Y(1)\) is observable only when both \(T=1\) and \(R=1\). Consequently, the SS causal inference problem mirrors a two-occasion missing data problem, akin to the structures found in \cite{robins1994estimation,bang2005doubly,sun2022high,shi2023data}, among others. Recognizing this, rate DR estimation becomes a natural adaptation of existing DR methods \citep{chernozhukov2018double,zhang2023double} {but now} within the {\it causal decaying MAR framework}.}

\subsection{{A rate DR estimator under the decaying MAR setting}}\label{sec:gen-DR-DMAR}

We split the samples into $\K\geq2$ parts, $\mathcal I_1,\dots,\mathcal I_{\K}$ of equal sizes $|\mathcal I_k|=M=N/\K$ and define $\mathcal I_{-k}=\mathcal I\setminus\mathcal I_k$, $\forall k\leq\K$. Let $\mhat^{(-k)}(\cdot)$ and $\gammahat_{N}^{(-k)}(\cdot)$ be estimators of $m(\cdot)$ and $\gamma_N(\cdot)$, respectively, using the samples from $\mathcal I_{-k}$, based on suitable (working) models.

We now define the \emph{rate DR (R-DR) estimator} $\thetahat_{1,\mbox{\tiny R-DR}}$ of
$\theta_1$ as
\begin{align}\label{eq:SS-est-def}
\thetahat_{1,\mbox{\tiny R-DR}} \; := \; \frac{1}{N} \sum_{k=1}^K\sum_{i\in\mathcal I_k}^N \left[\mhat^{(-k)}(\bX_i) + \frac{\Gamma_i\{Y_i - \mhat^{(-k)}(\bX_i)\}}{\gammahat_{N}^{(-k)}(\bX_i)}\right],\;\;\mbox{where}\;\;\Gamma_i:=T_iR_i.
\end{align}
Similarly, the R-DR estimator for $\theta_0$, denoted as $\thetahat_{0,\mbox{\tiny R-DR}}$, is obtained by replacing $T_i$ with $1-T_i$. The R-DR estimator for $\mu_0$ is then given by $\muhat_{\mbox{\tiny R-DR}} = \thetahat_{1,\mbox{\tiny R-DR}} - \thetahat_{0,\mbox{\tiny R-DR}}$.

\paragraph*{Asymptotic properties}

{The rate DR property of $\thetahat_{1,\mbox{\tiny R-DR}}$ can be established using the techniques from \cite{zhang2023double}, where the product $\Gamma_i := T_iR_i$ is treated as the effective labeling indicator for estimating $\theta_1=\E\{Y(1)\}$. Below, we {only} summarize the key findings, with a full description provided in Theorem \ref{t1} {(and Assumptions \ref{a2:high-level}-\ref{a3:L-F}) in Section \ref{sec:add-S1}} of the \hyperref[supp_mat]{Supplement}}.

\begin{theorem}[Asymptotic results for $\thetahat_{1,\mbox{\tiny R-DR}}$]\label{thm:gen-rate-DR}
Let Assumptions \ref{cond:basic} and \ref{a2:high-level}-\ref{a3:L-F} of the \hyperref[supp_mat]{Supplement} (Section \ref{sec:add-S1}) hold. Suppose that both nuisance models are correctly specified, i.e., $\gamma^*_N(\cdot)=\gamma_N(\cdot)$ and $m^*(\cdot) = m(\cdot)$. Let $\E\{\gamma_N^{-1}(j,\bX)\}<\infty$ {(for each $N$)} and {define:
\begin{align}
a_{N} \, := \, [\E\{\gamma_N^{-1}(\bX)\}]^{-1}.\label{eq:a_N-defn}
\end{align}
Assume that} $Na_N \rightarrow \infty$ and possibly, $a_N \rightarrow 0$, as $N \rightarrow \infty$.
Let $c<\E[\{Y(1)-m(\bX)\}^2\mid\bX]\leq C$ almost surely and $\Var\{Y(1)\}\leq C<\infty$ with constants $c,C>0$. Suppose that the following product-rate condition holds:
\begin{align}\label{cond:product-rate}
\E_{\bX}\{\mhat^{(-k)}(\bX) - m(\bX) \}^2\E_{\bX}\left\{1 - \frac{\gamma_N(\bX)}{\gammahat_N^{(-k)}(\bX)} \right\}^2 = o_p(1/(Na_N)).
\end{align}
{Define the DR score:} $\psi_{N}^{opt}(\bZ) := m(\bX) - \theta_1 + \Gamma\{Y(1) - m(\bX)\}/\gamma_N(\bX)$, {and let} $\Sigma_N^{opt} := \Var\{\psi^{opt}_{N}(\bZ)\}$.
Then, {$\Sigma_N^{opt} \asymp a_N^{-1}$,}
\begin{align}\label{result:normal1}
\sqrt{Na_N}(\thetahat_{1,\mbox{\tiny R-DR}} - \theta_1) \, = \, O_p(1), \;~\mbox{and}~\; \sqrt{N} \left(\Sigma^{opt}_N\right)^{-1/2}(\thetahat_{1,\mbox{\tiny R-DR}} - \theta_1) \, \xrightarrow{d} \, \mathcal N(0,1).
\end{align}
\end{theorem}

Here, \( Na_N \) denotes the \emph{effective sample size} under 
the decaying MAR {setting,} 
while \( a_N \) signifies the {\it deceleration factor} resulting from the decaying PS, giving rise to atypical convergence rates {of the form ${(Na_N)}^{-1/2}$}.

\begin{remark}[Necessity and usefulness of the decaying PS]\label{remark:dPS}
The introduced decaying MAR schema addresses the dependence of $a_N$, $\gamma_N(\bX)$, $R_i\equiv R_{N,i}$, and $R_i\mid(T_i,\bX_i)$ on $N$ in $N \gg n$ SS causal contexts. This schema accounts for the asymptotic scenario where $N\to\infty$ and $\P(R=1)\to0$. Previous studies may have overlooked this aspect, leading to limitations in exploring the crucial $N \gg n$ setting and only allowing $\P(R=1)=c=\lim_{N\to\infty}\E(n/N) >0$.
In conventional doubly robust (DR) literature, estimation error control relies on $a_N^{-2}\|\balphahat-\balpha^*\|_2 =o(1)$ for a nuisance parameter $\balpha^*$ \tcr{and its corresponding estimator $\balphahat$}; see, e.g., the control of $\mathcal I_2$ in Step 5 of the proof of Theorems 5.1 and 5.2 of \cite{chernozhukov2018double}. However, this condition becomes excessively demanding when $a_N\to0$, rendering accurate estimation of the nuisance parameters practically unattainable. For instance, as the expected inverse of the labeled sample size is $\E(1/n)\asymp 1/(Na_N)$, we have $\|\balphahat-\balpha^*\|_2=O_p((Na_N)^{-1/2})$ even in low dimensions. This leads to a very restrictive requirement on the decaying probability, $Na_N^5\to\infty$. Through refined analysis on the decaying PS, our results only require $Na_N\to\infty$.
\end{remark}

\begin{remark}[The rate DR property]\label{remark:rate-DR}
{The estimator \(\thetahat_{1,\mbox{\tiny R-DR}}\) is {\it \underline{rate DR}}, indicating that the asymptotic normality in \eqref{result:normal1} is achieved when the product-rate condition \eqref{cond:product-rate} holds for the nuisance estimation errors. Unlike the standard rate DR property seen in traditional supervised settings \citep{chernozhukov2018double, smucler2019unifying}, both the required product-rate condition and the resulting convergence rate of \(\thetahat_{1,\mbox{\tiny R-DR}}\) are influenced by the \emph{`effective sample size'} \(Na_N\), which adaptively incorporates the PS decay. Furthermore, the established rate and asymptotic variance in \eqref{result:normal1} aligns with the semi-parametric efficiency bound established by \cite{kallus2024role} when \(a_N \to 0\) as \(N \to \infty\). Notably, the nuisance estimation rates are also governed by \(Na_N\), rather than the total sample size \(N\); see, e.g., \cite{wang2020logistic, zhang2023double}.}
\end{remark}

{The asymptotic normality result \eqref{result:normal1} holds under the correct specification of both the OR and PS models. If one of these nuisance models is misspecified, \(\thetahat_{1,\mbox{\tiny R-DR}}\) remains consistent but converges at a slower rate, proportional to the estimation error of the correctly specified model, as outlined in Theorem \ref{t1}. Consequently, model misspecification can introduce significant bias, particularly when nuisance parameters cannot be estimated at a parametric rate, which is often the case in high-dimensional or non-parametric settings. In Section \ref{sec:MT}, we propose refined estimators that further achieve \emph{sparsity DR} and \emph{model DR}, thereby relaxing the model specification and complexity conditions required for CAN estimation.}

\section{The bias-reduced estimators}\label{sec:MT}

{In this section, we focus on developing robust inference methods for the ATE with enhanced robustness. We introduce two novel doubly robust (DR) estimators, representing the primary methodological contribution of this work. First, we propose the bias-reduced SS (BRSS) estimator in Section \ref{subsec:nuisance}, providing robust inference under possible model misspecification. Building on this, we {further} develop a de-coupled BRSS (DC-BRSS) estimator in Section \ref{subsec:SP}, which makes better use of the large number of unlabeled samples to strengthen robustness.}

\subsection{A bias-reduced SS (BRSS) ATE estimator}\label{subsec:nuisance}
We improve the estimator from Section \ref{sec:gen-DR-DMAR} by defining refined nuisance estimates to minimize bias {from model misspecification}. Consider the following \emph{working} models for $m(\cdot)$ and $\gamma_N(\cdot)$ with `targets':
\begin{equation}\label{eq:models_para}
m^*(\bx) = \varphi(\bx)^T\balpha^* 
\;\;\; \mbox{and} \;\;\; \gamma^*_N(\bx) = g(\varphi(\bx)^T\bbeta^*+\log\gamma_N) = \frac{\gamma_N\exp(\varphi(\bx)^T\bbeta^*)}{1+\gamma_N\exp(\varphi(\bx)^T\bbeta^*)},
\end{equation}
where $g(u):=\exp(u)/\{1+\exp(u)\}$ is the logistic function, $\gamma_N:=\E(\Gamma)$, and {$\varphi(\cdot):\R^p \mapsto \R^d$ represents a pre-specified feature map, enabling the use of series-type non-parametric nuisance estimators}. We extend the analysis to {allow for} {\it high-dimensional} settings, where $d$ is potentially large relative to the `effective sample size,' \(Na_N\). In \eqref{eq:models_para}, the \(\log\gamma_N\) term serves as an \emph{offset}, ensuring uniform decay in \(N\) for the working model \(\gamma^*_N(\bx)\) when \(\gamma_N \to 0\). This is particularly suitable when the true PS \(\gamma_N(\bx) = \P(TR=1 \mid \bX=\bx)\) also decays uniformly in \(N\), a common occurrence in SS settings with extreme outcome missingness. In the above, $\balpha^*,\bbeta^*\in\R^d$ are the \emph{targeted bias-reducing nuisance parameters} defined as follows:
\begin{align*}
\bbeta^* =~\arg\min_{\bbeta\in\R^d}\E\left\{l (\Gamma, \bX , \bbeta, \gamma_N) \right\} \quad \mbox{and} \quad
\balpha^* =~\arg\min_{\balpha\in\R^d}\E\left\{h(\Gamma, \bX, Y, \balpha, \bbeta^*, \gamma_N)\right\},
\end{align*}
with the loss functions $l$ and $h$ defined as
\begin{align}
l (\Gamma, \bX , \bbeta, \gamma_N)&=(1-\Gamma)\bS^T\bbeta+\frac{\Gamma}{\gamma_N}\exp(-\bS^T\bbeta) \quad \mbox{and}\label{def:loss-1}\\
h(\Gamma, \bX, Y, \balpha, \bbeta, \gamma_N)&=\frac{\Gamma}{\gamma_N}\exp(-\bS^T\bbeta )(Y-\bS^T\balpha)^2,\;\;\mbox{where}\;\;\bS:=\varphi(\bX).\label{def:loss-2}
\end{align}
Let $s_{\bbeta}:=\|\bbeta^*\|_0$ and $s_{\balpha}:=\|\balpha^*\|_0$ be the corresponding \emph{sparsity} levels, and assume $s_{\bbeta},s_{\balpha}\geq1$ for the sake of simplicity.
Note that $\balpha^*$ and $\bbeta^*$ always exist, and also equal the corresponding `true' model parameters when the working models in \eqref{eq:models_para} are correctly specified for $m(\cdot)$ or $\gamma_N(\cdot)$. In the following, we require only that {\it either} $m^*(\cdot)$ or $\gamma_N^*(\cdot)$ provides a good approximation of $m(\cdot)$ or $\gamma_N(\cdot)$, but {\it not} necessarily both.

We define {\it targeted bias-reducing nuisance estimators} as follows: split the samples into $\K=2$ equal parts, $\mathcal I_1$ and $\mathcal I_2$, with $|\mathcal I_1|=|\mathcal I_2|=M=N/2$. For each $k\in\{1,2\}$, define
\begin{align}
\gammahat_N^{(k)}&:=~M^{-1}\sum_{i\in\mathcal I_k}\Gamma_i,\label{def:gammahat}\\
\bbetahat^{(k)}&:=~\arg\min_{\bbeta\in\R^d}M^{-1}\sum_{i\in\mathcal I_k} l (\Gamma_i, \bX_i, \bbeta, \gammahat_N^{(k)}) +\lambda_{\bbeta}\|\bbeta\|_1\;\;\mbox{s.t.}\;\;\max_{i\in\mathcal I_k}|\bS_i^\top\bbeta|<C,\label{def:betahat} ~\; \mbox{and}\\
\balphahat^{(k)}&:=~\arg\min_{\balpha\in\R^d}M^{-1}\sum_{i\in\mathcal I_k} h(\Gamma_i, \bX_i, Y_i, \balpha, \bbetahat^{(k)}, \gammahat_N^{(k)})
 +\lambda_{\balpha}\|\balpha\|_1,\label{def:alphahat}
\end{align}
where $\lambda_{\balpha}, \lambda_{\bbeta} \geq 0$ denote the respective tuning parameters for the $\ell_1$-regularizations {and $C>0$ is any large enough constant satisfying $C>C_0$ with $C_0$ defined in Assumption \ref{cond:bound-SP}. Define \(\psi_{N}(\bZ;\balpha,\bbeta)=\bS^T\balpha-\theta_1+\Gamma(Y-\bS^T\balpha)/g(\bS^T\bbeta+\log\gamma_N)\) as the DR score function.
The nuisance estimates \eqref{def:gammahat}-\eqref{def:alphahat} and the corresponding loss functions \eqref{def:loss-1}-\eqref{def:loss-2} are constructed to ensure that the associated Karush-Kuhn-Tucker (KKT) conditions guarantee that the Neyman orthogonal condition \citep{chernozhukov2018double} \(\E\{\nabla_{\balpha}\psi_{N}(\bZ;\balpha^*,\bbeta^*)\}=\bzero\) and \(\E\{\nabla_{\bbeta}\psi_{N}(\bZ;\balpha^*,\bbeta^*)\}=\bzero\) are satisfied, \emph{even when the models are misspecified}. This approach effectively reduces bias in ATE estimation, particularly under model misspecification. To ensure robust ATE inference, \cite{tan2020model,smucler2019unifying,avagyan2021high} explored similar nuisance estimates, though only in degenerate supervised settings. We adapt these to the decaying MAR context with our proposed bias-reducing estimators. Notably, the PS estimation in \eqref{def:betahat} can be interpreted as a stabilized high-dimensional logistic regression designed for extreme label imbalance.} We incorporate an apriori chosen estimate $\gammahat_N^{(k)}$ to counterbalance the impact of diminishing PSs. In \eqref{def:betahat}, we aggregate two sums: $\sum_{i\in\mathcal I_{-k}:\Gamma_i=0}\bS_i^T\bbeta$ and $\sum_{i\in\mathcal I_{-k}:\Gamma_i=1}\exp(-\bS_i^T\bbeta)/\gammahat_N^{(k)}$. The set where $\Gamma_i=0$ predominates over the set where $\Gamma_i=1$. To compensate for this imbalance, we amplify the latter group by a large factor of $1/\gammahat_N^{(k)}$, ensuring a balanced influence from both groups.

With a special asymmetric cross-fitting, we propose the \emph{bias-reduced SS counterfactual mean estimator} $\thetahat_{\mbox{\tiny 1,BRSS}}$ of $\theta_1$ as:
 $\thetahat_{\mbox{\tiny 1,BRSS}}:=(\thetahat_{\mbox{\tiny 1,BRSS}}^{(1)}+\thetahat_{\mbox{\tiny 1,BRSS}}^{(2)})/2$, where {for} $k\neq k'\in\{1,2\}$,
\begin{align}
\thetahat_{\mbox{\tiny 1, BRSS}}^{(k)}~:=~M^{-1}\sum_{i\in\mathcal I_k}\left\{\bS_i^T\balphahat^{(k')}+\frac{\Gamma_i(Y_i-\bS_i^T\balphahat^{(k')})}{\gammahat_N^{(k)}(\bX_i)}\right\},\label{eq:SS-Lasso-est1-def}
\end{align}
with $\gammahat_N^{(k)}(\bX_i)=g(\bS_i^T\bbetahat^{(k)}+\log\gammahat_N^{(k)})$.
Analogously, we propose a bias-reduced estimator $\thetahat_{\mbox{\tiny 0,BRSS}}$ of $\theta_0$, and the \emph{bias-reduced SS (BRSS) ATE estimator} as:
\begin{align}\label{def:BRSS-ATE}
\muhat_{\mbox{\tiny BRSS}} ~:=~ \thetahat_{\mbox{\tiny 1,BRSS}}-\thetahat_{\mbox{\tiny 0,BRSS}}.
\end{align}

\paragraph*{The asymmetric cross-fitting}

In the following, we introduce the rationale behind the asymmetric cross-fitting strategy focusing on the simplest case where
both the OR and PS models are correctly specified. For $i\in\mathcal I_1$, the PS function $\gamma_N(\bX_i)$ is estimated using a non-cross-fitted $\bbetahat^{(1)}\not\independent\bZ_i$, whereas the OR function $m(\bX_i)$ is estimated using a cross-fitted $\balphahat^{(2)}\independent\bZ_i$. W.l.o.g., we consider $\thetahat_{\mbox{\tiny 1,BRSS}}^{(1)}~=~M^{-1}\sum_{i\in\mathcal I_1}\psi_{N}(\bZ;\balpha^*,\bbeta^*)+\delta_1+\delta_2+\delta_3,$ where
\begin{align}
\delta_1~:=&~-M^{-1}\sum_{i\in\mathcal I_1}\left\{\frac{\Gamma_i}{\gammahat_N^{(1)}(\bX_i)}-\frac{\Gamma_i}{\gamma_N^*(\bX_i)}\right\}\bS_i^T(\balphahat^{(2)}-\balpha^*),\nonumber\\
\delta_2~:=&~M^{-1}\sum_{i\in\mathcal I_1}\left\{1-\frac{\Gamma_i}{\gamma_N^*(\bX_i)}\right\}\bS_i^T(\balphahat^{(2)}-\balpha^*),\nonumber \\
\delta_3~:=&~M^{-1}\sum_{i\in\mathcal I_1}\left\{\frac{\Gamma_i}{\gammahat_N^{(1)}(\bX_i)}-\frac{\Gamma_i}{\gamma_N^*(\bX_i)}\right\}(Y_i-\bS_i^T\balpha^*). \nonumber
\end{align}

{We consider two distinct approaches for controlling the effects of nuisance estimation, depending on the approximation error of the OR model. When the OR model provides a good approximation of the truth, $\delta_3$ can be treated as the average of conditionally independent terms (given all pairs $(\bX_i, \Gamma_i)$) with means approaching zero. Importantly, the PS estimator does not require cross-fitting, as its construction is independent of the outcomes. The remaining terms, $\delta_1$ and $\delta_2$, can be jointly controlled using the KKT conditions for the PS estimate. In contrast, when the OR model deviates from the truth, our asymmetric cross-fitting approach effectively controls $\delta_2$ by ensuring that $\balphahat^{(2)}$ is independent of the sample $(\bX_i, \Gamma_i)_{i \in \mathcal{I}_1}$. The remaining terms, $\delta_1$ and $\delta_3$, are further managed via the KKT condition for the OR estimate. The asymmetry arises because the PS model estimation only involves $(\bX_i, \Gamma_i)$, whereas OR model estimation requires the entire triple $(\bX_i, \Gamma_i, \Gamma_i Y_i)$. This technique imposes less stringent sparsity conditions for valid inference, as demonstrated in Theorems \ref{cor:para} and \ref{t4-SP}.}

Our proposed approach innovatively combines the strengths of existing strategies to enhance robustness for ATE inference. We draw on the methodology of \cite{farrell2015robust,tan2020model,avagyan2021high} for non-cross-fitted PS estimates and the cross-fitted OR estimate approach of \cite{chernozhukov2018double,smucler2019unifying}. As opposed to previous works, our method relaxes the ultra-sparsity requirements on both $\balpha^*$ and $\bbeta^*$ by introducing asymmetric cross-fitting. This novel combination provides superior robustness under degenerate supervised settings and requires weaker sparsity conditions.
{Although \cite{bradic2019sparsity} employs cross-fitting akin to ours, our method and analysis streamline implementation, expand applicability, and reveal additional theoretical insights. See Table \ref{table:sparsity} and Section \ref{sec:compare-S2} for detailed comparisons.}

\subsection{A de-coupled BRSS (DC-BRSS) estimator}\label{subsec:SP}

{In this section, we introduce a de-coupled approach that, unlike the R-DR and BRSS methods in Sections \ref{sec:gen-DR-DMAR} and \ref{subsec:nuisance}, treats the missingness patterns of \(R\) and \(T\) separately, rather than combining them as \(\Gamma = R \cdot T\). This distinction provides several advantages. First, by fully utilizing all \(N\) of the \((\bX_i, T_i)\) pairs, we achieve more accurate estimation of \(\pi(\bx) = \P(T=1\mid\bX=\bx)\), thereby improving overall precision. In contrast, relying on \(\Gamma_i = R_i T_i\) often leads to information loss, as many of the \(R_i\) values are zero. Second, when \(N\) is substantially larger than the labeled sample size \(n\), fully non-parametric methods (e.g., random forests or neural networks) can be employed to estimate the treatment PS \(\pi(\bx)\), enhancing flexibility while maintaining precision. In contrast, due to the limited effective sample size for estimating the labeling PS \(p_{N}(\bx) = \P(R=1 \mid T=1, \bX = \bx)\), we use linear-form nuisance estimates to simplify the model and propose ATE estimators designed to mitigate bias from potential model over-simplification and to reduce the impact of misspecification errors. Finally, by de-coupling the product propensity score \(\gamma_N(\bx) = \pi(\bx) p_{N}(\bx)\), we reduce model complexity by extracting out the effect of the treatment PS \(\pi(\bx)\). This approach improves both estimation and approximation accuracy, especially when \(\pi(\bx)\) and \(p_{N}(\bx)\) are sparse with highly non-overlapping active features; see Remark \ref{remark:MDR1}.}

The de-coupled approach posits the following working models,
\begin{equation*}
\mtil^*(\bx) = \varphi(\bx)^T\balphatil^*,\;\;\gammatil_N^*(\bx) = \pi^*(\bx)p_{N}^*(\bx),\;\;\mbox{and}\;\;p_{N}^*(\bx)=g(\varphi(\bx)^T\bbeta_{p}^*+\log p_{N}).
\end{equation*}
where $p_{N} := \P(R=1 \mid T=1)$, $\pi^*(\cdot)$ is a (potentially) non-parametric model for $\pi(\cdot)$, the treatment PS \eqref{eq:key-defs2}, and $p_{N}^*(\cdot)$ is an offset-based logistic model for $p_{N}(\cdot)$, the labeling PS \eqref{eq:key-defs1}. {We introduce a novel re-parametrization based on an inverse probability weighting (IPW) approach leveraging the treatment PS. By replacing the effective labeling indicator $\Gamma$ in the loss functions \eqref{def:loss-1}-\eqref{def:loss-2} with a re-weighted version $\Gamma/\pi^*(\bX)$, we define the new (target) nuisance parameters $\balphatil^*,\bbeta_{p}^*\in\R^d$ as:}
\begin{align*}
\bbeta_{p}^* &= \arg\min_{\bbeta\in\R^d}\E\left\{l\Bigl(\frac{\Gamma}{\pi^*(\bX)}, \bX, \bbeta, p_{N}\Bigl)\right\}, \quad \mbox{and} \\
\balphatil^* &= \arg\min_{\balpha\in\R^d}\E \biggl\{h \Bigl(\frac{\Gamma}{\pi^*(\bX)},\bX,Y, \balpha, \bbeta_{p}^*, p_{N}\Bigl)\biggl\}.
\end{align*}
We define $s_{p}:=\|\bbeta_{p}^*\|_0$ and $s_{\balphatil}:=\|\balphatil^*\|_0$ as the corresponding sparsity levels, and assume $s_{p},s_{\balphatil}\geq1$ for the sake of simplicity. {{We clarify here that} the subscripts in \(\bbeta_p^*\) and \(s_p\) indicate their association with the function \(p_N(\cdot)\), rather than the dimension of the covariate vector.}

\paragraph*{The {de-coupled} bias-reduced SS ATE estimator}
Consider the {same 2-fold} sample-splitting strategy {as the one} outlined in Section \ref{subsec:nuisance}. Let \(\pihat^{(1)}(\cdot)\) and \(\pihat^{(2)}(\cdot)\) be {generic} user-specified, possibly non-parametric, estimates of \(\pi(\cdot)\), constructed using the pairs \((T_i, \bX_i)\) from the sub-samples \(\mathcal{I}_1\) and \(\mathcal{I}_2\), respectively. For $k\in\{1,2\}$, let
\begin{align*}
\phat_{N}^{(k)} & =\sum_{i\in\mathcal I_k}\Gamma_i/\sum_{i\in\mathcal I_k}T_i, \\
\bbetahat_{p}^{(k)}& =~\arg\min_{\bbeta\in\R^d}M^{-1}\sum_{i\in\mathcal I_k} l\Bigl(\frac{\Gamma_i}{\pihat^{(k)}(\bX_i)}, \bX_i , \bbeta, \phat_{N}^{(k)}\Bigl) +\lambda_{\bbeta}\|\bbeta\|_1\;\;\mbox{s.t.}\;\;\max_{i\in\mathcal I_k}|\bS_i^\top\bbeta|<C, ~\; \mbox{and}\\
\balphatil^{(k)}&=~\arg\min_{\balpha\in\R^d}M^{-1}\sum_{i\in\mathcal I_k} h\Bigl(\frac{\Gamma_i}{\pihat^{(k)}(\bX_i)}, \bX_i , Y_i, \balpha, \bbetahat_{p}^{(k)}, \phat_{N}^{(k)}\Bigl)+\lambda_{\balpha}\|\balpha\|_1,
\end{align*}
where with a slight abuse of notations, $\lambda_{\balpha}, \lambda_{\bbeta} \geq 0$ denote the corresponding tuning parameters. Then, the \emph{de-coupled bias-reduced SS counterfactual mean estimator} $\thetahat_{\mbox{\tiny 1,DC-BRSS}}$ of $\theta_1$ is {defined as:} $\thetahat_{\mbox{\tiny 1,DC-BRSS}}:=(\thetahat_{\mbox{\tiny 1,DC-BRSS}}^{(1)}+\thetahat_{\mbox{\tiny 1,DC-BRSS}}^{(2)})/2$, where for any $k\neq k'\in\{1,2\}$,
\begin{align*}
\thetahat_{\mbox{\tiny 1,DC-BRSS}}^{(k)}~:=~M^{-1}\sum_{i\in\mathcal I_k}\left\{\bS_i^T\balphatil^{(k')}+\frac{\Gamma_i(Y_i-\bS_i^T\balphatil^{(k')})}{\phat_{N}^{(k)}(\bX_i)\pihat^{(k)}(\bX_i)}\right\},
\end{align*}
with $\phat_{N}^{(k)}(\bX_i)=g(\bS_i^T\bbetahat_{p}^{(k)}+\log\phat_{N}^{(k)})$. Analogously, we define $\thetahat_{\mbox{\tiny 0,DC-BRSS}}$ and the ATE estimator $\muhat_{\mbox{\tiny DC-BRSS}}:=\thetahat_{\mbox{\tiny 1,DC-BRSS}}-\thetahat_{\mbox{\tiny 0,DC-BRSS}}$.
The variance estimator is also considered
\begin{align}
\Sigmahat_{\mbox{\tiny 1,DC-BRSS}}~:=~N^{-1}\sum_{k=1}^2\sum_{i\in\mathcal I_k}\left\{\bS_i^T\balphatil^{(k')}+\frac{\Gamma_i(Y_i-\bS_i^T\balphatil^{(k')})}{\phat_{N}^{(k)}(\bX_i)\pihat^{(k)}(\bX_i)}-\thetahat_{\mbox{\tiny 1,DC-BRSS}}\right\}^2.
\end{align}

\section{Theoretical results}\label{sec:theory}

{In Section \ref{subsec:result_BRSS}, we establish the theoretical properties of the proposed BRSS (Section \ref{subsec:nuisance}) and DC-BRSS (Section \ref{subsec:SP}) estimators, followed by the theoretical guarantees for the nuisance estimates used in their construction, in Section \ref{subsec:PS}.}

\subsection{Asymptotic properties of the bias-reduced SS estimators}\label{subsec:result_BRSS}

We assume the following conditions to establish robust inference for the ATE.

\begin{assumption}[Sub-Gaussian covariates]\label{cond:subG}
{Assume that $\bS=\varphi(\bX)$ is a sub-Gaussian random vector with
$\|\bS^T\bv\|_{\psi_2}\leq\sigma\|\bv\|_2$ for all $\bv\in\R^d $ with some constant $\sigma>0$. Additionally, for some constant $\kappa_l>0$, assume that $\inf_{\|\bv\|_2=1}\Var(\bS^T\bv\mid\Gamma=1)\geq\kappa_l$ and $\inf_{\|\bv\|_2=1}\Var(\bS^T\bv)\geq\kappa_l$.}
\end{assumption}

\begin{assumption}[Moment condition on product PS]\label{cond:tail}
There exists some $q>1$, such that $\E\{\gamma_N^q(\bX)\}\leq C\gamma_N^q,$  for some constant $C>0$.
\end{assumption}

\begin{assumption}[Non-parametric estimation of treatment PS]\label{cond:NP-est}
With some constant $c>0$, let: (a) $c_0<\pi^*(\bx)<1-c_0$ for all $\bx\in\mathcal X$ and (b) the events $\mathcal E_\pi:=\{c_0<\pihat^{(k)}(\bx)<1-c_0,\forall k\in\{1,2\},\bx\in\mathcal X\}$ and $\mathcal E_\zeta:=\{(M\gamma_N)^{-1}\sum_{i\in\mathcal I_k}\Gamma_i\{\pihat^{(k)}(\bX_i)-\pi^*(\bX_i)\}^2\leq\zeta_N^2,\forall k\in\{1,2\}\}$ occur with probability approaching one as $N\to\infty$ with some $\zeta_N=o(1/\sqrt{\log N})$.
\end{assumption}

\begin{assumption}[Decaying overlap of labeling PS]\label{cond:bound-SP}
{Assume \(|\bS^\top \bbeta_{p}^*| < C_0\) almost surely for some constant \(C_0 > 0\) independent of \(N\), while allowing for \(p_{N} \to 0\) as \(N \to \infty\).}
\end{assumption}

\begin{assumption}[Conditional sub-Gaussian noise]\label{cond:subG'}
Conditional on $\bX$, let the noise $Y(1)-m(\bX)$ be a sub-Gaussian random variable with a constant $\psi_2$-norm $\sigma_\varepsilon>0$.
\end{assumption}

\begin{assumption}[Approximation error]\label{cond:approx}
{Let $|\gamma_N(\bX)/\gammatil_N^*(\bX)-1|\leq e_\gamma$ and $|m(\bX)-m^*(\bX)|\leq e_m$ almost surely with some sequences $e_\gamma,e_m=O(1)$.}
\end{assumption}

\begin{assumption}[Lower bound]\label{cond:lowerbound}
Let $\Var\{Y(1)-m(\bX)\mid\Gamma=1\}\geq c_{\min}$.
\end{assumption}

If Assumptions \ref{cond:basic}-\ref{cond:subG} hold and \( p_{N}^*(\cdot) = p_{N}(\cdot) \), then Assumption \ref{cond:tail} is satisfied for any constant \( q > 1 \), although it is only required for some \( q > 1 \). Assumption \ref{cond:NP-est} places constraints on the non-parametric estimate \( \pihat(\cdot) \) and is specifically required for the DC-BRSS estimator. Conditions similar to Assumption \ref{cond:NP-est} (a) are found in \cite{tan2020model} and \cite{ning2020robust} as `overlap conditions' for the treatment PS, which account for model misspecification adaptively. Furthermore, Assumption \ref{cond:NP-est} (b) ensures the convergence of the PS estimate toward a (potentially misspecified) target \( \pi^*(\cdot) \) by controlling the in-sample mean squared error. {Assumption \ref{cond:bound-SP} generalizes the overlap condition to scenarios with decaying MAR by introducing a diverging offset term \( \log p_{N} \), which captures the decaying propensity score while assuming the remaining signal is bounded. In non-decaying setups with $p_{N}\asymp1$, this assumption is equivalent to the standard overlap condition \( c_0 < p_{N}^*(\bX)=\{1+p_{N}^{-1}\exp(-\bS^\top\bbeta_{p}^*)\}^{-1} < 1 - c_0 \).} Assumptions \ref{cond:subG'} and \ref{cond:lowerbound} serve as regularity conditions for the noise in the OR model. {Additionally, we introduce the sequences \( e_m \) and \( e_\gamma \) to regulate the approximation errors for the OR and PS models, respectively, as specified in Assumption \ref{cond:approx}.}

As we allow $\pi(\cdot)\neq\pi^*(\cdot)$, we can set $\pihat^{(k)}(\cdot)\equiv\pi^*(\cdot)\equiv1$, and $\thetahat_{\mbox{\tiny 1,BRSS}}$ can be viewed as a special case of the de-coupled version $\thetahat_{\mbox{\tiny 1,DC-BRSS}}$, apart from a minor difference between offset terms $\log\gammahat_N^{(k)}$ and $\log\phat_{N}^{(k)}$. We first introduce the theoretical properties for the BRSS estimator under the special case $\pi^*(\cdot)\equiv1$, $\balphatil^*=\balpha^*$, and $\bbeta_{p}^*=\bbeta^*$.

\begin{theorem}[BRSS]\label{cor:para}
Let Assumptions \ref{cond:basic}-\ref{cond:tail} and \ref{cond:bound-SP}-\ref{cond:lowerbound} hold with $\balphatil^*=\balpha^*$ and $\bbeta_{p}^*=\bbeta^*$. Let $N\gamma_N\gg\max(\log^{5/2} N,s_{\bbeta},s_{\balpha})\log d\log^{1/2}N$, $e_me_\gamma=o(N^{-1/2}(\|\balpha^*\|_2+\gamma_N^{-1/2}))$ {(where recall $\gamma_N$ from \eqref{eq:key-defs3} and the sparsities $(s_{\bbeta}, s_{\balpha})$ from Section \ref{subsec:nuisance}),} and
either one of the following hold:
\begin{align*}
&\text{(a)}\;s_{\balpha}\sqrt{s_{\bbeta}} \, = \, o(N\gamma_N/\log^{3/2}{d}),\;s_{\bbeta} \, = \, o(\sqrt{N\gamma_N}/\log{d}); \quad\mbox{or}\\
&\text{(b)}\;s_{\balpha} \, = \, o(\sqrt{N\gamma_N}/\log{d}),\;e_m^2s_{\bbeta} \, = \, o(1/\log{d}).
\end{align*}
{Then, with} $\psi_{N}^*(\bZ):=m^*(\bX)-\theta_1+\Gamma\{Y-m(\bX)\}/\gamma_N^*(\bX)$, {we have: as $N,d\to\infty$,} $\Sigma_{N}^*:=\Var\{\psi_{N}^*(\bZ)\}\asymp\|\balpha^*\|_2^2+\gamma_N^{-1}$ {with $\gamma_N\asymp a_N$ as in \eqref{eq:a_N-defn},} and
\begin{align}\label{normal:BRSS}
\sqrt N\left(\Sigma_{N}^*\right)^{-1/2}\left(\thetahat_{\mbox{\tiny 1,BRSS}}-\theta_1\right) \, \xrightarrow{d} \, \mathcal N(0,1).
\end{align}
\end{theorem}
{In contrast to the analysis provided in Theorem \ref{thm:gen-rate-DR}, Theorem \ref{cor:para} allows for a growing signal level \( \|\balpha^*\|_2 \) in the OR model while keeping the noise level constant for ease of presentation. Therefore, unlike the results presented in \cite{zhang2023double}, we show that the asymptotic variance \( \Sigma_{N}^* \asymp \|\balpha^*\|_2^2 + \gamma_N^{-1} \) depends not only on \( \gamma_N^{-1} \asymp a_N^{-1} \), but also on the signal level \( \|\balpha^*\|_2^2 \). Specifically, when the linear signal is strong, with \( \|\balpha^*\|_2^2 \gg \gamma_N^{-1} \), the convergence rate scales as \( \|\balpha^*\|_2/\sqrt{N} \); otherwise, it scales as \( 1/\sqrt{N\gamma_N} \).}

\begin{remark}[Robust inference]\label{remark:MDR}
{Below, we discuss the required conditions under the scenario where \( \|\balpha^*\|_2 \asymp 1 \) for simplicity. As demonstrated in Theorem \ref{cor:para}, the consistent and asymptotically normal (CAN) result in \eqref{normal:BRSS} requires the product of approximation errors to satisfy \( e_m e_\gamma = o((N\gamma_N)^{-1/2}) \). This condition holds under the following scenarios:

\par\smallskip
(a) \emph{Both nuisance models are correctly specified}. When $e_m = e_\gamma = 0$, the asymptotic normality result in \eqref{normal:BRSS} parallels that of the R-DR estimator in \eqref{result:normal1}, attaining the semi-parametric efficiency bound. However, the sparsity conditions required in Theorem \ref{cor:para} are weaker than those in Theorem \ref{thm:gen-rate-DR}. The asymptotic normality of the R-DR estimator relies on the product-rate condition \ref{cond:product-rate}, which translates to $s_{\balpha}s_{\bbeta} = o(N\gamma_N/\log^2d)$ when nuisance functions are estimated via standard $\ell_1$-regularized maximum likelihood estimators. In contrast, the BRSS estimator exhibits the {\it\underline{sparsity DR}} property -- it remains CAN as long as either (a) \(s_{\balpha}\sqrt{s_{\bbeta}}=o(N\gamma_N/\log^{3/2}{d})\) and \(s_{\bbeta}=o(\sqrt{N\gamma_N}/\log{d})\), or (b) \(s_{\balpha}=o(\sqrt{N\gamma_N}/\log{d})\). These conditions are strictly weaker than the conventional product-rate condition. A weaker form of this property was introduced by \cite{bradic2019sparsity}, where CAN results were established under more restrictive sparsity conditions in standard supervised settings. Detailed comparisons are provided in Table \ref{figure:sparsity1}.

\par\smallskip
(b) \emph{One of the nuisance models is correctly specified while the other may be misspecified}. In this scenario, the proposed estimator \( \thetahat_{\mbox{\tiny 1,BRSS}} \) remains CAN as long as either \( e_m = 0 \) or \( e_\gamma = 0 \). This result ensures that \( \thetahat_{\mbox{\tiny 1,BRSS}} \) exhibits the \underline{\textit{model DR}} property \citep{tan2020model, smucler2019unifying, dukes2020doubly}, achieving CAN results as long as at least one nuisance model is correctly specified. Consequently, the BRSS estimator offers more robust inference compared to the R-DR estimator introduced in Section \ref{sec:gen-DR-DMAR}, which requires both \( e_m = 0 \) and \( e_\gamma = 0 \).

\par\smallskip
(c) \emph{Both nuisance models are misspecified, but with decaying approximation errors}. We accommodate situations where both nuisance functions are misspecified, provided their approximation errors decay sufficiently fast as \( N, d \to \infty \). For example, if the true OR and PS functions are additive with Lipschitz components and sparsity levels \( s_m \) and \( s_\gamma \), respectively, additive piecewise-linear splines with \( K \) knots yield approximation errors \( e_m \asymp s_m/K \) and \( e_\gamma \asymp s_\gamma/K \). Furthermore, we have \( d = Kp \), \( s_{\balpha} \leq K s_m \) and \( s_{\bbeta} \leq K s_\gamma \). Setting \( K \asymp (s_m s_\gamma)^{1/2}(N \gamma_N)^{1/4}\log N \) ensures that all the required conditions on approximation errors and sparsity levels hold, guaranteeing the asymptotic normality \eqref{normal:BRSS}, as long as either \( N \gamma_N \gg s_m^6 s_\gamma^2 \) or \( N \gamma_N \gg s_m^2 s_\gamma^6 + s_m^{2.8} s_\gamma^2 \) (up to logarithmic terms).
Besides, we also allow for situations where one nuisance estimate experiences a non-decaying approximation error, as long as the other approximation error decays sufficiently fast with an \( o((N\gamma_N)^{-1/2}) \) rate.}
\end{remark}

In the following {result}, we {establish} the asymptotic properties of the DC-BRSS estimator.

\begin{theorem}[DC-BRSS]\label{t4-SP}
Let Assumptions \ref{cond:basic}-\ref{cond:approx} hold, {and recall $\gamma_N$ from \eqref{eq:key-defs3} and the sparsities $(s_{p},s_{\balphatil})$ from Section \ref{subsec:SP}.} Suppose that $N\gamma_N\gg\max(\log^{5/2} N,s_{p},s_{\balphatil})\log d\log^{1/2}N$. Let either one of the following hold:
\begin{align*}
&\text{(a)}\;s_{\balphatil}\sqrt{s_{p}}=o(N\gamma_N/\log^{3/2}{d}),\;s_{p}=o(\sqrt{N\gamma_N}/\log{d}),\;\zeta_N=o((N\gamma_N)^{-1/2}); \;~\mbox{or}\\
&\text{(b)}\;s_{\balphatil}=o(\sqrt{N\gamma_N}/\log{d}),\;e_m^2s_{p}=o(1/\log{d}),\;\;e_m\zeta_N=o((N\gamma_N)^{-1/2}).
\end{align*}
Then, with some $\lambda_{\bbeta}\asymp\lambda_{\balpha}\asymp\sqrt{\log{d}/(N\gamma_N)}$, as $N,d\to\infty$,
\begin{align*}
\thetahat_{\mbox{\tiny 1,DC-BRSS}}-\theta_1~=~N^{-1}\sum_{i\in\mathcal I}\psitil_{N}^*(\bZ_i)+o_p\left((N\gamma_N)^{-1/2}\right),
\end{align*}
where $\psitil_{N}^*(\bZ):=\mtil^*(\bX)-\theta_1+\Gamma\{Y-\mtil^*(\bX)\}/\gammatil_N^*(\bX)$ with $\E\{\psitil_{N}^*(\bZ)\}=O(e_me_\gamma)$ and
$N^{-1}\sum_{i\in\mathcal I}\psitil_{N}^*(\bZ_i)=O_p(e_me_\gamma+N^{-1/2}(\|\balphatil^*\|_2+\gamma_N^{-1/2})).$
If further, Assumption \ref{cond:lowerbound} holds and $e_me_\gamma=o(N^{-1/2}(\|\balphatil^*\|_2+\gamma_N^{-1/2}))$, then $\Sigmatil_{N}^*:=\Var\{\psitil_{N}^*(\bZ)\}\asymp\|\balphatil^*\|_2^2+\gamma_N^{-1}$ with $\gamma_N\asymp a_N$ {(from \eqref{eq:a_N-defn})}, and as $N,d\to \infty$,
\begin{align}
&\sqrt N\left(\Sigmatil_{N}^*\right)^{-1/2}\left(\thetahat_{\mbox{\tiny 1,DC-BRSS}}-\theta_1\right)\xrightarrow{d}\mathcal N(0,1),\label{result:BRSS2-SP1}\\
&\Sigmahat_{\mbox{\tiny 1,DC-BRSS}}=\Sigmatil_{N}^*\{1+o_p(1)\},\;\text{and}\;\sqrt N\left(\Sigmahat_{\mbox{\tiny 1,DC-BRSS}}\right)^{-1/2}\left(\thetahat_{\mbox{\tiny 1,DC-BRSS}}-\theta_1\right)\xrightarrow{d}\mathcal N(0,1).\label{result:BRSS2-SP}
\end{align}
\end{theorem}

\vspace{-0.15in} 
\begin{remark}[Theoretical challenges in ATE estimation]\label{remark:challenge-ATE}
{The asymmetric cross-fitting and the aim of enhanced robustness create unique challenges, setting our analysis apart from existing works such as \cite{zhang2023double,tan2020model,chernozhukov2018double}, which rely either on standard cross-fitting or no cross-fitting and establish valid inference under stronger assumptions. Since the PS estimates are constructed without cross-fitting, controlling their impact on the final estimator $\thetahat_{\mbox{\tiny 1,DC-BRSS}}$ requires additional effort. Meanwhile, the cross-fitted OR estimates also introduce further challenges in applying the corresponding KKT conditions to efficiently manage out-of-sample errors using in-sample data. To achieve sufficiently fast convergence rates, we first established key empirical process results in Lemmas \ref{lemma:convex}-\ref{lemma:emp}, which are crucial for addressing the above challenges and provide adaptive rates that account for the decaying PS. With more advanced theoretical analysis (see the error controls in Lemmas \ref{lemma:Delta1}-\ref{lemma:Delta6}), we provide valid inference under weaker sparsity conditions than those in \cite{bradic2019sparsity}, even in supervised settings. Furthermore, our derived rates are adaptive to the decaying PS, initial treatment PS estimation error $\zeta_N$, and approximation errors $e_m$ and $e_\gamma$, presenting additional analytical challenges. These adaptive rates expand the applicability of our results.}
\end{remark}

\begin{remark}[The treatment PS estimation]\label{remark:pi}
{The treatment PS estimate \( \pihat(\cdot) \) can be either parametric or non-parametric, as long as it satisfies the high-level conditions in Theorem \ref{t4-SP}. For example, when the number of covariates grows as \( p = O(N^c) \) for some constant \( c > 0 \), \cite{chi2022asymptotic} demonstrated that random forest estimates achieve an error rate of \( \zeta_N = O(N^{-1/(11.1\alpha_1 \lor 16.1)}) \), where \( \alpha_1 > 0 \) is the sufficient impurity decrease parameter (see Condition 1 therein). As a result, Conditions (a) and (b) of Theorem \ref{t4-SP} hold when \( \gamma_N = o(N^{2/(11.1\alpha_1 \lor 16.1) - 1}) \) and \( e_m^2 \gamma_N = o(N^{2/(11.1\alpha_1 \lor 16.1) - 1}) \), respectively. Neural networks can also be employed to estimate the treatment PS. For instance, \cite{bhattacharya2024deep} established a convergence rate of \( \zeta_N = O(\{p^{0.5}N^{-\beta_1/(4\beta_1 + 2)}+pN^{-\beta_2/(2\beta_2+2)}\}\log^{2.25}N) \) for two-way interaction models with \( p = o(N) \), where \(\beta_1\) and \(\beta_2\) denote the smoothness indices of the univariate and bivariate components, respectively. In this case, Conditions (a) and (b) hold when \( \gamma_N = o(\omega_N) \) and \( e_m^2 \gamma_N = o(\omega_N) \), respectively, where \( \omega_N = \{pN^{(\beta_1+1)/(2\beta_1 + 1)}+p^2N^{1/(\beta_2 + 1)}\}^{-1}\log^{-4.5}N \).

The PS estimate can also be constructed after applying initial dimension reduction techniques. Notably, we allow for \( \pi(\cdot) \neq \pi^*(\cdot) \), accommodating scenarios where variable selection consistency does not hold. However, if variable selection consistency holds and \( \pi(\cdot) \) belongs to a smooth function class, we can further ensure that \( \pi(\cdot) = \pi^*(\cdot) \), where accurate treatment PS estimation strengthens the robustness of ATE inference; see Remark \ref{remark:MDR1}.}
\end{remark}

\begin{remark}[Enriched robustness compared to the BRSS estimator]\label{remark:MDR1}
{Similar to Theorem \ref{cor:para}, the asymptotic normality results \eqref{result:BRSS2-SP1}-\eqref{result:BRSS2-SP} hold as long as the product of approximation errors satisfies $e_me_\gamma=o(N^{-1/2}(\|\balphatil^*\|_2+\gamma_N^{-1/2}))$. Importantly, the approximation error \( e_\gamma \) here depends on the distance between the ratio \( \gamma_N(\cdot)/\pi^*(\cdot) \) and its `best logistic approximation', \( \gammatil_N^*(\cdot) \). Specifically, \( e_\gamma = 0 \) occurs if and only if \( \gamma_N(\cdot)/\pi^*(\cdot) \) is logistic, i.e., there exists a \( \bbeta_{p}^0 \) such that \( \gamma_N(\bx)/\pi^*(\bx) = g(\varphi(\bx)^T\bbeta_{p}^0 + \log p_{N}) \).

Compared to the BRSS estimator, incorporating accurate treatment PS estimates significantly enhances robustness in inference. In many practical applications, {it is natural to expect that} only a limited number of confounders simultaneously affect both the treatment and labeling variables \( T \) and \( R \), while most confounders influence only one of the two. When \( \pi^*(\cdot) = \pi(\cdot) \), the complexity of the ratio \( \gamma_N(\cdot)/\pi^*(\cdot) = p_{N}(\cdot) \) depends solely on the confounders impacting \( R \), rather than those affecting both \( T \) and \( R \). Consequently, for any feature map \( \varphi(\cdot) \), the sparsity level \( s_{p} \) of the labeling PS tends to be much smaller than that of the product PS \( s_{\bbeta} \). In fact, as long as the treatment PS model \(\pi^*(\cdot)\) captures most of the confounders affecting \( T \), the complexity of estimating \( \gamma_N(\cdot)/\pi^*(\cdot) \) is reduced, even if \(\pi^*(\cdot)\) deviates from the true \(\pi(\cdot)\). This reduction in complexity makes it easier to meet the sparsity conditions in Theorem \ref{t4-SP} than those in Theorem \ref{cor:para}. Moreover, since the complexity of \( p_{N}(\cdot) \) is lower than that of the product \( \gamma_N(\cdot) \), achieving a reasonable approximation error also becomes more feasible as the feature map's dimension \( d \) increases.}
\end{remark}

\subsection{Theoretical properties of the targeted bias-reducing nuisance estimators}\label{subsec:PS}

In the following, we present the theoretical properties of the nuisance estimates {$\bbetahat_{p}^{(k)}$ and $\balphatil^{(k)}$} developed for the DC-BRSS estimator {in Section \ref{subsec:SP}}. The corresponding results for the {nuisance estimates $\bbetahat^{(k)}$ and $\balphahat^{(k)}$ involved in constructing the} BRSS estimator {in Section \ref{subsec:nuisance}} follow similarly, provided that we set \(\pihat^{(k)}(\cdot) \equiv 1\) {in the definitions of $\bbetahat_{p}^{(k)}$ and $\balphatil^{(k)}$ and then apply the results below}. {Theorems \ref{thm:PS-SP}-\ref{thm:OR-SP-mis} discuss nuisance estimation adaptive to the non-parametric initial step, high dimensionality, and the decaying PS setting, making them of independent interest given their applicability and the underlying challenges.}

\begin{theorem}[PS estimator]\label{thm:PS-SP}
Let Assumptions \ref{cond:basic}-\ref{cond:bound-SP} hold and suppose that $N\gamma_N\gg\max(\log N,s_{p})\log d$ with some $\lambda_{\bbeta}\asymp\sqrt{\log d/(N\gamma_N)}$. Then, as $N,d\to\infty$,
\begin{align*}
\|\bbetahat_{p}^{(k)}-\bbeta_{p}^*\|_2&=O_p\left(\sqrt\frac{s_{p}\log d}{N\gamma_N}+\zeta_N \right), \quad \mbox{and}\\
\|\bbetahat_{p}^{(k)}-\bbeta_{p}^*\|_1&=O_p\left(s_{p}\sqrt\frac{\log d}{N\gamma_N}+\zeta_N^2\sqrt\frac{N\gamma_N}{\log d}\right).
\end{align*}
\end{theorem}

\begin{remark}[Theoretical challenges in nuisance estimation]\label{remark:challenge-nuisance}
{Our de-coupled approach estimates \( \bbetahat_{p} \) using an initial non-parametric estimation \( \pihat(\cdot) \), introducing additional challenges for nuisance estimation analysis. This sets our method apart from existing works such as \cite{zhang2023double,tan2020model,bradic2019sparsity}, where the product PS model does not arise, and direct PS estimation is sufficient.

The dependence on \( \pihat(\cdot) \) shifts the error \( \bDelta_{\bbeta} = \bbetahat_{p} - \bbeta_{p}^* \) away from the conventional sparse cone set, \(\mathcal C_{k_0}(S) := \{\bDelta\in\R^d: \|\bDelta_{S^c}\|_1\leq k_0\|\bDelta_S\|_1\}\) \citep{wainwright2019high}, where \( S \) is the support set of \( \bbeta_{p}^* \). As shown in \eqref{eq:basic-PS} of the \hyperref[supp_mat]{Supplement}, unlike standard techniques for \(\ell_1\)-regularized estimators, our analysis requires controlling an additional error term \( \br_\pi^T\bDelta_{\bbeta} \), where \( \br_\pi \) (defined in \eqref{def:rpi}) captures the estimation error from \( \pihat(\cdot) \). Since this initial estimation error influences all coefficients \emph{simultaneously}, recovering the true support set is challenging even with \(\ell_1\)-regularization, preventing us from establishing \( \bDelta_{\bbeta} \in \mathcal C_{k_0}(S) \) unless the initial estimation error \( \zeta_N \) (in Assumption \ref{cond:NP-est}) is sufficiently small (e.g., \(\zeta_N=o((N\gamma_N)^{-1/2})\)).

To address this, we first control \( \br_\pi^T\bDelta_{\bbeta} \) using both \( \|\bDelta_{\bbeta}\|_1 \) and \( \|\bDelta_{\bbeta}\|_2 \); see Lemma \ref{lemma:rpi}. We then establish that \( \bDelta_{\bbeta} \) belongs to a cone set incorporating both \( \ell_1 \)- and \( \ell_2 \)-norm constraints: $\widetilde{\mathcal C}(r_N) := \{\bDelta\in\R^d: \|\bDelta\|_1\leq r_N\|\bDelta\|_2\}$, where \( r_N \asymp \sqrt{s_{p}} + \zeta_N\sqrt{N\gamma_N/\log d} \). This cone set is \emph{independent of the true support set \( S \)} and \emph{symmetric across all coordinates}, allowing it to accommodate the effects of initial estimation error. Finally, we establish convergence results based on this newly defined cone set; see Lemmas \ref{lemma:RSC-SP}-\ref{lemma:Fdelta} and the proof of Theorem \ref{thm:PS-SP}. Notably, Theorem \ref{thm:PS-SP} requires only \( \zeta_N = o(1/\sqrt{\log N}) \) for labeling PS estimation. This condition remains sufficient for the final ATE estimation when Condition (b) of Theorem \ref{t4-SP} holds and \( e_m = 0 \). More broadly, these techniques may also benefit other estimators that deviate from standard cone set conditions. Besides, achieving an `effective sample size' measure of \( N\gamma_N \) crucially depends on ensuring adaptability to a decaying \( \gamma_N \), which introduces additional challenges as discussed in Remark \ref{remark:RSC}.}
\end{remark}

\begin{theorem}[OR estimator]\label{thm:OR-SP-mis}
Let Assumptions \ref{cond:basic}-\ref{cond:approx} hold and suppose that $N\gamma_N\gg\max(\log N,s_{p},s_{\balphatil})\log d$ with some $\lambda_{\bbeta}\asymp\lambda_{\balpha}\asymp\sqrt{\log d/(N\gamma_N)}$. Then, as $N,d\to\infty$,
\begin{align*}
\|\balphatil^{(k)}-\balphatil^*\|_2&=O_p\left(\sqrt\frac{(s_{\balphatil}+e_m^2s_{p})\log d}{N\gamma_N}+e_m\zeta_N\right), \quad \mbox{and}\\
\|\balphatil^{(k)}-\balphatil^*\|_1&=O_p\left((s_{\balphatil}+e_m^2s_{p})\sqrt\frac{\log d}{N\gamma_N}+e_m^2\zeta_N^2\sqrt\frac{N\gamma_N}{\log d}\right).
\end{align*}
\end{theorem}

{While existing works have analyzed the properties of OR estimates constructed from initial PS estimates in supervised settings, our analysis provides additional insights into the effect of model approximation errors in the more challenging SS settings. For \( R \equiv 1 \) and \( \zeta_N = 0 \), \cite{tan2020model,ning2020robust} demonstrated that the OR estimation's consistency rates are influenced by \( s_{\balphatil} + s_{\bbeta} \), as in the above with \( e_m \asymp 1\). In contrast, \cite{bradic2019sparsity} found that the rates depend only on \( s_{\balphatil} \), assuming accurate PS and OR models, i.e., \( e_m = e_\gamma = 0 \). Our analysis carefully tracks the impact of the possibly decaying (but nonzero) approximation errors \( e_m \) and \( e_\gamma \) in nuisance estimation, presenting additional challenges beyond existing literature. Our results contribute by establishing consistency rates that adapt to the OR model's approximation error, showing dependence on \( s_{\balphatil} + e_m^2 s_{p} \). The established rates in Theorem \ref{thm:OR-SP-mis} are \emph{strictly faster} than those of \cite{tan2020model,smucler2019unifying} when \( e_m = o(1) \) and \( s_p \gg s_{\balphatil} \), and are the same as theirs otherwise, as both analyses are restricted to \( e_m = O(1) \). Notably, our analysis reveals that a more accurate linear approximation reduces the dependence of OR estimation rates on the complexity of the PS model. Moreover, we find that the consistency rates remain independent of the PS model's approximation error \( e_\gamma \).} On another front, we have integrated an innovative nuisance estimation step centered around \( \pi(\cdot) \) and achieved rates independent of the correctness of such estimates. Similarly, the dependence of the consistency rates on \( \zeta_N \) is influenced by the OR model's approximation error. When the OR model is correctly specified, the consistency rates no longer rely on \( \zeta_N \).

\section{Comparative Analysis}\label{sec:compare-S2}
{Our proposed framework unifies multiple strands of research. Below, we provide detailed comparisons with existing literature, highlighting specific instances that can be viewed as special cases of our framework.}

\paragraph*{Regular SS ATE estimation, a special case under MCAR}
\cite{cheng2021robust,hou2021efficient,zhang2022high,chakrabortty2022general} developed SS ATE estimators that improve efficiency compared to fully supervised methods. However, their approaches rely on the restrictive MCAR assumption, \(R \independent (Y, T, \mathbf{X})\), and introduces bias when labeling bias or data heterogeneity is present. Our methods retain the same efficiency under MCAR but are applicable to a wider range of scenarios, providing robust guarantees even in significantly more complex scenarios.

\paragraph*{Regular MAR in the causal context, a special case when \(\gamma_N \asymp 1\)}
\cite{kallus2024role} introduced a SS ATE estimator addressing labeling bias under MAR conditions. {While they provided insightful results on the role of surrogates in SS settings, their proposed estimator only achieved \emph{rate DR}, requiring both OR and PS models to be correctly specified for valid inference.} Restricting to low-dimensional parametric models, \cite{wei2022doubly} proposed a \emph{model DR} estimator, requiring either (a) the OR function to be correctly specified, or (b) both the labeling and treatment PS functions to be correctly specified. In contrast, our BRSS estimators achieve \emph{model DR} in the more challenging high-dimensional settings, where nuisance estimates typically converge with non-root-\(n\) rates. Moreover, our de-coupled approach allows fully non-parametric estimation of \(\pi(\cdot)\). Unlike their condition (b), we only require the fraction \(\gamma_N(\cdot) / \pi^*(\cdot)\) to be correctly specified, and \(\pi^*(\cdot) \neq \pi(\cdot)\) is permissible even when the OR model is misspecified; see the discussion in Remark \ref{remark:MDR}. {Furthermore, both \cite{kallus2024role} (in their earlier versions; see Section \ref{sec:setup}) and \cite{wei2022doubly} only provided rigorous analysis under a positivity condition, which restricts applicability in typical SS scenarios where \(N \gg n\).}

\paragraph*{Decaying MAR in the non-causal context, a special case when $T_i \equiv 1$}
{To bridge the gap between the SS and missing data literature, \cite{zhang2023double} recently explored a variation of the decaying MAR setting, albeit in a non-causal context, focusing on confounding related to the labeling variable \( R \). In contrast, our framework addresses confounders that affect both the labeling variable \( R \) and the treatment \( T \), posing a more complex challenge. While \cite{zhang2023double} established \emph{rate DR} results, our BRSS estimators are further \emph{model DR} and \emph{sparsity DR}, imposing weaker conditions on both model specifications and sparsity structures. Unlike the non-causal context, we address the estimation of \emph{product} PS models \eqref{eq:key-defs3}, where the efficient use of asymmetric \(T_i\) and \(R_i\) is the key challenge within the causal decaying MAR framework. This makes the proposed de-coupled approach unique in the considered setup.}

\begin{table}[t]
\caption{
Comparison of the sparsity conditions required (up to logarithmic terms) for the asymptotic normality of the DR estimators with $R\equiv1$. For brevity and without prejudice, we've abbreviated 'Avagyan' as 'Av.', 'Bradic' as 'Br.', 'Chakrabortty' as 'Chakr.', 'Chernozhukov' as 'Chern.', 'Kallus' as 'Kal.', and 'Vansteelandt' as 'Vdt.'.
}\label{table:sparsity}
\begin{tabular}{m{9em}|| m{9em}| m{9em} | m{9em} }
\toprule
Literature&Correct PS and OR& Misspecified PS & Misspecified OR\\
\hline
Farell (2015) &$N\gg s_{\balpha}^2+s_{\bbeta}^2$&$\xmark$&$\xmark$\\
\hline
Chern. et al.(2018) & &\\
Kal. and Mao (2020) & &\\
Hou et al. (2021)&\multirow{2}{*}{$N\gg s_{\balpha}s_{\bbeta}$}&\multirow{2}{*}{$\xmark$}&\multirow{2}{*}{$\xmark$}\\
Chakr. et al. (2022)& &\\
Zhang and Br. (2022)& &&\\
Zhang et al. (2023)& &&\\
\hline
Tan (2020)& &\\
Ning et al. (2020)&$N\gg s_{\balpha}^2+s_{\bbeta}^2$&$N\gg s_{\balpha}^2+s_{\bbeta}^2$&$N\gg s_{\balpha}^2+s_{\bbeta}^2$
\\
Av. and Vdt. (2021)
&&&\\
\hline
Dukes et al.(2020)&\multirow{2}{*}{$N\gg s_{\balpha}s_{\bbeta}+s_{\bbeta}^2$}&\multirow{2}{*}{$N\gg s_{\balpha}^2+s_{\bbeta}^2$}&\multirow{2}{*}{$N\gg s_{\balpha}^2+s_{\bbeta}^2$}
\\
Dukes and Vdt. (2021)
&&&\\
\hline
Smucler et al. (2019)&$N\gg s_{\balpha}s_{\bbeta}$&$N\gg s_{\balpha}s_{\bbeta}$&$N\gg s_{\balpha}s_{\bbeta}+s_{\bbeta}^2$\\
\hline
\multirow{2}{*}{Br. et al. (2019)}&$N\gg s_{\balpha}^{4/3}+s_{\bbeta}^2$ \footnote{\cite{bradic2019sparsity} require stronger sparsity conditions than ours. Specifically, when $s_{\bbeta}^3 \ll s_{\balpha}^2$, we have $s_{\balpha}s_{\bbeta}^{1/2} + s_{\bbeta}^2 \ll s_{\balpha}^{4/3} + s_{\bbeta}^2$. Otherwise, $s_{\balpha}s_{\bbeta}^{1/2} + s_{\bbeta}^2\asymp s_{\bbeta}^2\asymp s_{\balpha}^{4/3} + s_{\bbeta}^2$.}&\multirow{2}{*}{$\xmark$}&\multirow{2}{*}{$\xmark$}\\
&or $N\gg s_{\balpha}^2+s_{\bbeta}$&&\\
\hline
\multirow{2}{*}{This paper}&$N\gg s_{\balpha}s_{\bbeta}^{1/2}+s_{\bbeta}^2$ &$N\gg s_{\balpha}s_{\bbeta}^{1/2}+s_{\bbeta}^2$&\multirow{2}{*}{$N\gg s_{\balpha}s_{\bbeta}^{1/2}+s_{\bbeta}^2$}\\
&or $N\gg s_{\balpha}^2+s_{\bbeta}$&or $N\gg s_{\balpha}^2+s_{\bbeta}$&\\
\bottomrule
\end{tabular}
\end{table}

\paragraph*{Supervised ATE estimation, a special case when $R_i \equiv 1$}
In scenarios where all outcome variables \(Y_i\) are fully observed, i.e., \( N = n \), the problem reduces to standard ATE estimation in a supervised setting. Even in these cases, our approach surpasses existing methods by relaxing both model correctness and sparsity requirements; see Table \ref{table:sparsity}. Our proposal is closest to the \emph{model DR} estimators of \cite{smucler2019unifying}, \cite{tan2020model}, and \cite{avagyan2021high}. Beyond addressing the imbalance in labeling groups and effectively leveraging large unlabeled data under decaying MAR conditions, we introduce an asymmetric cross-fitting technique that further relaxes sparsity constraints. Unlike previous work that required the product of two sparsity levels to be on the order of \(N\), our approach permits this product to be on the order of \(N^{3/2}\) or \(N^{5/4}\), depending on which model is correctly specified. {While \cite{bradic2019sparsity} employs a similar cross-fitting strategy, their approach entails additional implementation challenges and requires correct specification of all nuisance models. Additionally, they impose stronger sparsity conditions even under the supervised setting with correctly specified models, as shown in Table \ref{table:sparsity}. Our results demonstrate that the asymmetric cross-fitting technique yields a stronger \emph{sparsity DR} property than that established by \cite{bradic2019sparsity}, without the need to solve non-convex optimization problems. Furthermore, we show that the effect of reducing the required sparsity condition remains valid not only under model misspecification but also when the PS functions decay, as commonly encountered in typical SS setups.}

\begin{figure}[b]
\begin{subfigure}{.24\linewidth}
\includegraphics[width=.9\linewidth]{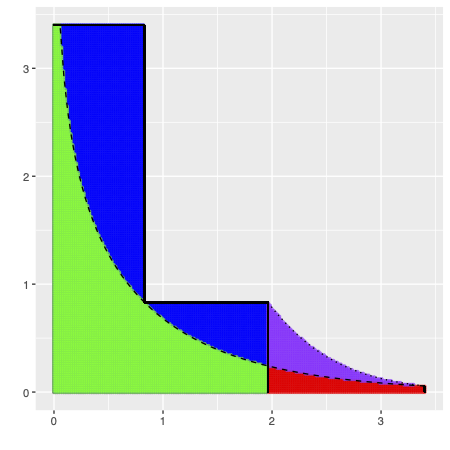}
\caption{$r=0.4$}
\end{subfigure}
\begin{subfigure}{.24\linewidth}
\includegraphics[width=.9\linewidth]{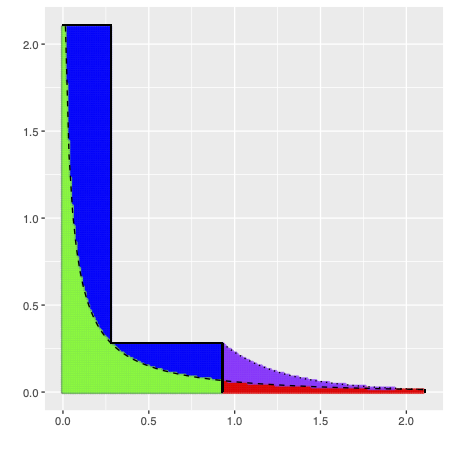}
\caption{$r=0.1$}
\end{subfigure}
\begin{subfigure}{.24\linewidth}
\includegraphics[width=.9\linewidth]{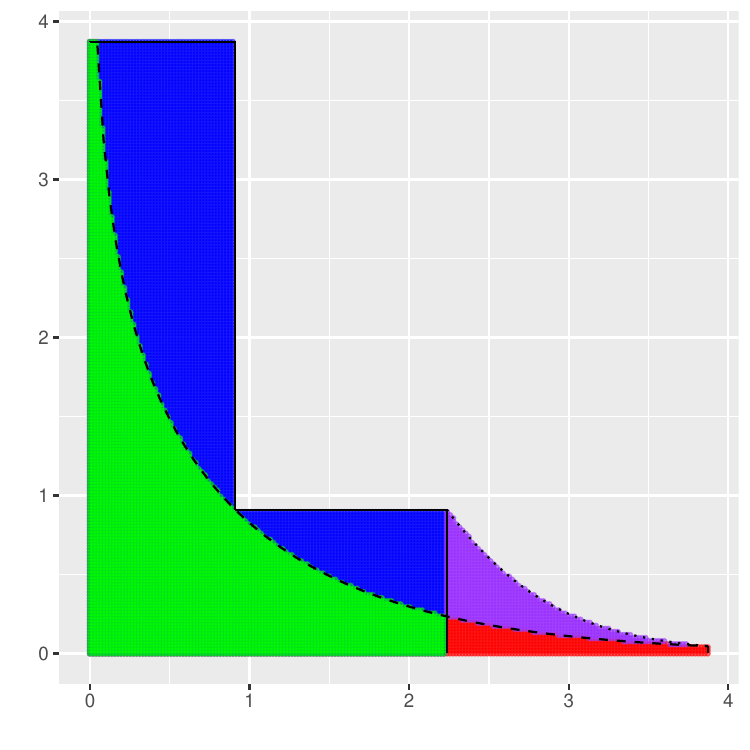}
\caption{$r=0.4$}
\end{subfigure}
\begin{subfigure}{.24\linewidth}
\includegraphics[width=.9\linewidth]{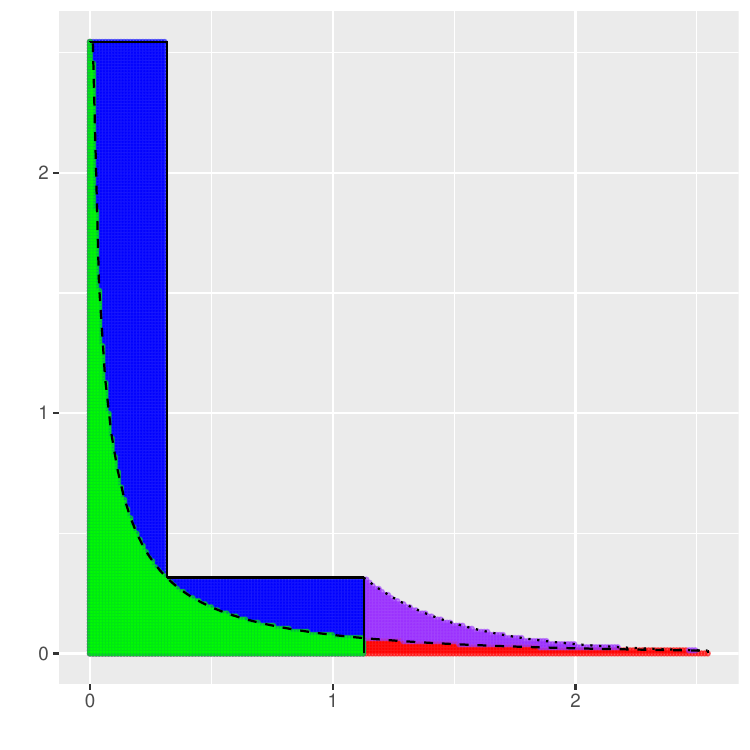}
\caption{$r=0.1$}
\end{subfigure}
\caption{Plots of $\log(1+\bs)=(\log(1+s_{\balpha}),\log(1+s_{\bbeta}))$ satisfying sparsity conditions with $N=500$, $d=1000$ for (a)-(b) and $N=1000$, $d=5000$ for (c)-(d) : Green = \{$f_1(\bs)\leq r$, $f_2(\bs)\leq r$, and $f_3(\bs)\leq r$\}; Red = \{$f_1(\bs)\leq r$, $f_3(\bs)\leq r$, and $f_2(\bs)>r$\}; Blue = \{$f_2(\bs)\leq r$, $f_3(\bs)\leq r$, and $f_1(\bs)>r$\}; Purple = \{$f_3(\bs)\leq r$, $f_1(\bs)>r$, and $f_2(\bs)>r$\}. Lines
 \{$f_1(\bs)=r$\}, $\{f_2(\bs)=r\}$, and $\{f_3(\bs)=r\}$ are dashed, solid, and dotted. With $R \equiv 1$ and all models well specified, \cite{chernozhukov2018double} and the R-DR estimator in Section \ref{sec:gen-DR-DMAR} require $f_1(\bs)=o(1)$ (green + red), \cite{bradic2019sparsity} requires $f_2(\bs)=o(1)$ (green + blue), and the BRSS estimators only require $f_3(\bs)=o(1)$ (green + red + blue + purple).}
\label{figure:sparsity1}
\end{figure}

We clarify the statements above and consider the following special cases: {\bf (M1)} both nuisance models are correctly specified, {\bf (M2)} only the OR model is correctly specified, and {\bf (M3)} only the PS model is correctly specified. Important quantities with $\bs=(s_{\balpha},s_{\bbeta})$ are
\begin{align*}
f_1(\bs)&:={s_{\balpha}\log d/N}\lor{s_{\bbeta}\log d/N}\lor{\sqrt{s_{\balpha}s_{\bbeta}}\log d/\sqrt N},\\
f_2(\bs)&:=\left({s_{\balpha}\log d/N^{3/4}}\lor{s_{\bbeta}\log d/\sqrt N}\right)\land\left({s_{\balpha}\log d/\sqrt N}\lor{s_{\bbeta}\log d/N}\right),\;\;\mbox{and}\\
f_3(\bs)&:=f_1(\bs)\land f_2(\bs)\land\left\{\left(s_{\balpha}\log d/N\right)\lor\left(s_{\bbeta}\log d/\sqrt{N}\right)\lor\left(s_{\balpha}^{2/3}s_{\bbeta}^{1/3}\log d/N^{2/3}\right)\right\}.
\end{align*}

{\bf (M1)} $m(\cdot)=m^*(\cdot)$ and $\gamma_N(\cdot)=\gamma_N^*(\cdot)$.
While \cite{chernozhukov2018double} (and the R-DR estimator in Section \ref{sec:gen-DR-DMAR}) requires $f_1(\bs)=o(1)$ for \emph{rate DR}, and \cite{bradic2019sparsity} relies on $f_2(\bs)=o(1)$ for \emph{sparsity DR}, our method
allows for a more flexible and general sparsity setting $f_3(\bs)=o(1)$, as illustrated in Figure \ref{figure:sparsity1} -- our work combines all four colors, where the area colored in purple denotes the sparsity scenario new to the literature.

{\bf (M2)} $m(\cdot)=m^*(\cdot)$ and $\gamma_N(\cdot)\neq\gamma_N^*(\cdot)$ with $e_\gamma\asymp1$.
We achieve the coveted \emph{model DR} property while requiring the same sparsity as in (M1), i.e., $f_3(\bs)=o(1)$. Comparatively, the best result in the existing literature, \cite{smucler2019unifying}, still necessitates a stronger condition of $f_1(\bs)=o(1)$, while \cite{bradic2019sparsity} is valid only for correctly specified models. Further details and comparisons can be found in Table \ref{table:sparsity}.

{\bf (M3)} $\gamma_N(\cdot)=\gamma_N^*(\cdot)$ and $m(\cdot)\neq m^*(\cdot)$ with $e_m\asymp1$. Our sparsity conditions are once again weaker.
The correctness of OR and PS models affects the required sparsity conditions differently. When the PS model is misspecified, we do not require any additional assumptions compared with (M1). In contrast, when the OR model is misspecified, both \cite{smucler2019unifying} and our proposed method further require an ultra-sparse PS. By using the non-cross-fitted PS estimate, we allow a weaker product-rate condition $s_{\balpha}\sqrt{s_{\bbeta}}=o(N)$ (omitting the logarithm terms) instead of the usual product-rate condition $s_{\balpha}s_{\bbeta}=o(N)$ -- a condition that is always required in \cite{smucler2019unifying}. Although \cite{tan2020model,ning2020robust,avagyan2021high,dukes2021inference,dukes2020doubly} also provide robust inference for the ATE when the OR model is misspecified, they require ultra-sparse conditions for both nuisances, whereas we only need the PS model to be ultra-sparse.

{Beyond the aforementioned works, another line of research has explored optimal sparsity conditions for achieving \(\sqrt{N}\)-inference in supervised settings. To the best of our knowledge, \cite{bradic2019minimax} is the only study that has investigated minimax sparsity conditions; however, it focuses on a different class of functions with `approximate sparsity' structures. \cite{liu2023root} established \(\sqrt{N}\)-inference under the condition \(s_{\balpha} \land s_{\bbeta} = o(\sqrt{N}/\log d)\) but additionally imposed a bounded signal requirement on the OR model, assuming conditions such as \(|Y(j)|<c_0\) almost surely. The minimal sparsity conditions necessary to achieve \(\sqrt{N}\)-inference without such bounded signal assumptions, even in degenerate supervised settings, remain an open question and merit further investigation.

\paragraph*{Limited overlap} The potential violation of the (strict) positivity condition and its associated challenges have been studied by another line of work, often referred to as `limited overlap'; see, for example, \cite{scharfstein1999adjusting}, \cite{kang2007demystifying}, \cite{robins2007comment}, \cite{d2021overlap}, and \cite{mou2023kernel}. While efficient estimation becomes unattainable in the worst-case scenario when the strict positivity condition is violated \citep{khan2010irregular}, a common approach is to trim the PS function and estimate a re-weighted ATE, which down-weights samples with small PS values; see, e.g., \cite{crump2009dealing}, \cite{rothe2017robust}, and \cite{yang2018asymptotic}. This line of work is suitable when the PS function approaches zero at specific \(\bx\). These approaches typically model PS functions as dependent only on \(\bx\), without accounting for their dependency on \(N\). In contrast, SS datasets typically involve a labeled sample size \(n\) that is much smaller than the total sample size \(N\), necessitating the modeling of the distribution of the labeling indicator \(R\) as a function of \(N\), so that \(\E(R) = \E(n/N) \to 0\) as \(N \to \infty\) is permissible. When the PS function decays \emph{uniformly} across the support, trimming either excludes a substantial portion of the data if the threshold is not small enough, or fails to resolve the decaying issue if the threshold is small, as the post-trimming PS values are still small. Our framework differs from the `limited overlap' approach in that we treat the labeling PS function as a function of both \(\bx\) and \(N\), making it especially suitable for SS scenarios where \(N \gg n\). Our results demonstrate that CAN estimation of the original ATE parameter is feasible under the decaying MAR framework, without the need for trimming techniques, even when the positivity condition is violated.}

\section{Simulation studies}\label{sec:sim}
We evaluate the R-DR and bias-reduced estimators under various data-generating processes (DGPs), focusing on the estimation of the ATE, $\mu_0 = \theta_1 - \theta_0$.

\subsection{Results under the decaying MAR setting}\label{subsec:sim_DMAR}

We consider three types of DGPs in our simulation studies: (a) linear OR models, logistic (product) PS models; (b) linear OR models, non-logistic PS models; and (c) non-linear OR models, logistic PS models. In this section, we focus on linear and logistic working models (excluding estimates based on non-parametric models) for the OR and PS functions. Consequently, DGP (a) represents the case where both models are correctly specified, while DGPs (b) and (c) correspond to scenarios where only one of the nuisance models is correctly specified. We consider i.i.d. truncated normal covariates $\bX_{ij}\sim^\mathrm{iid}Z_{\mathrm{trun},2}$ and $\bX_{i1}=1$ for each $i\in\{1,\dots,N\}$ and $j\in\{2,\dots,d\}$, where $Z_{\mathrm{trun},2}\sim Z\mid\{|Z|<2\}$ and $Z\sim\mathcal N(0,1)$.

\paragraph*{(a) Linear OR models, logistic PS models} For $j\in\{0,1\}$ and any $\bx\in\R^d$, we set
\begin{align}
&\gamma_N(j,\bx)=g(\bx^T\bbeta(j)), \;\; \pi(\bx)=\frac{\gamma_N(1,\bx)+1-\gamma_N(0,\bx)}{2},\label{DGP_logistic_PS1}\\
&p_N(1,\bx)=\frac{\gamma_N(1,\bx)}{\pi(\bx)},\;\; \mbox{and} \;\; p_N(0,\bx)=\frac{\gamma_N(0,\bx)}{1-\pi(\bx)}.\label{DGP_logistic_PS2}
\end{align}
\begin{table}[t]
\begin{center}
\caption{Simulation results for DGP (a) with sparsity levels $s_{\balpha}=s_{\bbeta}=3$. Bias: empirical bias; RMSE: root mean square error; Length: average length of the $95\%$ confidence intervals; Coverage: average coverage of the $95\%$ confidence intervals. {Bias, RMSE, and Length are calculated based on medians.}} \label{table:DGP-a}
\resizebox{\columnwidth}{!}{%
\begin{tabular}{>{\centering\arraybackslash}p{3.8em}>{\centering\arraybackslash}p{3.4em}>{\centering\arraybackslash}p{3.4em}>{\centering\arraybackslash}p{3.4em}>{\centering\arraybackslash}p{3.4em}c>{\centering\arraybackslash}p{3.4em}>{\centering\arraybackslash}p{3.4em}>{\centering\arraybackslash}p{3.4em}>{\centering\arraybackslash}p{3.4em}}
\hline
\toprule
Estimator&Bias&RMSE&Length&Coverage&&Bias&RMSE&Length&Coverage\\
\hline
&\multicolumn{4}{c}{\cellcolor{gray!50} DGP (a) $N\gamma_N=500$, $d=51$, $\gamma_N=0.05$}&&\multicolumn{4}{c}{\cellcolor{gray!50} DGP (a) $N\gamma_N=1000$, $d=51$, $\gamma_N=0.1$}\\
\cline{2-5}\cline{7-10}
$\muhat_{\mbox{\tiny oracle}}$&0.007&0.102&0.540&0.948&&0.006&0.075&0.428&0.956\\
$\muhat_{\mbox{\tiny MCAR}}$&0.003&0.090&0.374&0.828&&0.003&0.067&0.329&0.902\\
$\muhat_{\mbox{\tiny SS-Lasso}}$&0.008&0.100&0.470&0.906&&0.004&0.073&0.391&0.934\\
$\muhat_{\mbox{\tiny SS-RF}}$&-0.026&0.151&0.619&0.832&&-0.023&0.097&0.422&0.848\\
$\muhat_{\mbox{\tiny BRSS}}$&0.004&0.102&0.503&0.908&&0.008&0.073&0.409&0.942\\
\hline
&\multicolumn{4}{c}{\cellcolor{gray!50} DGP (a) $N\gamma_N=500$, $d=201$, $\gamma_N=0.1$}&&\multicolumn{4}{c}{\cellcolor{gray!50} DGP (a) $N\gamma_N=1000$, $d=201$, $\gamma_N=0.1$}\\
\cline{2-5}\cline{7-10}
$\muhat_{\mbox{\tiny oracle}}$&-0.014&0.098&0.597&0.954&&-0.013&0.080&0.424&0.964\\
$\muhat_{\mbox{\tiny MCAR}}$&-0.011&0.091&0.466&0.878&&-0.010&0.069&0.329&0.900\\
$\muhat_{\mbox{\tiny SS-Lasso}}$&-0.011&0.095&0.522&0.936&&-0.012&0.074&0.379&0.946\\
$\muhat_{\mbox{\tiny BRSS}}$&-0.012&0.099&0.555&0.924&&-0.014&0.079&0.402&0.958\\
\bottomrule
\end{tabular}
}
\end{center}
\end{table}
We then consider
\begin{align}
T_i\mid\bX_i\sim\mathrm{Bernoulli}(\pi(\bX_i))\;\; \mbox{and} \;\; R_i\mid(\bX_i,T_i=j)\sim\mathrm{Bernoulli}(p_N(j,\bX_i)).\label{DGP_logistic_PS3}
\end{align}
Finally, we set linear OR models as follows:
\begin{align}
Y_i(j)=\bX_i^T\balpha(j)+\delta_i,\;\; \mbox{and} \;\; R_iY_i=R_iY_i(T_i), \;\;\;\mbox{where}\;\;\delta_i\sim\mathcal N(0,1).\label{DGP_linear_OR}
\end{align}

\paragraph*{(b) Linear OR models, non-logistic PS models}
 We use $\pi(\bx)=0.3\sin(\bx^T\bomega)+0.5$ and $p_N(j,\bx)=g(\bx^T\bbeta(j))$. The treatment and missingness indicators $(T_i,R_i)$ follow \eqref{DGP_logistic_PS3}, and the outcomes $Y_i(j)$ are determined as in \eqref{DGP_linear_OR}.

\paragraph*{(c) Non-linear OR models, logistic PS models}
We consider \eqref{DGP_logistic_PS1}-\eqref{DGP_logistic_PS3} as in part (a) but set quadratic OR models with $\bX_i^2:=(\bX_{i1}^2,\dots,\bX_{i d}^2)^T$ and
\begin{align}
Y_i(j)=\bX_i^T\balpha(j)+(\bX_i^2)^T\etabold(j)+\delta_i,\;\;R_iY_i=R_iY_i(T_i).\label{DGP_quad_OR}
\end{align}

\par\smallskip
The parameter values across all the DGPs above are chosen as follows:
$\balpha(1):=3(1,1,\boldsymbol{1}_{s_{\balpha}-1}^T/\sqrt{s_{\balpha}-1},0,\dots,0)^T\in\R^d,\;\;\balpha(0):=-\balpha(1), $
 $ \bbeta(1):=(\beta_N(1),1,\boldsymbol{1}_{s_{\bbeta}-1}^T/(s_{\bbeta}-1),0,\dots,0)^T\in\R^d, $
$\bbeta(0):=(\beta_N(0),-1,-\boldsymbol{1}_{s_{\bbeta}-1}^T/(s_{\bbeta}-1),0,\dots,0)^T\in\R^d,$
$\bomega:=(0,1,$ $ \boldsymbol{1}_{s_{\bbeta}-1}^T/(s_{\bbeta}-1),0,\dots,0)^T\in\R^d, $
$\etabold(1):=(0,1,\boldsymbol{1}_{s_{\balpha}-1}^T/\sqrt{s_{\balpha}-1},0,\dots,0)^T\in\R^d,\;\;\etabold(0):=-\etabold(1),
$
where for any positive integer $s\geq1$, $\boldsymbol{1}_s:=\{1,\dots,1\}\in\R^s$ and $\boldsymbol{1}_0:=\emptyset$, and $\beta_N(1),\beta_N(0)$ are chosen such that $\E(RT)=\E\{R(1-T)\}=\gamma_{N}$.

We consider the following estimators:
(1) Oracle estimator \(\muhat_{\text{oracle}}\): R-DR estimator with nuisances set as the true values;
(2) \(\muhat_{\text{MCAR}}\): SS estimator that assumes missingness is MCAR (ignoring labeling bias), with OR and \emph{treatment PS} models estimated via \(\ell_1\)-regularized linear/logistic regression;
(3) \(\muhat_{\text{SS-Lasso}}\): R-DR estimator with OR and \emph{product PS} models estimated using \(\ell_1\)-regularized linear/logistic regression;
(4) \(\muhat_{\text{SS-RF}}\): R-DR estimator with random forest nuisance estimates (up to \(d=51\) dimensions); {and}
(5) \(\muhat_{\mbox{\tiny BRSS}}\): BRSS estimator of \eqref{def:BRSS-ATE}, with the feature map \(\varphi(\cdot)\) set to the identity function. The tuning parameters are chosen using 5-fold cross-validation and results, repeated $500$ times, are presented in Tables \ref{table:DGP-a}-\ref{table:DGP-c}.

\begin{table}[t]
\begin{center}
\caption{Simulation results for DGP (b) with sparsity levels $s_{\balpha}=2$ and $s_{\bbeta}=6$. The rest of the caption details remain the same as those in Table \ref{table:DGP-a}.} \label{table:DGP-b}
\resizebox{\columnwidth}{!}{%
\begin{tabular}{>{\centering\arraybackslash}p{3.8em}>{\centering\arraybackslash}p{3.4em}>{\centering\arraybackslash}p{3.4em}>{\centering\arraybackslash}p{3.4em}>{\centering\arraybackslash}p{3.4em}c>{\centering\arraybackslash}p{3.4em}>{\centering\arraybackslash}p{3.4em}>{\centering\arraybackslash}p{3.4em}>{\centering\arraybackslash}p{3.4em}}
\hline
\toprule
Estimator&Bias&RMSE&Length&Coverage&&Bias&RMSE&Length&Coverage\\
\hline
&\multicolumn{4}{c}{\cellcolor{gray!50} DGP (b) $N\gamma_N=500$, $d=51$, $\gamma_N=0.05$}&&\multicolumn{4}{c}{\cellcolor{gray!50} DGP (b) $N\gamma_N=1000$, $d=51$, $\gamma_N=0.1$}\\
\cline{2-5}\cline{7-10}
$\muhat_{\mbox{\tiny oracle}}$&-0.012&0.084&0.565&0.964&&-0.007&0.075&0.412&0.954\\
$\muhat_{\mbox{\tiny MCAR}}$&-0.005&0.067&0.286&0.812&&-0.001&0.061&0.250&0.856\\
$\muhat_{\mbox{\tiny SS-Lasso}}$&-0.004&0.079&0.390&0.912&&-0.005&0.065&0.312&0.916\\
$\muhat_{\mbox{\tiny SS-RF}}$&-0.023&0.101&0.403&0.824&&-0.015&0.070&0.297&0.852\\
$\muhat_{\mbox{\tiny BRSS}}$&0.001&0.075&0.470&0.946&&-0.005&0.068&0.366&0.942\\
\hline
&\multicolumn{4}{c}{\cellcolor{gray!50} DGP (b) $N\gamma_N=500$, $d=201$, $\gamma_N=0.1$}&&\multicolumn{4}{c}{\cellcolor{gray!50} DGP (b) $N\gamma_N=1000$, $d=201$, $\gamma_N=0.1$}\\
\cline{2-5}\cline{7-10}
$\muhat_{\mbox{\tiny oracle}}$&-0.008&0.106&0.572&0.936&&0.001&0.078&0.412&0.952\\
$\muhat_{\mbox{\tiny MCAR}}$&-0.008&0.078&0.354&0.872&&-0.004&0.062&0.270&0.880\\
$\muhat_{\mbox{\tiny SS-Lasso}}$&-0.006&0.089&0.418&0.884&&-0.001&0.066&0.300&0.920\\
$\muhat_{\mbox{\tiny BRSS}}$&-0.010&0.089&0.480&0.918&&0.001&0.074&0.357&0.952\\
\bottomrule
\end{tabular}
}
\end{center}
\end{table}

\begin{table}[t]
\begin{center}
\caption{Simulation results for DGP (c) with sparsity levels $s_{\balpha}=6$ and $s_{\bbeta}=2$. The rest of the caption details remain the same as those in Table \ref{table:DGP-a}.} \label{table:DGP-c}
\resizebox{\columnwidth}{!}{%
\begin{tabular}{>{\centering\arraybackslash}p{3.8em}>{\centering\arraybackslash}p{3.4em}>{\centering\arraybackslash}p{3.4em}>{\centering\arraybackslash}p{3.4em}>{\centering\arraybackslash}p{3.4em}c>{\centering\arraybackslash}p{3.4em}>{\centering\arraybackslash}p{3.4em}>{\centering\arraybackslash}p{3.4em}>{\centering\arraybackslash}p{3.4em}}
\hline
\toprule
Estimator&Bias&RMSE&Length&Coverage&&Bias&RMSE&Length&Coverage\\
\hline
&\multicolumn{4}{c}{\cellcolor{gray!50} DGP (c) $N\gamma_N=500$, $d=51$, $\gamma_N=0.05$}&&\multicolumn{4}{c}{\cellcolor{gray!50} DGP (c) $N\gamma_N=1000$, $d=51$, $\gamma_N=0.1$}\\
\cline{2-5}\cline{7-10}
$\muhat_{\mbox{\tiny oracle}}$&-0.013&0.077&0.463&0.948&&-0.010&0.063&0.390&0.954\\
$\muhat_{\mbox{\tiny MCAR}}$&-0.722&0.722&0.484&0.002&&-0.642&0.642&0.394&0.000\\
$\muhat_{\mbox{\tiny SS-Lasso}}$&-0.255&0.263&0.752&0.712&&-0.181&0.182&0.576&0.740\\
$\muhat_{\mbox{\tiny SS-RF}}$&-0.282&0.288&0.749&0.668&&-0.186&0.188&0.525&0.674\\
$\muhat_{\mbox{\tiny BRSS}}$&-0.116&0.139&0.637&0.884&&-0.044&0.095&0.493&0.930\\
\hline
&\multicolumn{4}{c}{\cellcolor{gray!50} DGP (c) $N\gamma_N=1000$, $d=201$, $\gamma_N=0.1$}&&\multicolumn{4}{c}{\cellcolor{gray!50} DGP (c) $N\gamma_N=2000$, $d=201$, $\gamma_N=0.1$}\\
\cline{2-5}\cline{7-10}
$\muhat_{\mbox{\tiny oracle}}$&-0.009&0.064&0.391&0.964&&-0.010&0.050&0.276&0.948\\
$\muhat_{\mbox{\tiny MCAR}}$&-0.637&0.637&0.395&0.000&&-0.628&0.628&0.278&0.000\\
$\muhat_{\mbox{\tiny SS-Lasso}}$&-0.223&0.223&0.556&0.660&&-0.168&0.168&0.407&0.622\\
$\muhat_{\mbox{\tiny BRSS}}$&-0.065&0.091&0.488&0.894&&-0.033&0.063&0.349&0.930\\
\bottomrule
\end{tabular}
}
\end{center}
\end{table}

Among the four estimators, $\muhat_{\mbox{\tiny SS-RF}}$ exhibits larger bias and RMSE due to slower convergence rates, while the remaining estimators demonstrate smaller biases compared to RMSE in Tables \ref{table:DGP-a} and \ref{table:DGP-b}. In contrast, in Table \ref{table:DGP-c}, notable biases and larger RMSEs are observed for $\muhat_{\mbox{\tiny SS-RF}}$, $\muhat_{\mbox{\tiny SS-Lasso}}$, and $\muhat_{\mbox{\tiny MCAR}}$, while $\muhat_{\mbox{\tiny BRSS}}$ outperforms them significantly. The poor performance of $\muhat_{\mbox{\tiny MCAR}}$ in all DGPs is attributed to its incorrect treatment of the labeling indicator's PS as a constant, resulting in significant undercoverage, particularly in Table \ref{table:DGP-c}. On the other hand, $\muhat_{\mbox{\tiny SS-Lasso}}$ provides valid inference with correct nuisance model specification (as per Theorem \ref{thm:gen-rate-DR}), but its reliability diminishes when model misspecification occurs, leading to underestimation of variance and large bias in Tables \ref{table:DGP-b} and \ref{table:DGP-c}. In contrast, the $\muhat_{\mbox{\tiny SS-RF}}$ estimator, relying on non-parametric nuisance estimators, fails to satisfy the `product-rate' condition, resulting in undercoverage in all DGPs. The coverage results in Tables \ref{table:DGP-a}-\ref{table:DGP-c} support the inference quality of $\muhat_{\mbox{\tiny BRSS}}$ when one nuisance is misspecified, as per Theorem \ref{cor:para}. Notably, Table \ref{table:DGP-b} exhibits excellent coverage with minor deviations for smaller effective sample size and higher dimension, while in Table \ref{table:DGP-c}, $\muhat_{\mbox{\tiny BRSS}}$ consistently delivers strong performance for higher sample sizes where $\muhat_{\mbox{\tiny MCAR}}$, $\muhat_{\mbox{\tiny SS-Lasso}}$, and $\muhat_{\mbox{\tiny SS-RF}}$ fail, respectively.

\subsection{A degenerate setting with outcomes fully observed}\label{subsec:sim_ful}

In the setting of fully observed outcomes ($R_i=1$ for all $i\in{1,\dots,N}$), we examine the cases where one of the nuisance models is misspecified and highlight what is different from Section \ref{subsec:sim_DMAR}.
\vskip 2pt
\paragraph*{(d) Linear OR models, non-logistic PS model}

$$\pi(\bx)=\exp(\bx^T\bbeta_d)/\{1+\exp(\bx^T\bbeta_d)\}\{0.3\sin(\bx^T\bbeta_d)+0.7\}\;\;\mbox{and}\;\;Y_i(j)=\bS_i^T\balpha_d(j)+\delta_i.$$

\paragraph*{(e) Non-linear OR models, logistic PS model}

$$\pi(\bx)=\exp(\bx^T\bbeta_e)/\{1+\exp(\bx^T\bbeta_e)\}\;\;\mbox{and}\;\;Y_i(j)=\bX_i^T\balpha_e(j)+(\bX_i^2)^T\etabold_e(j)+\delta_i.$$

\begin{table}[t]
\begin{center}
\caption{Simulation results for DGPs (d)-(e). The rest of the caption details remain the same as those in Table \ref{table:DGP-a}.} \label{table:DGP-d_e}
\resizebox{\columnwidth}{!}{%
\begin{tabular}{>{\centering\arraybackslash}p{3.8em}>{\centering\arraybackslash}p{3.4em}>{\centering\arraybackslash}p{3.4em}>{\centering\arraybackslash}p{3.4em}>{\centering\arraybackslash}p{3.4em}c>{\centering\arraybackslash}p{3.4em}>{\centering\arraybackslash}p{3.4em}>{\centering\arraybackslash}p{3.4em}>{\centering\arraybackslash}p{3.4em}}
\hline
\toprule
Estimator&Bias&RMSE&Length&Coverage&&Bias&RMSE&Length&Coverage\\
\hline
&\multicolumn{4}{c}{\cellcolor{gray!50} DGP (d) $N=300$, $d=51$}&&\multicolumn{4}{c}{\cellcolor{gray!50} DGP (e) $N=400$, $d=51$}\\
\cline{2-5}\cline{7-10}
$\muhat_{\mbox{\tiny oracle}}$&0.002&0.411&2.518&0.948&&-0.012&0.375&2.305&0.958\\
$\muhat_{\mbox{\tiny Smucler}}$&-0.560&0.637&2.626&0.870&&0.484&0.584&2.479&0.896\\
$\muhat_{\mbox{\tiny BRSS}}$&-0.208&0.460&2.471&0.928&&0.104&0.471&2.315&0.942\\
\bottomrule
\end{tabular}
}
\end{center}
\end{table}

\begin{table}[t]
\begin{center}
\caption{Simulation results for DGP (f). The rest of the caption details remain the same as those in Table \ref{table:DGP-a}.} \label{table:DGP-f}
\resizebox{\columnwidth}{!}{%
\begin{tabular}{>{\centering\arraybackslash}p{3.8em}>{\centering\arraybackslash}p{3.4em}>{\centering\arraybackslash}p{3.4em}>{\centering\arraybackslash}p{3.4em}>{\centering\arraybackslash}p{3.4em}c>{\centering\arraybackslash}p{3.4em}>{\centering\arraybackslash}p{3.4em}>{\centering\arraybackslash}p{3.4em}>{\centering\arraybackslash}p{3.4em}}
\hline
\toprule
Estimator&Bias&RMSE&Length&Coverage&&Bias&RMSE&Length&Coverage\\
\hline
&\multicolumn{4}{c}{\cellcolor{gray!50} DGP (f) $N\gamma_N=500$, $d=31$, $\gamma_N=0.05$}&&\multicolumn{4}{c}{\cellcolor{gray!50} DGP (f) $N\gamma_N=1000$, $d=51$, $\gamma_N=0.1$}\\
\cline{2-5}\cline{7-10}
$\muhat_{\mbox{\tiny oracle}}$&-0.007&0.083&0.500&0.954&&-0.018&0.074&0.405&0.944\\
$\muhat_{\mbox{\tiny MCAR}}$&-0.724&0.724&0.493&0.004&&-0.614&0.614&0.402&0.002\\
$\muhat_{\mbox{\tiny SS-Lasso}}$&-0.224&0.233&0.774&0.736&&-0.255&0.255&0.521&0.534\\
$\muhat_{\mbox{\tiny SS-RF}}$&-0.172&0.198&0.736&0.774&&-0.151&0.161&0.498&0.732\\
$\muhat_{\mbox{\tiny BRSS}}$&-0.170&0.180&0.640&0.800&&-0.135&0.143&0.478&0.760\\
$\muhat_{\mbox{\tiny DC-BRSS}}$&-0.127&0.158&1.144&0.988&&-0.084&0.109&0.799&0.992\\
\bottomrule
\end{tabular}
}
\end{center}
\end{table}

The parameters for DGPs (d) and (e) are $\balpha_d(1):=3(1,0.9^1,\dots,0.9^{d-1})^T\in\R^d,\;\;\balpha_d(0)=-\balpha_d(1),\;\;\balpha_e(j)=\balpha_d(j)$, $\etabold_e(1):=(0,0.9^1,\dots,0.9^{d-1})^T\in\R^d,\;\;\etabold_e(0)=-\etabold_e(1)$, $\bbeta_d:=(0.99,0.5\cdot0.7^1,\dots,0.5\cdot0.7^{d-1})^T\in\R^d,\;\;\bbeta_e:=(0.2247,0.7^1,\dots,0.7^{d-1})^T\in\R^d$. Here, we compare the numerical performance of our proposed bias-reduced estimator $\muhat_{\mbox{\tiny BRSS}}$ with $\muhat_{\mbox{\tiny Smucler}}$ by \cite{smucler2019unifying}; as per Table \ref{table:sparsity} their estimator is the most competitive in the existing literature. The nuisance parameters in DGPs (d) and (e) exhibit a `weakly sparse' nature, with bounded $\ell_1$-norms but $\ell_0$-norms equal to the dimension. In Table \ref{table:DGP-d_e}, it is observed that the estimator $\muhat_{\mbox{\tiny Smucler}}$ suffers from substantial biases. Consequently, the coverages based on $\muhat_{\mbox{\tiny Smucler}}$ are relatively poor. In contrast, our proposed estimator $\muhat_{\mbox{\tiny BRSS}}$ exhibits significantly smaller biases, leading to improved coverages. Additionally, $\muhat_{\mbox{\tiny BRSS}}$ achieves smaller RMSEs compared to $\muhat_{\mbox{\tiny Smucler}}$ for both DGPs (d) and (e).

\subsection{Results based on the de-coupled approach}

In this section, we further examine the behavior of the de-coupled bias-reduced SS estimator, $\muhat_{\tiny\mbox{DC-BRSS}}$, proposed in Section \ref{subsec:SP}. The following case is considered.

\paragraph*{(f) Non-linear OR models, non-logistic treatment PS model, logistic labeling PS model} Generate $\pi(\bx)=0.3\cos(\bx^T\bbeta)+0.5$, $p_N(j,\bx)=\exp\{\bx^T\bbeta'(j)\}/[1+\exp\{\bx^T\bbeta'(j)\}]$, and \eqref{DGP_quad_OR}. Choose $\balpha(j)$ and $\etabold(j)$ as in Section \ref{subsec:sim_DMAR} with $s_{\balpha}=5$, $\bbeta:=2(0,0.9^1,\dots,0.9^{d-1})^T$, $\bbeta'(1):=(\bbeta_N'(1),1,\boldsymbol{1}_{4}^T/4,0,\dots,0)^T\in\R^d$, and $\bbeta'(0):=(\bbeta_N'(0),-1,-\boldsymbol{1}_{4}^T/4,0,\dots,0)^T\in\R^d$, where $\bbeta_N'(1)$ and $\bbeta_N'(0)$ are chosen such that $\E(RT)=\E\{R(1-T)\}=\gamma_{N}$.

We compare the numerical performance of the estimators considered in Section \ref{subsec:sim_DMAR} with $\muhat_{\tiny\mbox{DC-BRSS}}$, where the treatment PS function $\pi(\cdot)$ is estimated using random forests. As shown in Table \ref{table:DGP-f}, $\muhat_{\tiny\mbox{MCAR}}$ provides large biases and very poor coverages since the true labeling mechanism is not MCAR. The performance of $\muhat_{\tiny\mbox{SS-Lasso}}$ is also relatively poor since both nuisance models are misspecified -- the OR models are non-linear and the product PS models are non-logistic. The BRSS estimator $\muhat_{\tiny\mbox{BRSS}}$ provides slightly smaller biases and RMSEs, although the working models are still misspecified. The fully non-parametric estimator $\muhat_{\tiny\mbox{SS-RF}}$ performs slightly worse than $\muhat_{\tiny\mbox{BRSS}}$. Although the non-parametric nuisance estimates are consistent, the convergence rates are relatively slow as the `effective sample size' for the OR and product PS estimation is only $500$ when $d=31$ and $1000$ when $d=51$, resulting in relatively large biases for the final ATE estimation. Lastly, by fully leveraging the large-sized \((\bX_i, T_i)\) samples and estimating \(\pi(\cdot)\) non-parametrically while keeping other working models parametric, the proposed DC-BRSS estimator \(\muhat_{\tiny\mbox{DC-BRSS}}\) achieves lower bias and RMSE than all other ATE estimators, even when the OR models are misspecified. We observe that confidence intervals based on the DC-BRSS method tend to exhibit slight over-coverage, likely due to the instability of the treatment PS estimate, where random forests were used as the initial estimator.

\section{Applications to a pseudo-random dataset}\label{sec:data_analysis}

We compare the performance of the proposed estimators using a synthetic dataset obtained from the Atlantic Causal Inference Conference (ACIC) 2019 Data Challenge.\footnote{The high-dimensional datasets (with continuous outcomes) provided by ACIC 2019 are available at: \url{https://sites.google.com/view/acic2019datachallenge/data-challenge}, and these are constructed based on 16 scenarios.} We focus on 14 scenarios from the ACIC 2019 dataset, as two scenarios share the same covariate matrix. To examine the performance under model misspecification, we construct $T_i$ and $R_i$ as in \eqref{DGP_logistic_PS1}-\eqref{DGP_logistic_PS3} and generate the outcome variable $Y_i$ as in \eqref{DGP_quad_OR} -- that is, we consider a correctly specified (logistic) PS model and a misspecified (quadratic) OR model. We set
$\balpha(1):=(2,2,1,1,1,1,0,\dots,0)^T\in\R^d,\;\;\balpha(0):=-\balpha(1), $
$\bbeta(1):=(-1.5,0.5,0,\dots,0)^T\in\R^d,\;\;\bbeta(0):=(-1.5,-0.5,0,\dots,0)^T\in\R^d, $ and
$\etabold(1):=0.35(0,2,1,1,1,1,0,\dots,0)^T\in\R^d,\;\;\etabold(0):=-\etabold(1).
$ We generate 100 sets for each scenario resulting in
 $14*100=1400$ pseudo-random datasets in total. The covariate matrices have dimensions of $(N,d)=(1000,201)$ in Scenarios 1, 2, 3, 6, 8, and 14, and $(N,d)=(2000,201)$ in the other scenarios.
We consider $\muhat_{\mbox{\tiny MCAR}}$, $\muhat_{\mbox{\tiny SS-Lasso}}$, and $\muhat_{\mbox{\tiny BRSS}}$. The results are reported in Table \ref{table:pseudo}.

\begin{table}[b]
\begin{center}
\caption{Results for the pseudo-random dataset. Bias: empirical bias; RMSE: root mean square error; Length: average length of the $95\%$ confidence intervals; Coverage: average coverage of the $95\%$ confidence intervals.} \label{table:pseudo}
\begin{tabular}{cccccccccc}
\hline
\toprule
Estimator&Bias&RMSE&Length&Coverage&&Bias&RMSE&Length&Coverage\\
\hline
&\multicolumn{4}{c}{\cellcolor{gray!50} Scenario 1 $N=1000$, $d=201$}&&\multicolumn{4}{c}{\cellcolor{gray!50} Scenario 2 $N=1000$, $d=201$}\\
\cline{2-5}\cline{7-10}
$\muhat_{\mbox{\tiny MCAR}}$&-0.341&0.403&1.209&0.880&&-0.231&0.307&0.994&0.900\\
$\muhat_{\mbox{\tiny SS-Lasso}}$&-0.257&0.398&1.248&0.900&&-0.157&0.260&1.057&0.960\\
$\muhat_{\mbox{\tiny BRSS}}$&-0.263&0.322&1.142&0.940&&-0.164&0.244&1.004&0.940\\
\hline
&\multicolumn{4}{c}{\cellcolor{gray!50} Scenario 3 $N=1000$, $d=201$}&&\multicolumn{4}{c}{\cellcolor{gray!50} Scenario 4 $N=2000$, $d=201$}\\
\cline{2-5}\cline{7-10}
$\muhat_{\mbox{\tiny MCAR}}$&-0.080&0.400&1.879&0.960&&-0.237&0.273&0.652&0.760\\
$\muhat_{\mbox{\tiny SS-Lasso}}$&-0.102&0.408&1.958&0.920&&-0.144&0.200&0.696&0.930\\
$\muhat_{\mbox{\tiny BRSS}}$&-0.114&0.411&1.906&0.930&&-0.095&0.161&0.672&0.960\\
\hline
&\multicolumn{4}{c}{\cellcolor{gray!50} Scenario 5 $N=2000$, $d=201$}&&\multicolumn{4}{c}{\cellcolor{gray!50} Scenario 6 $N=1000$, $d=201$}\\
\cline{2-5}\cline{7-10}
$\muhat_{\mbox{\tiny MCAR}}$&-0.290&0.316&0.665&0.650&&-0.247&0.313&0.997&0.860\\
$\muhat_{\mbox{\tiny SS-Lasso}}$&-0.187&0.231&0.717&0.870&&-0.168&0.266&1.058&0.920\\
$\muhat_{\mbox{\tiny BRSS}}$&-0.118&0.167&0.686&0.980&&-0.171&0.245&1.006&0.950\\
\hline
&\multicolumn{4}{c}{\cellcolor{gray!50} Scenario 7 $N=2000$, $d=201$}&&\multicolumn{4}{c}{\cellcolor{gray!50} Scenario 8 $N=1000$, $d=201$}\\
\cline{2-5}\cline{7-10}
$\muhat_{\mbox{\tiny MCAR}}$&-0.254&0.290&0.656&0.670&&-0.272&0.337&0.956&0.810\\
$\muhat_{\mbox{\tiny SS-Lasso}}$&-0.152&0.217&0.709&0.880&&-0.180&0.285&1.019&0.910\\
$\muhat_{\mbox{\tiny BRSS}}$&-0.095&0.158&0.679&0.980&&-0.186&0.270&0.975&0.940\\
\hline
&\multicolumn{4}{c}{\cellcolor{gray!50} Scenario 9 $N=2000$, $d=201$}&&\multicolumn{4}{c}{\cellcolor{gray!50} Scenario 12 $N=2000$, $d=201$}\\
\cline{2-5}\cline{7-10}
$\muhat_{\mbox{\tiny MCAR}}$&-0.247&0.289&0.666&0.710&&-0.250&0.283&0.662&0.710\\
$\muhat_{\mbox{\tiny SS-Lasso}}$&-0.140&0.209&0.718&0.900&&-0.147&0.198&0.711&0.890\\
$\muhat_{\mbox{\tiny BRSS}}$&-0.082&0.156&0.687&0.960&&-0.093&0.159&0.678&0.990\\
\hline
&\multicolumn{4}{c}{\cellcolor{gray!50} Scenario 13 $N=2000$, $d=201$}&&\multicolumn{4}{c}{\cellcolor{gray!50} Scenario 14 $N=1000$, $d=201$}\\
\cline{2-5}\cline{7-10}
$\muhat_{\mbox{\tiny MCAR}}$&-0.236&0.275&0.665&0.720&&-0.211&0.297&0.956&0.910\\
$\muhat_{\mbox{\tiny SS-Lasso}}$&-0.132&0.199&0.714&0.860&&-0.122&0.259&1.020&0.950\\
$\muhat_{\mbox{\tiny BRSS}}$&-0.084&0.155&0.682&0.990&&-0.134&0.242&0.970&0.960\\
\hline
&\multicolumn{4}{c}{\cellcolor{gray!50} Scenario 15 $N=2000$, $d=201$}&&\multicolumn{4}{c}{\cellcolor{gray!50} Scenario 16 $N=2000$, $d=201$}\\
\cline{2-5}\cline{7-10}
$\muhat_{\mbox{\tiny MCAR}}$&-0.245&0.278&0.661&0.730&&-0.356&0.388&0.832&0.590\\
$\muhat_{\mbox{\tiny SS-Lasso}}$&-0.148&0.202&0.708&0.900&&-0.228&0.321&0.911&0.820\\
$\muhat_{\mbox{\tiny BRSS}}$&-0.087&0.149&0.679&0.990&&-0.159&0.204&0.812&0.980\\
\bottomrule
\end{tabular}
\end{center}
\end{table}

Except for Scenario 3, the MCAR estimator exhibits poor performance with large biases and inadequate coverage. In Scenarios 2, 3, 4, 6, 8, and 14, both the SS-Lasso and BRSS estimators perform well, showing similar biases, RMSEs, and coverages close to 95\%. However, in Scenarios 5, 7, 9, 12, 13, 15, and 16, BRSS outperforms SS-Lasso with smaller biases, RMSEs, and improved coverages. Additionally, in Scenario 1, BRSS achieves a smaller RMSE and a coverage closer to 95\%, while its bias remains comparable to SS-Lasso. Overall, the MCAR estimator's poor performance is attributed to the mischaracterization of the labeling PS function, while the BRSS estimator exhibits greater stability compared to SS-Lasso due to the misspecified OR model.

\section{Discussion}\label{sec:disc}

{This paper studies ATE estimation in SS settings with extreme outcome missingness, where labeling bias may arise. We introduce a {\it `decaying' MAR setting} for causal inference, unifying regular labeling bias and missing outcome scenarios into a broader framework. As explored in Section \ref{subsec:misT} of the \hyperref[supp_mat]{Supplement}, this framework also extends to settings with missing treatment variables. We propose bias-reduced estimators using two distinct nuisance estimation strategies, with the de-coupled approach particularly suited for SS ATE estimation. By integrating IPW with bias-reducing techniques, it enhances robustness of inference and therefore leading to more reliable causal conclusions. Our results underscore the critical role of nuisance estimation design in achieving reliable ATE estimation and inference.}

In addition to the ATE estimation problem, our framework opens up possibilities for studying estimation and inference of other parameters of interest in the decaying MAR setting. One such parameter is the decaying PS function, which serves as a vital intermediate step for ATE estimation. Further research is needed to explore the validity of non-parametric PS estimators in the decaying PS setup. Additionally, alternative methods like inverse probability weighting and residual balancing can be employed in degenerate supervised settings, but their validity and theoretical properties in this context require further investigation for future research.

\section*{Acknowledgement}
This work was supported in part by the National Natural Science Foundation of China (NSFC) grant 12301390 (YZ) and the National Science Foundation grants NSF DMS-2113768 (AC) and NSF DMS-1712481 (JB).


\section*{Supplementary Material}\label{supp_mat}

\noindent \textbf{Supplement to `The Decaying Missing-at-Random Framework: {Model} Doubly Robust Causal Inference with Partially Labeled Data'}.
In the \hyperref[supp_mat]{Supplement}, we present additional theoretical results, discussions, and proofs that support our main findings. Section \ref{subsec:misT} introduces {settings with missing treatments and discusses extensions of our BRSS and DC-BRSS methodology therein,} while Sections \ref{sec:lemmas}-\ref{sec:proof_BRDR-SP} provide detailed proofs for the main results on bias-reduced estimators. {{The general theoretical properties} for the R-DR estimator {of Section \ref{sec:gen-DR-DMAR}} are offered in Section \ref{sec:add-S1}, with corresponding proofs found in Section \ref{sec:proof_sec:gen}.}

\appendix
\bibliographystyle{imsart-nameyear} 
\bibliography{mar}

\renewcommand{\thelemma}{S.\arabic{lemma}}
\renewcommand{\thetable}{S.\arabic{table}}


\clearpage\newpage  
\par\bigskip
\begin{center}
\textbf{\uppercase{Supplement to `The Decaying Missing-at-Random Framework: Model Doubly Robust Causal Inference with Partially Labeled Data'}}
\end{center}

\par\medskip
This supplementary document contains further discussions, additional theoretical results, and proofs of the main results that could not be accommodated in the main paper. All results and notations are numbered and used as in the main text unless stated otherwise.

\begin{table}[b]
\caption{Table of notations (we let $j\in\{0,1\}$ and $k\in\{1,2\}$, wherever applicable, in the following table)} \label{table:notations}
\begin{tabular}{| l | l |}
\hline
Notation & Description\\
\hline
$\bX_i,\bX$ & The vectors of covariates\\
$\varphi(\cdot)$ & The feature map\\
$\bS_i=\varphi(\bX_i),\bS=\varphi(\bX)$ & The feature vectors \\
$T_i,T$ & The treatment indicators\\
$R_i,R$ & The labeling indicators\\
$\Gamma_i=T_iR_i,\Gamma=TR$ & The product indicators\\
$Y_i,Y$ & The {observable} outcome of interest\\
$Y_i(j),Y(j)$ & The potential outcomes \\
$N, n=\sum_{i=1}^NR_i$ & The total and labeled sample sizes\\
$p,d$ & The dimensions of covariates $\bX$ and feature vector $\varphi(\bX)$\\
$\K$ & Number of folds\\
$M=N/\K$ & Number of samples in each fold\\
$p_N$ & The labeling probability: $\P(R=1\mid T=1)$\\
$\gamma_N$ & The product {indicator's} probability: $\P(\Gamma=1)$\\
$a_{N}$ & The inverse of the inverse product PS function's expectation\\
$Na_N$ & The `effective sample size'\\
$\theta_j$ & The counterfactual means $\theta_j:=\E\{Y(j)\}$ \\
$\mu_0$ & The ATE parameter $\mu_0:=\theta_1-\theta_0$\\
$m(\cdot)$ & The true OR function\\
$\pi(\cdot),p_N(\cdot),\gamma_N(\cdot)$ & The true treatment, labeling, and product PS functions\\
$\mhat^{(-k)}(\cdot),\gammahat_N^{(-k)}(\cdot)$ & The cross-fitted OR and product PS estimators\\
$m^*(\cdot),\gamma_N^*(\cdot)$ & The limiting OR and product PS functions\\
$\thetahat_{j,\mbox{\tiny R-DR}}$ & The R-DR estimator for the counterfactual mean\\
$\muhat_{\mbox{\tiny R-DR}}$ & The R-DR estimator for the ATE\\
$g(\cdot)$ & The logistic function\\
$\balpha^*,\bbeta^*$ & The nuisance parameters for the BRSS estimator\\
$\balphatil^*,\bbeta_{p}^*$ & The nuisance parameters for the DC-BRSS estimator\\
$s_{\balpha},s_{\bbeta},s_{\balphatil},s_{p}$ & The sparsity levels of $\balpha^*,\bbeta^*,\balphatil^*,\bbeta_{p}^*$\\
$\gammahat_N^{(k)},\gammahat_N:=\gammahat_N^{(1)},\gammabar_N$ & The estimates of $\gamma_N$\\
$\phat_{N}^{(k)},\phat_{N}:=\phat_{N}^{(1)},\pbar_{N}$ & The estimates of $p_{N}$\\
$\balphahat^{(k)},\balphahat:=\balphahat^{(1)}$ & The OR estimators for the BRSS approach\\
$\bbetahat^{(k)},\bbetahat:=\bbetahat^{(1)}$ & The PS estimators for the BRSS approach\\
$\balphatil^{(k)},\balphatil:=\balphatil^{(1)}$ & The OR estimators for the DC-BRSS approach\\
$\bbetahat_{p}^{(k)},\bbetahat_{p}:=\bbetahat_{p}^{(1)}$ & The PS estimators for the DC-BRSS approach\\
$\thetahat_{\mbox{\tiny j,BRSS}}^{(k)},\thetahat_{\mbox{\tiny j,BRSS}},\thetahat_{\mbox{\tiny j,DC-BRSS}}^{(k)},\thetahat_{\mbox{\tiny j,DC-BRSS}}$&The bias-reduced SS estimators for $\theta_j$\\
$\muhat_{\mbox{\tiny BRSS}},\muhat_{\mbox{\tiny DC-BRSS}}$&The bias-reduced SS estimators for the ATE\\
\hline
\end{tabular}
\end{table}

\paragraph*{Organization}

{The structure of the document is as follows. In Section \ref{subsec:misT}, we introduce {settings with missing treatments and discusses extensions of our BRSS and DC-BRSS methodology therein,} Before establishing the results for bias-reduced estimators in Section \ref{sec:MT}, we outline several key preliminary lemmas in Section \ref{sec:lemmas}. Since the DC-BRSS estimator extends the BRSS estimator, we focus primarily on the de-coupled approach, detailing the theoretical properties of the PS and OR nuisance estimators in Section \ref{sec:proof_nuisance_SP}. Using these nuisance estimation results, we establish the asymptotic properties of the DC-BRSS estimator in Section \ref{sec:proof_BRDR-SP}. {The general theoretical properties for the R-DR estimator of Section \ref{sec:gen-DR-DMAR} are provided in Section \ref{sec:add-S1},} with corresponding proofs in Section \ref{sec:proof_sec:gen}. The key notations used throughout the main paper and the supplement are summarized in Table \ref{table:notations}.}

\section{Missing treatment settings}\label{subsec:misT}

{}
In the main paper, we focus on semi-supervised (SS) settings where the outcome \( Y \) is subject to missingness. In this section, we extend our analysis to settings where the treatment variable \( T \) may also be missing and consider the following four scenarios:

\begin{itemize}
\item\textbf{Setting a (missing outcome).} Observe \( (R,T,RY,\bX) \), where \( R=1 \) indicates that \( Y \) is observed.
\item\textbf{Setting b (missing treatment).} Observe \( (R,RT,Y,\bX) \), where \( R=1 \) indicates that \( T \) is observed.
\item\textbf{Setting c (simultaneously missing treatment and outcome).} Observe \( (R,RT,RY,\bX) \), where \( Y \) and \( T \) are jointly missing, and \( R=1 \) indicates that both are observed.
\item\textbf{Setting d (non-simultaneously missing treatment and outcome).} Observe \( (R_T,R_Y,\) \(R_TT,R_YY,\bX) \), where \( R_Y=1 \) if \( Y \) is observed, \( R_T=1 \) if \( T \) is observed, and \( R:=R_TR_Y=1 \) if both are observed.
\end{itemize}

Setting a above corresponds to the SS setup discussed in the main {paper.} Since Settings a-c are special cases of Setting d, we now concentrate on Setting d and provide a unified approach for identifying $\theta_1=\E\{Y(1)\}$ across all these settings in the following corollary.

\begin{corollary}\label{thm:DR_d}
Let Assumption \ref{cond:basic} hold with $R:=R_TR_Y$.
Define $\Gamma:=TR=TR_TR_Y$, $\gamma_N(\bx):=\P(\Gamma=1\mid\bX=\bx)$, and $m(\bx):=\E\{Y(1)\mid\bX\}$. Then, $m(\bx)=\E(Y\mid\bX=\bx,\Gamma=1)$. Additionally, for any arbitrary functions $m^*(\cdot)$ and $\gamma_N^*(\cdot)$, we have
\[\E\left\{\psi^*_{N} (\bZ)\right\} =0,\;\;\mbox{with}\;\;\psi^*_{N}(\bZ)=m^*(\bX) - \theta_1+\frac{\Gamma}{\gamma_N^*(\bX)} \left\{Y - m^*(\bX)\right\},\]
as long as either $m^*(\cdot)=m(\cdot)$ or $\gamma_N^*(\cdot)=\gamma_N(\cdot)$.
\end{corollary}

Corollary \ref{thm:DR_d} assumes an R-MAR condition, \( R_T R_Y \perp Y \mid (T, \bX) \). Unlike prior studies that impose restrictive monotonicity conditions \citep{manski1997monotone,manski2000monotone,molinari2010missing,mebane2013causal}, our results avoid such restrictions between treatment and potential outcomes. Additionally, \cite{zhang2016causal} consider another MAR condition, $R\independent T\mid(\bX,Y)$. However, such an approach is generally inappropriate as $Y$ is typically evaluated after the treatment assignment.

{Under Setting d (and thus Settings b and c), the problem can be framed as a {\it three-occasion} missing data problem, where \( \Gamma = TR_T R_Y \) serves as the \emph{effective labeling indicator}. Redefining \( R = R_T R_Y \) as the labeling indicator for simultaneously observing \( T \) and \( Y \), the R-DR and BRSS estimators from Sections \ref{sec:gen-DR-DMAR} and \ref{subsec:nuisance} remain valid, and Theorems \ref{thm:gen-rate-DR} and \ref{cor:para} continue to hold under this updated definition. Therefore, as discussed in Remark \ref{remark:MDR}, the BRSS estimator offers enhanced robustness in both model and sparsity structures compared to the preliminary R-DR estimator under the missing treatment setup.

When applying the de-coupled approach from Section \ref{subsec:SP}, the additional missingness in \( T \) necessitates a modified factorization of the product propensity score function as below:
\[\gamma_N(\bx) ~=~ \P(TR_T = 1 \mid \bX = \bx) \cdot \P(R_Y = 1 \mid TR_T = 1, \bX = \bx).\]
The DC-BRSS estimator remains valid after replacing \( (T, R) \) with \( (TR_T, R_Y) \), and the results in Theorem \ref{t4-SP} still hold under this modification. However, its performance now also depends on the missingness mechanism of \( T \). If only a small or moderate proportion of \( T \) is missing, the groups \( TR_T = 1 \) and \( TR_T = 0 \) remain relatively balanced, enabling accurate estimation of \( \pi(\bx) := \P(TR_T = 1 \mid \bX = \bx) \) when \( N \gg n \). When this initial estimation is accurate, the DC-BRSS estimator offers improved robustness, as highlighted in Remark \ref{remark:MDR1}. In contrast, if \( T \) suffers from extreme missingness, leading to severe imbalance between these groups, estimating \( \pi(\bx) \) becomes challenging, as its convergence rate depends on the size of the smaller group. In such cases, the BRSS method from Section \ref{subsec:nuisance} may be a more suitable alternative.}

\section{Preliminary lemmas}\label{sec:lemmas}
We first provide some preliminary lemmas before we analyze the biased reduced SS estimators and the considered nuisance estimators.

\begin{lemma}\label{lemma:subG}
The following are some useful properties regarding the $\psi_\alpha$-norms.

(a) If $|X|\leq|Y|$ almost surely, then $\|X\|_{\psi_2}\leq\|Y\|_{\psi_2}$. If $|X|\leq C$ almost surely for some constant $C>0$, then $\|X\|_{\psi_2}\leq\{\log(2)\}^{-1/2}C$.

(b) If $\|X\|_{\psi_2}\leq\sigma$, then $\P(|X|>t)\leq2\exp(-t^2/\sigma^2)$ for all $t\geq0$.

(c) If $\|X\|_{\psi_\alpha}\leq\sigma$ for some $(\alpha,\sigma)>0$, then $\E(|X|^m)\leq C_{\alpha}^m\sigma^mm^{m/\alpha}$ for all $m\geq1$, for some constant $C_\alpha$ depending only on $\alpha$. In particular, if $\|X\|_{\psi_2}\leq\sigma$, $\E(|X|^m)\leq2\sigma^m\Gamma(m/2+1)$, for all $m\geq1$, where $\Gamma(a):=\int_0^\infty x^{a-1}\exp(-x)dx$ denotes the Gamma function. Hence, $\E(|X|)\leq\sigma\sqrt\pi$ and $\E(|X|^m)\leq2\sigma^m(m/2)^{m/2}$ for $m\geq2$. 

(d) For any $\alpha,\beta>0$, let $\gamma:=(\alpha^{-1}+\beta^{-1})^{-1}$. Then, for any $X,Y$ with $\|X\|_{\psi_\alpha}<\infty$ and $\|Y\|_{\psi_\beta}<\infty$, $\|XY\|_{\psi_\gamma}<\infty$ and $\|XY\|_{\psi_\gamma}<\|X\|_{\psi_\alpha}\|Y\|_{\psi_\beta}$.

(e) Let $\bX\in\R^p$ be a random vector with $\sup_{1\leq j\leq p}\|\bX(j)\|_{\psi_\alpha}\leq\sigma$. Then, $\|\|\bX\|_\infty\|_{\psi_\alpha}\leq\sigma\{\log{p}+2\}^{1/\alpha}$, where we denote $\bX(j)$ as the $j$-th element of a vector $\bX$.
\end{lemma}
Lemma \ref{lemma:subG} follows from Lemma D.1 of \cite{chakrabortty2019high}.

\begin{lemma}\label{integrating}
If $X\in\R$ is a random variable and there exists constants $a_1,a_2,a_3,a_4>0$, such that
$$\P(|X|>a_1u^2+a_2u+a_3)\leq a_4\exp(-u^2),\quad\forall u>0.$$
Then,
$$\E(|X|)\leq a_3+a_4(4a_1+\sqrt\pi a_2).$$
\end{lemma}

\begin{proof}[Proof of Lemma \ref{integrating}]
Observe that
\begin{align*}
\E(|X|)&=\int_0^\infty\P(|X|>t)dt\\
&=\int_0^{a_3}\P(|X|>t)dt+\int_{a_3}^\infty\P(|X|>t)dt\\
&\leq a_3+\int_{0}^\infty\P(|X|>t+a_3)dt.
\end{align*}
For any $t>0$, let $u>0$ satisfies $a_1u^2+a_2u=t$. Then,
$$u=\frac{\sqrt{a_2^2+4a_1t}-a_2}{2a_1}=\frac{2t}{\sqrt{a_2^2+4a_1t}+a_2}\geq\frac{2t}{\sqrt{4a_1t}+2a_2}.$$
Hence,
\begin{align*}
&\int_0^\infty\P(|X|>t+a_3)dt=\int_0^\infty\P(|X|>a_1u^2+a_2u+a_3)dt\leq\int_0^\infty a_4\exp(-u^2)dt\\
&\qquad\leq a_4\int_0^\infty\exp\left\{-\frac{4t^2}{(\sqrt{4a_1t}+2a_2)^2}\right\}\leq a_4\int_0^\infty\exp\left(-\frac{t^2}{2a_1t+2a_2^2}\right)\\
&\qquad=a_4\int_0^{a_2^2/a_1}\exp\left(-\frac{t^2}{2a_1t+2a_2^2}\right)+a_4\int_{a_2^2/a_1}^\infty\exp\left(-\frac{t^2}{2a_1t+2a_2^2}\right)\\
&\qquad\leq a_4\int_0^{a_2^2/a_1}\exp\left(-\frac{t^2}{4a_2^2}\right)+a_4\int_{a_2^2/a_1}^\infty\exp\left(-\frac{t}{4a_1}\right).
\end{align*}
Notice that
$$\int_0^{a_2^2/a_1}\exp\left(-\frac{t^2}{4a_2^2}\right)\leq2\sqrt\pi a_2\frac{1}{2}\int_{-\infty}^\infty\frac{1}{\sqrt{2\pi}\sqrt2a_2}\exp\left(-\frac{t^2}{4a_2^2}\right)=\sqrt\pi a_2$$
and
$$\int_{a_2^2/a_1}^\infty\exp\left(-\frac{t}{4a_1}\right)=4a_1\exp\left(-\frac{a_2^2}{4a_1^2}\right)\leq4a_1.$$
Therefore,
$$\E(|X|)\leq a_3+a_4(4a_1+\sqrt\pi a_2).$$
\end{proof}

We establish some empirical process results in Lemmas \ref{lemma:convex}, \ref{lemma:l1-l2}, and \ref{lemma:emp}.

\begin{lemma}\label{lemma:convex}
Let $\bOmega$ be a $d$ by $d$ matrix, $k_0>0$, and $s\geq1$. Then, for any $\delta\in(0,1)$, there exists some $k\asymp s$ such that
\begin{align}
\sup_{\bDelta_1,\bDelta_2\in\mathcal C(s,k_0)\cap\S^{d-1}}|\bDelta_1^T\bOmega\bDelta_2|&\leq(1-\delta)^{-1}\sup_{\bDelta_1,\bDelta_2\in\Theta_k}|\bDelta_1^T\bOmega\bDelta_2|,\label{bound:Thetak}
\end{align}
where $\S^{d-1}:=\{\bDelta\in\R^d:\|\bDelta\|_2=1\}$, $\mathcal C(s ,k_0):=\{\bDelta\in\R^d:\exists S\subset\{1,\dots,d\},\;|S|\leq s ,\;\text{s.t.}\;\|\bDelta_{S^c}\|_1\leq k_0\|\bDelta_S\|_1\}$, and $\Theta_k:=\{\bDelta\in\R^d:\|\bDelta\|_0\leq k,\|\bDelta\|_2=1\}$.
\end{lemma}
Lemma \ref{lemma:convex} holds by repeating the proof of Lemma 16 of \cite{bradic2019sparsity}.

\begin{lemma}\label{lemma:l1-l2}
Let $\bOmega$ be a $d$ by $d$ matrix. Then, for any $\bDelta\in\R^d$ and $k_0>0$,
\begin{align}
|\bDelta^T\bOmega\bDelta|\leq\inf_{s\geq1}\left\{\left(\frac{6\|\bDelta\|_1^2}{k_0^2s}+4\|\bDelta\|_2^2\right)\sup_{\bdelta\in\mathcal C(s,k_0)\cap\S^{d-1}}|\bdelta^T\bOmega\bdelta|\right\}.\label{bound:DeltaRd}
\end{align}
\end{lemma}

\begin{proof}[Proof of Lemma \ref{lemma:l1-l2}]
For any $s\leq d$, choose some $S\subset\{1,2,\dots,d\}$ satisfying $|S|=s$. For any $\bDelta\in\R^d$, define $\bdeltatil=(\bdeltatil_S^T,\bdeltatil_{S^c}^T)^T\in\R^d$ as
\begin{align}\label{def:deltatil}
\bdeltatil_{S}=(k_0s)^{-1}\|\bDelta\|_1(1,\dots,1)^T\in\R^{s},\quad\bdeltatil_{S^c}=\bDelta_{S^c}\in\R^{d-s}.
\end{align}
Then, $\|\bdeltatil_{S}\|_1=\|\bDelta\|_1/k_0$ and $\|\bdeltatil_{S}\|_2=\|\bDelta\|_1/(k_0\sqrt s)$. Hence,
$$\|\bdeltatil_{S^c}\|_1=\|\bDelta_{S^c}\|_1\leq\|\bDelta\|_1=k_0\|\bdeltatil_{S}\|_1.$$
That is, $\bdeltatil\in\mathcal C(s,k_0)$. Also, note that $\bDelta-\bdeltatil\in\mathcal C(s,k_0)$ since $(\bDelta-\bdeltatil)_{S^c}=\bzero$. Hence,
\begin{align*}
&|\bDelta^T\bOmega\bDelta|\leq|\bdeltatil^T\bOmega\bdeltatil|+|(\bDelta-\bdeltatil)^T\bOmega(\bDelta-\bdeltatil)|+2|\bdeltatil^T\bOmega(\bDelta-\bdeltatil)|\\
&\qquad\leq(\|\bdeltatil\|_2^2+\|\bDelta-\bdeltatil\|_2^2+2\|\bdeltatil\|_2\|\bDelta-\bdeltatil\|_2)\sup_{\bdelta\in\mathcal C(s,k_0)\cap\S^{d-1}}\bdelta^T\bOmega\bdelta\\
&\qquad\leq\left(\frac{6\|\bDelta\|_1^2}{k_0^2s}+4\|\bDelta\|_2^2\right)\sup_{\bdelta\in\mathcal C(s,k_0)\cap\S^{d-1}}\sum_{i\in\mathcal I_2}(\bU_i^T\bdelta)^2v_i,
\end{align*}
since
\begin{align*}
&\|\bdeltatil\|_2^2+\|\bDelta-\bdeltatil\|_2^2+2\|\bdeltatil\|_2\|\bDelta-\bdeltatil\|_2\leq2\|\bdeltatil\|_2^2+2\|\bDelta-\bdeltatil\|_2^2\\
&\qquad=2\|\bdeltatil_{S}\|_2^2+2\|\bdeltatil_{S^c}\|_2^2+2\|\bDelta_S-\bdeltatil_{S}\|_2^2\leq2\|\bdeltatil_{S}\|_2^2+2\|\bDelta_{S^c}\|_2^2+4\|\bDelta_S\|_2^2+4\|\bdeltatil_{S}\|_2^2\\
&\qquad\leq6\|\bdeltatil_{S}\|_2^2+4\|\bDelta\|_2^2=6\|\bDelta\|_1^2/(k_0^2s)+4\|\bDelta\|_2^2.
\end{align*}

When $s>d$, since $\bDelta\in\R^d=\mathcal C(d,k_0)=\mathcal C(s,k_0)$, we also have
\begin{align*}
|\bDelta^T\bOmega\bDelta|\leq\|\bDelta\|_2^2\sup_{\bdelta\in\mathcal C(s,k_0)\cap\S^{d-1}}|\bdelta^T\bOmega\bdelta|.
\end{align*}
To sum up, we conclude that \eqref{bound:DeltaRd} holds.
\end{proof}

\begin{lemma}\label{lemma:emp}
Let Assumptions \ref{cond:subG} and \ref{cond:tail} hold. Let $j\in\{1,2\}$.

\par\smallskip
(a) For any $s=o(M/\log{d})$, we have
\begin{align*}
\sup_{\bDelta\in\R^d\setminus\{\bzero\}}\frac{M^{-1}\sum_{i\in\mathcal I_j}(\bS_i^T\bDelta)^2}{\|\bDelta\|_1^2/s+\|\bDelta\|_2^2}=O_p(1).
\end{align*}

\par\smallskip
(b) Let $s_0:=\lceil N\gamma_N/\{\log{d}\log{N}\}\rceil$, then
\begin{align}\label{bound:Deltahat}
\sup_{\bDelta\in\R^d\setminus\{\bzero\}}\frac{(M\gamma_N)^{-1}\sum_{i\in\mathcal I_j}\Gamma_i(\bS_i^T\bDelta)^2}{\|\bDelta\|_1^2/s_0+\|\bDelta\|_2^2}\leq c\left\{1+\sqrt\frac{t}{N\gamma_N}+\frac{t\log{N}}{N\gamma_N}\right\}
\end{align}
with probability at least $1-3\exp(-t)$ and some constant $c>0$.

\par\smallskip
(c) Let $f(\cdot):\mathcal X\mapsto\R$ and $|f(\bx)|\leq c$ for all $\bx\in\mathcal X$ with some constant $c>0$. Then, for any $s\geq1$, as $N,d\to\infty$,
\begin{align*}
&\sup_{\bDelta_1,\bDelta_2\in\R^d\setminus\{\bzero\}}\frac{(M\gamma_N)^{-1}\left|\bDelta_1^T\left[\sum_{i\in\mathcal I_1}\Gamma_if(\bX_i)\bS_i\bS_i^T-\E\{\Gamma f(\bX)\bS\bS^T\}\right]\bDelta_2\right|}{\|\bDelta_1\|_2\|\bDelta_2\|_2\{\|\bDelta_1\|_1^2/(s\|\bDelta_1\|_2^2)+\|\bDelta_2\|_1^2/(s\|\bDelta_2\|_2^2)+1\}}\\
&\qquad=O_p\left(\sqrt\frac{s\log{d}}{N\gamma_N}+\frac{s\log{d}\log{N}}{N\gamma_N}\right).
\end{align*}
(d) Let $r_0\geq2$ be a constant. Then, for any $s\geq1$,
\begin{align*}
\sup_{\bDelta\in\R^d\setminus\{\bzero\}}\frac{(M\gamma_N)^{-1}\sum_{i\in\mathcal I_j}\Gamma_i|\bS_i^T\bDelta|^{r_0}}{\|\bDelta\|_1^{r_0}/s^{\frac{r_0}{2}}+\|\bDelta\|_2^{r_0}}\leq c\left(1+\sqrt\frac{s\log{d}}{M\gamma_N}+\frac{(s\log{M}\log{d})^{r_0/2}}{M\gamma_N}\right)
\end{align*}
with probability at least $1-2\exp(-t)$ and some constant $c>0$. Hence, as $N,d\to\infty$,
\begin{align*}
\sup_{\bDelta\in\R^d\setminus\{\bzero\}}\frac{(M\gamma_N)^{-1}\sum_{i\in\mathcal I_j}\Gamma_i|\bS_i^T\bDelta|^{r_0}}{\|\bDelta\|_1^{r_0}/s^{\frac{r_0}{2}}+\|\bDelta\|_2^{r_0}}=O_p\left(1+\sqrt\frac{s\log{d}}{M\gamma_N}+\frac{(s\log{M}\log{d})^{r_0/2}}{M\gamma_N}\right).
\end{align*}

\par\smallskip
(e) Let $(\epsilon_i)_{i\in\mathcal I_j}\in\R$ be i.i.d. random variables satisfying $\E(|\epsilon_i|\mid\bX_i,\Gamma_i=1)=\E(|\epsilon_i|\mid\bX_i)$ and $\|\epsilon\|_{\psi_\alpha}\leq\sigma_{\epsilon}$ with some constants $\alpha,\epsilon>0$. Then, for any $s\geq1$,
\begin{align*}
&(M\gamma_N)^{-1}\sum_{i\in\mathcal I_j}\Gamma_i|\epsilon_i|(\bS_i^T\bDelta)^2\\
&\;\;\leq c\|\bDelta\|_2^2+c\left[\sqrt\frac{t+s\log{d}}{N\gamma_N}+\frac{\log^{1+1/\alpha}{M}(t+s\log{d})^{1+1/\alpha}}{N\gamma_N}\right]\left(\frac{\|\bDelta\|_1^2}{s}+\|\bDelta\|_2^2\right),
\end{align*}
uniformly for all $\bDelta\in\R^d$ with probability at least $1-3\exp(-t)$ and some constant $c>0$.

(f) Let $(\epsilon_i)_{i\in\mathcal I_j}\in\R$ be i.i.d. random variables satisfying $\E\{\epsilon_i^2\mid(\Gamma_i,\bX_i)\}<C$. Then,
\begin{align*}
\sup_{\bDelta\in\R^d\setminus\{\bzero\}}\frac{(M\gamma_N)^{-1}\sum_{i\in\mathcal I_j}\Gamma_i\epsilon_i^2(\bS_i^T\bDelta)^2}{\|\bDelta\|_1^2/s_0+\|\bDelta\|_2^2}\leq c\left\{1+\sqrt\frac{t}{N\gamma_N}+\frac{t\log{N}}{N\gamma_N}\right\}
\end{align*}
with probability at least $1-3\exp(-t)-t^{-1}$ and some constant $c>0$, where $s_0$ is defined as in part (b).
\end{lemma}

\begin{proof}[Proof of Lemma \ref{lemma:emp}]
(a) By Theorem 15 of \cite{rudelson2012reconstruction}, for any $s=o(M/\log{d})$, we have
\begin{align*}
\sup_{\bDelta\in\mathcal C(s,3)\cap\S^{d-1}}\frac{\sum_{i\in\mathcal I_j}(\bS_i^T\bDelta)^2}{M\E(\bS^T\bDelta)^2}=O_p(1),
\end{align*}
where $\sup_{\bDelta\in\mathcal C(s,3)\cap\S^{d-1}}\E(\bS^T\bDelta)^2=O(1)$ under Assumption \ref{cond:subG}. Therefore,
\begin{align*}\sup_{\bDelta\in\mathcal C(s,3)\cap\S^{d-1}}M^{-1}\sum_{i\in\mathcal I_j}(\bS_i^T\bDelta)^2=O_p(1).
\end{align*}
Together with Lemma \ref{lemma:l1-l2},
\begin{align*}
\sup_{\bDelta\in\R^d\setminus\{\bzero\}}\frac{M^{-1}\sum_{i\in\mathcal I_j}(\bS_i^T\bDelta)^2}{\|\bDelta\|_1^2/s+\|\bDelta\|_2^2}=O_p(1).
\end{align*}

(b) Fix $\delta\in(0,1)$ and let $\bOmega=M^{-1}\sum_{i\in\mathcal I_j}\Gamma_i\bS_i\bS_i^T-\E(\Gamma\bS\bS^T)$. For any $s\geq1$, by Lemma \ref{lemma:convex}, \eqref{bound:Thetak} holds with some $k\asymp s$. Since $\bS$ is sub-Gaussian, under Assumption \ref{cond:tail}, we have
$$\sup_{\bDelta\in\Theta_k}\Var\left\{(\Gamma\bS^T\bDelta)^2\right\}\leq\sup_{\bDelta\in\Theta_k}\E\left\{\Gamma(\bS^T\bDelta)^4\right\}=\sup_{\bDelta\in\Theta_k}\E\left\{\gamma_N(\bX)(\bS^T\bDelta)^4\right\}=O(\gamma_N),$$
where $\Theta_k=\{\bDelta\in\R^d:\|\bDelta\|_0\leq k,\|\bDelta\|_2=1\}$. By Theorem 4.3 of \cite{kuchibhotla2022moving},
\begin{align*}
\P_{\mathcal D_N'}\left(\sup_{\bDelta\in\Theta_k}\left|\bDelta^T\bOmega\bDelta\right|>c\left[\sqrt\frac{\gamma_N\{t+k\log{d}\}}{M}+\frac{\log{N}\{t+k\log{d}\}}{M}\right]\right)\geq3\exp(-t),
\end{align*}
with some constant $c>0$. Together with \eqref{bound:Thetak}, Lemma \ref{lemma:l1-l2} and note that $k\asymp s$ and $M\asymp N$, we have for any $s\geq1$,
\begin{align}
\sup_{\bDelta\in\R^d}\frac{\left|\bDelta^T\bOmega\bDelta\right|}{\|\bDelta\|_1^2/s+\|\bDelta\|_2^2}\leq c\left[\sqrt\frac{\gamma_N(t+s\log{d})}{N}+\frac{\log{N}(t+s\log{d})}{N}\right]\label{bound:Csk0-Omega}
\end{align}
with probability at least $1-3\exp(-t)$ and some constant $c>0$. Note that
\begin{align*}
\gamma_N^{-1}\E\{\Gamma(\bS^T\bDelta)^2\}\leq\gamma_N^{-1}\|\gamma_N(\bX)\|_{\P,q}\|\bS^T\bDelta\|_{\P,2r}^2=O(1).
\end{align*}
Hence, we conclude that uniformly for all $\bDelta\in\R^d$,
\begin{align*}
\frac{\sum_{i\in\mathcal I_j}\Gamma_i(\bS_i^T\bDelta)^2}{M\gamma_N}\leq c\left[1+\sqrt\frac{t+s\log{d}}{N\gamma_N}+\frac{\log{N}(t+s\log{d})}{N\gamma_N}\right]\left(\frac{\|\bDelta\|_1^2}{s}+\|\bDelta\|_2^2\right)
\end{align*}
with probability at least $1-3\exp(-t)$ and some constant $c>0$. Choose $s=s_0:=\lceil N\gamma_N/\{\log{N}\log{d}\}\rceil$, then it follows that \eqref{bound:Deltahat} holds with probability at least $1-3\exp(-t)$ and some constant $c>0$.

(c) Fix $\delta\in(0,1)$ and let $\bOmega=M^{-1}\sum_{i\in\mathcal I_j}\Gamma_if(\bX_i)\bS_i\bS_i^T-\E\{\Gamma f(\bX)\bS\bS^T\}$. For any $s\geq1$, by Lemma \ref{lemma:convex}, \eqref{bound:Thetak} holds with some $k\asymp s$. Since $\bX$ is sub-Gaussian and the function $f(\cdot)$ is bounded, under Assumption \ref{cond:tail}, we also have
\begin{align*}
&\sup_{\bDelta\in\Theta_k}\Var\left[\{\Gamma f(\bX)\bS^T\bDelta)^2\}\right]\leq\sup_{\bDelta\in\Theta_k}\E\left\{\Gamma f(\bX)(\bS^T\bDelta)^4\right\}\\
&\qquad=\sup_{\bDelta\in\Theta_k}\E\left\{\gamma_N(\bX)f^2(\bX)(\bS^T\bDelta)^4\right\}=O(\gamma_N).
\end{align*}
By Theorem 4.3 of \cite{kuchibhotla2022moving},
\begin{align*}
\P_{\mathcal D_N'}\left(\sup_{\bDelta\in\Theta_k}\left|\bDelta^T\bOmega\bDelta\right|>c\left[\sqrt\frac{\gamma_N\{t+k\log{d}\}}{M}+\frac{\log{N}\{t+k\log{d}\}}{M}\right]\right)\geq3\exp(-t),
\end{align*}
with some constant $c>0$. Together with \eqref{bound:Thetak}, Lemma \ref{lemma:l1-l2} and note that $k\asymp s$ and $M\asymp N$, we have for any $s\geq1$,
\begin{align*}
\sup_{\bDelta\in\R^d}\frac{\left|\bDelta^T\bOmega\bDelta\right|}{\|\bDelta\|_1^2/s+\|\bDelta\|_2^2}=O_p\left(\sqrt\frac{\gamma_N(t+s\log{d})}{N}+\frac{\log{N}(t+s\log{d})}{N}\right).
\end{align*}
For any $\bDelta_1,\bDelta_2\in\R^d\setminus\{\bzero\}$, note that
\begin{align*}
\left\|\frac{\bDelta_1}{\|\bDelta_1\|_2}+\frac{\bDelta_2}{\|\bDelta_2\|_2}\right\|_1&\leq\frac{\|\bDelta_1\|_1}{\|\bDelta_1\|_2}+\frac{\|\bDelta_2\|_1}{\|\bDelta_2\|_2},\\
\left\|\frac{\bDelta_1}{\|\bDelta_1\|_2}-\frac{\bDelta_2}{\|\bDelta_2\|_2}\right\|_1&\leq\frac{\|\bDelta_1\|_1}{\|\bDelta_1\|_2}+\frac{\|\bDelta_2\|_1}{\|\bDelta_2\|_2},\\
\left\|\frac{\bDelta_1}{\|\bDelta_1\|_2}+\frac{\bDelta_2}{\|\bDelta_2\|_2}\right\|_2&\leq\frac{\|\bDelta_1\|_2}{\|\bDelta_1\|_2}+\frac{\|\bDelta_2\|_2}{\|\bDelta_2\|_2}=2,\\
\left\|\frac{\bDelta_1}{\|\bDelta_1\|_2}-\frac{\bDelta_2}{\|\bDelta_2\|_2}\right\|_2&\leq\frac{\|\bDelta_1\|_2}{\|\bDelta_1\|_2}+\frac{\|\bDelta_2\|_2}{\|\bDelta_2\|_2}=2.
\end{align*}
Hence, uniformly for all $\bDelta_1,\bDelta_2\in\R^d$,
\begin{align*}
&\gamma_N^{-1}|\bDelta_1^T\bOmega\bDelta_2|=\gamma_N^{-1}\|\bDelta_1\|_2\|\bDelta_2\|_2\left|\frac{\bDelta_1^T}{\|\bDelta_1\|_2}\bOmega\frac{\bDelta_2}{\|\bDelta_2\|_2}\right|\\
&\quad\leq(2\gamma_N)^{-1}\|\bDelta_1\|_2\|\bDelta_2\|_2\left|\left(\frac{\bDelta_1}{\|\bDelta_1\|_2}+\frac{\bDelta_2}{\|\bDelta_2\|_2}\right)^T\bOmega\left(\frac{\bDelta_1}{\|\bDelta_1\|_2}+\frac{\bDelta_2}{\|\bDelta_2\|_2}\right)\right|\\
&\quad\quad+(2\gamma_N)^{-1}\|\bDelta_1\|_2\|\bDelta_2\|_2\left|\left(\frac{\bDelta_1}{\|\bDelta_1\|_2}-\frac{\bDelta_2}{\|\bDelta_2\|_2}\right)^T\bOmega\left(\frac{\bDelta_1}{\|\bDelta_1\|_2}-\frac{\bDelta_2}{\|\bDelta_2\|_2}\right)\right|\\
&\quad=O_p\left(\left(\sqrt\frac{s\log{d}}{N\gamma_N}+\frac{s\log{d}\log{N}}{N\gamma_N}\right)\|\bDelta_1\|_2\|\bDelta_2\|_2\left(\frac{\|\bDelta_1\|_1^2}{s\|\bDelta_1\|_2^2}+\frac{\|\bDelta_2\|_1^2}{s\|\bDelta_2\|_2^2}+1\right)\right).
\end{align*}

(d) For any $r_0\geq2$, we have $\sup_{\bDelta\in\S^{d-1}}\|\Gamma|\bS^T\bDelta|^{r_0}\|_{\psi_{2/r_0}}\leq\sigma^{r_0}$ using part (d) of Lemma \ref{lemma:subG}. By the H\"older inequality and part (c) of Lemma \ref{lemma:subG},
\begin{align}
&\sup_{\bDelta\in\S^{d-1}}\E\{\Gamma|\bS^T\bDelta|^{r_0}\}=\sup_{\bDelta\in\S^{d-1}}\E\{\gamma_N(\bX)|\bS^T\bDelta|^{r_0}\}\nonumber\\
&\qquad\leq\|\gamma_N(\bX)\|_{\P,q}\sup_{\bDelta\in\S^{d-1}}\|\bS^T\bDelta\|_{\P,r_0r}^{r_0}=O(\gamma_N),\label{eq:E4th}
\end{align}
with $r>0$ satisfying $1/r+1/q=1$. Let $W=W(\bDelta)=\Gamma|\bS^T\bDelta|^{r_0}-\E\{\Gamma|\bS^T\bDelta|^{r_0}\}$. By part (d) of Lemma \ref{lemma:subG} and note that $\gamma_N\leq1$, we have $\sup_{\bDelta\in\S^{d-1}}\|W\|_{\psi_{2/r_0}}=O(1)$. Additionally, we have
$$\sup_{\bDelta\in\S^{d-1}}\E(W^2)\leq\sup_{\bDelta\in\S^{d-1}}\E\left\{\Gamma|\bS^T\bDelta|^{2r_0}\right\}\leq\|\gamma(\bX)\|_{\P,q}\sup_{\bDelta\in\S^{d-1}}\|\bS^T\bDelta\|_{\P,2r_0r}^{2r_0}=O(\gamma_N).$$
By Theorem 3.2 and Proposition A.3 of \cite{kuchibhotla2022moving}, for any $\bDelta\in\S^{d-1}$ and $u>0$,
\begin{align}
\P_{\mathcal D_N'}\biggr(&\left|(M\gamma_N)^{-1}\sum_{i\in\mathcal I_j}\Gamma_i|\bS_i^T\bDelta|^{r_0}-\gamma_N^{-1}\E\{\Gamma|\bS^T\bDelta|^{r_0}\}\right|\nonumber\\
&\qquad>c\left[\sqrt\frac{u}{M\gamma_N}+\frac{\{\log{M}\}^{r_0/2}u^{r_0/2}}{M\gamma_N}\right]\biggr)\leq2\exp(-u),\label{eq:En4th}
\end{align}
with some constant $c>0$ independent of $\bDelta$. For any $s\geq1$, repeating Step 2 of proof of Lemma 17 in \cite{bradic2019sparsity} (with $\|\cdot\|_4$ replaced by $\|\cdot\|_{r_0}$), we obtain
\begin{align}\label{eq:bradic17}
\sup_{\bDelta\in\mathcal C(s,3)\cap\S^{d-1}}\left\{\sum_{i\in\mathcal I_j}\Gamma_i(\bS_i^T\bDelta)^{r_0}\right\}^{1/r_0}\leq(1-\delta)^{-1}\max_{\bDelta\in\mathcal T}\left\{\sum_{i\in\mathcal I_j}\Gamma_i(\bS_i^T\bDelta)^{r_0}\right\}^{1/r_0},
\end{align}
with some constant $\delta\in(0,1)$ and a set $\mathcal T\subset\S^{d-1}$ satisfying
$$|\mathcal T|\leq(c\delta^{-1}d)^{cs},$$
where $c>0$ is a constant. By the union bound and \eqref{eq:En4th}, choosing $u=cs\log(c\delta^{-1}d)+t$ with $t>0$, we have
\begin{align}
&\P_{\mathcal D_N'}\biggr(\max_{\bDelta\in\mathcal T}\left|(M\gamma_N)^{-1}\sum_{i\in\mathcal I_j}\Gamma_i|\bS_i^T\bDelta|^{r_0}-\gamma_N^{-1}\E\{\Gamma|\bS^T\bDelta|^{r_0}\}\right|\nonumber\\
&\qquad\qquad>c'\left[\sqrt\frac{s\log{d}+t}{M\gamma_N}+\frac{\{\log{M}\}^{r_0/2}\{s\log{d}+t\}^{r_0/2}}{M\gamma_N}\right]\biggr)\nonumber\\
&\qquad\leq2(c\delta^{-1}d)^{cs}\exp(-u)=2\exp(-t),\label{bound:net}
\end{align}
with some constant $c'>0$. Together with \eqref{eq:E4th} and \eqref{eq:bradic17}, we have
\begin{align}
&\sup_{\bDelta\in\mathcal C(s,3)\cap\S^{d-1}}(M\gamma_N)^{-1}\sum_{i\in\mathcal I_j}\Gamma_i|\bS_i^T\bDelta|^{r_0}\leq c\left(1+\sqrt\frac{s\log{d}}{N\gamma_N}+\frac{\{s\log{N}\log{d}\}^{r_0/2}}{N\gamma_N}\right),\label{rate:sup-r0}
\end{align}
with probability at least $1-2\exp(-t)$ and some constant $c>0$. Now, for any $\bDelta\in\R^d$, define $\bdeltatil$ as in \eqref{def:deltatil} with $k_0=3$. Then, we also have $\|\bdeltatil_{S}\|_1=\|\bDelta\|_1/3$, $\|\bdeltatil_{S}\|_2=\|\bDelta\|_1/(3\sqrt s)$, and $\bdeltatil,\bDelta-\bdeltatil\in\mathcal C(s,3)$. Therefore,
\begin{align}
&\sum_{i\in\mathcal I_j}\Gamma_i|\bS_i^T\bDelta|^{r_0}\leq2^{r_0-1}\left(\sum_{i\in\mathcal I_j}\Gamma_i|\bS_i^T\bdeltatil|^{r_0}+\sum_{i\in\mathcal I_j}\Gamma_i|\bS_i^T(\bDelta-\bdeltatil)|^{r_0}\right)\nonumber\\
&\qquad\leq2^{r_0-1}\left(\|\bdeltatil\|_2^{r_0}+\|\bDelta-\bdeltatil\|_2^{r_0}\right)\sup_{\bdelta\in\mathcal C(s,3)\cap\S^{d-1}}\sum_{i\in\mathcal I_j}\Gamma_i|\bS_i^T\bdelta|^{r_0}.\label{bound:Delta-r0}
\end{align}
Note that
\begin{align*}
&\|\bdeltatil\|_2^{r_0}+\|\bDelta-\bdeltatil\|_2^{r_0}=\|\bdeltatil\|_2^{r_0}+(\|\bDelta_S-\bdeltatil_S\|_2)^{r_0}\\
&\qquad\leq(\|\bdeltatil_S\|_2+\|\bdeltatil_{S^c}\|_2)^{r_0}+(\|\bDelta_S\|_2+\|\bdeltatil_S\|_2)^{r_0}\\
&\qquad=\left\{\|\bDelta\|_1/(3\sqrt s)+\|\bDelta_{S^c}\|_2\right\}^{r_0}+\left\{\|\bDelta_S\|_2+\|\bDelta\|_1/(3\sqrt s)\right\}^{r_0}\\
&\qquad\leq2^{r_0-1}\left\{\|\bDelta\|_1^{r_0}/(3\sqrt s)^{r_0}+\|\bDelta_{S^c}\|_2^{r_0}+\|\bDelta_S\|_2^{r_0}+\|\bDelta\|_1^{r_0}/(3\sqrt s)^{r_0}\right\}\\
&\qquad\leq2^{r_0}\left\{\|\bDelta\|_1^{r_0}/(3\sqrt s)^{r_0}+\|\bDelta\|_2^{r_0}\right\}.
\end{align*}
Together with \eqref{rate:sup-r0} and \eqref{bound:Delta-r0}, we have
\begin{align*}
\sup_{\bDelta\in\R^d\setminus\{\bzero\}}\frac{(M\gamma_N)^{-1}\sum_{i\in\mathcal I_j}\Gamma_i|\bS_i^T\bDelta|^{r_0}}{\|\bDelta\|_1^{r_0}/s^{\frac{r_0}{2}}+\|\bDelta\|_2^{r_0}}\leq c\left(1+\sqrt\frac{s\log{d}}{M\gamma_N}+\frac{(s\log{M}\log{d})^{r_0/2}}{M\gamma_N}\right),
\end{align*}
with probability at least $1-2\exp(-t)$ and some constant $c>0$.

(e) By the tower rule, for any $\bDelta\in\R^d$,
\begin{align}\label{bound:outsampleE}
&\E\{\Gamma|\epsilon|(\bS^T\bDelta)^2\}=\E[\E\{\Gamma|\epsilon|(\bS^T\bDelta)^2\mid\bX\}]=\E\{\gamma_N(\bX)|\epsilon|(\bS^T\bDelta)^2\}\leq c\gamma_N\|\bDelta\|_2^2,
\end{align}
with some constant $c>0$ under Assumption \ref{cond:tail}, as $\bS$ is sub-Gaussian and $\epsilon$ has a bounded $\psi_\alpha$-norm. Let $\bOmega=M^{-1}\sum_{i\in\mathcal I_j}\Gamma_i|\epsilon_i|\bS_i\bS_i^T-\E(\Gamma|\epsilon|\bS\bS^T)$. For any $\delta\in(0,1)$ and $s\geq1$, by Lemma \ref{lemma:convex}, we have \eqref{bound:Thetak} holds with $k_0=3$ and some $k\asymp s$. Similar to \eqref{bound:outsampleE}, we also have
$$\sup_{\bDelta\in\Theta_k}\Var\left\{\Gamma|\epsilon|(\bS^T\bDelta)^2\right\}\leq\sup_{\bDelta\in\Theta_k}\E\left\{\Gamma\epsilon^2(\bS^T\bDelta)^4\right\}=\sup_{\bDelta\in\Theta_k}\E\left\{\gamma_N(\bX)\epsilon^2(\bS^T\bDelta)^4\right\}\leq c\gamma_N,$$
with some constant $c>0$. Note that $|\epsilon|=\max(\epsilon,-\epsilon)$. By part (e) of Markov's inequality, $\||\epsilon|\|_{\psi_\alpha}\leq\sigma_\epsilon\{\log(2)+2\}^{1/\alpha}$. Together with the definition of $\psi_\alpha$-norms, we further have $\|\sqrt{|\epsilon|}\|_{\psi_{2\alpha}}\leq\sqrt{\sigma_\epsilon}\{\log(2)+2\}^{1/(2\alpha)}$. By part (d) of Lemma \ref{lemma:subG}, $\|\sqrt{|\epsilon|}\bS^T\bv\|_{\psi_{\frac{2\alpha}{1+\alpha}}}\leq \sigma\sqrt{\sigma_\epsilon}\{\log(2)+2\}^{1/(2\alpha)}$ for any $\|\bv\|_2\leq1$. By Theorem 4.3 of \cite{kuchibhotla2022moving} and note that $k\asymp s$ and $M\asymp N$, we have
\begin{align}
\sup_{\bDelta\in\Theta_k}\left|\bDelta^T\bOmega\bDelta\right|\leq c\left[\sqrt\frac{\gamma_N(t+s\log{d})}{N}+\frac{\log^{1+1/\alpha}{M}(t+s\log{d})^{1+1/\alpha}}{N}\right]\label{bound:Omega'}
\end{align}
with probability at least $1-3\exp(-t)$ and some constant $c>0$. Together with \eqref{bound:Thetak}, \eqref{bound:outsampleE}, and Lemma \ref{lemma:l1-l2},
\begin{align*}
&(M\gamma_N)^{-1}\sum_{i\in\mathcal I_j}\Gamma_i|\epsilon_i|(\bS_i^T\bDelta)^2\\
&\;\;\leq c\|\bDelta\|_2^2+c\left[\sqrt\frac{t+s\log{d}}{N\gamma_N}+\frac{\log^{1+1/\alpha}{M}(t+s\log{d})^{1+1/\alpha}}{N\gamma_N}\right]\left(\frac{\|\bDelta\|_1^2}{s}+\|\bDelta\|_2^2\right),
\end{align*}
uniformly for all $\bDelta\in\R^d$ with probability at least $1-3\exp(-t)$ and some constant $c>0$.

(f) For any $\bDelta\in\R^d$,
\begin{align*}
&\E_{\mathcal D_N'}\left\{M^{-1}\sum_{i\in\mathcal I_j}\Gamma_i\epsilon_i^2(\bS_i^T\bDelta)^2\mid(\Gamma_i,\bX_i)_{i\in\mathcal I_j}\right\}\\
&\qquad=M^{-1}\sum_{i\in\mathcal I_j}\Gamma_i(\bS_i^T\bDelta)^2\E_{\mathcal D_N'}\left\{\epsilon_i^2\mid(\Gamma_i,\bX_i)_{i\in\mathcal I_j}\right\}\leq CM^{-1}\sum_{i\in\mathcal I_j}\Gamma_i(\bS_i^T\bDelta)^2.
\end{align*}
By Markov's inequality,
\begin{align}
\P_{\mathcal D_N'}\biggr(&M^{-1}\sum_{i\in\mathcal I_j}\Gamma_i\epsilon_i^2(\bS_i^T\bDelta)^2\geq tCM^{-1}\sum_{i\in\mathcal I_j}\Gamma_i(\bS_i^T\bDelta)^2\mid(\Gamma_i,\bX_i)_{i\in\mathcal I_j}\biggr)\leq t^{-1}.\label{bound:prob_eps}
\end{align}
Together with \eqref{bound:Deltahat}, we conclude that
\begin{align*}
\sup_{\bDelta\in\R^d\setminus\{\bzero\}}\frac{(M\gamma_N)^{-1}\sum_{i\in\mathcal I_j}\Gamma_i\epsilon_i^2(\bS_i^T\bDelta)^2}{\|\bDelta\|_1^2/s_0+\|\bDelta\|_2^2}\leq c\left\{1+\sqrt\frac{t}{N\gamma_N}+\frac{t\log{N}}{N\gamma_N}\right\},
\end{align*}
with probability at least $1-3\exp(-t)-t^{-1}$, $s_0=\lceil N\gamma_N/\{\log{d}\log{N}\}\rceil$, and some constant $c>0$.
\end{proof}

\section{Proof of the properties of the nuisance estimator for the DC-BRSS estimator}\label{sec:proof_nuisance_SP}

In this section, we analyze the properties of $\phat_{N}=\phat_{N}^{(1)}$, $\bbetahat_{p}=\bbetahat_{p}^{(1)}$, $\balphatil=\balphatil^{(1)}$, with analogous results applicable to $\phat_{N}^{(2)}$, $\bbetahat_{p}^{(2)}$, and $\balphatil^{(2)}$. We also denote $\pihat(\cdot)=\pihat^{(1)}$, $\mathcal D_N':=\mathcal D_N^{(1)}$, $\mathcal J:=\mathcal I_1$, and $M:=|\mathcal J|=N/2$ throughout.

\subsection{Preliminary analysis of the de-coupled PS estimator}\label{sec:nuisance-PS-SP}

We first study a general restricted strong convexity (RSC) property. For any $v\in[0,1]$, $\bbeta^*,\bDelta\in\R^d $, and $\phi:\R\to(0,\infty)$, we define
$$f(\bDelta,a,v,\bbeta^*,\phi(\cdot)):=(a M)^{-1}\sum_{i\in\mathcal J}\Gamma_i\phi\left(\bS_i^T(\bbeta^*+v\bDelta)\right)(\bS_i^T\bDelta)^2.$$

\begin{lemma}\label{lemma:RSC_gen}
Let Assumptions \ref{cond:subG}, \ref{cond:tail}, and \ref{cond:bound-SP} hold and $\phi(\cdot)$ is a continuous function. Let $\kappa_0>0$ be any fixed number, let $\kappa_1,\kappa_2,C_1,C_2,c_1,c_2>0$ be some constants depending only on $(\sigma,\kappa_0,\kappa_l,v,\phi(\cdot))$. For any (possibly random) $a\in(0,1]$, when $ M\gamma_N>\max\{C_2,C_1\log{M}\log{d}\}$ and $k\leq2$,
\begin{align*}
&\P_{\mathcal D_N'}\left(f(\bDelta,a,v,\bbeta^*,\phi(\cdot))\geq\frac{\gamma_N}{a}\left\{\kappa_1\|\bDelta\|_2^2-\kappa_2\frac{\log{d}}{M\gamma_N}\|\bDelta\|_1^2\right\},\;\;\forall\|\bDelta\|_2\leq\kappa_0\right)\\
&\qquad\geq1-c_1\exp(-c_2 M\gamma_N).
\end{align*}
\end{lemma}

The following lemma demonstrates the properties of $\gammahat_N$.
\begin{lemma}\label{lemma:gammaN-est}
For any $t>0$, define
\begin{align}\mathcal E_\gamma:=\left\{|\gammahat_N-\gamma_N|\leq2\sqrt\frac{t\gamma_N}{M}+\frac{t}{M}\right\}.\label{def:Egamma}
\end{align}
Then,
$$\P_{\mathcal D_N'}\left(\mathcal E_\gamma\right)\leq1-2\exp(-t).$$
On $\mathcal E_\gamma$, when $0<t<0.01M\gamma_N$, we have
\begin{align*}
\left|\frac{\gamma_N}{\gammahat_N}-1\right|\leq2.66\sqrt\frac{t}{M\gamma_N},\quad
0.79\gamma_N\leq\gammahat_N\leq1.21\gamma_N.
\end{align*}
\end{lemma}

The following lemma demonstrates the properties of $\phat_{N}$.
\begin{lemma}\label{lemma:phatN-est}
For any $t>0$, define
\begin{align}
\mathcal E_p:=\left\{|\phat_{N}-p_{N}|\leq\frac{5.32}{\E(T)}\sqrt\frac{t\gamma_N}{M}\right\}.\label{def:Ep}
\end{align}
Then,
$$\P_{\mathcal D_N'}\left(\mathcal E_p\right)\leq1-4\exp(-t).$$
On $\mathcal E_p$, when $0<t<0.01M\gamma_N$, we have
\begin{align*}
\left|\frac{p_{N}}{\phat_{N}}-1\right|\leq12\sqrt\frac{t}{M\gamma_N},\quad
0.46p_{N}\leq\phat_{N}\leq1.54p_{N}.
\end{align*}
\end{lemma}

For any $\bbeta\in\R^d$, $a\in(0,1]$, and $\pitil(\cdot):\mathcal X\mapsto[0,1]$, define the following loss function
\begin{align*}
\elltil_{\bbeta}(\bbeta;a;\pitil):=M^{-1}\sum_{i\in\mathcal J}\left[\left\{1-\frac{\Gamma_i}{\pitil(\bX_i)}\right\}\bS_i^T\bbeta+\frac{\Gamma_i\exp(\bS_i^T\bbeta)}{a\pitil(\bX_i)}\right].
\end{align*}
In addition, for any $\bbeta,\bDelta\in\R^d$, $a\in(0,1]$, and $\pitil(\cdot):\mathcal X\mapsto[0,1]$, further define
\begin{align*}
\delta\elltil_{\bbeta}(\bDelta;a;\bbeta;\pitil)&~:=~\elltil_{\bbeta}(\bbeta+\bDelta;a;\pitil)-\elltil_{\bbeta}(\bbeta;a;\pitil)-\nabla_{\bbeta}\elltil_{\bbeta}(\bbeta;a;\pitil)^T\bDelta.
\end{align*}

In the following, we show the RSC property required for the labeling PS estimation and control the gradient $\|\nabla_{\bbeta}\elltil_{\bbeta}(\bbeta_{p}^*;\phat_{N};\pi^*)\|_\infty$.

\begin{lemma}[The PS model's RSC property]\label{lemma:RSC-SP}
Let Assumption \ref{cond:basic} hold. Then, $p_{N}\asymp\gamma_N$. Further let Assumptions \ref{cond:subG}, \ref{cond:tail}, and \ref{cond:bound-SP} hold. For any constant $\kappa_0>1$ and some $\kappa_1,\kappa_2>0$, define
\begin{align}
\Bcaltil_1:=\left\{\delta\elltil_{\bbeta}(\bDelta;\phat_{N};\bbeta_{p}^*;\pihat)\geq\frac{p_{N}}{\phat_{N}}\left\{\kappa_1\|\bDelta\|_2^2-\kappa_2\frac{\log{d}}{M\gamma_N}\|\bDelta\|_1^2\right\},\quad\forall\|\bDelta\|_2\leq\kappa_0\right\}.\label{def:B1}
\end{align}
Then, with some constants $C_1,C_2>0$, when $ M\gamma_N>\max\{C_2,C_1\log{M}\log{d}\}$,
$$\P_{\mathcal D_N'}(\Bcaltil_1\mid\mathcal E_\pi)\geq1-c_1\exp(-c_2 M\gamma_N),$$
with some constants $c_1,c_2>0$.
\end{lemma}

\begin{remark}[Technical challenges of showing Lemma \ref{lemma:RSC-SP}]\label{remark:RSC}
The proof of Lemma \ref{lemma:RSC-SP} is an analog of showing the RSC property of a generalized linear model (GLM) \citep{negahban2010unified,wainwright2019high}. However, we have additional challenges because of the non-standard loss function and also the decaying PS. By Taylor's theorem, we have
$$\delta\elltil_{\bbeta}(\bDelta;\phat_{N};\bbeta_{p}^*;\pihat)=(M\phat_{N})^{-1}\sum_{i\in\mathcal J}\frac{\Gamma_i}{\pihat(\bX_i)}\exp\{-\bS_i^T(\bbeta_{p}^*+v\bDelta)\}(\bS_i^T\bDelta)^2,$$
with some $v\in(0,1)$. Unlike a generalized linear model (GLM), $\delta\elltil_{\bbeta}(\bDelta;\phat_{N};\bbeta_{p}^*;\pihat)$ is also a function of $\Gamma_i$. Without the presence of $\Gamma_i$ (or when $\Gamma_i\equiv1$), $\delta\elltil_{\bbeta}(\bDelta;\phat_{N};\bbeta_{p}^*;\pihat)$ can be directly lower bounded by, e.g., Proposition 2 of \cite{negahban2010unified} or Theorem 9.36 of \cite{wainwright2019high}, as $\pihat(\bX_i)$ is bounded below under Assumption \ref{cond:NP-est}. However, with the presence of $\Gamma_i$ and the decaying MAR setting mechanism, we need to carefully track on the impact of $\Gamma_i$ since only a few of them are non-zero that $\E(\Gamma)=\gamma_N\to0$ as $N,d\to\infty$. The usual Bernstein inequality with the usage of $\psi_\alpha$-norms is suboptimal here. To construct a tight lower bound for $\delta\elltil_{\bbeta}(\bDelta;\phat_{N};\bbeta_{p}^*;\pihat)$, we utilize concentration inequalities that involve random variables' $\psi_2$-norm and also second moment \citep{kuchibhotla2022moving}. Here, the second moment helps us to capture the decaying value $\gamma_N$; see more details in Section \ref{sec:proof_nuisance_SP}. In addition, the same challenges also arise in the proofs of Lemmas \ref{lemma:gradient-SP}-\ref{lemma:gradient3-SP} below.
\end{remark}

\begin{lemma}[Upper bound for $\|\nabla_{\bbeta}\elltil_{\bbeta}(\bbeta_{p}^*;\phat_{N};\pi^*)\|_\infty$]\label{lemma:gradient-SP}
Let Assumptions \ref{cond:basic}, \ref{cond:NP-est}, and \ref{cond:bound-SP} hold. For any $0<t<M\gamma_N/\{100+\log{M}\log{d}\}$, define
\begin{align}
\Bcaltil_2:=\left\{\|\nabla_{\bbeta}\elltil_{\bbeta}(\bbeta_{p}^*;\phat_{N};\pi^*)\|_\infty\leq\kappa_3\sqrt\frac{t+\log{d}}{M\gamma_N}\right\},\label{def:B2}
\end{align}
with some constant $\kappa_3>0$ and $\gammahat_N$ is defined as \eqref{def:betahat}. Then, when $M\gamma_N>C_1\log{M}\log{d}$,
\begin{align*}
\P_{\mathcal D_N'}(\Bcaltil_2)\geq1-10\exp(-t).
\end{align*}
\end{lemma}

Define
\begin{align}
\br_\pi:=&\nabla_{\bbeta}\elltil_{\bbeta}(\bbeta_{p}^*;\phat_{N};\pihat)-\nabla_{\bbeta}\elltil_{\bbeta}(\bbeta_{p}^*;\phat_{N};\pi^*)\nonumber\\
=&-M^{-1}\sum_{i\in\mathcal J}\left\{\frac{1}{\pihat(\bX_i)}-\frac{1}{\pi^*(\bX_i)}\right\}\Gamma_i\left\{1+\phat_{N}^{-1}\exp(-\bS_i^T\bbeta_{p}^*)\right\}\bS_i.\label{def:rpi}
\end{align}

The following lemma controls the treatment PS's estimation error's effect on the labeling PS's estimation.

\begin{lemma}\label{lemma:rpi}
Let Assumptions \ref{cond:subG}, \ref{cond:tail}, and \ref{cond:NP-est} hold. Then, on the event $\mathcal E_\zeta$,
\begin{align}
\Bcaltil_3:=\left\{|\br_\pi^T\bDelta|\leq c\zeta_N\left(\|\bDelta\|_1\sqrt\frac{\log{d}\log{N}}{M\gamma_N}+\|\bDelta\|_2\right),\;\;\forall\bDelta\in\R^d\right\}\label{def:B3}
\end{align}
occurs with some constant $c>0$ with probability at least $1-3\exp\{N\gamma_N/\log{N}\}$.
\end{lemma}

In order to utilize the RSC property of Lemma \ref{lemma:RSC-SP}, we need to first show that the error $\|\bbetahat_{p}-\bbeta_{p}^*\|_2$ is bounded. For any $\bDelta\in\R^d$, define
\begin{align}
\mathcal F(\bDelta):=&\delta\elltil_{\bbeta}(\bDelta;\phat_{N};\bbeta_{p}^*;\pihat)+\lambda_{\bbeta}\|\bbeta_{p}^*+\bDelta\|_1+\nabla_{\bbeta}\elltil_{\bbeta}(\bbeta_{p}^*;\phat_{N};\pihat)^T\bDelta-\lambda_{\bbeta}\|\bbeta_{p}^*\|_1\nonumber\\
=&\delta\elltil_{\bbeta}(\bDelta;\phat_{N};\bbeta_{p}^*;\pihat)+\lambda_{\bbeta}\|\bbeta_{p}^*+\bDelta\|_1+\nabla_{\bbeta}\elltil_{\bbeta}(\bbeta_{p}^*;\phat_{N};\pi^*)^T\bDelta\nonumber\\
&\quad+\br_\pi^T\bDelta-\lambda_{\bbeta}\|\bbeta_{p}^*\|_1.\label{def:FDelta}
\end{align}
By construction, we have $\mathcal F(\bbetahat_{p}-\bbeta_{p}^*)\leq0$. That is,
\begin{align}
&\delta\elltil_{\bbeta}(\bbetahat_{p}-\bbeta_{p}^*;\phat_{N};\bbeta_{p}^*;\pihat)+\lambda_{\bbeta}\|\bbetahat_{p}\|_1\nonumber\\
&\qquad\leq\lambda_{\bbeta}\|\bbeta_{p}^*\|_1-\nabla_{\bbeta}\elltil_{\bbeta}(\bbeta_{p}^*;\phat_{N};\pi^*)^T(\bbetahat_{p}-\bbeta_{p}^*)-\br_\pi^T(\bbetahat_{p}-\bbeta_{p}^*).\label{eq:basic-PS}
\end{align}
For any $r_N>0$, define $\mathcal K(r_N,1):=\{\bDelta\in\R^d:\|\bDelta\|_1\leq r_N\|\bDelta\|_2,\|\bDelta\|_2=1\}$. The following Lemma shows that the function $\mathcal F(\bDelta)$ is strictly positive with high probability for any $\bDelta\in\mathcal K(r_N,1)$. As we will show that $\|\bDelta\|_1\leq r_N\|\bDelta\|_2$ with some $r_N>0$ (see \eqref{bound:rN}), the following Lemma indicates that $\|\bbetahat_{p}-\bbeta_{p}^*\|_2\neq1$. In fact, we can further show that $\|\bbetahat_{p}-\bbeta_{p}^*\|_2\leq1$ using the convexity of the function $\mathcal F(\cdot)$; see details in the proof of Theorem \ref{thm:PS-SP}.

\begin{lemma}\label{lemma:Fdelta}
Let $s_{p}=o(M\gamma_N/\log{d})$. For $0<t<M\gamma_N/\{100+\log{M}\log{d}\}$, choose any $\lambda_{\bbeta}\asymp\sqrt{\log{d}/(M\gamma_N)}$ satisfying $\lambda_{\bbeta}>2\kappa_3\sqrt{\{t+\log{d}\}/(M\gamma_N)}$. Then, on the event $\Bcaltil_2\cap\Bcaltil_3$, when $N$ is large enough,
\begin{align*}
\mathcal F(\bDelta)\geq\delta\elltil_{\bbeta}(\bDelta;\phat_{N};\bbeta_{p}^*;\pihat)+\lambda_{\bbeta}\|\bDelta\|_1/4-(2\sqrt{s_{p}}\lambda_{\bbeta}+c\zeta_N)\|\bDelta\|_2,\;\;\forall\bDelta\in\R^d.
\end{align*}
Further condition on the event $\Bcaltil_1\cap\mathcal E_p$. Let $r_N=o(\sqrt{M\gamma_N/\log{d}})$, when $N$ is large enough, we have $\mathcal F(\bDelta)>0$, $\forall\bDelta\in\mathcal K(r_N,1)$.
\end{lemma}

\subsection{Preliminary analysis of the OR estimator}\label{sec:nuisance-OR-SP}

For any $\balpha,\bbeta,\bDelta\in\R^d $, $a\in(0,1]$, and function $\pitil(\cdot):\mathcal X\mapsto[0,1]$, define the loss function for the corresponding OR model as
\begin{align*}
\elltil_{\balpha}(\balpha;a,\bbeta;\pitil)&:=M^{-1}\sum_{i\in\mathcal J}\frac{\Gamma_i}{a\pitil(\bX_i)}\exp(-\bS_i^T\bbeta)(Y_i-\bS_i^T\balpha)^2.
\end{align*}
In addition, let
\begin{align*}
&\delta\elltil_{\balpha}(\bDelta;a,\bbeta;\balpha;\pitil):=\elltil_{\balpha}(\balpha+\bDelta;a,\bbeta;\pitil)-\elltil_{\balpha}(\balpha;a,\bbeta;\pitil)-\nabla_{\balpha}\elltil_{\balpha}(\balpha;a,\bbeta;\pitil)^T\bDelta\\
&\qquad=M^{-1}\sum_{i\in\mathcal J}\frac{\Gamma_i}{a\pitil(\bX_i)}\exp(-\bS_i^T\bbeta)(\bS_i^T\bDelta)^2.
\end{align*}
We first characterize the RSC property for the OR model in the following Lemma.

\begin{lemma}[The OR model's RSC property]\label{lemma:RSC2-SP}
Let Assumptions \ref{cond:basic}, \ref{cond:subG}, \ref{cond:tail}, and \ref{cond:bound-SP} hold. For some constants $\kappa_1,\kappa_2>0$, define
\begin{align*}
\Acaltil_1:=\left\{\delta\elltil_{\balpha}(\bDelta;\phat_{N},\bbetahat_{p};\balphatil^*;\pihat)\geq\frac{p_{N}}{\phat_{N}}\left\{\kappa_1\|\bDelta\|_2^2-\kappa_2\frac{\log{d}}{M\gamma_N}\|\bDelta\|_1^2\right\},\quad\forall\bDelta\in\R^d\right\}.
\end{align*}
Then, with some constants $C_1,C_2>0$, when $ M\gamma_N>\max\{C_2,C_1\log{M}\log{d}\}$,
$$\P_{\mathcal D_N'}(\Acaltil_1\mid\mathcal E_\pi)\geq1-c_1\exp(-c_2 M\gamma_N),$$
with some constants $c_1,c_2>0$.
\end{lemma}

{In the following, we control the effect of the gradient $\nabla_{\balpha}\elltil_{\balpha}(\balphatil^*;\phat_{N},\bbetahat_{p};\pihat)$. Define $\varepsilon_{i,1}:=Y_i(1)-m(\bX_i)$ and $\varepsilon_{i,2}:=m(\bX_i)-m^*(\bX_i)$, representing respectively the stochastic and approximation errors. Then $\varepsilon_i=\varepsilon_{i,1}+\varepsilon_{i,2}$. Consider the following decomposition:
\begin{align}\label{rep:gradient}
\nabla_{\balpha}\elltil_{\balpha}(\balphatil^*;\phat_{N},\bbetahat_{p};\pihat)=(M\phat_{N})^{-1}\sum_{i\in\mathcal J}(\bV_{i,1}+\bV_{i,2})+\br_{p}+\br_{p,2},
\end{align}
where
\begin{align*}
\bV_{i,1}:=&-2\Gamma_i\exp(-\bS_i^T\bbetahat_{p})\bS_i\varepsilon_{i,1}/\pihat(\bX_i),\\
\bV_{i,2}:=&-2\Gamma_i\exp(-\bS_i^T\bbeta_{p}^*)\bS_i\varepsilon_{2,1}/\pi^*(\bX_i)\\
\br_{p}:=&-2(M\phat_{N})^{-1}\sum_{i\in\mathcal J}\left\{\frac{1}{\pihat(\bX_i)}-\frac{1}{\pi^*(\bX_i)}\right\}\exp(-\bS_i^T\bbeta_{p}^*)\Gamma_i\bS_i\varepsilon_{i,2},\\
\br_{p,2}:=&-2(M\phat_{N})^{-1}\sum_{i\in\mathcal J}\left\{\exp(-\bS_i^T\bbetahat_{p})-\exp(-\bS_i^T\bbeta_{p}^*)\right\}\Gamma_i\bS_i\varepsilon_{i,2}/\pihat(\bX_i).
\end{align*}}

\vspace{-1.5em}

\begin{lemma}[Controlling the effect of stochastic error]\label{lemma:gradient2-SP}
Let Assumptions \ref{cond:basic}, \ref{cond:bound-SP}, and \ref{cond:subG'} hold. For any $0<t<M\gamma_N/\{100+\log{M}\log{d}\}$, define
$$\Acaltil_2:=\left\{\left\|(M\phat_{N})^{-1}\sum_{i\in\mathcal J}\bV_{i,1}\right\|_\infty\leq\kappa_4\sqrt\frac{t+\log{d}}{M\gamma_N}\right\},$$
with some constant $\kappa_4>0$. Then, when $\lambda_{\bbeta}>2\kappa_3\sqrt{\{t+\log{d}\}(M\gamma_N)}$,
$$\P_{\mathcal D_N'}(\Acaltil_2\mid\mathcal E_\pi)\geq1-2\exp\{-t\log{d}\}-4\exp(-t)-t^{-1}.$$
\end{lemma}

Now, we further consider the case that the OR model is misspecified. We control the gradient $\|\nabla_{\balpha}\elltil_{\balpha}(\balphatil^*;\phat_{N},\bbeta_{p}^*;\pi^*)\|_\infty$ in the following lemma.

\begin{lemma}[Controlling the effect of approximation error]\label{lemma:gradient3-SP}
Let Assumptions \ref{cond:tail}, \ref{cond:NP-est}, \ref{cond:bound-SP}, and \ref{cond:approx} hold. For any $0<t<M\gamma_N/\{100+\log{M}\log{d}\}$, define
$$\Acaltil_3:=\left\{\left\|(M\phat_{N})^{-1}\sum_{i\in\mathcal J}\bV_{i,2}\right\|_\infty\leq\kappa_5\sqrt\frac{t+\log{d}}{M\gamma_N}\right\},$$
with some constant $\kappa_5>0$. Let $M\gamma_N>C_1\log{M}\log{d}$ and $\lambda_{\bbeta}>2\kappa_3\sqrt{\{t+\log{d}\}(M\gamma_N)}$, then
$$\P_{\mathcal D_N'}(\mathcal A_3)\geq1-7\exp(-t).$$
\end{lemma}

\subsection{Proof of the results in Section \ref{sec:nuisance-PS-SP}}

\begin{proof}[Proof of Lemma \ref{lemma:RSC_gen}]
For any $\bDelta\in\R^d$, define
$$A_i=\Gamma_i\phi\left(\bS_i^T(\bbeta^*+v\bDelta)\right)(\bS_i^T\bDelta)^2.$$
Then,
$$f(\bDelta,\gamma_N,v,\bbeta^*,\phi(\cdot))=(a M)^{-1}\sum_{i\in\mathcal J}A_i.$$
For truncation levels $T\geq\tau>0$, we define the following truncation functions
$$\varphi_\tau(u)=u^2\mathbbm1_{|u|\leq\tau/2}+(\tau-u)^2\mathbbm1_{\tau/2<|u|\leq\tau},\qquad\alpha_\tau(u)=u\mathbbm1_{|u|\leq\tau}.$$
Now, we show that, for each $i\in\mathcal J$,
\begin{equation}\label{lower:Ai}
A_i\geq\Gamma_i\phi\left(\alpha_T(\bS_i^T\bbeta^*)+v\alpha_\tau(\bS_i^T\bDelta)\right)\varphi_\tau\left(\bS_i^T\bDelta\mathbbm1_{|\bS_i^T\bbeta^*|\leq T}\right).
\end{equation}

\textbf{Case 1:} $|\bS_i^T\bbeta^*|>T$ or $|\bS_i^T\bDelta|>\tau$. We can see that $\varphi_\tau\left(\bS_i^T\bDelta\mathbbm1_{|\bS_i^T\bbeta^*|\leq T}\right)=0$ and hence \eqref{lower:Ai} follows.

\textbf{Case 2:} $|\bS_i^T\bbeta^*|\leq T$, and $|\bS_i^T\bDelta|\leq\tau$. Then,
$$\alpha_T(\bS_i^T\bbeta^*)=\bS_i^T\bbeta^*,\quad\alpha_\tau(\bS_i^T\bDelta)=\bS_i^T\bDelta,\quad(\bS_i^T\bDelta)^2\geq\varphi_\tau\left(\bS_i^T\bDelta\mathbbm1_{|\bS_i^T\bbeta^*|\leq T}\right),$$
and hence \eqref{lower:Ai} follows. Define
$$K_3:=8r\sigma^2\log\left\{2^{4+1/(2r)}C^{1/q}r\sigma^2\kappa_l^{-1}\right\},\qquad L_\phi:=\min_{|u|\leq \sqrt{K_3}(1+\kappa_0)}\phi(u).$$
For some $\delta\in(0,\kappa_0]$, we choose the truncation levels as
$$T^2=K_3,\qquad\tau^2=\tau^2(\delta)=K_3\delta^2.$$
Then, for each $i\in\mathcal J$,
$$\phi\left(\alpha_T(\bS_i^T\bbeta^*)+v\alpha_\tau(\bS_i^T\bDelta)\right)\geq L_\phi.$$
Note that, for any $\kappa_1,\kappa_2>0$ and $\bDelta\in\R^d $,
$$\kappa_1\|\bDelta\|_2^2-\kappa_2\frac{\log{d}}{M\gamma_N}\|\bDelta\|_1^2\leq2\kappa_1\|\bDelta\|_2^2-2\sqrt\frac{\kappa_2\log{d}}{M\gamma_N}\|\bDelta\|_1^2.$$
Hence, it suffices to show that
\begin{align*}
&M^{-1}\sum_{i\in\mathcal J}\varphi_{\tau(\delta)}\left(\bS_i^T\bDelta\mathbbm1_{|\bS_i^T\bbeta^*|\leq T}\right)\\
&\qquad\geq2L_\phi^{-1}\gamma_N\left\{\kappa_1\|\bDelta\|_2^2-\sqrt\frac{\kappa_2\log{d}}{M\gamma_N}\|\bDelta\|_1^2\right\},\quad\forall\delta\in(0,\kappa_0],\;\;\|\bDelta\|_2=\delta,
\end{align*}
with high probability. By rescaling the vector that $\widetilde{\bDelta}=\bDelta/\|\bDelta\|_2$ and notice that $\varphi_{\tau(\delta)}(\delta^2u)=\delta^2\varphi_{\tau(1)}(u)$, it suffices to show
\begin{align}
& M^{-1}\sum_{i\in\mathcal J}\varphi_{\tau(1)}\left(\bS_i^T\bDelta\mathbbm1_{|\bS_i^T\bbeta^*|\leq T}\right)\geq2L_\phi^{-1}\gamma_N\left\{\kappa_1-\sqrt\frac{\kappa_2\log{d}}{M\gamma_N}\|\bDelta\|_1\right\},\;\;\forall\|\bDelta\|_2=1.\label{ineq:sufficient}
\end{align}
Hence, we will restrict $\delta=1$ and $\tau=\tau(1)=\sqrt{K_3}$. Define $g_{\bDelta}(\bx)=\varphi_\tau\left(\varphi(\bx)^T\bDelta\mathbbm1_{|\varphi(\bx)^T\bbeta^*|\leq T}\right)$, $\kappa_1=\kappa_l/(2L_\phi)$, and $\kappa_2=2c_1'^2/L_\phi$.

\textbf{Step 1.} We first demonstrate that, for any $\|\bDelta\|_2=1$,
\begin{equation}\label{step1}
\E\left\{\Gamma g_{\bDelta}(\bX)\right\}\geq\gamma_N\kappa_l/2.
\end{equation}
Under Assumption \ref{cond:subG}, we have
$$\E\left\{\Gamma(\bS^T\bDelta)^2\right\}=\gamma_N\E\left\{(\bS^T\bDelta)^2\mid\Gamma=1\right\}\geq\gamma_N\Var(\bS^T\bDelta\mid\Gamma=1)\geq\gamma_N\kappa_l.$$
Besides, observe that, for $r>0$ satisfying $1/r+1/q=1$,
\begin{align*}
&\E\left[\Gamma\left\{(\bS^T\bDelta)^2-g_{\bDelta}(\bX)\right\}\right]=\E\left[\gamma_N(\bX)\left\{(\bS^T\bDelta)^2-g_{\bDelta}(\bX)\right\}\right]\\
&\qquad\leq\|\gamma_N(\cdot)\|_{\P,q}\left\|(\bS^T\bDelta)^2-g_{\bDelta}(\bX)\right\|_{\P,r}\\
&\qquad\leq C^{1/q}\gamma_N\left\|(\bS^T\bDelta)^2\mathbbm1_{|\bS^T\bbeta^*|\geq T}\right\|_{\P,r}+c\gamma_N\left\|(\bS^T\bDelta)^2\mathbbm1_{|\bS^T\bDelta|\geq\tau/2}\right\|_{\P,r}\\
&\qquad\leq C^{1/q}\gamma_N\|(\bS^T\bDelta)^2\|_{\P,2r}\left[\left\{\P(|\bS^T\bbeta^*|\geq T)\right\}^{1/(2r)}+\left\{\P(|\bS^T\bDelta|\geq \tau/2)\right\}^{1/(2r)}\right].
\end{align*}
Under Assumptions \ref{cond:subG} and \ref{cond:bound-SP}, $\|\bS^T\bDelta\|_{\psi_2}\leq\sigma$ and $\|\bS^T\bbeta^*\|_{\psi_2}\leq\sigma$. By part (b) of Lemma \ref{lemma:subG},
$$\|(\bS^T\bDelta)^2\|_{\P,2r}\leq2^{1+1/(2r)}r\sigma^2.$$
By part (b) of Lemma \ref{lemma:subG},
\begin{align*}
\P(|\bS^T\bbeta^*|\geq T)&\leq2\exp\left(-\frac{K_3}{\sigma^2}\right)\leq2\exp\left(-\frac{K_3}{4\sigma^2}\right),\\
\P(|\bS^T\bDelta|\geq\tau/2)&\leq2\exp\left(-\frac{K_3}{4\sigma^2}\right).
\end{align*}
Hence, by the construction of $K_3$,
$$\E\left[\Gamma\left\{(\bS^T\bDelta)^2-g_{\bDelta}(\bX)\right\}\right]\leq2^{3+1/(2r)}C^{1/q}r\sigma^2\exp\left(-\frac{K_3}{8r\sigma^2}\right)\gamma_N\leq\gamma_N\kappa_l/2.$$
Therefore,
$$\E\left\{\Gamma g_{\bDelta}(\bX)\right\}=\E\left\{\Gamma(\bS^T\bDelta)^2\right\}-\E\left[\Gamma\left\{(\bS^T\bDelta)^2-g_{\bDelta}(\bX)\right\}\right]\geq\gamma_N\kappa_l/2.$$
\textbf{Step 2.} Define
$$f_N(\bDelta)=\left| M^{-1}\sum_{i\in\mathcal J}\Gamma_ig_{\bDelta}(\bX_i)-\E\{\Gamma g_{\bDelta}(\bX)\}\right|,\quad Z(t)=\sup_{\|\bDelta\|_2=1,\|\bDelta\|_1\leq t}f_N(\bDelta).$$
We prove that, with some constant $C>0$, when $ M\gamma_N$ is large enough,
\begin{equation}\label{EZt}
\E_{\mathcal D_N'}\{Z(t)\}\leq1016K_3r\sigma t\sqrt\frac{C^{1/q}\gamma_N\log{d}}{M}.
\end{equation}
Let $(\epsilon_i)_{i\in\mathcal I}$ be i.i.d. Rademacher variables. With a slight abuse of notation, we let $\bZ=(R,T,Y(0),Y(1),\bX,\epsilon)$ and $\mathcal D_N'=\{\bZ_i\}_{i\in\mathcal J}$. Then,
\begin{align*}
\E_{\mathcal D_N'}\{Z(t)\}&\leq 2\E_{\mathcal D_N'}\left\{\sup_{\|\bDelta\|_2=1,\|\bDelta\|_1\leq t}\left| M^{-1}\sum_{i\in\mathcal J}\epsilon_i\Gamma_ig_{\bDelta}(\bX_i)\right|\right\}\\
&=2\E_{\mathcal D_N'}\left\{\sup_{\|\bDelta\|_2=1,\|\bDelta\|_1\leq t}\left| M^{-1}\sum_{i\in\mathcal J}\epsilon_i\Gamma_i\varphi_\tau\left(\bS_i^T\bDelta\mathbbm1_{|\bS_i^T\bbeta^*|\leq T}\right)\right|\right\}\\
&\overset{(i)}{\leq}8K_3\E_{\mathcal D_N'}\left\{\sup_{\|\bDelta\|_2=1,\|\bDelta\|_1\leq t}\left| M^{-1}\sum_{i\in\mathcal J}\epsilon_i\Gamma_i\bS_i^T\bDelta\mathbbm1_{|\bS_i^T\bbeta^*|\leq T}\right|\right\}\\
&\leq8K_3t\E_{\mathcal D_N'}\left\| M^{-1}\sum_{i\in\mathcal J}\epsilon_i\Gamma_i\bS_i\right\|_\infty,
\end{align*}
where (i) holds by utilizing Ledoux-Talagrand contraction inequality in \cite{ledoux2013probability}, notice that $\varphi_\tau(\cdot)$ is $(2K_3)$-Lipschitz continuous and $\varphi_\tau(0)=0$. For any vector $\bx$, denote $\bx(j)$ as its $j$-th element throughout. Here, $\E(\epsilon\Gamma\bS)=\boldsymbol0$, $\sup_{1\leq j\leq d}\|\epsilon\Gamma\bS(j)\|_{\psi_2}\leq\sigma$ as $|\epsilon\Gamma|\leq1$, and with $r>0$ satisfying $1/r+1/q=1$,
\begin{align*}
&\sup_{1\leq j\leq d}\E\left\{\epsilon\Gamma\bS(j)\right\}^2=\sup_{1\leq j\leq d}\E\left\{\gamma_N(\bX)\bS^2(j)\right\}\\
&\qquad\leq\|\gamma_N(\cdot)\|_{\P,q}\|\bS(j)\|_{\P,2r}^2\overset{(i)}{\leq} C^{1/q}2^{1/r}r^2\sigma^2\gamma_N\leq2C^{1/q}r^2\sigma^2\gamma_N,
\end{align*}
where (i) holds by Assumption \ref{cond:tail} and $\|\bS(j)\|_{\P,2r}\leq2^{1/(2r)}r\sigma$ using part (c) of Lemma \ref{lemma:subG}. By Theorem 3.4 of \cite{kuchibhotla2022moving}, for any $u\geq0$, with probability at least $1-3\exp(-u)$,
\begin{align*}
\left\| M^{-1}\sum_{i\in\mathcal J}\epsilon_i\Gamma_i\bS_i\right\|_\infty&\leq7r\sigma\sqrt\frac{2C^{1/q}\gamma_N\{u+\log{d}\}}{M}+\frac{c_1'\sigma\sqrt{\log{M}}\{u+\log{d}\}}{M}\\
&\leq a_1u+a_2\sqrt u+a_3,
\end{align*}
where $c_1'>0$ is a constant and
\begin{align*}
a_1&=\frac{c_1'\sigma\sqrt{\log{M}}}{M},\\
a_2&=7r\sigma\sqrt\frac{2C^{1/q}\gamma_N}{M},\\
a_3&=7r\sigma\sqrt\frac{2C^{1/q}\gamma_N\log{d}}{M}+\frac{c_1'\sigma\sqrt{\log{M}}\log{d}}{M}.
\end{align*}
By Lemma \ref{integrating},
\begin{align*}
&\E_{\mathcal D_N'}\left\| M^{-1}\sum_{i\in\mathcal J}\epsilon_i\Gamma_i\bS_i\right\|_\infty\leq12a_1+3\sqrt\pi a_2+a_3\\
&\qquad=7r\sigma\sqrt\frac{2C^{1/q}\gamma_N}{M}\{\sqrt{\log{d}}+3\sqrt\pi\}+\frac{c_1'\sigma\sqrt{\log{M}}\{\log{d}+12\}}{M}\\
&\qquad\leq14r\sigma\sqrt\frac{2C^{1/q}\gamma_N}{M}\{\sqrt{\log{d}}+3\sqrt\pi\}\leq127r\sigma\sqrt\frac{C^{1/q}\gamma_N\log{d}}{M},
\end{align*}
when $ M\gamma_N>C_1\log{M}\log{d}$, where $C_1=c_1'^2/(4C^{1/q}r^2)$. Hence, \eqref{EZt} follows.

\textbf{Step 3.} We show that, with some constants $c_2',c_1,c_2>0$,
\begin{equation}\label{step3}
\P_{\mathcal D_N'}\left(Z(t)\geq\frac{\gamma_N\kappa_l}{4}+2c_2'\|\bDelta\|_1\sqrt\frac{\gamma_N\log{d}}{M},\forall t>0\right)\leq c_1\exp\left(-c_2 M\gamma_N\right).
\end{equation}
Define $z^*(t)=\gamma_N\{\kappa_l/8+2K_3\sigma t\sqrt{\log{d}/( M\gamma_N)}\}$ and $\mathcal F=\{\pm f(\cdot):f(\bx)=\Gamma g_{\bDelta}(\bx)-\E\{\Gamma g_{\bDelta}(\bX)\},\|\bDelta\|_2=1,\|\bDelta\|_1=t\}$. Notice that $0\leq g_{\bDelta}(\bx)\leq K_3$ for all $\bx\in\R^p$, it follows that $|f(\bx)|\leq K_3$ for all $f\in\mathcal F$. Besides, notice that
$$\sup_{f\in\mathcal F}\E\{f^2(\bX)\}\leq\sup_{\|\bDelta\|_2=1,\|\bDelta\|_1=t}\E\left\{\Gamma g_{\bDelta}^2(\bX)\right\}\leq K_3^2\gamma_N.$$
By Theorem 3.27 of \cite{wainwright2019high} and \eqref{EZt}, we have
\begin{align*}
&\P_{\mathcal D_N'}\left(Z(t)\geq\E_{\mathcal D_N'}\{Z(t)\}+z^*(t)\right)\\
&\qquad\leq\exp\left(-\frac{M\{z^*(t)\}^2}{8e[K_3^2\gamma_N+2K_3\E\{Z(t)\}]+4K_3z^*(t)}\right)\\
&\qquad\leq\exp\left\{-\frac{M\gamma_N\kappa_l^2/64+4K_3^2\sigma^2t^2\log{d}}{22K_3^2+2032K_3^2rt\sigma C^{1/(2q)}\sqrt\frac{\log{d}}{M\gamma_N}+K_3\kappa_l/2+8K_3^2\sigma t\sqrt\frac{\log{d}}{M\gamma_N}}\right\}\\
&\qquad=\exp\left\{-\frac{M\gamma_N\kappa_l^2/64+K_3^2\sigma^2t^2\log{d}}{22K_3^2+K_3\kappa_l/2+(8+2032rC^{1/(2q)})K_3^2\sigma t\sqrt\frac{\log{d}}{M\gamma_N}}\right\}.
\end{align*}
It follows that
$$\P_{\mathcal D_N'}\left(\sup_{\|\bDelta\|_2=1,\|\bDelta\|_1\leq t}f_N(\bDelta)\geq g(t)\right)\leq h(t),$$
where $a=K_3^2\sigma^2$, $b=\kappa_l^2/64$, $c=(8+2032rC^{1/(2q)})K_3^2\sigma$, $d=22K_3^2+K_3\kappa_l/2$, $c_2'=(2+1016rC^{1/(2q)})K_3\sigma$, $c_3'=\kappa_l/8$,
$$g(t)=c_2't\sqrt\frac{\gamma_N\log{d}}{M}+c_3'\gamma_N,\;\;h(t)=\exp\left\{-\frac{at^2\log{d}+b M\gamma_N}{ct\sqrt\frac{\log{d}}{M\gamma_N}+d}\right\}.$$
Now we apply a peeling argument to extend the radii $\|\bDelta\|_1$. For each $m\geq1$, define
$$\mathcal A_m=\left\{\bDelta\in\R^d :\|\bDelta\|_2=1,2^{m-1}c_3'\gamma_N\leq g(\|\bDelta\|_1)\leq2^mc_3'\gamma_N\right\}.$$
Notice that $g(t)$ is a strictly increasing function, $g^{-1}(t)=(c_2')^{-1}(t-c_3'\gamma_N)\sqrt{M/\{\gamma_N\log{d}\}}$, and $g(t)\geq c_3'\gamma_N$ for all $t>0$. Let $c_4'=\log_2\{c_2'/(\sigma c_3')+1\}$ and $t_m=c_3'(2^m-1)/c_2'\geq c_3'm/c_2'$ for all $m\geq1$. Then,
\begin{align*}
&\P_{\mathcal D_N'}\left(\exists\bDelta\in\R^d \;\;\text{s.t.}\;\;\|\bDelta\|_2=1,\;f_N(\bDelta)\geq2g(\|\bDelta\|_1)\right)\\
&\qquad\overset{(i)}{\leq}\sum_{m=1}^\infty\P_{\mathcal D_N'}\left(\exists\bDelta\in\mathcal A_m\;\;\text{s.t.}\;\;f_N(\bDelta)\geq2g(\|\bDelta\|_1)\right)\\
&\qquad\overset{(ii)}{\leq}\sum_{m=1}^\infty\P_{\mathcal D_N'}\left(\sup_{\|\bDelta\|_2=1,\|\bDelta\|_1\leq g^{-1}(2^m\gamma_Nc_2)}f_N(\bDelta)\geq2g(\|\bDelta\|_1)\right)\\
&\qquad\leq\sum_{m=1}^\infty h(g^{-1}(2^m\gamma_N\kappa_l/4))=\sum_{m=1}^\infty h\left(t_m\sqrt\frac{M}{\gamma_N\log{d}}\right)\\
&\qquad=\sum_{m=1}^\infty\exp\left(- M\gamma_N\frac{at_m^2+b}{ct_m+d}\right)\\
&\qquad\leq\sum_{t_m\leq\sigma^{-1}}\exp\left(-M\gamma_N\frac{b}{c\sigma^{-1}+d}\right)+\sum_{t_m>\sigma^{-1}}\exp\left\{- M\gamma_N\frac{at_m}{c+d\sigma}\right\}\\
&\qquad=\sum_{m\leq c_4'}\exp\left(- M\gamma_N\frac{b}{c\sigma^{-1}+d}\right)+\sum_{m>c_4'}\exp\left\{- M\gamma_N\frac{at_m}{c+d\sigma}\right\}\\
&\qquad\leq c_4'\exp\left(- M\gamma_N\frac{b}{c\sigma^{-1}+d}\right)+\sum_{m=2}^\infty\exp\left\{- M\gamma_N\frac{ac_3'm}{c_2'(c+d\sigma)}\right\}\\
&\qquad=c_4'\exp\left(- M\gamma_N\frac{b}{c\sigma^{-1}+d}\right)+\frac{\exp\left\{- M\gamma_N\frac{2ac_3'}{c_2'(c+d\sigma)}\right\}}{1-\exp\left\{- M\gamma_N\frac{ac_3'}{c_2'(c+d\sigma)}\right\}}\\
&\qquad\leq c_1\exp\left(-c_2 M\gamma_N\right),
\end{align*}
when $M\gamma_N>c_2'(c+d\sigma)\log(2)/(ac_3')$, where $c_1=c_4'+2$ and $c_2=\max[b/(c\sigma^{-1}+d),2ac_3'/\{c_2'(c+d\sigma)\}]$. Here, (i) holds since $\P_{\mathcal D_N'}(\cup_{m=1}^\infty\mathcal A_m)=1$, (ii) holds since $\mathcal A_m\subseteq\{\bDelta:\|\bDelta\|_2=1,\|\bDelta\|_1\leq g^{-1}(2^m\gamma_Nc_3')\}$. Hence, \eqref{step3} holds.

Combining \eqref{step1} and \eqref{step3}, with probability at least $1-c_1\exp\left(-c_2M\gamma_N\right)$,
$$ M^{-1}\sum_{i\in\mathcal I_k}\Gamma_ig_{\bDelta}(\bX_i)\geq\gamma_N\left\{\kappa_l/4-2c_2'\sqrt\frac{\log{d}}{M\gamma_N}\|\bDelta\|_1\right\},\quad\forall\|\bDelta\|_2=1.$$
It follows that
$$\P_{\mathcal D_N'}\left(\frac{1}{M}\sum_{i\in\mathcal I_k}\Gamma_ig_{\bDelta}(\bX_i)\geq\gamma_N\left\{\frac{\kappa_l}{4}-2c_2'\sqrt\frac{\log{d}}{M\gamma_N}\|\bDelta\|_1\right\}\right)\geq1-c_1\exp\left(-c_2 M\gamma_N\right).$$
Choose $a=\gamma_N$, $\kappa_1=\kappa_l/8$ and $\kappa_2=c_2'^2$, when $ M\gamma_N>\max\{C_2,c_1'^2\log{M}\log{d}/(4C^{1/q}r^2)\}$, where $C_2=\frac{c_2'(c+d\sigma)}{ac_3'}\log(2)$, we have \eqref{ineq:sufficient} holds with probability at least $1-c_1\exp(-c_2 M\gamma_N)$.
\end{proof}

\begin{proof}[Proof of Lemma \ref{lemma:gammaN-est}]
By Theorem 1 of \cite{van2013bernstein}, we have
$$\P_{\mathcal D_N'}\left(\mathcal E_\gamma\right)\leq1-2\exp(-t).$$
Moreover, On $\mathcal E_\gamma$, when $0<t< 0.01M\gamma_N$,
\begin{align*}
|\gammahat_N-\gamma_N|\leq2\sqrt\frac{0.01M\gamma_N^2}{M}+\frac{0.01M\gamma_N}{M}=0.21\gamma_N.
\end{align*}
Hence,
\begin{align*}
0.79\gamma_N\leq\gammahat_N\leq1.21\gamma_N.
\end{align*}
In addition,
\begin{align*}
\left|\frac{\gamma_N}{\gammahat_N}-1\right|&=\frac{|\gammahat_N-\gamma_N|}{\gammahat_N}\leq\frac{2\sqrt{t\gamma_N/M}+\sqrt{t\cdot0.01M\gamma_N}/M}{0.79\gamma_N}\leq2.66\sqrt\frac{t}{M\gamma_N}.
\end{align*}
\end{proof}

\begin{proof}[Proof of Lemma \ref{lemma:phatN-est}]
For any $0<t<0.01M\gamma_N\leq0.01M\E(T)$, similar to Lemma \ref{lemma:gammaN-est}, we also have $\mathcal E_T:=\{|M^{-1}\sum_{i\in\mathcal J}T_i-\E(T)|\leq2\sqrt{t\E(T)/M}+t/M\}$ occurs with probability at least $1-2\exp(-t)$. On the event $\mathcal E_T$,
\begin{align*}
\left|\frac{\E(T)}{M^{-1}\sum_{i\in\mathcal J}T_i}-1\right|\leq2.66\sqrt\frac{t}{M\E(T)},\quad0.79\E(T)\leq M^{-1}\sum_{i\in\mathcal J}T_i\leq1.21\E(T).
\end{align*}
Hence, on the event $\mathcal E_T\cap\mathcal E_\gamma$,
\begin{align*}
&|\phat_{N}-p_{N}|=\left|\frac{M^{-1}\sum_{i\in\mathcal J}\Gamma_i}{M^{-1}\sum_{i\in\mathcal J}T_i}-\frac{\E(\Gamma)}{\E(T)}\right|\leq\frac{|\gammahat_N-\gamma_N|}{M^{-1}\sum_{i\in\mathcal J}T_i}+p_{N}\left|\frac{\E(T)}{M^{-1}\sum_{i\in\mathcal J}T_i}-1\right|\\
&\qquad\leq\frac{1}{0.79\E(T)}\left(2\sqrt\frac{t\gamma_N}{M}+\frac{t}{M}\right)+2.66p_{N}\sqrt\frac{t}{M\E(T)}\\
&\qquad\leq\frac{2.1}{0.79\E(T)}\sqrt\frac{t\gamma_N}{M}+\frac{2.66}{\E(T)}\sqrt \frac{t\gamma_Np_{N}}{M}\leq\frac{5.32}{\E(T)}\sqrt\frac{t\gamma_N}{M},
\end{align*}
since $0<t<0.01M\gamma_N$, $\gamma_N=p_{N}E(T)$, and $p_{N}\leq1$. It follows that
\begin{align*}
|\phat_{N}-p_{N}|\leq\frac{5.32}{\E(T)}\sqrt\frac{0.01M\gamma_N^2}{M}\leq0.54p_{N}.
\end{align*}
Hence,
\begin{align*}
0.46p_{N}\leq\phat_{N}\leq1.54p_{N}.
\end{align*}
Moreover,
\begin{align*}
\left|\frac{p_{N}}{\phat_{N}}-1\right|&=\frac{|\phat_{N}-p_{N}|}{\phat_{N}}\leq\frac{5.32\{\E(T)\}^{-1}\sqrt{t\gamma_N/M}}{0.46p_{N}}\leq12\sqrt\frac{t}{M\gamma_N}.
\end{align*}
\end{proof}

\begin{proof}[Proof of Lemma \ref{lemma:RSC-SP}]
Under Assumption \ref{cond:basic}, we have
\begin{align*}
p_{N}&=\E\{p_N(\bX)\}=\E\{\gamma_N(\bX)/\pi(\bX)\}\geq\E\{\gamma_N(\bX)\}=\gamma_N\;\;\mbox{and}\\
p_{N}&=\E\{\gamma_N(\bX)/\pi(\bX)\}\leq\E\{\gamma_N(\bX)\}/c_0=\gamma_N/c_0.
\end{align*}
Therefore, $p_{N}\asymp\gamma_N$. By Taylor's theorem, we have
\begin{align*}
\delta\elltil_{\bbeta}(\bDelta;\phat_{N};\bbeta_{p}^*;\pihat)&=(M\phat_{N})^{-1}\sum_{i\in\mathcal J}\frac{\Gamma_i}{\pihat(\bX_i)}\exp\{-\bS_i^T(\bbeta_{p}^*+v\bDelta)\}(\bS_i^T\bDelta)^2\\
&\geq (c_0M\phat_{N})^{-1}\sum_{i\in\mathcal J}\Gamma_i\exp\{-\bS_i^T(\bbeta_{p}^*+v\bDelta)\}(\bS_i^T\bDelta)^2,
\end{align*}
with some $v\in(0,1)$ on the event $\mathcal E_\pi$. Hence, Lemma \ref{lemma:RSC-SP} follows directly from Lemma \ref{lemma:RSC_gen}, with $a=\phat_{N}$, $v=v$, and $\phi(u)\equiv\exp(-u)$ using the fact that $p_{N}\asymp\gamma_N$.
\end{proof}

\begin{proof}[Proof of Lemma \ref{lemma:gradient-SP}]
For any $a\in(0,1]$,
$$\nabla_{\bbeta}\elltil_{\bbeta}(\bbeta_{p}^*;a;\pi^*)= M^{-1}\sum_{i\in\mathcal J}\left\{1-\frac{\Gamma_i}{\pi^*(\bX_i)}-\frac{\Gamma_i}{a\pi^*(\bX_i)}\exp(-\bS_i^T\bbeta_{p}^*)\right\}\bS_i.$$
Let $\bU_i=[1-\Gamma_i/\pi^*(\bX_i)-\Gamma_i\exp(-\bS_i^T\bbeta_{p}^*)/\{p_{N}\pi^*(\bX_i)\}]\bS_i$ and set $\bU$ as an independent copy of $\bU_i$. By the construction of $\bbeta_{p}^*$, $\E(\bU)=\boldsymbol0$. For each $1\leq j\leq d$,
\begin{align*}
\sup_{1\leq j\leq d}\|\bU(j)\|_{\psi_2}&\leq\{\log(2)\}^{-1/2}(1+c_0^{-1})+(cp_{N})^{-1}\exp(C_0)\sigma=O(\gamma_N^{-1}),
\end{align*}
since $p_{N}\asymp\gamma_N$ as in Lemma \ref{lemma:RSC-SP}. In addition,
\begin{align*}
&\sup_{1\leq j\leq d}\E\{\bU^2(j)\}\\
&\qquad\leq2\E\left\{1+\Gamma/\pi^*(\bX)\right\}^2+2\sup_{1\leq j\leq d}\E\left\{(c_0p_{N})^{-2}\gamma_N(\bX)\exp(-2\bS^T\bbeta^*)\bS^2(j)\right\}\\
&\qquad\leq(1+c_0^{-1})^2+2(c_0p_{N})^{-2}\exp(2C_0)\|\gamma_N(\cdot)\|_{\P,q}\sup_{1\leq j\leq d}\|\bS(j)\|_{\P,2q/(q-1)}^2=O(\gamma_N^{-1}).
\end{align*}
By Theorem 3.4 of \cite{kuchibhotla2022moving},
$$\P_{\mathcal D_N'}\left(\left\|M^{-1}\sum_{i\in\mathcal J}\bU_i\right\|_\infty>c\sqrt\frac{t+\log{d}}{M\gamma_N}+\frac{c\sqrt{\log{M}}\{t+\log{d}\}}{M\gamma_N}\right)\leq3\exp(-t),$$
with some constant $c>0$. When $M\gamma_N>C_1\log{M}\log{d}$, we have
\begin{align*}
\P_{\mathcal D_N'}\left(\left\|\nabla_{\bbeta}\ell_{\bbeta}(\bbeta^*;\gamma_N)\right\|_\infty>c\sqrt\frac{t+t^2+\log{d}}{M\gamma_N}\right)\leq3\exp(-t),
\end{align*}
with some constant $c>0$. Besides, observe that
\begin{align*}
&\|\nabla_{\bbeta}\elltil_{\bbeta}(\bbeta_{p}^*;\phat_{N};\pi^*)-\nabla_{\bbeta}\elltil_{\bbeta}(\bbeta^*;p_{N};\pi^*)\|_\infty\\
&\qquad=\left|\frac{\phat_{N}}{p_{N}}-1\right| \left\|(Mp_{N})^{-1}\sum_{i\in\mathcal I_k}\frac{\Gamma_i}{\pi^*(\bX_i)}\exp(-\bS_i^T\bbeta^*)\bS_i\right\|_\infty.
\end{align*}
Let $\bW_i=\Gamma_i\exp(-\bS_i^T\bbeta^*)\bS_i/\pi^*(\bX_i)$ set $\bW$ as an independent copy of $\bW_i$. Note that
\begin{align*}
\sup_{1\leq j\leq d}\left\|\bW(j)\right\|_{\psi_2}&\leq c_0^{-1}\exp(C_0)\sup_{1\leq j\leq d}\left\|\bS(j)\right\|_{\psi_2}=O(1),\\
\sup_{1\leq j\leq d}\E\{\bW^2(j)\}&\leq c_0^{-1}\exp(2C_0)\|\gamma_N(\cdot)\|_{\P,q}\sup_{1\leq j\leq d}\|\bS(j)\|_{\P,2q/(q-1)}^2=O(\gamma_N).
\end{align*}
By Theorem 3.4 of \cite{kuchibhotla2022moving} and note that $p_{N}\asymp\gamma_N$, we have
\begin{align*}
\P_{\mathcal D_N'}\biggr(&\left\|(Mp_{N})^{-1}\sum_{i\in\mathcal J}\bW_i-p_{N}^{-1}\E(\bW)\right\|_\infty\\
&\qquad>c\sqrt\frac{t+\log{d}}{M\gamma_N}+\frac{c\sqrt{\log{M}}\{t+\log{d}\}}{M\gamma_N}\biggr)\leq3\exp(-t),
\end{align*}
with some constant $c>0$. Notice that, we also have
$$\|p_{N}^{-1}\E(\bW)\|_\infty\leq c_0^{-1}\exp(C_0)\sup_{1\leq j\leq d}\E|\bS(j)|=O(1).$$
Hence, when $0<t<M\gamma_N/\{\log{M}\log{d}\}$ and $M\gamma_N>C_1\log{M}\log{d}$ with some constant $C_1>0$, we have
\begin{align*}
&\P_{\mathcal D_N'}\left(\left\|(M\gamma_N)^{-1}\sum_{i\in\mathcal J}\bW_i\right\|_\infty>c\left\{1+\sqrt\frac{\log{d}}{M\gamma_N}\right\}\right)\leq3\exp(-t),
\end{align*}
with some constant $c>0$. For any $0<t<0.01M\gamma_N$, by Lemma \ref{lemma:phatN-est}, we have
$$\left|\frac{p_{N}}{\phat_{N}}-1\right|\leq12\sqrt\frac{t}{Mp_{N}}\asymp\sqrt\frac{t}{M\gamma_N},$$
with probability at least $1-4\exp(-t)$. It follows that, when $0<t<M\gamma_N/\{100+\log{M}\log{d}\}$ and $M\gamma_N>C_1\log{M}\log{d}$,
\begin{align*}
\P_{\mathcal D_N'}\left(\|\nabla_{\bbeta}\elltil_{\bbeta}(\bbeta_{p}^*;\phat_{N};\pi^*)-\nabla_{\bbeta}\elltil_{\bbeta}(\bbeta_{p}^*;p_{N};\pi^*)\|_\infty>c\sqrt\frac{t+\log{d}}{M\gamma_N}\right)\leq7\exp(-t),
\end{align*}
with some constant $c>0$. Therefore, we have
\begin{align*}
\P_{\mathcal D_N'}\left(\|\nabla_{\bbeta}\elltil_{\bbeta}(\bbeta_{p}^*;\phat_{N};\pi^*)\|_\infty>\kappa_3\sqrt\frac{t+\log{d}}{M\gamma_N}\right)\leq10\exp(-t),
\end{align*}
with some constant $\kappa_3>0$.
\end{proof}

\begin{proof}[Proof of Lemma \ref{lemma:rpi}]
By the Cauchy-Schwarz inequality, for any $\bDelta\in\R^d$,
\begin{align*}
|\br_\pi^T\bDelta|^2&\leq M^{-2}\sum_{i\in\mathcal J}\Gamma_i\left\{\frac{1}{\pihat(\bX_i)}-\frac{1}{\pi^*(\bX_i)}\right\}^2\sum_{i\in\mathcal J}\Gamma_i\left\{1+\phat_{N}^{-1}\exp(-\bS_i^T\bbeta_{p}^*)\right\}^2|\bS_i^T\bDelta|^2\\
&\leq 2c_0^{-4}M^{-2}\{1+\phat_{N}^{-2}\exp(2C_{\bX}C_{\bbeta})\}\sum_{i\in\mathcal J}\Gamma_i\left\{\pihat(\bX_i)-\pi^*(\bX_i)\right\}^2\sum_{i\in\mathcal J}\Gamma_i|\bS_i^T\bDelta|^2.
\end{align*}
By Lemma \ref{lemma:phatN-est} and note that $p_{N}\asymp\gamma_N$, we have $\phat_{N}^{-2}=O_p(\gamma_N^{-2})$. Under Assumptions \ref{cond:NP-est}, we have
 $(M\gamma_N)^{-1}\sum_{i\in\mathcal J}\Gamma_i\left\{\pihat(\bX_i)-\pi^*(\bX_i)\right\}^2\leq\zeta_N^2$ on the event $\mathcal E_\zeta$, with $\P_{\mathcal D_N'}(\mathcal E_\zeta)=1-o(1)$. In addition, by part (b) of Lemma \ref{lemma:emp}, we also have
\begin{align*}
\frac{\sum_{i\in\mathcal I_j}\Gamma_i|\bS_i^T\bDelta|^2}{M\gamma_N}\leq c\left\{1+\sqrt\frac{t}{N\gamma_N}+\frac{t\log{N}}{N\gamma_N}\right\}\left(\frac{\|\bDelta\|_1^2\log{d}\log{N}}{N\gamma_N}+\|\bDelta\|_2^2\right),
\end{align*}
uniformly for all $\bDelta\in\R^d$ with probability at least $1-3\exp(-t)$ and some constant $c>0$. Set $t=N\gamma_N/\log{N}$ and note that $M\asymp N$, we have
\begin{align*}
\P_{\mathcal D_N'}\left(\sup_{\bDelta\in\R^d\setminus\{\bzero\}}\frac{(M\gamma_N)^{-1}\sum_{i\in\mathcal I_j}\Gamma_i|\bS_i^T\bDelta|^2}{\|\bDelta\|_1^2\log{d}\log{N}/(M\gamma_N)+\|\bDelta\|_2^2}\leq c\right)\geq1-3\exp\{N\gamma_N/\log{N}\}.
\end{align*}
Combining the results above, on the event $\mathcal E_\zeta$, we have uniformly for all $\bDelta\in\R^d$,
\begin{align*}
|\br_\pi^T\bDelta|\leq c\zeta_N\left(\|\bDelta\|_1\sqrt\frac{\log{d}\log{N}}{M\gamma_N}+\|\bDelta\|_2\right)
\end{align*}
with probability at least $1-3\exp\{N\gamma_N/\log{N}\}$ and some constant $c>0$.
\end{proof}

\begin{proof}[Proof of Lemma \ref{lemma:Fdelta}]
Let $S\subset\{1,\dots,d\}$ be the support set of $\bbeta_{p}^*$. For any $\bDelta\in\R^d$, we have $\|\bbeta_{p}^*+\bDelta\|_1=\|\bbeta_{p,1,S}^*+\bDelta_S\|_1+\|\bDelta_{S^c}\|_1\geq\|\bbeta_{p,1,S}^*\|_1-\|\bDelta_S\|_1+\|\bDelta_{S^c}\|_1=\|\bbeta_{p}^*\|_1-\|\bDelta_S\|_1+\|\bDelta_{S^c}\|_1$. On the event $\Bcaltil_2$, we also have
$\|\nabla_{\bbeta}\elltil_{\bbeta}(\bbeta_{p}^*;\phat_{N};\pi^*)\|_\infty\leq\kappa_3\sqrt{\{t+\log{d}\}/(M\gamma_N)}$. Choose $\lambda_{\bbeta}>2\kappa_3\sqrt{\{t+\log{d}\}/(M\gamma_N)}$, then it follows that $|\nabla_{\bbeta}\elltil_{\bbeta}(\bbeta_{p}^*;\phat_{N};\pi^*)^T\bDelta|\leq\lambda_{\bbeta}\|\bDelta\|_1/2\leq\lambda_{\bbeta}(\|\bDelta_{S^c}\|_1+\|\bDelta_S\|_1)/2$. Together with \eqref{def:FDelta}, we have
\begin{align}
\mathcal F(\bDelta)&\geq\delta\elltil_{\bbeta}(\bDelta;\phat_{N};\bbeta_{p}^*;\pihat)+\lambda_{\bbeta}(\|\bDelta_{S^c}\|_1-\|\bDelta_S\|_1-\|\bDelta_{S^c}\|_1/2-\|\bDelta_S\|_1/2)+\br_\pi^T\bDelta\nonumber\\
&=\delta\elltil_{\bbeta}(\bDelta;\phat_{N};\bbeta_{p}^*;\pihat)+\lambda_{\bbeta}(\|\bDelta\|_1-4\|\bDelta_S\|_1)/2+\br_\pi^T\bDelta.\label{eq:basic-SP}
\end{align}
Moreover, on the event $\Bcaltil_3$,
\begin{align}
|\br_\pi^T\bDelta|\leq c\zeta_N\left(\|\bDelta\|_1\sqrt\frac{\log{d}\log{N}}{M\gamma_N}+\|\bDelta\|_2\right)\leq\lambda_{\bbeta}\|\bDelta\|_1/4+c\zeta_N\|\bDelta\|_2,\label{bound:rpi-Delta}
\end{align}
since $\lambda_{\bbeta}>2\kappa_3\sqrt{\log{d}/(M\gamma_N)}$ and $\zeta_N\leq\kappa_3/\{2c\sqrt{\log{N}}\}$ under Assumption \ref{cond:NP-est} when $N$ is large enough. Together with \eqref{eq:basic-SP} and \eqref{bound:deltaell-SP}, we have
\begin{align}
\mathcal F(\bDelta)&\geq \delta\elltil_{\bbeta}(\bDelta;\phat_{N};\bbeta_{p}^*;\pihat)+\lambda_{\bbeta}(\|\bDelta\|_1-4\|\bDelta_S\|_1)/2-\lambda_{\bbeta}\|\bDelta\|_1/4-c\zeta_N\|\bDelta\|_2\nonumber\\
&\geq\delta\elltil_{\bbeta}(\bDelta;\phat_{N};\bbeta_{p}^*;\pihat)+\lambda_{\bbeta}\|\bDelta\|_1/4-(2\sqrt{s_{p}}\lambda_{\bbeta}+c\zeta_N)\|\bDelta\|_2,\label{bound:FDelta}
\end{align}
since $\|\bDelta_S\|_1\leq\sqrt{s_{p}}\|\bDelta_S\|_2\leq\sqrt{s_{p}}\|\bDelta\|_2$.

Additionally, if $\bDelta\in\mathcal K(r_N,1)$, we have $\|\bDelta\|_1\leq r_N\|\bDelta\|_2$ and $\|\bDelta\|_2=1$. Hence, on the event $\Bcaltil_1\cap\mathcal E_p$,
\begin{align}
&\delta\elltil_{\bbeta}(\bDelta;\phat_{N};\bbeta_{p}^*;\pihat)\geq c\left\{\|\bDelta\|_2^2-\frac{\log{d}}{M\gamma_N}\|\bDelta\|_1^2\right\}\nonumber\\
&\qquad\geq c\left\{\|\bDelta\|_2^2-\frac{r_N^2\log{d}}{M\gamma_N}\|\bDelta\|_2^2\right\}\geq c\|\bDelta\|_2^2/2=c/2,\label{bound:deltaell-SP}
\end{align}
when $0<t< 0.01M\gamma_N$ and $N$ is large enough, with some constant $c>0$. Besides, we also have $(2\sqrt{s_{p}}\lambda_{\bbeta}+c\zeta_N)\|\bDelta\|_2\leq c/2$ when $N$ is large enough since $\|\bDelta\|_2=1$ and $2\sqrt{s_{p}}\lambda_{\bbeta}+c\zeta_N\asymp\sqrt{s_{p}\log{d}/(M\gamma_N)}+\zeta_N=o(1)$. Together with \eqref{bound:FDelta} and \eqref{bound:deltaell-SP}, we have
\begin{align*}
\mathcal F(\bDelta)\geq c/4+\lambda_{\bbeta}\|\bDelta\|_1/4>0.
\end{align*}
\end{proof}

\subsection{Proof of Theorem \ref{thm:PS-SP}}
In the following proofs, we will consider the events $\mathcal E_\gamma$, $\mathcal E_p$, $\Bcaltil_1$, $\Bcaltil_2$, and $\Bcaltil_3$, defined as in \eqref{def:Egamma}-\eqref{def:B3}. Additionally, the events $\mathcal E_\pi$ and $\mathcal E_\zeta$ are defined in Assumption \ref{cond:NP-est}, and $\Ecaltil_{\bbetahat}$ is defined in Theorem \ref{thm:PS-SP}.
\begin{proof}[Proof of Theorem \ref{thm:PS-SP}]
{Under Assumption \ref{cond:bound-SP}, \( |\bS^\top\bbeta_{p}^*| < C_0 < C \) almost surely. Consequently, \( \bbeta = \bbeta_{p}^* \) satisfies the constraint \( \max_{i\in\mathcal{I}_k} |\bS_i^\top\bbeta| < C \) almost surely, ensuring that it belongs to the feasible set of the optimization problem defined for \( \bbetahat_{p} \). This guarantees that the feasible set is nonempty. Moreover, by the construction of \( \bbetahat_{p} \), we have}
\begin{align}
&\max_{i\in\mathcal I_1}|\bS_i^\top\bbetahat_{p}|<C,\label{constraint}\\
&\elltil_{\bbeta}(\bbetahat_{p};\phat_{N};\pihat)+\lambda_{\bbeta}\|\bbetahat_{p}\|_1\leq\elltil_{\bbeta}(\bbeta_{p}^*;\phat_{N};\pihat)+\lambda_{\bbeta}\|\bbeta_{p}^*\|_1\nonumber.
\end{align}

Let $\bDelta=\bbetahat_{p}-\bbeta_{p}^*$, then we have $\mathcal F(\bDelta)\leq0$. Condition on the event $\Bcaltil_2\cap\Bcaltil_3$. Together with Lemma \ref{lemma:Fdelta}, when $N$ is large enough, we have
\begin{align}
0\geq\mathcal F(\bDelta)\geq\delta\elltil_{\bbeta}(\bDelta;\phat_{N};\bbeta_{p}^*;\pihat)+\lambda_{\bbeta}\|\bDelta\|_1/4-(2\sqrt{s_{p}}\lambda_{\bbeta}+c\zeta_N)\|\bDelta\|_2.\label{eq:basic-SP'}
\end{align}
Since the loss function $\ell_{\bbeta}(\cdot;\phat_{N};\pihat)$ is convex, $\delta\elltil_{\bbeta}(\bDelta;\phat_{N};\bbeta_{p}^*;\pihat)\geq0$, and it follows that $\lambda_{\bbeta}\|\bDelta\|_1/4\leq (2\sqrt{s_{p}}\lambda_{\bbeta}+c\zeta_N)\|\bDelta\|_2$. That is,
\begin{align}
\|\bDelta\|_1\leq4(2\sqrt{s_{p}}+c\zeta_N/\lambda_{\bbeta})\|\bDelta\|_2\leq\left(8\sqrt{s_{p}}+\frac{2c\zeta_N}{\kappa_3}\sqrt\frac{M\gamma_N}{\log{d}}\right)\|\bDelta\|_2,\label{bound:rN}
\end{align}
since $\lambda_{\bbeta}>2\kappa_3\sqrt{\log{d}/(M\gamma_N)}$. Let $r_N=8\sqrt{s_{p}}+\kappa_3^{-1}2c\zeta_N\sqrt{M\gamma_N/\log{d}}$. Then $\|\bDelta\|_1\leq r_N\|\bDelta_2\|_2$ with $r_N=o(\sqrt{M\gamma_N/\log{d}})$ under Assumption \ref{cond:NP-est} and since $s_{p}=o(M\gamma_N/\log{d})$.

In the following, we further prove that $\|\bDelta\|_2\leq1$ by contradiction. Suppose that $\|\bDelta\|_2>1$. Define $\bDeltatil:=\bDelta/\|\bDelta\|_2$, then $\|\bDeltatil\|_2=1$ and $\|\bDeltatil\|_1=\|\bDelta\|_1/\|\bDelta\|_2\leq r_N=r_N\|\bDeltatil\|_2$. That is, $\bDeltatil\in\mathcal K(r_N,1)$. By Lemma \ref{lemma:Fdelta}, $\mathcal F(\bDeltatil)>0$. Define $u=1/\|\bDelta\|_2$, then $0<u<1$. Note that $\mathcal F(\cdot)$ is a convex function, $\mathcal F(\bzero)=0$, and $\mathcal F(\bDelta)\leq0$. Hence,
\begin{align*}
0<\mathcal F(\bDeltatil)=\mathcal F(u\bDelta+(1-u)\bzero)\leq u\mathcal F(\bDelta)+(1-u)\mathcal F(\bzero)=u\mathcal F(\bDelta)\leq0.
\end{align*}
Therefore, we must have $\|\bDelta\|_2\leq1$.

Now, further condition on $\Bcaltil_1\cap\mathcal E_p$. As in \eqref{bound:deltaell-SP}, we also have
\begin{align}
&\delta\elltil_{\bbeta}(\bDelta;\phat_{N};\bbeta_{p}^*;\pihat)\geq c\left\{\|\bDelta\|_2^2-\frac{\log{d}}{M\gamma_N}\|\bDelta\|_1^2\right\}\nonumber\\
&\qquad\geq c\left\{\|\bDelta\|_2^2-\frac{r_N^2\log{d}}{M\gamma_N}\|\bDelta\|_2^2\right\}\geq c\|\bDelta\|_2^2/2,\label{bound:deltaell-SP-2norm}
\end{align}
when $N$ is large enough, with some constant $c>0$. Together with \eqref{eq:basic-SP'}, we have
\begin{align*}
c\|\bDelta\|_2^2/2+\lambda_{\bbeta}\|\bDelta\|_1/4\leq(2\sqrt{s_{p}}\lambda_{\bbeta}+c\zeta_N)\|\bDelta\|_2.
\end{align*}
It follows that $\bDelta=\bbetahat_{p}-\bbeta_{p}^*$ satisfies
\begin{align}
\|\bbetahat_{p}-\bbeta_{p}^*\|_2&\leq4\sqrt{s_{p}}\lambda_{\bbeta}/c+2\zeta_N,\label{bound:beta-2error-SP}\\
\|\bbetahat_{p}-\bbeta_{p}^*\|_1&\leq r_N\|\bDelta\|_2\leq\left(8\sqrt{s_{p}}+\frac{2c\zeta_N}{\kappa_3}\sqrt\frac{M\gamma_N}{\log{d}}\right)\left(4\sqrt{s_{p}}\lambda_{\bbeta}/c+2\zeta_N\right).\label{bound:beta-1error-SP}
\end{align}
Since $\lambda_{\bbeta}\asymp\sqrt{\log{d}/(M\gamma_N)}$ and $\Bcaltil_1\cap\Bcaltil_2\cap\Bcaltil_3\cap\mathcal E_p\cap\mathcal E_\pi=1-o(1)$, we conclude that
\begin{align*}
\|\bDelta\|_2=O_p\left(\sqrt\frac{s_{p}\log{d}}{M\gamma_N}+\zeta_N\right),\;\;\|\bDelta\|_1=O_p\left(s_{p}\sqrt\frac{\log{d}}{M\gamma_N}+\zeta_N^2\sqrt\frac{M\gamma_N}{\log{d}}\right).
\end{align*}
\end{proof}

\subsection{Proof of the results in Section \ref{sec:nuisance-OR-SP}}

\begin{proof}[Proof of Lemma \ref{lemma:RSC2-SP}]
By Lemma \ref{lemma:RSC-SP}, $p_{N}\asymp\gamma_N$. By the construction of $\bbetahat_{p}$, $\max_{i\in\mathcal J}|\bS_i^T\bbetahat_{p}|<C_0$. Under Assumptions \ref{cond:NP-est} and on the event $\mathcal E_\pi$, we have
\begin{align}
&\delta\elltil_{\balpha}(\bDelta;\phat_{N},\bbetahat;\balphatil^*;\pihat)=(M\phat_{N})^{-1}\sum_{i\in\mathcal J}\frac{\Gamma_i}{\pihat(\bX_i)}\exp(-\bS_i^T\bbetahat_{p})(\bS_i^T\bDelta)^2\nonumber\\
&\qquad\geq\exp(-C_0)(c_0M\phat_{N})^{-1}\sum_{i\in\mathcal J}\Gamma_i(\bS_i^T\bDelta)^2.\label{bound:deltaell-SP'}
\end{align}
Hence, Lemma \ref{lemma:RSC2-SP} follows directly from Lemma \ref{lemma:RSC_gen}, with $a=\phat_{N}$, $v=0$, and $\phi(u)\equiv1$ using the fact that $p_{N}\asymp\gamma_N$.
\end{proof}

\begin{proof}[Proof of Lemma \ref{lemma:gradient2-SP}]
For each $1\leq j\leq d$ and $i\in\mathcal J$, we have
\begin{align*}
&\E_{\mathcal D_N'}\left\{\bV_{i,1}(j)\mid(R_i,T_i,\bX_i)_{i\in\mathcal J}\right\}\\
&\qquad\overset{(i)}{=}-2\Gamma_i\pihat^{-1}(\bX_i)\exp(-\bS_i^T\bbetahat_{p})\bS_i(j)\E_{\mathcal D_N'}\left\{\varepsilon_{i,1}\mid(R_i,T_i,\bX_i)_{i\in\mathcal J}\right\}\\
&\qquad\overset{(ii)}{=}-2\Gamma_i\pihat^{-1}(\bX_i)\exp(-\bS_i^T\bbetahat_{p})\bS_i(j)\E_{\mathcal D_N'}\left(\varepsilon_{i,1}\mid\bX_i\right)=0,
\end{align*}
where we denote $\bx(j)$ as the $j$-th element of any vector $\bx$. Here, (i) holds since by construction, $\pihat(\cdot)$ and $\bbetahat_{p}$ only depend on $(R_i,T_i,\bX_i)_{i\in\mathcal J}$; (ii) holds under Assumption \ref{cond:basic}. Note that $(\bV_{i,1}(j))_{i\in\mathcal J}$ are independent conditional on $(R_i,T_i,\bX_i)_{i\in\mathcal J}$, by Proposition 2.5 (Hoeffding bound) of \cite{wainwright2019high}, for any $u>0$,
$$\P_{\mathcal D_N'}\left(\left|\sum_{i\in\mathcal J}\bV_{i,1}(j)\right|>u\mid(R_i,T_i,\bX_i)_{i\in\mathcal J}\right)\leq2\exp\left(-\frac{u^2}{2\sum_{i\in\mathcal J}\sigma_{ij}^2}\right),$$
if $\bV_{i,1}(j)$ is sub-gaussian with parameter $\sigma_{ij}>0$ conditional on $(R_i,T_i,\bX_i)_{i\in\mathcal J}$. Condition on $\mathcal E_\pi$. Under Assumption \ref{cond:subG'}, together with \eqref{constraint}, $\sigma_{ij}^2$ can be chosen such that
$$\sigma_{ij}^2\leq4\sigma_\varepsilon^2\Gamma_i\pihat^{-2}(\bX_i)\exp(-2\bS_i^T\bbetahat_{p})\bS_i^2(j)\leq c^2\Gamma_i\bS_i^2(j),$$
with some constant $c>0$. Let $u=c(1+t)\sqrt{\log{d}\sum_{i\in\mathcal J}\Gamma_i\bS_i^2(j)}$. By the union bound,
\begin{align*}
\P_{\mathcal D_N'}\left(\left\|\sum_{i\in\mathcal J}\bV_{i,1}\right\|_\infty>u\mid(R_i,T_i,\bX_i)_{i\in\mathcal J}\right)\leq2\exp\{-t\log{d}\}.
\end{align*}
By Lemmas \ref{lemma:phatN-est} and \ref{lemma:RSC-SP}, with probability at least $1-4\exp(-t)$, $\phat_{N}\geq0.46p_{N}\asymp\gamma_N$. Under Assumptions \ref{cond:subG} and \ref{cond:tail},
\begin{align*}
\E\left\{\sum_{i\in\mathcal J}\Gamma_i\bS_i^2(j)\right\}&=M\E\{\gamma_N(\bX)\bS_i^2(j)\}=O(M\gamma_N),\\
\Var\left\{\sum_{i\in\mathcal J}\Gamma_i\bS_i^2(j)\right\}&=M\Var\{\Gamma_i\bS_i^2(j)\}\leq M\E\{\gamma_N(\bX)\bS_i^4(j)\}=O(M\gamma_N).
\end{align*}
By Chebyshev's inequality, $\sum_{i\in\mathcal J}\Gamma_i\bS_i^2(j)\leq c(M\gamma_N+\sqrt{tM\gamma_N})\leq c'M\gamma_N$ with probability at least $1-t^{-1}$ and some constant $c,c'>0$ when $0<t<0.01Mp_{N}\asymp M\gamma_N$. Hence, with probability at least $1-4\exp(-t)-t^{-1}$ and some constant $\kappa_4>0$,
\begin{align*}
\frac{u}{M\phat_{N}}&\leq\frac{\kappa_4(1+t)\sqrt{\log{d}M\gamma_N}}{M\gamma_N}=\kappa_4(1+t)\sqrt\frac{\log{d}}{M\gamma_N}.
\end{align*}
Therefore, on the event $\mathcal E_\pi$,
$$\left\|(M\phat_{N})^{-1}\sum_{i\in\mathcal J}\bV_{i,1}\right\|_\infty\leq\kappa_4(1+t)\sqrt\frac{\log{d}}{M\gamma_N},$$
with probability at least $1-2\exp\{-t\log{d}\}-4\exp(-t)-t^{-1}$.
\end{proof}

\begin{proof}[Proof of Lemma \ref{lemma:gradient3-SP}]
{By the construction of $\balphatil^*$, we have $\E(\bV_{i,2})=\bzero\in\R^d$. Under Assumptions \ref{cond:NP-est}, \ref{cond:bound-SP}, and \ref{cond:approx}, we have $\sup_{1\leq j\leq d}\|\bV_{i,2}(j)\|_{\psi_2}\leq2c_0^{-1}\exp(C_0)\sigma e_m$. In addition, under Assumption \ref{cond:tail}, with some $r>1$ satisfies $1/r+1/q=1$,
\begin{align*}
&\sup_{1\leq j\leq d}\E\{\bV_{i,2}^2(j)\}\leq4c_0^{-2}\exp(2C_0)e_m^2\E\{\gamma_N(\bX_i)\bS_i^2(j)\}\\
&\qquad\leq4c_0^{-2}e_m^2\exp(2C_0)\|\gamma_N(\cdot)\|_{\P,q}\|\bS_i(j)\|_{\P,2r}^2=O(\gamma_Ne_m^2).
\end{align*}
By Theorem 3.4 of \cite{kuchibhotla2022moving}, with some constant $c>0$,
\begin{align*}
&\P\left(\left\|M^{-1}\sum_{i\in\mathcal J}\bV_{i,2}\right\|_\infty>ce_m\sqrt\frac{\gamma_N\{t+\log{d}\}}{M}+ce_m\frac{\sqrt{\log{M}}\{t+\log{d}\}}{M}\right)\\
&\qquad\leq3\exp(-t).
\end{align*}
Together with Lemma \ref{lemma:phatN-est}, when $0<t<M\gamma_N/\{100+\log{M}\log{d}\}$ and $M\gamma_N>C_1\log{M}\log{d}$,
$$\P\left(\left\|(M\phat_{N})^{-1}\sum_{i\in\mathcal J}\bV_{i,2}\right\|_\infty>\kappa_5\sqrt\frac{t+\log{d}}{M\gamma_N}\right)\leq7\exp(-t),$$
with some constant $\kappa_5>0$, as long as $e_m=O(1)$.}
\end{proof}

\subsection{Proof of Theorem \ref{thm:OR-SP-mis}}

\begin{proof}[Proof of Theorem \ref{thm:OR-SP-mis}]
For any $0<t<M\gamma_N/\{100+\log{M}\log{d}\}$, consider some $\lambda_{\bbeta}>2\kappa_3\sqrt{\{t+\log{d}\}(M\gamma_N)}$ and $\lambda_{\balpha}>4(\kappa_4+\kappa_5)\sqrt{\{t+\log{d}\}(M\gamma_N)}$. Condition on the event $\Acaltil_2\cap\Acaltil_3$. Then, we have
\begin{align}\label{event:gradient}
4\left\|(M\phat_{N})^{-1}\sum_{i\in\mathcal J}(\bV_{i,1}+\bV_{i,2})\right\|_\infty\leq\lambda_{\balpha}.
\end{align}
By Lemmas \ref{lemma:gradient2-SP} and \ref{lemma:gradient3-SP}, $\P_{\mathcal D_N'}(\Acaltil_2\cap\Acaltil_3)\geq1-2\exp\{-t\log{d}\}-11\exp(-t)-t^{-1}-o(1).$
By the construction of $\balphatil$, we have
\begin{align*}
&\elltil_{\balpha}(\balphatil;\phat_{N},\bbetahat_{p};\pihat)+\lambda_{\balpha}\|\balphatil\|_1\leq\elltil_{\balpha}(\balphatil^*;\phat_{N},\bbetahat_{p};\pihat)+\lambda_{\balpha}\|\balphatil^*\|_1.
\end{align*}
Let $\bDelta=\balphatil-\balphatil^*$. Then,
$$\delta\elltil_{\balpha}(\bDelta;\phat_{N},\bbetahat;\balphatil^*;\pihat)+\lambda_{\balpha}\|\balphatil\|_1\leq\nabla_{\balpha}\elltil_{\balpha}(\balphatil^*;\phat_{N},\bbetahat_{p};\pihat)^T\bDelta+\lambda_{\balpha}\|\balphatil^*\|_1.$$
{Together with \eqref{rep:gradient} and \eqref{event:gradient},
\begin{align}
&\delta\elltil_{\balpha}(\bDelta;\phat_{N},\bbetahat;\balphatil^*;\pihat)+\lambda_{\balpha}\|\balphatil\|_1\nonumber\\
&\qquad\leq(M\phat_{N})^{-1}\sum_{i\in\mathcal J}(\bV_{i,1}+\bV_{i,2})^T\bDelta+(\br_{p}+\br_{p,2})^T\bDelta+\lambda_{\balpha}\|\balphatil^*\|_1\nonumber\\
&\qquad\leq\lambda_{\balpha}\|\bDelta\|_1/4+(\br_{p}+\br_{p,2})^T\bDelta+\lambda_{\balpha}\|\balphatil^*\|_1.\label{eq:basic-OR-SP}
\end{align}

Now, we first consider the error term $\br_{p}^T\bDelta$. For any $a_1>0$, on the event $\mathcal E_\zeta$, we have
\begin{align*}
|\br_{p}^T\bDelta|&\leq(M\phat_{N})^{-1}\sum_{i\in\mathcal J}\Gamma_i\left[a_1\left\{\frac{1}{\pihat(\bX_i)}-\frac{1}{\pi^*(\bX_i)}\right\}^2\varepsilon_{i,2}^2+\exp(-2\bS_i^T\bbeta_{p}^*)|\bS_i^T\bDelta|^2/a_1\right]\\
&\leq(\gamma_N/\phat_{N})a_1t\zeta_N^2e_m^2+\exp(2C_0)a_1^{-1}(M\phat_{N})^{-1}\sum_{i\in\mathcal J}\Gamma_i|\bS_i^T\bDelta|^2.
\end{align*}
Condition on the event $\mathcal E_p$. By Lemma \ref{lemma:phatN-est}, $\phat_{N}^{-1}\leq(0.46p_{N})^{-1}\leq2.2\gamma_N^{-1}$ as $\gamma_N\leq p_{N}$. Additionally, by \eqref{bound:deltaell-SP'}, we have
\begin{align*}
(M\phat_{N})^{-1}\sum_{i\in\mathcal J}\Gamma_i|\bS_i^T\bDelta|^2\leq c_0\exp(C_0)\delta\elltil_{\balpha}(\bDelta;\phat_{N},\bbetahat;\balphatil^*;\pihat).
\end{align*}
Choose $a_1=4c_0\exp(3C_0)$. Then,
\begin{align}
|\br_{p}^T\bDelta|\leq 2.2a_1t\zeta_N^2e_m^2+\delta\elltil_{\balpha}(\bDelta;\phat_{N},\bbetahat;\balphatil^*;\pihat)/4.\label{bound:rp1}
\end{align}
Now, we control the term $\br_{p,2}^T\bDelta$. By Taylor's theorem, with some $\bbetatil$ lies between $\bbeta_{p}^*$ and $\bbetahat_{p}$, we have
\begin{align*}
|\br_{p,2}^T\bDelta|&=\left|2(M\phat_{N})^{-1}\sum_{i\in\mathcal J}\exp(-\bS_i^T\bbetatil)\Gamma_i\bS_i^T(\bbetahat_{p}-\bbeta_{p}^*)\bS_i^T\bDelta\varepsilon_{i,2}/\pihat(\bX_i)\right|\\
&\leq\exp(C_0)(c_0M\phat_{N})^{-1}\sum_{i\in\mathcal J}\Gamma_i\left[a_2e_m^2\left\{\bS_i^T(\bbetahat_{p}-\bbeta_{p}^*)\right\}^2+a_2^{-1}|\bS_i^T\bDelta|^2\right],
\end{align*}
for any $a_2>0$ since $\max_{i\in\mathcal J}|\bS_i^T\bbetatil|\leq\max(\max_{i\in\mathcal J}|\bS_i^T\bbeta^*|,\max_{i\in\mathcal J}|\bS_i^T\bbetahat_{p}|)<C_0$ almost surely. Choose $s_0:=\lceil N\gamma_N/\{\log{d}\log{N}\}\rceil$, then $s_{p}=o(s_0)$ and $\zeta_N^2\lambda_{\bbeta}^2/s_0=o(1)$ since $N\gamma_N\gg\log{d}\log{N}$, $\zeta_N=o(1)$, $\lambda_{\bbeta}=o(1)$, and $s_{p}=o(N\gamma_N/\{\log{d}\log{N}\})$. When \eqref{bound:beta-2error-SP} and \eqref{bound:beta-1error-SP} hold, we have
\begin{align*}
&\frac{\|\bbetahat_{p}-\bbeta_{p}^*\|_1^2}{s_0}+\|\bbetahat_{p}-\bbeta_{p}^*\|_2^2=O\left((s_{p}\lambda_{\bbeta}^2+\zeta_N^2)(1+s_{p}/s_0)+(\lambda_{\bbeta}^2+\zeta_N^2)\zeta_N^2\lambda_{\bbeta}^2/s_0\right)\\
&\qquad=O(s_{p}\lambda_{\bbeta}^2+\zeta_N^2).
\end{align*}
Together with part (b) of Lemma \ref{lemma:emp} and the fact that $\phat_{N}^{-1}\leq2.2\gamma_N^{-1}$, we also have
\begin{align*}
(M\phat_{N})^{-1}\sum_{i\in\mathcal J}\Gamma_i\left\{\bS_i^T(\bbetahat_{p}-\bbeta_{p}^*)\right\}^2\leq c(s_{p}\lambda_{\bbeta}^2+\zeta_N^2)
\end{align*}
with probability at least $1-3\exp(-t)$ and some constant $c>0$ when $N$ is large enough, conditional on $\Bcaltil_1\cap\Bcaltil_2\cap\Bcaltil_3\cap\mathcal E_p\cap\mathcal E_\pi$ for any $0<t<M\gamma_N/\{100+\log{M}\log{d}\}$. In addition, chose $a_2=4c_0\exp(2C_0)$. Together with \eqref{bound:deltaell-SP'}, we also have
\begin{align*}
\exp(C_0)(a_2c_0M\phat_{N})^{-1}\sum_{i\in\mathcal J}\Gamma_i|\bS_i^T\bDelta|^2\leq\delta\elltil_{\balpha}(\bDelta;\phat_{N},\bbetahat;\balphatil^*;\pihat)/4.
\end{align*}
Therefore,
$$|\br_{p,2}^T\bDelta|\leq4\exp(3C_0)c(s_{p}\lambda_{\bbeta}^2+\zeta_N^2)e_m^2+\delta\elltil_{\balpha}(\bDelta;\phat_{N},\bbetahat;\balphatil^*;\pihat)/4.$$
Together with \eqref{bound:rp1}, with some constant $c>0$,
\begin{align*}
|(\br_{p}+\br_{p,2})^T\bDelta|\leq c\{s_{p}\lambda_{\bbeta}^2+(1+t)\zeta_N^2\}e_m^2+\delta\elltil_{\balpha}(\bDelta;\phat_{N},\bbetahat;\balphatil^*;\pihat)/2.
\end{align*}
Besides, with $S\subseteq\{1,\dots,d\}$ denoting the support set of $\balphatil^*$, we have $\|\balphatil\|_1=\|\balphatil_S\|_1+\|\bDelta_{S^c}\|_1\geq\|\balphatil_S^*\|_1-\|\bDelta_S\|_1+\|\bDelta_{S^c}\|_1=\|\balphatil^*\|_1-\|\bDelta_S\|_1+\|\bDelta_{S^c}\|_1$. Together with \eqref{eq:basic-OR-SP}, it follows that
\begin{align*}
&\delta\elltil_{\balpha}(\bDelta;\phat_{N},\bbetahat;\balphatil^*;\pihat)+\lambda_{\balpha}(\|\bDelta_{S^c}\|_1-\|\bDelta_S\|_1)\\
&\qquad\leq\lambda_{\balpha}\|\bDelta\|_1/2+c\{s_{p}\lambda_{\bbeta}^2+(1+t)\zeta_N^2\}e_m^2+\delta\elltil_{\balpha}(\bDelta;\phat_{N},\bbetahat;\balphatil^*;\pihat)/2.
\end{align*}
That is,
\begin{align}
\delta\elltil_{\balpha}(\bDelta;\phat_{N},\bbetahat;\balphatil^*;\pihat)+\lambda_{\balpha}\|\bDelta_{S^c}\|_1\leq3\lambda_{\balpha}\|\bDelta_S\|_1+2c\{s_{p}\lambda_{\bbeta}^2+(1+t)\zeta_N^2\}e_m^2.\label{eq:basic-OR-SP'}
\end{align}
\textbf{Case 1.} If $\lambda_{\balpha}\|\bDelta_S\|_1\leq2c\{s_{p}\lambda_{\bbeta}^2+(1+t)\zeta_N^2\}e_m^2$. Then, we have
$$\lambda_{\balpha}\|\bDelta_{S^c}\|_1\leq8c\{s_{p}\lambda_{\bbeta}^2+(1+t)\zeta_N^2\}e_m^2,$$
and hence
$$\|\bDelta\|_1=\|\bDelta_{S}\|_1+\|\bDelta_{S^c}\|_1\leq ce_m^2\left\{(t+1)\zeta_N^2+\frac{s_{p}\log{d}}{M\gamma_N}\right\}\sqrt{\frac{M\gamma_N}{\log{d}}},$$
with some constant $c>0$ when $\lambda_{\balpha}\asymp\lambda_{\bbeta}\asymp\sqrt\frac{t+\log{d}}{M\gamma_N}$. Additionally, by \eqref{eq:basic-OR-SP'} and on the event $\Acaltil_1\cap\mathcal E_p$, we have
\begin{align*}
\|\bDelta\|_2^2&\leq\frac{\phat_{N}}{\kappa_1p_{N}}\left\{\delta\elltil_{\balpha}(\bDelta;\phat_{N},\bbetahat;\balphatil^*;\pihat)+\frac{\kappa_2\log{d}\|\bDelta\|_1^2}{M\gamma_N}\right\}\\
&\leq c\{(t+1)\zeta_N^2+s_{p}\lambda_{\bbeta}^2\}e_m^2+ce_m^4\left\{(t+1)\zeta_N^2+\frac{s_{p}\log{d}}{M\gamma_N}\right\}^2\frac{M\gamma_N}{\log{d}}\cdot\frac{\log{d}}{M\gamma_N},
\end{align*}
with some constant $c>0$. As $\zeta_N=o(1)$ and $s_{p}=o(M\gamma_N/\log{d})$, when $N$ is large enough, we have
\begin{align*}
\|\bDelta\|_2\leq ce_m\left\{(t+1)\zeta_N+\sqrt\frac{s_{p}\log{d}}{M\gamma_N}\right\},
\end{align*}
with some constant $c>0$.

\textbf{Case 2.} If $\lambda_{\balpha}\|\bDelta_S\|_1>2c\{s_{p}\lambda_{\bbeta}^2+(1+t)\zeta_N^2\}e_m^2$. Then, \eqref{eq:basic-OR-SP'} implies that $\|\bDelta_{S_{\balpha}^c}\|_1\leq 4\|\bDelta_{S_{\balpha}}\|_1$. Hence,
\begin{align}
\|\bDelta\|_1=\|\bDelta_{S}\|_1+\|\bDelta_{S^c}\|_1\leq5\|\bDelta_{S}\|_1\leq5\sqrt{s_{\balphatil}}\|\bDelta_{S}\|_2\leq5\sqrt{s_{\balphatil}}\|\bDelta\|_2.\label{eq:l2l1norm-SP}
\end{align}
Besides, on the event $\Acaltil_1\cap\mathcal E_p$ and together with \eqref{eq:basic-OR-SP'}, we also have
\begin{align*}
\|\bDelta\|_2^2&\leq c\left\{\delta\elltil_{\balpha}(\bDelta;\phat_{N},\bbetahat;\balphatil^*;\pihat)+\frac{\log{d}}{M\gamma_N}\|\bDelta\|_1^2\right\}\\
&\leq c\left\{4\lambda_{\balpha}\|\bDelta_S\|_1+\frac{\log{d}}{M\gamma_N}\|\bDelta\|_1^2\right\}\leq c\left\{4\lambda_{\balpha}\|\bDelta\|_1+\frac{\log{d}}{M\gamma_N}\|\bDelta\|_1^2\right\},
\end{align*}
with some constant $c>0$. Meanwhile, by \eqref{eq:l2l1norm-SP}, we also have $\|\bDelta\|_2^2\geq\|\bDelta\|_1^2/(25s_{\balphatil})$. Since $s_{\balphatil}=o(M\gamma_N/\log{d})$ and $\lambda_{\balpha}\asymp\sqrt{\log{d}/(M\gamma_N)}$, we have
\begin{align*}
\|\bDelta\|_1\leq cs_{\balphatil}\sqrt\frac{\log{d}}{M\gamma_N},
\end{align*}
with some constant $c>0$ when $N$ is large enough. Together with \eqref{eq:l2l1norm-SP}, we also have
\begin{align*}
\|\bDelta\|_2\leq c\sqrt\frac{s_{\balphatil}\log{d}}{M\gamma_N},
\end{align*}
with some constant $c>0$.

Lastly, combining the results in Lemmas \ref{lemma:phatN-est}, \ref{lemma:RSC-SP}, \ref{lemma:gradient-SP}, \ref{lemma:RSC2-SP}, \ref{lemma:gradient2-SP}, \ref{lemma:gradient3-SP}, and Theorem \ref{thm:PS-SP}, we conclude that
\begin{align*}
\|\balphatil-\balphatil^*\|_1&=O_p\left((s_{\balphatil}+e_m^2s_{p})\sqrt\frac{\log{d}}{M\gamma_N}+e_m^2\zeta_N^2\sqrt\frac{M\gamma_N}{\log{d}}\right),\\
\|\balphatil-\balphatil^*\|_2&=O_p\left(\sqrt\frac{(s_{\balphatil}+e_ms_{p})\log{d}}{M\gamma_N}+e_m\zeta_N\right).
\end{align*}}
\end{proof}

\section{Proof of the properties of the de-coupled BRSS estimator}\label{sec:proof_BRDR-SP}

Recall that $M=N/2$. For each $k\in\{1,2\}$ and $k'=3-k$, we have
$$\thetahat_{\mbox{\tiny 1,DC-BRSS}}^{(k)}= M^{-1}\sum_{i\in\mathcal I_k}\left\{\bS_i^T\balphatil^{(k')}+\frac{\Gamma_i}{\phat_{N}^{(k)}(\bX_i)\pihat^{(k)}(\bX_i)}(Y_i-\bS_i^T\balphatil^{(k')})\right\},$$
where $\phat_{N}^{(k)}(\bX_i)=g(\bS_i^T\bbetahat_{p}^{(k)}+\log(\phat_{N}^{(k)}))$. In this section, we consider the case that $k=1$, the case for $k=2$ will follow analogously. {Here, we have
\begin{align}
\thetahat_{\mbox{\tiny 1,DC-BRSS}}^{(1)}-\theta_1= M^{-1}\sum_{i\in\mathcal I_1}\psitil_{N}^*(\bZ_i)+\Delta_1+\Delta_2+\Delta_3,\label{eq:Delta123}
\end{align}
where
\begin{align*}
\psitil_{N}^*(\bZ_i)&:=\bS_i^T\balphatil^*+\frac{\Gamma_i}{\gammatil_N^*(\bX_i)}(Y_i-\bS_i^T\balphatil^*)-\theta_1,\\
\Delta_1&:=M^{-1}\sum_{i\in\mathcal I_1}\left\{1-\frac{\Gamma_i}{\gammatil_N^{(1)}(\bX_i)}\right\}\bS_i^T(\balphatil^{(2)}-\balphatil^*),\\
\Delta_2&:=M^{-1}\sum_{i\in\mathcal I_1}\left\{\frac{\Gamma_i}{\gammatil_N^{(1)}(\bX_i)}-\frac{\Gamma_i}{\gammatil_N^*(\bX_i)}\right\}\varepsilon_{i,1},\\
\Delta_3&:=M^{-1}\sum_{i\in\mathcal I_1}\left\{\frac{\Gamma_i}{\gammatil_N^{(1)}(\bX_i)}-\frac{\Gamma_i}{\gammatil_N^*(\bX_i)}\right\}\varepsilon_{i,2},
\end{align*}
with $\gammatil_N^{(1)}(\cdot):=\pihat^{(1)}(\cdot)\phat_{N}^{(1)}(\cdot)$ and $\gammatil_N^*(\cdot):=\pi^*(\cdot)p_{N}^*(\cdot)$. Additionally, recall that $\varepsilon_{i,1}:=Y_i(1)-m(\bX_i)$ and $\varepsilon_{i,2}:=m(\bX_i)-m^*(\bX_i)$ with $\varepsilon_i=\varepsilon_{i,1}+\varepsilon_{i,2}$.}

Based on the constructions of $\bbetahat_{p}^{(1)}$ and $\balphatil^{(2)}$ and by the Karush-Kuhn-Tucker (KKT) conditions, we have the following \emph{covariate balancing equations}:
\begin{align}
\left\| M^{-1}\sum_{i\in\mathcal I_1}\left\{1-\frac{\Gamma_i}{\gammatil_N^{(1)}(\bX_i)}\right\}\bS_i\right\|_\infty&\leq\lambda_{\bbeta},\label{bound:KKT1}\\
\left\| M^{-1}\sum_{i\in\mathcal I_2}\frac{\Gamma_i\exp(-\bX_i\bbetahat_{p}^{(2)})}{\phat_{N}^{(2)}\pihat^{(2)}(\bX_i)}(Y_i-\bS_i^T\balphatil^{(2)})\bS_i\right\|_\infty&\leq\lambda_{\balpha}.\label{bound:KK2}
\end{align}
Additionally, by Theorems \ref{thm:PS-SP} and \ref{thm:OR-SP-mis}, with some $\lambda_{\bbeta}\asymp\lambda_{\balpha}\asymp\sqrt{\log{d}/(N\gamma_N)}$,
\begin{align}
\|\bbetahat_{p}^{(1)}-\bbeta_{p}^*\|_1&=O_p\left(s_{p}\sqrt\frac{\log{d}}{N\gamma_N}+\zeta_N^2\sqrt\frac{N\gamma_N}{\log{d}}\right),\label{rate:betatil1}\\
\|\bbetahat_{p}^{(1)}-\bbeta_{p}^*\|_2&=O_p\left(\sqrt\frac{s_{p}\log{d}}{N\gamma_N}+\zeta_N\right),\label{rate:betatil2}\\
\|\balphatil^{(2)}-\balphatil^*\|_1&=O_p\left((s_{\balphatil}+e_m^2s_{p})\sqrt\frac{\log d}{N\gamma_N}+e_m^2\zeta_N^2\sqrt\frac{N\gamma_N}{\log d}\right),\label{rate:alphatil1}\\
\|\balphatil^{(2)}-\balphatil^*\|_2&=O_p\left(\sqrt\frac{(s_{\balphatil}+e_m^2s_{p})\log d}{N\gamma_N}+e_m\zeta_N\right).\label{rate:alphatil2}
\end{align}

\begin{lemma}\label{lemma:Delta1}
{Let Assumptions \ref{cond:basic}-\ref{cond:approx} hold and $N\gamma_N\gg\max(\log N,s_{p},s_{\balphatil})\log d$ with $\lambda_{\bbeta}\asymp\lambda_{\balpha}\asymp\sqrt{\log d/(N\gamma_N)}$. Then, as $N,d\to\infty$,
$$\Delta_1=O_p\left(\frac{(s_{\balphatil}+e_m^2s_{p})\log d}{N\gamma_N}+e_m^2\zeta_N^2\right)=o_p\left((N\gamma_N)^{-1/2}\right),$$
when $s_{\balphatil}+e_m^2s_{p}=o(\sqrt{N\gamma_N}/\log{d})$ and $e_m^2\zeta_N^2=o((N\gamma_N)^{-1/2})$.}
\end{lemma}

\begin{proof}[Proof of Lemma \ref{lemma:Delta1}]
{Observe that
\begin{align*}
|\Delta_1|&=\left|M^{-1}\sum_{i\in\mathcal I_1}\left\{1-\frac{\Gamma_i}{\gammatil_N^{(1)}(\bX_i)}\right\}\bS_i^T(\balphatil^{(2)}-\balphatil^*)\right|\\
&\leq\left\| M^{-1}\sum_{i\in\mathcal I_1}\left\{1-\frac{\Gamma_i}{\gammatil_N^{(1)}(\bX_i)}\right\}\bS_i\right\|_\infty\|\balphatil^{(2)}-\balphatil^*\|_1\\
&\overset{(i)}{\leq}\lambda_{\bbeta}\|\balphatil^{(2)}-\balphatil^*\|_1\overset{(ii)}{=}O_p\left(\frac{(s_{\balphatil}+e_m^2s_{p})\log d}{N\gamma_N}+e_m^2\zeta_N^2\right),
\end{align*}
where (i) holds by \eqref{bound:KKT1}; (ii) holds by \eqref{rate:alphatil1} and $\lambda_{\bbeta}\asymp\sqrt{\log{d}/(N\gamma_N)}$.}
\end{proof}

\begin{lemma}\label{lemma:Delta2}
Let Assumptions \ref{cond:basic}-\ref{cond:subG'} hold, $N\gamma_N\gg\max(\log N,s_{p}\log^{1/2}N)\log{d}$, and $\zeta_N=o(\log^{-1/4}{N})$ with some $\lambda_{\bbeta}\asymp\sqrt{\log{d}/(N\gamma_N)}$. Then, as $N,d\to\infty$,
$$\Delta_2=o_p\left((N\gamma_N)^{-1/2}\right).$$
\end{lemma}

\begin{proof}[Proof of Lemma \ref{lemma:Delta2}]
By construction, note that both $\pihat^{(1)}(\cdot)$ and $\phat_{N}^{(1)}(\cdot)$ are constructed based on training samples $\mathcal S:=(R_i,T_i,\bX_i)_{i\in\mathcal D_N^{(1)}}$. Hence, under Assumption \ref{cond:basic},
\begin{align*}
\E_{\mathcal D_N^{(1)}}(\Delta_2\mid\mathcal S)&=\E_{\mathcal D_N^{(1)}}\left[ M^{-1}\sum_{i\in\mathcal I_1}\left\{\frac{\Gamma_i}{\gammatil_N^{(1)}(\bX_i)}-\frac{\Gamma_i}{\gammatil_N^*(\bX_i)}\right\}\varepsilon_{i,1}\mid\mathcal S\right]=0.
\end{align*}
It follows that
\begin{align*}
&\E_{\mathcal D_N^{(1)}}\left(\Delta_2^2\mid\mathcal S\right)=\E_{\mathcal D_N^{(1)}}\left(\left[ M^{-1}\sum_{i\in\mathcal I_1}\left\{\frac{\Gamma_i}{\gammatil_N^{(1)}(\bX_i)}-\frac{\Gamma_i}{\gammatil_N^*(\bX_i)}\right\}\varepsilon_{i,1}\right]^2\mid\mathcal S\right)\\
&\qquad= M^{-2}\sum_{i\in\mathcal I_1}\E(\varepsilon_{i,1}^2|\bX_i)\Gamma_i\left\{\frac{1}{\gammatil_N^{(1)}(\bX_i)}-\frac{1}{\gammatil_N^*(\bX_i)}\right\}^2\\
&\qquad=O\left(M^{-2}\sum_{i\in\mathcal I_1}\Gamma_i\left\{\frac{1}{\gammatil_N^{(1)}(\bX_i)}-\frac{1}{\gammatil_N^*(\bX_i)}\right\}^2\right)
\end{align*}
under Assumption \ref{cond:subG'}. On the event $\mathcal E_\pi$ and under Assumptions \ref{cond:NP-est} and \ref{cond:bound-SP}, we have $1/\pihat^{(1)}(\bX_i)\leq c_0^{-1}$, $1/\pi^*(\bX_i)\geq c_0^{-1}$, and $1/p_{N}^*(\bX_i)=1+p_{N}^{-1}\exp(-\bS_i^T\bbeta_{p}^*)\leq p_{N}^{-1}\{1+\exp(C_0)\}$ for all $i\in\mathcal I_1$ almost surely. Hence,
\begin{align*}
&\left\{\frac{1}{\gammatil_N^{(1)}(\bX_i)}-\frac{1}{\gammatil_N^*(\bX_i)}\right\}^2\\
&\quad\leq\frac{2}{\{\pihat^{(1)}(\bX_i)\}^2}\left\{\frac{1}{\phat_{N}^{(1)}(\bX_i)}-\frac{1}{p_{N}^*(\bX_i)}\right\}^2+\frac{2}{\{p_{N}^*(\bX_i)\}^2}\left\{\frac{1}{\pihat^{(1)}(\bX_i)}-\frac{1}{\pi^*(\bX_i)}\right\}^2\\
&\quad\leq2c_0^{-2}\left\{\frac{1}{\phat_{N}^{(1)}(\bX_i)}-\frac{1}{p_{N}^*(\bX_i)}\right\}^2+\frac{2\{1+\exp(C_0)\}^2}{c_0^4p_{N}^2}\left\{\pihat^{(1)}(\bX_i)-\pi^*(\bX_i)\right\}^2.
\end{align*}
Let $\bDelta_{\bbeta}^{(j)}:=\bbetahat_{p}^{(j)}-\bbeta_{p}^*+\log(\Delta_{p}^{(j)})\be_1$ and $\Delta_{p}^{(j)}:=\phat_{N}^{(j)}/p_{N}$ for each $j\in\{1,2\}$, where $\be_l\in\R^d$ denotes the $l$-th column of an identity matrix. By Lemma \ref{lemma:phatN-est} and the fact that $1-u^{-1}\leq\log(u)\leq u-1$ for all $u>0$, we have $\phat_{N}^{(j)}=O_p(\gamma_N)$ and
\begin{align}
&\{\log(\Delta_{p}^{(j)})\}^2\leq(\Delta_{p}^{(j)}-1)^2+(1-1/\Delta_{p}^{(j)})^2\nonumber\\
&\qquad=\left(\frac{\phat_{N}^{(j)}-p_{N}}{p_{N}}\right)^2+\left(\frac{\phat_{N}^{(j)}-p_{N}}{\phat_{N}^{(j)}}\right)^2=O_p\left((N\gamma_N)^{-1}\right).\label{logdelta}
\end{align}
Together with \eqref{rate:betatil1} and \eqref{rate:betatil2}, we have
\begin{align}
&\|\bDelta_{\bbeta}^{(j)}\|_1=O_p\left(s_{p}\sqrt\frac{\log{d}}{N\gamma_N}+\zeta_N^2\sqrt\frac{N\gamma_N}{\log{d}}\right),\;\;\|\bDelta_{\bbeta}^{(j)}\|_2=O_p\left(\sqrt\frac{s_{p}\log{d}}{N\gamma_N}+\zeta_N\right).\label{rate:betatil-Delta}
\end{align}
By Taylor's theorem, for each $i\in\mathcal I_1$, with some $v_i\in(0,1)$,
\begin{align*}
&\left\{\frac{1}{\phat_{N}^{(1)}(\bX_i)}-\frac{1}{p_{N}^*(\bX_i)}\right\}^2= p_{N}^{-2}\exp(-2\bS_i^T\bbeta_{p}^*)\left\{\exp(-\bS_i^T\bDelta_{\bbeta}^{(1)})-1\right\}^2\\
&\qquad= p_{N}^{-2}\exp(-2\bS_i^T\bbeta_{p}^*)\exp(-2v_i\bS_i^T\bDelta_{\bbeta}^{(1)})(\bS_i^T\bDelta_{\bbeta}^{(1)})^2\\
&\qquad= p_{N}^{-2}\exp\{-2(1-v_i)\bS_i^T\bbeta_{p}^*-2v_i\bS_i^T\bbetahat_{p}^{(1)}\}(\Delta_{p}^{(1)})^{-2v_i}(\bS_i^T\bDelta_{\bbeta}^{(1)})^2\\
&\qquad\leq p_{N}^{-2}\exp(2C_0)(\bS_i^T\bDelta_{\bbeta}^{(1)})^2\max\{1,(\Delta_{p}^{(1)})^{-2}\}.
\end{align*}
Further conditional on the event $\mathcal E_\zeta$, we have
\begin{align}
&M^{-2}\sum_{i\in\mathcal I_1}\left\{\frac{\Gamma_i}{\gammatil_N^{(1)}(\bX_i)}-\frac{\Gamma_i}{\gammatil_N^*(\bX_i)}\right\}^2\nonumber\\
&\quad=O\left(\max\{1,(\Delta_{p}^{(1)})^{-2}\} (M\gamma_N)^{-2}\sum_{i\in\mathcal I_1}\Gamma_i(\bS_i^T\bDelta_{\bbeta}^{(1)})^2+(M\gamma_N)^{-1}\zeta_N^2\right),\label{bound:delta3_moment}
\end{align}
By Lemma \ref{lemma:phatN-est}, $\max\{1,(\Delta_{p}^{(1)})^{-2}\}=O_p(1)$. Together with part (b) of Lemma \ref{lemma:emp} and \eqref{rate:betatil-Delta},
\begin{align}
&(M\gamma_N)^{-1}\sum_{i\in\mathcal I_1}\Gamma_i(\bS_i^T\bDelta_{\bbeta}^{(1)})^2\nonumber\\
&\qquad=O_p\left(\frac{s_{p}^2\log^2{d}\log{N}}{(N\gamma_N)^2}+\frac{s_{p}\log{d}}{N\gamma_N}+\zeta_N^4\log{N}+\zeta_N^2\right).\label{rate:insample1}
\end{align}
Together with \eqref{bound:delta3_moment}, we have
\begin{align}
&(M\gamma_N)M^{-2}\sum_{i\in\mathcal I_1}\left\{\frac{\Gamma_i}{\gammatil_N^{(1)}(\bX_i)}-\frac{\Gamma_i}{\gammatil_N^*(\bX_i)}\right\}^2\nonumber\\
&\qquad=O_p\left(\frac{s_{p}^2\log^2{d}\log{N}}{(N\gamma_N)^2}+\frac{s_{p}\log{d}}{N\gamma_N}+\zeta_N^4\log{N}+\zeta_N^2\right)=o_p(1),\label{rate:gammadiff_insample}
\end{align}
as long as $s_{p}=o(N\gamma_N/(\log{d}\log^{1/2}{N}))$ and $\zeta_N=o(\log^{-1/4}{N})$. Hence, $\E_{\mathcal D_N^{(1)}}(\Delta_2^2\mid\mathcal S)=o_p\left((N\gamma_N)^{-1}\right).$ By Markov's inequality,
$$\Delta_2=o_p\left((N\gamma_N)^{-1/2}\right).$$
\end{proof}

\begin{lemma}\label{lemma:Delta3}
{Let Assumptions \ref{cond:basic}-\ref{cond:bound-SP} and \ref{cond:approx} hold, $N\gamma_N\gg s_{p}\log d\log N$, and $\zeta_N=o(\log^{-1/2}N)$ with some $\lambda_{\bbeta}\asymp\sqrt{\log{d}/(N\gamma_N)}$. Then, as $N,d\to\infty$,
$$\Delta_3=O_p\left(e_m\sqrt\frac{s_{p}\log{d}}{N\gamma_N}+e_m\zeta_N\right)=o_p\left((N\gamma_N)^{-1/2}\right),$$
when $e_m=o(1/\sqrt{s_{p}\log{d}})$ and $e_m\zeta_N=o(1/\sqrt{N\gamma_N})$.}
\end{lemma}

\begin{proof}[Proof of Lemma \ref{lemma:Delta3}]
{By the Cauchy-Schwarz inequality,
\begin{align*}
|\Delta_3|&=\left|M^{-1}\sum_{i\in\mathcal I_1}\left\{\frac{\Gamma_i}{\gammatil_N^{(1)}(\bX_i)}-\frac{\Gamma_i}{\gammatil_N^*(\bX_i)}\right\}\varepsilon_{i,2}\right|\\
&\leq \sqrt{M^{-1}\sum_{i\in\mathcal I_1}\left\{\frac{\Gamma_i}{\gammatil_N^{(1)}(\bX_i)}-\frac{\Gamma_i}{\gammatil_N^*(\bX_i)}\right\}^2M^{-1}\sum_{i\in\mathcal I_1}\Gamma_i\varepsilon_{i,2}^2}\\
&\overset{(i)}{=}O_p\left(\sqrt{\gamma_N^{-1}\left\{\frac{s_{p}^2\log^2{d}\log{N}}{(N\gamma_N)^2}+\frac{s_{p}\log{d}}{N\gamma_N}+\zeta_N^4\log{N}+\zeta_N^2\right\}\gamma_Ne_m^2}\right)\\
&\overset{(ii)}{=}O_p\left(e_m\sqrt\frac{s_{p}\log{d}}{N\gamma_N}+e_m\zeta_N\right),
\end{align*}
where (i) holds by \eqref{rate:gammadiff_insample}, Lemma \ref{lemma:gammaN-est}, and under Assumption \ref{cond:approx}; (ii) holds when $N\gamma_N\gg s_{p}\log{d}\log{N}$ and $\zeta_N=o(\log^{-1/2}{N})$.}
\end{proof}

\vspace{-0.1in} 
{In the following, we consider another decomposition:
\begin{align}
\thetahat_{\mbox{\tiny 1,DC-BRSS}}^{(1)}-\theta_1= M^{-1}\sum_{i\in\mathcal I_1}\psitil_{N}^*(\bZ_i)+\Delta_4+\Delta_5+\Delta_6,\label{eq:Delta456}
\end{align}
where
\begin{align*}
\Delta_4&:=M^{-1}\sum_{i\in\mathcal I_1}\left\{1-\frac{\Gamma_i}{\gammatil_N^*(\bX)}\right\}\bS_i^T(\balphatil^{(2)}-\balphatil^*),\\
\Delta_5&:=M^{-1}\sum_{i\in\mathcal I_1}\frac{1}{p_{N}^*(\bX_i)}\left\{\frac{\Gamma_i}{\pihat^{(1)}(\bX_i)}-\frac{\Gamma_i}{\pi^*(\bX_i)}\right\}\varepsilontil_i,\\
\Delta_6&:=M^{-1}\sum_{i\in\mathcal I_1}\frac{1}{\pihat^{(1)}(\bX_i)}\left\{\frac{\Gamma_i}{\phat_{N}^{(1)}(\bX_i)}-\frac{\Gamma_i}{p_{N}^*(\bX_i)}\right\}\varepsilontil_i,
\end{align*}
with $\varepsilontil_i:=Y_i(1)-\bS_i^T\balphatil^{(2)}$ and let $\varepsilontil$ be an independent copy of $\varepsilontil_i$.}

\begin{lemma}\label{lemma:Delta4}
Let Assumptions \ref{cond:basic}-\ref{cond:bound-SP} and \ref{cond:approx} hold, $N\gamma_N\gg\max(\log N,s_{p},s_{\balphatil})\log d$, and $e_m\zeta_N=o(1)$ with some $\lambda_{\bbeta}\asymp\lambda_{\balpha}\asymp\sqrt{\log{d}/(N\gamma_N)}$. Then, as $N,d\to\infty$,
$$\Delta_4=o_p\left((N\gamma_N)^{-1/2}\right).$$
\end{lemma}

\vspace{-0.25in} 
\begin{proof}[Proof of Lemma \ref{lemma:Delta4}]
Since $\balphatil^{(2)}$ is independent of $\mathcal D_N^{(1)}$, we have
$$\E_{\mathcal D_N^{(1)}}(\Delta_4)=\E_{\bX,\Gamma}\left[\left\{1-\frac{\Gamma}{\gammatil_N^*(\bX)}\right\}\bS\right]^T(\balphatil^{(2)}-\balphatil^*)=0,$$
by the construction of $\bbeta_{p}^*$. For any function $f(\cdot)$ and constant $r>0$, denote $\| f(\cdot) \|_{\P_{\bX},r} := \{\E_{\bX}| f(\bX)|^r\}^{1/r}$. Then, with $r>0$ satisfying $1/r+1/q=1$,
\begin{align*}
&\E_{\mathcal D_N^{(1)}}(\Delta_4^2)= M^{-1}\E_{\bX,\Gamma}\left(\left[\left\{1-\frac{\Gamma}{\gammatil_N^*(\bX)}\right\}\bS^T(\balphatil^{(2)}-\balphatil^*)\right]^2\right)\\
&\qquad\leq M^{-1}\E_{\bX,\Gamma}\left(\left[1+\frac{\Gamma}{\{\gammatil_N^*(\bX)\}^2}\right]\left\{\bS^T(\balphatil^{(2)}-\balphatil^*)\right\}^2\right)\\
&\qquad= M^{-1}\E_{\bX,\Gamma}\left(\left[1+\frac{\left\{1+\gamma_N^{-2}\exp(-2\bS^T\bbeta_{p}^*)\right\}\gamma_N(\bX)}{\{\pi^*(\bX)\}^2}\right]\left\{\bS^T(\balphatil^{(2)}-\balphatil^*)\right\}^2\right)\\
&\qquad\leq M^{-1}\left\|1+\frac{\left\{1+\gamma_N^{-2}\exp(2C_0)\right\}\gamma_N(\bX)}{c_0^2}\right\|_{\P,q}\left\|\bS^T(\balphatil^{(2)}-\balphatil^*)\right\|_{\P_{\bX},2r}^2\\
&\qquad={O_p\left((N\gamma_N)^{-1}\left\{\frac{(s_{\balphatil}+e_m^2s_{p})\log{d}}{N\gamma_N}+e_m^2\zeta_N^2\right\}\right)},
\end{align*}
under Assumptions \ref{cond:subG} and \ref{cond:tail} together with \eqref{rate:alphatil2}. By Markov's inequality,
\begin{align}\label{rate:Delta2}
\Delta_4=O_p\left((N\gamma_N)^{-1/2}\sqrt{\frac{(s_{\balphatil}+e_m^2s_{p})\log{d}}{N\gamma_N}+e_m^2\zeta_N^2}\right)=o_p\left((N\gamma_N)^{-1/2}\right),
\end{align}
as long as $s_{\balphatil}+e_m^2s_{p}=o(N\gamma_N/\log{d})$ and $e_m\zeta_N=o(1)$.
\end{proof}

\begin{lemma}\label{lemma:Delta5}
Let Assumptions \ref{cond:basic}-\ref{cond:approx} hold, $N\gamma_N\gg\max(\log N,s_{p},s_{\balphatil})\log d$ with some $\lambda_{\bbeta}\asymp\lambda_{\balpha}\asymp\sqrt{\log{d}/(N\gamma_N)}$. Then, as $N,d\to\infty$,
$$\Delta_5=O_p(\zeta_N)=o_p\left((N\gamma_N)^{-1/2}\right),$$
when $\zeta_N=o((N\gamma_N)^{-1/2})$.
\end{lemma}

\begin{proof}[Proof of Lemma \ref{lemma:Delta5}]
Condition on the event $\mathcal E_\pi\cap\mathcal E_\zeta$. By the Cauchy-Schwarz inequality, under Assumptions \ref{cond:NP-est} and \ref{cond:bound-SP}, we have
\begin{align*}
\Delta_5^2&\leq M^{-1}\sum_{i\in\mathcal I_1}\frac{\Gamma_i\varepsilontil_i^2}{\{p_{N}^*(\bX_i)\}^2}M^{-1}\sum_{i\in\mathcal I_1}\Gamma_i\left\{\frac{1}{\pihat^{(1)}(\bX_i)}-\frac{1}{\pi^*(\bX_i)}\right\}^2\\
&=O\left(\gamma_N^{-2}M^{-1}\sum_{i\in\mathcal I_1}\Gamma_i\varepsilontil_i^2M^{-1}\sum_{i\in\mathcal I_1}\Gamma_i\left\{\pihat^{(1)}(\bX_i)-\pi^*(\bX_i)\right\}^2\right)\\
&=O\left(\zeta_N^2(M\gamma_N)^{-1}\sum_{i\in\mathcal I_1}\Gamma_i\varepsilontil_i^2\right).
\end{align*}
Denote $\P_{\bX,\Gamma,Y}$ as the joint distribution of $(\bX,\Gamma,Y)$ and let $\E_{\bX,\Gamma,Y}(\cdot)$ be the corresponding expectation. Note that
\begin{align*}
&\E_{\mathcal D_N^{(1)}}\left\{(M\gamma_N)^{-1}\sum_{i\in\mathcal I_1}\Gamma_i\varepsilontil_i^2\right\}=\gamma_N^{-1}\E_{\bX,\Gamma,Y}(\Gamma\varepsilontil^2)\\
&\qquad=\gamma_N^{-1}\E_{\bX,\Gamma,Y}\left[\gamma_N(\bX)\left\{\varepsilon-\bS^T(\balphatil^{(2)}-\balphatil^*)\right\}^2\right]=O_p(1+\|\balphatil^{(2)}-\balphatil^*\|_2^2)=O_p(1),
\end{align*}
under Assumptions \ref{cond:basic}, \ref{cond:subG}, \ref{cond:tail}, and \ref{cond:subG'}, together with the fact that $\|\balphatil^{(2)}-\balphatil^*\|_2=o_p(1)$. By Markov's inequality, $(M\gamma_N)^{-1}\sum_{i\in\mathcal I_1}\Gamma_i\varepsilontil_i^2=O_p(1)$. Hence, as long as $\zeta_N=o((N\gamma_N)^{-1/2})$,
\begin{align*}
\Delta_5=O_p(\zeta_N).
\end{align*}
\end{proof}

\begin{lemma}\label{lemma:Delta6}
{Let Assumptions \ref{cond:basic}-\ref{cond:approx} hold, $N\gamma_N\gg\max(\log^{5/2} N,s_{p},s_{\balphatil})\log d\log^{1/2}N$ with some $\lambda_{\bbeta}\asymp\lambda_{\balpha}\asymp\sqrt{\log{d}/(N\gamma_N)}$. Then, as $N,d\to\infty$,
$$\Delta_6=o_p\left((N\gamma_N)^{-1/2}\right),$$
when $s_{p}=o(\sqrt{N\gamma_N}/\log{d})$, $s_{\balphatil}\sqrt{s_{p}}=o(N\gamma_N/\log^{3/2}{d})$, and $\zeta_N=o((N\gamma_N)^{-1/2})$.}
\end{lemma}

\begin{proof}[Proof of Lemma \ref{lemma:Delta6}]
By Taylor's theorem, there exists some $v\in(0,1)$ such that
\begin{align*}
\Delta_6&=M^{-1}\sum_{i\in\mathcal I_1}\frac{\Gamma_i\exp(-\bS_i^T\bbetahat_{p}^{(1)})\{\exp(\bS_i^T\bDelta_{\bbeta}^{(1)})-1\}\varepsilontil_i}{\phat_{N}^{(1)}\pihat^{(1)}(\bX_i)}=\Delta_{6,1}+\Delta_{6,2},
\end{align*}
where
\begin{align*}
\Delta_{6,1}&=M^{-1}\sum_{i\in\mathcal I_1}\frac{\Gamma_i\exp(-\bS_i^T\bbetahat_{p}^{(1)})\bS_i^T\bDelta_{\bbeta}^{(1)}\varepsilontil_i}{\phat_{N}^{(1)}\pihat^{(1)}(\bX_i)},\\
\Delta_{6,2}&=M^{-1}\sum_{i\in\mathcal I_1}\frac{\Gamma_i\exp(-\bS_i^T\bbetahat_{p}^{(1)}+v\bS_i^T\bDelta_{\bbeta}^{(1)})(\bS_i^T\bDelta_{\bbeta}^{(1)})^2\varepsilontil_i}{\phat_{N}^{(1)}\pihat^{(1)}(\bX_i)},
\end{align*}
with $\bDelta_{\bbeta}^{(1)}=\bbetahat_{p}^{(1)}-\bbeta_{p}^*+\log(\Delta_{p}^{(1)})\be_1$ and $\Delta_{p}^{(1)}=\phat_{N}^{(1)}/p_{N}$. Condition on the event that $\{0.46p_{N}\leq\phat_{N}^{(j)}\leq1.54p_{N}\}$, which occurs with probability approaching one by Lemma \ref{lemma:phatN-est}. Then, for any $i\in\mathcal I_1$, $\exp(-\bS_i^T\bbetahat_{p}^{(1)}+v\bS_i^T\bDelta_{\bbeta}^{(1)})\leq\max(0.46^{-1},1.54)\exp(C_0)$. It follows that
\begin{align}
|\Delta_{6,2}|&=O\left((M\gamma_N)^{-1}\sum_{i\in\mathcal I_1}\Gamma_i(\bS_i^T\bDelta_{\bbeta}^{(1)})^2|\varepsilontil_i|\right).\label{RN}
\end{align}
Define $\mathcal E_{\balphatil}:=\{\|\balphatil^{(2)}-\balphatil^*\|_2\leq1\}$. Then, $\P_{\mathcal D_N^{(2)}}(\mathcal E_{\balphatil})=1-o(1)$. On the event $\mathcal E_{\balphatil}$ and conditional on $\mathcal D_N^{(2)}$, we have
\begin{align}
\|\varepsilontil\|_{\psi_2}&=\|\varepsilon-\bS^T(\balphatil^{(2)}-\balphatil^*)\|_{\psi_2}\leq\|\varepsilon\|_{\psi_2}+\|\bS^T(\balphatil^{(2)}-\balphatil^*)\|_2\nonumber\\
&\leq\|\varepsilon\|_{\psi_2}+\sigma\|\balphatil^{(2)}-\balphatil^*\|_2=O(1).\label{bound:psi-2-Y-alpha}
\end{align}
Note that under Assumption \ref{cond:basic},
\begin{align*}
&\E_{\bX,\Gamma,Y}(|\varepsilontil|\mid\bX,\Gamma=1)=\E_{\bX,\Gamma,Y}(|Y-\bS^T\balphatil^{(2)}|\mid\bX,R=1,T=1)\nonumber\\
&\qquad\overset{(i)}{=}\E_{\bX,T,Y}\{|Y(1)-\bS^T\balphatil^{(2)}|\mid\bX,T=1\}\overset{(ii)}{=}\E_{\bX,Y(1)}(|\varepsilontil|\mid\bX),
\end{align*}
where (i) holds since $R\independent Y\mid(T,\bX)$ and $Y=Y(T)$; (ii) holds since $T\independent Y(1)\mid\bX$. In the above, $\E_{\bX,\Gamma,Y}(\cdot)$, $\E_{\bX,T,Y}$, and $\E_{\bX,Y(1)}$ denote the expectations with respect to the joint distribution of $(\bX,\Gamma,Y)$, $(\bX,T,Y)$, and $(\bX,Y(1))$, respectively.

Choose some $s\asymp\sqrt{N\gamma_N}/\{\log{d}\log^{3/2}{N}\}$, then $\sqrt{s\log{d}/(N\gamma_N)}+(s\log{N}\log{d})^{3/2}/(N\gamma_N)\asymp[1/\{N\gamma_N\log^3{N}\}]^{1/4}=o(1)$. By part (e) of Lemma \ref{lemma:emp} and \eqref{rate:betatil-Delta}, we have
\begin{align}
\Delta_{6,2}&=O_p\left(\frac{\|\bDelta_{\bbeta}^{(1)}\|_1^2}{s(N\gamma_N)^{1/4}\log^{3/4}{N}}+\|\bDelta_{\bbeta}^{(1)}\|_2^2\{1+o(1)\}\right)\nonumber\\
&=O_p\left(\left(\frac{s_{p}^2\log{d}}{N\gamma_N}+\frac{\zeta_N^4N\gamma_N}{\log{d}}\right)\frac{\log d\log^{3/4}N}{(N\gamma_N)^{3/4}}+\frac{s_{p}\log{d}}{N\gamma_N}+\zeta_N^2\right)\nonumber\\
&\overset{(i)}{=}O_p\left(\left(\frac{1}{\log d}+\frac{1}{N\gamma_N\log d}\right)\frac{\log d\log^{3/4}N}{(N\gamma_N)^{3/4}}+\frac{s_{p}\log{d}}{N\gamma_N}+\zeta_N^2\right)\nonumber\\
&\overset{(ii)}{=}o_p\left((N\gamma_N)^{-1/2}\right),\label{rate:Delta52}
\end{align}
where (i) holds when $s_{p}=o(\sqrt{N\gamma_N}/\log{d})$ and $\zeta_N=o((N\gamma_N)^{-1/2})$; (ii) holds if additionally that $N\gamma_N\gg\log^3N$.

Now, we control the term $\Delta_{6,1}$. Define $\bDelta_{\balpha}^{(2)}:=\balphatil^{(2)}-\balphatil^*$. We consider the following representation:
\begin{align*}
\Delta_{6,1}=&\Delta_{6,1,1}+\sum_{k=2}^5\Delta_{6,1,k}^{(1)}-\sum_{k=2}^5\Delta_{6,1,k}^{(2)},\;\;\mbox{where}\\
\Delta_{6,1,1}:=&M^{-1}\sum_{i\in\mathcal I_2}\frac{\Gamma_i\exp(-\bS_i^T\bbetahat_{p}^{(2)})\bS_i^T\bDelta_{\bbeta}^{(1)}\varepsilontil_i}{\phat_{N}^{(2)}\pihat^{(2)}(\bX_i)},\;\;\mbox{and for each}\;\;j\in\{1,2\},\\
\Delta_{6,1,2}^{(j)}:=&M^{-1}\sum_{i\in\mathcal I_j}\frac{\Gamma_i\bS_i^T\bDelta_{\bbeta}^{(1)}\varepsilontil_i}{\pihat^{(j)}(\bX_i)}\left\{\frac{\exp(-\bS_i^T\bbetahat_{p}^{(j)})}{\phat_{N}^{(j)}}-\frac{\exp(-\bS_i^T\bbeta_{p}^*)}{p_{N}}\right\},\\
\Delta_{6,1,3}^{(j)}:=&M^{-1}\sum_{i\in\mathcal I_j}\frac{\Gamma_i\exp(-\bS_i^T\bbeta_{p}^*)\bS_i^T\bDelta_{\bbeta}^{(1)}\varepsilontil_i}{p_{N}}\left\{\frac{1}{\pihat^{(j)}(\bX_i)}-\frac{1}{\pi^*(\bX_i)}\right\},\\
\Delta_{6,1,4}^{(j)}:=&M^{-1}\sum_{i\in\mathcal I_j}\frac{\Gamma_i\exp(-\bS_i^T\bbeta_{p}^*)\bS_i^T\bDelta_{\bbeta}^{(1)}\varepsilon_i}{p_{N}\pi^*(\bX_i)},\\
\Delta_{6,1,5}^{(j)}:=&M^{-1}\sum_{i\in\mathcal I_j}\frac{\Gamma_i\exp(-\bS_i^T\bbeta_{p}^*)\bS_i^T\bDelta_{\bbeta}^{(1)}\bS_i^T\bDelta_{\balpha}^{(2)}}{p_{N}\pi^*(\bX_i)}\\
&\qquad-\E_{\bX,\Gamma}\left\{\frac{\Gamma\exp(-\bS^T\bbeta_{p}^*)\bS^T\bDelta_{\bbeta}^{(1)}\bS^T\bDelta_{\balpha}^{(2)}}{p_{N}\pi^*(\bX)}\right\}.
\end{align*}
By \eqref{bound:KK2}, \eqref{rate:betatil-Delta}, and $\lambda_{\balpha}\asymp\sqrt{\log{d}/(N\gamma_N)}$, we have
\begin{align*}
\Delta_{6,1,1}&\leq\left\| M^{-1}\sum_{i\in\mathcal I_2}\frac{\Gamma_i\exp(-\bS_i^T\bbetahat_{p}^{(2)})}{\phat_{N}^{(2)}\pihat^{(2)}(\bX_i)}\varepsilontil_i\bS_i\right\|_\infty\|\bDelta_{\bbeta}^{(1)}\|_1\\
&=O_p\left(\frac{s_{p}\log{d}}{N\gamma_N}+\zeta_N^2\right)=o_p\left((N\gamma_N)^{-1/2}\right),
\end{align*}
when $s_{p}=o(\sqrt{N\gamma_N}/\log{d})$ and $\zeta_N=o((N\gamma_N)^{-1/2})$.

For each $j\in\{1,2\}$, by Taylor's theorem, with some $v_j\in(0,1)$,
\begin{align*}
|\Delta_{6,1,2}^{(j)}|&=\left|M^{-1}\sum_{i\in\mathcal I_j}\frac{\Gamma_i\exp[-\bS_i^T\{v_j\bbetahat_{p}^{(j)}+(1-v_j)\bbeta_{p}^*\}]\bS_i^T\bDelta_{\bbeta}^{(1)}\bS_i^T\bDelta_{\bbeta}^{(j)}\varepsilontil_i}{(\phat_{N}^{(j)})^{v_j}p_{N}^{1-v_j}\pihat^{(j)}(\bX_i)}\right|\\
&=O_p\left((M\gamma_N)^{-1}\sum_{i\in\mathcal I_j}\Gamma_i|\bS_i^T\bDelta_{\bbeta}^{(1)}\bS_i^T\bDelta_{\bbeta}^{(j)}\varepsilontil_i|\right)\\
&=O_p\left((M\gamma_N)^{-1}\sum_{i\in\mathcal I_j}\Gamma_i(\bS_i^T\bDelta_{\bbeta}^{(1)})^2|\varepsilontil_i|+(M\gamma_N)^{-1}\sum_{i\in\mathcal I_j}\Gamma_i(\bS_i^T\bDelta_{\bbeta}^{(j)})^2|\varepsilontil_i|\right)\\
&=o_p\left((N\gamma_N)^{-1/2}\right),
\end{align*}
using part (e) of Lemma \ref{lemma:emp} as in \eqref{rate:Delta52}.

By the Cauchy-Schwarz inequality,
\begin{align*}
|\Delta_{6,1,3}^{(j)}|^2&=O_p\left(\max_{i\in\mathcal I_j}|\bS_i^T\bDelta_{\bbeta}^{(1)}|^2(M\gamma_N)^{-2}\sum_{i\in\mathcal I_j}\Gamma_i\varepsilontil_i^2\sum_{i\in\mathcal I_j}\Gamma_i\left\{\pihat^{(j)}(\bX_i)-\pi^*(\bX_i)\right\}^2\right)\\
&=O_p\left(\zeta_N^2\max_{i\in\mathcal I_j}|\bS_i^T\bDelta_{\bbeta}^{(1)}|^2(M\gamma_N)^{-1}\sum_{i\in\mathcal I_j}\Gamma_i\varepsilontil_i^2\right).
\end{align*}
For each $j\in\{1,2\}$, $(M\gamma_N)^{-1}\E(\sum_{i\in\mathcal I_j}\Gamma_i\varepsilon_i^2)=\gamma_N^{-1}\E\{\gamma_N(\bX)\varepsilon^2\}=O(1)$. By Markov's inequality,
$$(M\gamma_N)^{-1}\sum_{i\in\mathcal I_j}\Gamma_i\varepsilon_i^2=O_p(1).$$
Additionally, by part (b) of Lemma \ref{lemma:emp}, together with \eqref{rate:alphatil1} and \eqref{rate:alphatil2},
\begin{align*}
&(M\gamma_N)^{-1}\sum_{i\in\mathcal I_j}\Gamma_i\{\bS_i^T(\balphatil^{(2)}-\balphatil^*)\}^2\\
&\qquad={O_p\left(\frac{(s_{\balphatil}+e_m^2s_{p})\log{d}}{N\gamma_N}+\frac{(s_{\balphatil}+e_m^2s_{p})^2\log^2{d}\log{N}}{(N\gamma_N)^2}+e_m^2\zeta_N^2+e_m^4\zeta_N^4\log{N}\right)}\\
&\qquad=o_p(1),
\end{align*}
{when $\sbar:=s_{\balphatil}+e_m^2s_{p}=o(N\gamma_N/(\log{d}\log^{1/2}{N}))$, $s_{p}=o(\sqrt{N\gamma_N}/\log{d})$, and $e_m\zeta_N=o(\log^{-1/4}N)$.} Hence,
\begin{align}
(M\gamma_N)^{-1}\sum_{i\in\mathcal I_j}\Gamma_i\varepsilontil_i^2&\leq 2(M\gamma_N)^{-1}\sum_{i\in\mathcal I_j}\Gamma_i\varepsilon_i^2+2(M\gamma_N)^{-1}\sum_{i\in\mathcal I_j}\Gamma_i\{\bS_i^T(\balphatil^{(2)}-\balphatil^*)\}^2\nonumber\\
&=O_p(1).\label{bound:epsilontil}
\end{align}
Conditional on $\mathcal D_N^{(1)}$, $\|\bS_i^T\bDelta_{\bbeta}^{(1)}\|_{\psi_2}=O(\|\bDelta_{\bbeta}^{(1)}\|_2)$ for any $i\in\mathcal I_2$. By part (e) of Lemma \ref{lemma:subG}, $\|\max_{i\in\mathcal I_2}|\bS_i^T\bDelta_{\bbeta}^{(1)}|\|_{\psi_2}=O(\|\bDelta_{\bbeta}^{(1)}\|_2\sqrt{\log N})$. Together with part (c) of Lemma \ref{lemma:subG}, $\E_{\mathcal D_N^{(2)}}(\max_{i\in\mathcal I_2}|\bS_i^T\bDelta_{\bbeta}^{(1)}|)=O(\|\bDelta_{\bbeta}^{(1)}\|_2\sqrt{\log N})$. By Markov's inequality,
$$\max_{i\in\mathcal I_2}|\bS_i^T\bDelta_{\bbeta}^{(1)}|=O_p(\|\bDelta_{\bbeta}^{(1)}\|_2\sqrt{\log N})=O_p\left(\sqrt\frac{s_{p}\log{d}\log{N}}{N\gamma_N}+\zeta_N\log{N}\right)$$
since $N\gamma_N\gg s_{p}\log{d}\log{N}$ and $\zeta_N=o(1/\log{N})$. To sum up,
$$\Delta_{6,1,3}^{(2)}=O_p\left(\zeta_N\sqrt\frac{s_{p}\log{d}\log{N}}{N\gamma_N}+\zeta_N^2\log{N}\right)=o_p\left((N\gamma_N)^{-1/2}\right),$$
when $\zeta_N=o_p((N\gamma_N)^{-1/2})$ and $N\gamma_N\gg s_{p}\log{d}\log{N}+\log^2N$.

On the other hand, on the event $\{0.46p_{N}\leq\phat_{N}^{(1)}\leq1.54p_{N}\}$, together with \eqref{constraint} and Assumption \ref{cond:bound-SP}, we have $\max_{i\in\mathcal I_1}|\bS_i^T\bDelta_{\bbeta}^{(1)}|\leq|\log(\Delta_{p}^{(1)})|+\max_{i\in\mathcal I_1}(|\bS_i^T\bbetahat_{p}^{(1)}|+|\bS_i^T\bbeta_{p}^*|)\leq\log(0.46^{-1}\lor1.54)+2C_0$ almost surely. Together with \eqref{bound:epsilontil},
$$\Delta_{6,1,3}^{(1)}=O_p\left(\zeta_N^2\right)=o_p\left((N\gamma_N)^{-1/2}\right),$$
as long as $\zeta_N=o_p((N\gamma_N)^{-1/4})$.

Let $\bV_i=\Gamma_i\exp(-\bS_i^T\bbeta_{p}^*)\bS_i\varepsilon_i/\pi^*(\bX_i)$. By the construction of $\bbeta_{p}^*$, we have $\E(\bV_i)=\bzero$. In addition, $\sup_{1\leq j\leq d}\|\bV_i(j)\|_{\psi_1}\leq c_0^{-1}\exp(C_0)\sigma\|\varepsilon\|_{\psi_2}=O(1)$ and $\sup_{1\leq j\leq d}\E\{\bV_i(j)\}^2\leq c\gamma_N$ with some constant $c>0$. By Theorem 3.4 of \cite{kuchibhotla2022moving},
\begin{align*}
\left\|(Mp_{N})^{-1}\sum_{i\in\mathcal I_j}\bV_i\right\|_\infty=O_p\left(\sqrt\frac{\log{d}}{N\gamma_N}+\frac{\log{N}\sqrt{\log{d}}}{N\gamma_N}\right)=O_p\left(\sqrt\frac{\log{d}}{N\gamma_N}\right),
\end{align*}
since $N\gamma_N\gg\log^2{N}$. Together with \eqref{rate:betatil-Delta}, we have
\begin{align*}
|\Delta_{6,1,4}^{(j)}|&\leq\left\|(Mp_{N})^{-1}\sum_{i\in\mathcal I_j}\bV_i\right\|_\infty\|\bDelta_{\bbeta}^{(1)}\|_\infty\\
&=O_p\left(\frac{s_{p}\log{d}}{N\gamma_N}+\zeta_N^2\right)=o_p\left((N\gamma_N)^{-1/2}\right),
\end{align*}
since $s_{p}=o(\sqrt{N\gamma_N}/\log{d})$ and $\zeta_N=o((N\gamma_N)^{-1/2})$.

Lastly, by part (c) of Lemma \ref{lemma:emp}, \eqref{rate:betatil-Delta}, \eqref{rate:alphatil1}, and \eqref{rate:alphatil2} with {$s\asymp\sbar+\zeta_N^2N\gamma_N/\log{d}$ and $\sbar=s_{\balphatil}+e_m^2s_{p}$,}
\begin{align*}
\Delta_{6,1,5}^{(j)}&=\left\{\sqrt\frac{s\log{d}}{N\gamma_N}+\frac{s\log{d}\log{N}}{N\gamma_N}\right\}\\
&\qquad\cdot O_p\left(\|\bDelta_{\bbeta}^{(1)}\|_2\|\bDelta_{\balpha}^{(2)}\|_2\left(\frac{\|\bDelta_{\bbeta}^{(1)}\|_1^2}{s\|\bDelta_{\bbeta}^{(1)}\|_2^2}+\frac{\|\bDelta_{\balpha}^{(2)}\|_1^2}{s\|\bDelta_{\balpha}^{(2)}\|_2^2}+1\right)\right)\\
&={\left\{\sqrt\frac{\sbar\log{d}}{N\gamma_N}+\zeta_N+\frac{\sbar\log{d}\log{N}}{N\gamma_N}+\zeta_N^2\log{N}\right\}}\\
&\qquad\cdot{O_p\left(\left\{\sqrt\frac{s_{p}\log{d}}{N\gamma_N}+\zeta_N\right\}\left\{\sqrt\frac{\sbar\log{d}}{N\gamma_N}+\zeta_N\right\}\right)}\\
&=o_p\left((N\gamma_N)^{-1/2}\right),
\end{align*}
since $N\gamma_N\gg\log{d}\log{N}$, $\sbar=o(N\gamma_N/(\log{d}\log^{1/2}{N}))$, $s_{p}=o(\sqrt{N\gamma_N}/\log{d})$, $s_{\balphatil}\sqrt{s_{p}}=o(N\gamma_N/\log^{3/2}{d})$, and $\zeta_N=o((N\gamma_N)^{-1/2})$.
To sum up, we conclude that
\begin{align*}
\Delta_6=\Delta_{6,1,1}+\sum_{k=2}^5\Delta_{6,1,k}^{(1)}-\sum_{k=2}^5\Delta_{6,1,k}^{(2)}+\Delta_{6,2}=o_p\left((N\gamma_N)^{-1/2}\right).
\end{align*}
\end{proof}

\subsection{Asymptotic normality of the influence function (IF)}
\begin{lemma}\label{lemma:normal_t4}
Let Assumptions \ref{cond:basic}-\ref{cond:approx} hold. Then, $\E\{\psitil_{N}^*(\bZ)\}=O(e_me_\gamma)$, $\Sigmatil_{N}^*=\Var\{\psitil_{N}^*(\bZ)\}=O(\gamma_N^{-1})$, and
$$N^{-1}\sum_{i\in\mathcal I}\psitil_{N}^*(\bZ_i)=O_p\left(e_me_\gamma+N^{-1/2}(\|\balphatil^*\|_2+\gamma_N^{-1/2})\right).$$
Furthermore, let Assumption \ref{cond:lowerbound} hold and $e_me_\gamma=o(N^{-1/2}(\|\balphatil^*\|_2+\gamma_N^{-1/2}))$. Then, $\Sigmatil_{N}^*\asymp N^{-1/2}(\|\balphatil^*\|_2+\gamma_N^{-1/2})$, and we have the following asymptotic normality: as $N,d\to\infty$,
\begin{equation}\label{norm_SSL}
\sqrt N\left(\Sigmatil_{N}^*\right)^{-1/2}N^{-1}\sum_{i\in\mathcal I}\psitil_{N}^*(\bZ_i)\xrightarrow{d}\mathcal N(0,1).
\end{equation}
\end{lemma}

\vspace{-1em}

\begin{proof}[Proof of Lemma \ref{lemma:normal_t4}]
{Notice that, under Assumption \ref{cond:basic},
\begin{align*}
\omega:=&\E\{\psitil_{N}^*(\bZ)\}=\E\left[\mtil^*(\bX)+\frac{\Gamma}{\gammatil_N^*(\bX)}\{Y-\mtil^*(\bX)\}\right]-\theta_1\\
=&-\E\left[\left\{\frac{\gamma_N(\bX)}{\gammatil_N^*(\bX)}-1\right\}\left\{m(\bX)-\mtil^*(\bX)\right\}\right]=O(e_me_\gamma).
\end{align*}
Besides, with ${r}>0$ satisfies $1/r+1/q=1$,
\begin{align}
&\Sigmatil_{N}^*\leq\E\left[\mtil^*(\bX)+\frac{\Gamma\{Y-\mtil^*(\bX)\}}{\gammatil_N^*(\bX)}\right]^2\leq2\E\{\mtil^*(\bX)\}^2+2\E\left[\frac{\Gamma\{Y-\mtil^*(\bX)\}^2}{\{\gammatil_N^*(\bX)\}^2}\right]\nonumber\\
&\qquad=2\E(\bS^T\balphatil^*)^2+2\E\left[\gamma_N(\bX)\{1+\gamma_N^{-1}\exp(-\bS^T\bbeta_{p}^*)\}^2\varepsilon^2\right]\nonumber\\
&\qquad\leq2\E(\bS^T\balphatil^*)^2+2\{1+\gamma_N^{-1}\exp(C_0)\}^2\|\gamma_N(\cdot)\|_{\P,q}\|\varepsilon\|_{\P,2r}^2=O(\|\balphatil^*\|_2^2+\gamma_N^{-1}),\label{bound:upper-Sigma}
\end{align}
under Assumptions \ref{cond:subG}, \ref{cond:tail}, \ref{cond:bound-SP}, and \ref{cond:subG'}. By Chebyshev's inequality,
$$N^{-1}\sum_{i\in\mathcal I}\psitil_{N}^*(\bZ_i)=O_p\left(e_me_\gamma+N^{-1/2}(\|\balphatil^*\|_2+\gamma_N^{-1/2})\right).$$
Now, we construct a lower bound for $\Sigmatil_{N}^*$ under an additional Assumption \ref{cond:lowerbound}. Let $\varepsilon_1:=Y(1)-m(\bX)$ and $\varepsilon_2:=m(\bX)-\mtil^*(\bX)$. Then,
\begin{align}
\Sigmatil_{N}^*&=\Var\left[\frac{\Gamma\varepsilon_1}{\gammatil_N^*(\bX)}+\frac{\{\gamma_N(\bX)-\Gamma\}\varepsilon_2}{\gammatil_N^*(\bX)}-\left\{1-\frac{\gamma_N(\bX)}{\gammatil_N^*(\bX)}\right\}\varepsilon_2+m(\bX)-\theta_1\right]\nonumber\\
&\overset{(i)}{=}\Var\left\{\frac{\Gamma\varepsilon_1}{\gammatil_N^*(\bX)}\right\}+\Var\left[\frac{\{\gamma_N(\bX)-\Gamma\}\varepsilon_2}{\gammatil_N^*(\bX)}\right]+\Var\left\{\mtil^*(\bX)-\frac{\gamma_N(\bX)\varepsilon_2}{\gammatil_N^*(\bX)}\right\}\label{bound:lower-Sigma}\\
&\geq\Var\left\{\frac{\Gamma\varepsilon_1}{\gammatil_N^*(\bX)}\right\}=\E\left[\frac{\gamma_N(\bX)\varepsilon_1^2}{\{\gammatil_N^*(\bX)\}^2}\right]\overset{(ii)}\geq c_0^{-2}\exp(-C_0)\gamma_N^{-2}\E\{\gamma_N(\bX)\varepsilon_1^2\}\nonumber\\
&\overset{(iii)}\geq c_{\min}c_0^{-2}\exp(-C_0)\gamma_N^{-1},\label{bound:lower-Sigma-2}
\end{align}
where (i) holds under Assumption \ref{cond:basic}; (ii) holds since $\{\gammatil_N^*(\bX)\}^{-2}=\{\pi^*(\bX)\}^{-2}\{1+\gamma_N^{-2}\exp(-2\bS^T\bbeta_{p}^*)\}\geq c_0^{-2}\gamma_N^{-2}\exp(-C_0)$ almost surely under Assumptions \ref{cond:NP-est} and \ref{cond:bound-SP}; (iii) holds since, under Assumptions \ref{cond:basic} and \ref{cond:lowerbound}, $\E\{\gamma_N(\bX)\varepsilon_1^2\}=\E(\Gamma\varepsilon_1^2)=\E(\varepsilon_1^2\mid\Gamma=1)\E(\Gamma)\geq c_{\min}\gamma_N$. Moreover, by \eqref{bound:lower-Sigma}, we also have
\begin{align*}
\Sigmatil_{N}^*&\geq\Var\left\{\mtil^*(\bX)-\frac{\gamma_N(\bX)\varepsilon_2}{\gammatil_N^*(\bX)}\right\}\\
&\geq\Var\{\mtil^*(\bX)\}-2\sqrt{\Var\{\mtil^*(\bX)\}\Var\left\{\frac{\gamma_N(\bX)\varepsilon_2}{\gammatil_N^*(\bX)}\right\}}.
\end{align*}
Here,
$$\Var\left\{\frac{\gamma_N(\bX)\varepsilon_2}{\gammatil_N^*(\bX)}\right\}\leq\E\left\{\frac{\gamma_N(\bX)\varepsilon_2}{\gammatil_N^*(\bX)}\right\}^2\leq e_m^2(1+e_\gamma)^2.$$
When $\|\balphatil^*\|_2>4e_m(1+e_\gamma)/\sqrt{\kappa_l}$, under Assumption \ref{cond:subG},
$$\Var\{\mtil^*(\bX)\}=\Var(\bS^\top\balphatil^*)\geq\kappa_l\|\balphatil^*\|_2^2\geq16\Var\left\{\frac{\gamma_N(\bX)\varepsilon_2}{\gammatil_N^*(\bX)}\right\}.$$
Hence,
\begin{align*}
\Sigmatil_{N}^*&\geq\Var\{\mtil^*(\bX)\}-2\sqrt{\Var\{\mtil^*(\bX)\}\cdot\Var\{\mtil^*(\bX)\}/16}\\
&=\Var\{\mtil^*(\bX)\}/2\geq\kappa_l\|\balphatil^*\|_2^2/2,
\end{align*}
On the other hand, when $\|\balphatil^*\|_2\leq4e_m(1+e_\gamma)/\sqrt{\kappa_l}$, we have $\|\balphatil^*\|_2^2=O(1)=O(\gamma_N^{-1})$. Together with \eqref{bound:lower-Sigma-2},
$$(\Sigmatil_{N}^*)^{-1}=O\left((\|\balphatil^*\|_2^2+\gamma_N^{-1})^{-1}\right).$$
Further combining with \eqref{bound:upper-Sigma}, we have 
$$\Sigmatil_{N}^*\asymp\|\balphatil^*\|_2^2+\gamma_N^{-1}.$$
Additionally, for any $c>0$ and with $r>0$ satisfying $1/r+1/q=1$,
\begin{align*}
&\|\psitil_{N}^*(\cdot)\|_{\P,2+c}=\left\|\mtil^*(\bX)-\theta_1+\frac{\Gamma\varepsilon}{\gammatil_N^*(\bX)}\right\|_{\P,2+c}\leq\left\|\mtil^*(\cdot)\right\|_{\P,2+c}+|\theta_1|+\left\|\frac{\Gamma\varepsilon}{\gammatil_N^*(\bX)}\right\|_{\P,2+c}\\
&\qquad\overset{(i)}{=}\left\|\bS^T\balphatil^*\right\|_{\P,2+c}+|\E(\bS^T\balphatil^*)|+O\left(\gamma_N^{1/(2+c)-1}\right)\overset{(ii)}{=}O\left(\|\balphatil^*\|_2+\gamma_N^{1/(2+c)-1}\right),
\end{align*}
where (i) holds since, by part (c) of Lemma \ref{lemma:subG} and under Assumption \ref{cond:tail}, we have $\{\E(|\Gamma\varepsilon|^{2+c})\}^{1/(2+c)}=[\E\{\gamma_N(\bX)|\varepsilon|^{2+c}\}]^{1/(2+c)}\leq\|\gamma_N(\bX)\|_{\P,q}^{1/(2+c)}\|\varepsilon\|_{\P,r(2+c)}=O(\gamma_N^{1/(2+c)})$ and, under Assumptions \ref{cond:NP-est} and \ref{cond:bound-SP}, $1/\gammatil_N^*(\bX)=\{1+\gamma_N^{-1}\exp(-\bS^T\bbeta_{p}^*)\}/\pi^*(\bX)\leq c_0^{-1}+(c_0\gamma_N)^{-1}\exp(C_0)$ almost surely; (ii) holds by part (c) of Lemma \ref{lemma:subG}. Additionally, $\omega=O(e_me_\gamma)=O(1)$ and it follows that
$$\|\psitil_{N}^*(\cdot)-\omega\|_{\P,2+c}=O\left(\|\balphatil^*\|_2+\gamma_N^{1/(2+c)-1}\right).$$
Hence,
\begin{align}
&N^{-c/2}(\Sigmatil_{N}^*)^{1+c/2}\E|\psitil_{N}^*(\bZ)-\omega|^{2+c}=O\left(\frac{(\|\balphatil^*\|_2+\gamma_N^{1/(2+c)-1})^{2+c}}{N^{c/2}(\|\balphatil^*\|_2^2+\gamma_N^{-1})^{1+c/2}}\right)\nonumber\\
&\qquad=O\left(\frac{\|\balphatil^*\|_2^{2+c}+\gamma_N^{-1-c}}{N^{c/2}(\|\balphatil^*\|_2^{2+c}+\gamma_N^{-1-c/2})}\right)=O\left((N\gamma_N)^{-c/2}\right)=o(1).\label{moments_psistar}
\end{align}
It follows that, for any $\delta>0$ as $N,d\to\infty$,
$$\gamma_N^{-1}\E\left[\left\{\psitil_{N}^*(\bZ)-\omega\right\}^2\mathbbm{1}_{|\psitil_{N}^*(\bZ)-\omega|>\delta\sqrt{N/\gamma_N}
}\right]=o(1).$$
By Proposition 2.27 (Lindeberg-Feller theorem) of \cite{van2000asymptotic},
$$\sqrt N\left(\Sigmatil_{N}^*\right)^{-1/2}\left\{N^{-1}\sum_{i\in\mathcal I}\psitil_{N}^*(\bZ_i)-\omega\right\}\xrightarrow{d}\mathcal N(0,1).$$
Hence, \eqref{norm_SSL} holds when $\omega=O(e_me_\gamma)=o(N^{-1/2}(\|\balphatil^*\|_2+\gamma_N^{-1/2}))$.}
\end{proof}

\subsection{Asymptotic variance estimation}
\begin{lemma}\label{lemma:var-est}
Let the assumptions of Theorems \ref{thm:PS-SP}, \ref{thm:OR-SP-mis}, and Lemmas \ref{lemma:Delta2}, \ref{lemma:normal_t4} hold. Let $s_{\balphatil}+s_{p}=o(N\gamma_N/(\log{d}\log^{1/2}{N}))$ and $\thetahat_{\mbox{\tiny 1,DC-BRSS}}-\theta_1=o_p(\gamma_N^{-1/2})$. Then, with some $\lambda_{\bbeta}\asymp\lambda_{\balpha}\asymp\sqrt{\log{d}/(N\gamma_N)}$, as $N,d\to\infty$,
\begin{align}\label{result:var-est}
\Sigmahat_{\mbox{\tiny 1,DC-BRSS}}=\Sigmatil_{N}^*\{1+o_p(1)\}.
\end{align}
\end{lemma}

\begin{proof}[Proof of Lemma \ref{lemma:var-est}]
For each $i\in\mathcal I_k$, define $g_i:=\pi^*(\bX_i)g(\bS_i^T\bbeta_{p}^*+\log(\gamma_N))$, $\ghat_i^{(k)}:=\pihat^{(k)}(\bX_i)g(\bS_i^T\bbetahat_{p}^{(k)}+\log(\phat_{N}^{(k)}))$, where $g(\cdot)$ is the logistic function and let
$$\psihat_{N}^{(k)}(\bZ_i):=\bS_i^T\balphahat^{(k')}+\Gamma_i(Y_i-\bS_i^T\balphahat^{(k')})/\ghat_i^{(k)}-\thetahat_{\mbox{\tiny 1,DC-BRSS}},$$
with $k'\neq k\in\{1,2\}$. Then,
\begin{align}\label{bound:gbar}
1/\ghat_i^{(k)}=\{1+\exp(-\bS_i^T\bbetahat_{p}^{(k)})/\phat_{N}^{(k)}\}/\pi^{(k)}(\bX_i)\leq\{1+\exp(C_0)/\phat_{N}^{(k)}\}/c_0.
\end{align}
Observe that
$$\psihat_{N}^{(k)}(\bZ_i)-\psitil_{N}^*(\bZ_i)=\Deltahat_{1,i}^{(k)}+\Deltahat_{2,i}^{(k)},$$
where
\begin{align*}
\Deltahat_{1,i}^{(k)}:=(1-\Gamma_i/\ghat_i^{(k)})\bS_i^T(\balphahat^{(k')}-\balphatil^*),\quad\Deltahat_{2,i}^{(k)}:=\Gamma_i(1/\ghat_i^{(k)}-1/g_i)\varepsilon_i.
\end{align*}
Note that
\begin{align}
&M^{-1}\sum_{i\in\mathcal I_k}(\Deltahat_{1,i}^{(k)})^2=M^{-1}\sum_{i\in\mathcal I_k}\left\{1-2\Gamma_i/\ghat_i^{(k)}+\Gamma_i/(\ghat_i^{(k)})^2\right\}\left\{\bS_i^T(\balphahat^{(k')}-\balphatil^*)\right\}^2\nonumber\\
&\qquad\overset{(i)}{\leq}M^{-1}\sum_{i\in\mathcal I_k}\left\{\bS_i^T(\balphahat^{(k')}-\balphatil^*)\right\}^2\nonumber\\
&\qquad\qquad+c_0^{-2}\left\{1+\exp(C_0)/\phat_{N}^{(k)}\right\}^{2}M^{-1}\sum_{i\in\mathcal I_k}\Gamma_i\left\{\bS_i^T(\balphahat^{(k')}-\balphatil^*)\right\}^2.\label{bound:Deltabar1}
\end{align}
where (i) holds by \eqref{bound:gbar}. {By part (a) of Lemma \ref{lemma:emp} choosing $s=\sbar+e_m^2\zeta_N^2N\gamma_N/\log{d}=o(N/\log{d})$, together with \eqref{rate:alphatil1} and \eqref{rate:alphatil2},
\begin{align*}
M^{-1}\sum_{i\in\mathcal I_k}\{\bS_i^T(\balphahat^{(k')}-\balphatil^*)\}^2=O_p\left(\frac{\sbar\log{d}}{N\gamma_N}+e_m^2\zeta_N^2\right)=o_p(1),
\end{align*}
as long as $\sbar=o(N\gamma_N/\log{d})$ and $e_m\zeta_N=o(1)$. Additionally, by part (b) of Lemma \ref{lemma:emp}, together with \eqref{rate:alphatil1} and \eqref{rate:alphatil2}, we also have
\begin{align*}
&(M\gamma_N)^{-1}\sum_{i\in\mathcal I_k}\Gamma_i\{\bS_i^T(\balphahat^{(k')}-\balphatil^*)\}^2\\
&\qquad=O_p\left(\frac{\sbar^2\log^2{d}\log{N}}{(N\gamma_N)^2}+e_m^4\zeta_N^4\log{N}+\frac{\sbar\log{d}}{N\gamma_N}+e_m^2\zeta_N^2\right)=o_p(1),
\end{align*}
if $\sbar=o(N\gamma_N/(\log{d}\log^{1/2}{N}))$ and $e_m\zeta_N=o(\log^{-1/4}{N})$.}

By Lemma \ref{lemma:phatN-est} and note that $p_{N}\asymp\gamma_N$, we have $\phat_{N}^{(k)}=O_p(\gamma_N)$ and $1/\phat_{N}^{(k)}=O_p(\gamma_N^{-1})$. Together with \eqref{bound:Deltabar1}, we obtain
$$M^{-1}\sum_{i\in\mathcal I_k}(\Deltahat_{1,i}^{(k)})^2=o_p(\gamma_N^{-1}).$$
Besides, under Assumption \ref{cond:basic}, we have
$$\E_{\mathcal D_N}\left\{M^{-1}\sum_{i\in\mathcal I_k}(\Deltahat_{1,2}^{(k)})^2\mid(R_i,T_i,\bX_i)_{i\in\mathcal D_N^{(k)}}\right\}=M^{-1}\sum_{i\in\mathcal I_k}\Gamma_i(1/\ghat_i^{(k)}-1/g_i)^2\E(\varepsilon_i^2|\bX_i).$$
Under Assumptions \ref{cond:subG'} and \ref{cond:approx}, $\sup_{i\in\mathcal I_k}\E(\varepsilon_i^2|\bX_i)=O(1)$. By \eqref{rate:gammadiff_insample}, we also have
$$M^{-1}\sum_{i\in\mathcal I_k}\Gamma_i(1/\ghat_i^{(k)}-1/g_i)^2=o_p(\gamma_N^{-1}),$$
and it follows that
$$\E_{\mathcal D_N}\left\{M^{-1}\sum_{i\in\mathcal I_k}(\Deltahat_{1,2}^{(k)})^2\mid(R_i,T_i,\bX_i)_{i\in\mathcal D_N^{(k)}}\right\}=o_p(\gamma_N^{-1}).$$
By Markov's inequality,
$$M^{-1}\sum_{i\in\mathcal I_k}(\Deltahat_{1,2}^{(k)})^2=o_p(\gamma_N^{-1}).$$
Therefore,
\begin{align*}
M^{-1}\sum_{i\in\mathcal I_k}\left\{\psihat^{(k)}_{N}(\bZ_i)-\psitil_{N}^*(\bZ_i)\right\}^2\leq2M^{-1}\sum_{i\in\mathcal I_k}(\Deltahat_{1,1}^{(k)})^2+2M^{-1}\sum_{i\in\mathcal I_k}(\Deltahat_{1,2}^{(k)})^2=o_p(\gamma_N^{-1}).
\end{align*}
In addition, by Lemma \ref{lemma:normal_t4}, we have $\Sigmatil_{N}^*\asymp\|\balphatil^*\|_2^2+\gamma_N^{-1}$. Together with \eqref{moments_psistar} and Lemma D.4 of \cite{zhang2023double}, we have
\begin{align*}
M^{-1}\sum_{i\in\mathcal I_k}\left\{\psihat^{(k)}_{N}(\bZ_i)\right\}^2=\E\left[\{\psitil_{N}^*(\bZ)\}^2\right]\{1+o_p(1)\}\overset{(i)}{=}\Sigmatil_{N}^*\{1+o_p(1)\},
\end{align*}
where (i) holds since $\E[\{\psitil_{N}^*(\bZ)\}^2]=\Sigmatil_{N}^*+[\E\{\psitil_{N}^*(\bZ)\}]^2=\Sigmatil_{N}^*\{1+o_p(1)\}$ when $e_me_\gamma=o(N^{-1/2}(\|\balphatil^*\|_2+\gamma_N^{-1/2}))$. Therefore,
\begin{align*}
&\Sigmahat_{\mbox{\tiny 1,DC-BRSS}}-\Sigmatil_{N}^*=N^{-1}\sum_{k=1}^2\sum_{i\in\mathcal I_k}\left\{\psihat^{(k)}_{N}(\bZ_i)\right\}^2-\Sigmatil_{N}^*=o_p(\Sigmatil_{N}^*),
\end{align*}
and \eqref{result:var-est} follows.
\end{proof}

\subsection{Proofs of Theorems \ref{t4-SP} and \ref{cor:para}}
\begin{proof}[Proof of Theorem \ref{t4-SP}]
Theorem \ref{t4-SP} follows directly from Lemmas \ref{lemma:Delta1}-\ref{lemma:var-est}.
\end{proof}

\begin{proof}[Proof of Theorem \ref{cor:para}]
Theorem \ref{cor:para} follows directly from Theorem \ref{t4-SP} with $\balphatil^*=\balpha^*$, $\bbeta_{p}^*=\bbeta^*$, and $\zeta_N=0$.
\end{proof}

\section{Extended properties of the preliminary R-DR estimator}\label{sec:add-S1}

In this section, we extend the theoretical properties of the R-DR estimator, establishing both consistency and asymptotic normality. Theorem \ref{thm:gen-rate-DR} in the main document provides a concise summary of the results presented here. We assume the following conditions.

\begin{assumption}[High-level conditions on the nuisance function estimators]\label{a2:high-level}
Let $m^*(\cdot)$ and $\gamma_N^*(\cdot)$ be approximations of the true nuisance functions $m(\cdot)$ and $\gamma_N(\cdot)$, respectively. For each $k\leq K$, assume the following conditions hold (recall $\E_{\bX}(\cdot)$ from Section \ref{sec:intro}):
\begin{align}
& \E_{\bX}\left[ \frac{a_{N}}{\gamma_N(\bX)} \{\mhat^{(-k)}(\bX) - m^*(\bX) \}^2 \right] = O_p(\alpha_{N}^2) \quad \mbox{with some} \;\; \alpha_{N} = o(1), \label{drthm:ratecond1}\\
& \E_{\bX}\left[ \frac{a_{N}}{\gamma_N(\bX)} \left\{1 - \frac{\gamma_N^*(\bX)}{\gammahat_N^{(-k)}(\bX)} \right\}^2\right] = O_p(\beta_{N}^2) \quad \mbox{with some} \;\; \beta_{N} = o(1),\label{drthm:ratecond2} \\
& \E_{\bX}\{\mhat^{(-k)}(\bX) - m^*(\bX) \}^2 = O_p(c_{N}^2) \quad \mbox{with some} \;\; c_{N} = o(1),\;\;\mbox{and}\label{drthm:ratecond3}\\
& \E_{\bX}\left\{1 - \frac{\gamma_N^*(\bX)}{\gammahat_N^{(-k)}(\bX)} \right\}^2 = O_p(d_{N}^2) \quad \mbox{with some} \;\; d_{N} = o(1).\label{drthm:ratecond4}
\end{align}
\end{assumption}

\begin{assumption}[Tail condition 1]\label{a3:L-F}
For any $\delta>0$,
$$a_{N}^{-1}\E\left[\left\{\psi_{N}^{opt}(\bZ)\right\}^2\mathbbm{1}_{|\psi_{N}^{opt}(\bZ)|>\delta\sqrt{N/a_{N}}}\right]\to0\;\;\text{as}\;\;N\to\infty.$$
\end{assumption}

\begin{assumption}[Tail condition 2]\label{cond:est_var}
For a constant $c>0$, let
$$N^{-c/2}a_{N}^{1+c/2}\E\left\{|\psi_{N}^{opt}(\bZ)|^{2+c}\right\}\to0\;\;\text{as}\;\;N\to\infty.$$
\end{assumption}

Assumptions \ref{a2:high-level}-\ref{cond:est_var} are analogous to the required conditions in Theorem 3.2 of \cite{zhang2023double}, where the authors considered the estimation of $\E(Y)$ in a non-causal setup. In our causal setup, if we treat $\Ytil=Y(1)$ as the outcome of interest, the estimation of $\theta_1=\E(\Ytil)$ can also be seen as a mean estimation problem with missing outcomes. However, in our causal \emph{and} missing data problem, $\Ytil$ suffers from two different types of missingness: (a) missing due to the labeling procedure (occurs when $R=0$) and (b) missing due to the treatment assignment (occurs when $T=0$). We can only observe $\Ytil$ if $\Gamma=TR=1$. Hence, as also indicated in Lemma \ref{lemma:rep}, $\Gamma$ can be viewed as the effective labeling indicator. Based on the triple $(\Ytil,\Gamma,\bX)$, we show Theorem \ref{t1} applying the results developed in Theorem 3.2 of \cite{zhang2023double}.

Notably, while Theorem \ref{thm:gen-rate-DR} in the main file relies only on Assumptions \ref{a2:high-level}-\ref{a3:L-F}, Assumption \ref{cond:est_var} is additionally required to ensure a consistent estimate of the asymptotic variance in Theorem \ref{t1}, the full version presented below.

\begin{theorem}[Asymptotic results for $\thetahat_{1,\mbox{\tiny R-DR}}$]\label{t1}
Let Assumptions \ref{cond:basic} and \ref{a2:high-level} hold. Let $a_{N}=[\E\{\gamma_N^{-1}(\bX)\}]]^{-1}$, assume $Na_{N} \rightarrow \infty$ and possibly, $a_{N} \rightarrow 0$, as $N \rightarrow \infty$. Further, with $m^*(\cdot)$ and $\gamma_N^*(\cdot)$ as in Assumption \ref{a2:high-level}, define functions $\psi^*_{N}(\bZ):=m^*(\bX) - \theta_1 + \Gamma\{Y(1) - m^*(\bX)\}/\gamma_N^*(\bX)$ and $\psi^{opt}_{N}(\bZ):=m(\bX) - \theta_1 + \Gamma\{Y(1) - m(\bX)\}/\gamma_N(\bX)$. Assume that $\psi^*_{N}(\bZ) \in \Lsc_2(\P_{\bX})$ and note that $\E\{\psi_{N}^*(\bZ)\} = 0$ whenever $m^*(\cdot) = m(\cdot)$ or $\gamma_N^*(\cdot) = \gamma_N(\cdot)$ but not necessarily both. Let $\Sigma_{N}^* := \Var\{\psi^*_{N}(\bZ)\}$ and $\Sigma_{N}^{opt} := \Var\{\psi^{opt}_{N}(\bZ)\}$. Then the properties of $\thetahat_{1,\mbox{\tiny R-DR}}$ as defined in \eqref{eq:SS-est-def} under different cases are as follows.

(a) Suppose $\gamma^*_N(\cdot) = \gamma_N(\cdot)$ and $m^*(\cdot) = m(\cdot)$. Assume $\E[\{Y(1)-m(\bX)\}^2|\bX]\leq C<\infty$. Then, $\thetahat_{1,\mbox{\tiny R-DR}}$ satisfies the following (`optimal') asymptotic linear expansion:
\begin{align*}
& (\thetahat_{1,\mbox{\tiny R-DR}} - \theta_1) = \frac{1}{N} \sum_{i=1}^N \psi^{opt}_{N}(\bZ_i) + O_p\left(\frac{\alpha_{N}}{\sqrt{Na_{N}}} + \frac{\beta_{N}}{\sqrt{Na_{N}}}\right) + D_N, \;\;\mbox{where}\\
& \;\; |D_N| := \left| \frac{1}{N}\sum_{k=1}^K\sum_{i\in\mathcal I_k} \left\{\frac{\Gamma_i}{\gammahat_N^{(-k)}(\bX_i)} - \frac{\Gamma_i}{\gamma_N^*(\bX_i)}\right\} \left\{\mhat^{(-k)}(\bX_i) - m^*(\bX_i) \right\} \right|= O_p\left(c_{N}d_{N}\right),
\end{align*}
and we note further that $\E\{\psi^{opt}_{N}(\bZ)\} = 0$ and $\Sigma_{N}^{opt} := \Var\{\psi^{opt}_{N}(\bZ)\} \asymp a_{N}^{-1}$. Further, assume $\E[\{Y(1)-m(\bX)\}^2|\bX]\geq c>0$ and $\Var\{Y(1)\}\leq C<\infty$. Then, under Assumption \ref{a3:L-F}, and as long as the product rate $c_{N}d_{N}$ from \eqref{drthm:ratecond3} and \eqref{drthm:ratecond4} additionally satisfies: $c_{N}d_{N}= o(1/\sqrt{Na_{N}})$, we have:
\begin{align*}
& \sqrt{Na_{N}}(\thetahat_{1,\mbox{\tiny R-DR}} - \theta_1) = O_p(1) \;\; \mbox{and} \;\; \sqrt{N} \left(\Sigma^{opt}_{N}\right)^{-1/2}(\thetahat_{1,\mbox{\tiny R-DR}} - \theta_1)\xrightarrow{d}\mathcal N(0,1).
\end{align*}
Further, let Assumption \ref{cond:est_var} hold. Then, as $N,d\to\infty$,
$$\Sigmahat_{N}=\Sigma_{N}^{opt}\{1+o_p(1)\} \;\; \mbox{and} \;\; \sqrt{N} \left(\Sigmahat^{opt}_{N}\right)^{-1/2}(\thetahat_{1,\mbox{\tiny R-DR}} - \theta_1)\xrightarrow{d}\mathcal N(0,1),$$
where we consider the following plug-in estimate of the asymptotic variance, $\Sigma_{N}^{opt}$:
$$\Sigmahat^{opt}_{N}:=\frac{1}{N}\sum_{k=1}^K\sum_{i\in\mathcal I_k}\left[\mhat^{(-k)}(\bX_i)+\frac{\Gamma_i\{Y_i-\mhat^{(-k)}(\bX_i)\}}{\gammahat_N^{(-k)}(\bX_i)}-\thetahat_{1,\mbox{\tiny R-DR}}\right]^2.$$
(b) Suppose either $\gamma_N^*(\cdot) = \gamma_N(\cdot)$ or $m^*(\cdot) = m(\cdot)$, but not necessarily both. Assume $\gamma_N^*(\cdot)\geq c\gamma_N(\cdot)$ with some $c>0$.
Then, $\thetahat_{1,\mbox{\tiny R-DR}}$ satisfies the following expansion:
\begin{align*}
&(\thetahat_{1,\mbox{\tiny R-DR}} - \theta_1) = \frac{1}{N} \sum_{i=1}^N \psi^{*}_{N}(\bZ_i) + O_p\left(\frac{\alpha_{N}}{\sqrt{Na_{N}}} + \frac{\beta_{N}}{\sqrt{Na_{N}}}\right) + D_N+ \Deltahat_{N},
\end{align*}
where
\begin{align}
&\Deltahat_{N} := \frac{1}{N} \sum_{k=1}^K\sum_{i\in\mathcal I_k} \frac{\Gamma_i}{\gamma_N(\bX_i)}\left\{1 - \frac{\gamma_N(\bX_i)}{\gammahat_N^{(-k)}(\bX_i)}\right\}\{m^*(\bX_i) - m(\bX_i)\} \;\; \mbox{if} \;\; \gamma_N^*(\cdot) = \gamma_N(\cdot),\label{def:Deltahat1}\\
&\mbox{or}\;\;\Deltahat_{N} := \frac{1}{N} \sum_{k=1}^K\sum_{i\in\mathcal I_k} \frac{\Gamma_i}{\gamma_N(\bX_i)}\left\{1 - \frac{\gamma_N(\bX_i)}{\gamma_N^*(\bX_i)}\right\}\{\mhat^{(-k)}(\bX_i) - m(\bX_i)\} \;\; \mbox{if} \;\; m^*(\cdot) = m(\cdot).\label{def:Deltahat1'}
\end{align}
Assume $0<c\leq\E[\{Y(1)-m^*(\bX)\}^2|\bX]\leq C<\infty$, $\Var\{Y(1)\}\leq C<\infty$, and $\gamma_N^*(\cdot)\leq C\gamma_N(\cdot)$. We note further that $\E\{\psi^{*}_{N}(\bZ_i)\} = 0$ and $\Sigma_{N}^{*} := \Var\{\psi^*_{N}(\bZ)\} \asymp a_{N}^{-1}$, and in general, under Assumption \ref{a2:high-level} and accounting for both cases \eqref{def:Deltahat1} and \eqref{def:Deltahat1'} above, $\Deltahat_{N}$ always satisfies:
$$
\Deltahat_{N} = O_p\left(d_{N}\mathbbm{1}_{m^*(\cdot) \neq m(\cdot)} + c_{N} \mathbbm{1}_{\gamma_N^*(\cdot) \neq \gamma_N(\cdot)}\right).
$$
\end{theorem}

\section{Proof of results in Sections \ref{sec:psetup} and \ref{sec:add-S1}}\label{sec:proof_sec:gen}

\begin{proof}[Proof of Lemma \ref{lemma:rep}]
Since $Y(1)\independent T\mid\bX$, we have
$$m(\bX)=\E\{Y(1)\mid\bX\}=\E\{Y(1)\mid\bX,T=1\}.$$
Additionally, since $R\independent Y\mid(T,\bX)$, we have $R\independent Y(1)\mid(\bX,T=1)$. Therefore,
\begin{align}
\E(\Gamma Y\mid\bX)&=\E(RTY\mid\bX)=\E\{Y(1)\mid\bX,T=1,R=1\}\P\{RT=1\mid\bX\}\\
&=\E\{Y(1)\mid\bX,T=1\}\P(\Gamma=1\mid\bX)=m(\bX)\gamma_N(\bX).\label{eq:tower}
\end{align}
Let either $m^*(\cdot) = m(\cdot)$ or $\gamma_N^*(\cdot) = \gamma_N(\cdot)$ hold. Then, it follows that
\begin{align*}
&\E\left[m^*(\bX)+\frac{\Gamma}{\gamma_N^*(\bX)}\left\{Y-m^*(\bX)\right\}\right]-\theta_1\\
&\qquad=\E\left\{\E\left(\left[m^*(\bX)+\frac{\Gamma}{\gamma_N^*(\bX)}\left\{Y-m^*(\bX)\right\}\right]\mid\bX\right)\right\}-\theta_1\\
&\qquad\overset{(i)}{=}\E\left[m^*(\bX)+\frac{\gamma_N(\bX)}{\gamma_N^*(\bX)}\left\{m(\bX)-m^*(\bX)\right\}\right]-\E\{m(\bX)\}\\
&\qquad=\E\left[\left\{\frac{\gamma_N(\bX)}{\gamma_N^*(\bX)}-1\right\}\left\{m(\bX)-m^*(\bX)\right\}\right]=0,
\end{align*}
where (i) holds by \eqref{eq:tower} and $\E(\Gamma\mid\bX)=\gamma_N(\bX)$.
\end{proof}

\begin{proof}[Proof of Theorem \ref{t1}]
We first show that $Y(1)\independent\Gamma\mid\bX$ under Assumption \ref{cond:basic}. Observe that
\begin{align*}
&\E\{Y(1)\mid\bX,\Gamma=1\}\overset{(i)}{=}\E(Y\mid\bX,R=T=1)\overset{(ii)}{=}\E(Y\mid\bX,T=1)\\
&\qquad\overset{(iii)}{=}\E\{Y(1)\mid\bX,T=1\}\overset{(iv)}{=}\E\{Y(1)\mid\bX\},
\end{align*}
where (i) holds since $Y=Y(T)$ and $\Gamma=TR$; (ii) holds since $R\independent Y\mid(T,X)$; (iii) holds since $Y=Y(T)$; (iv) holds since $T\independent Y(1)\mid\bX$. In addition, by the tower rule,
\begin{align*}
\E\{Y(1)\mid\bX\}&=\E\{Y(1)\mid\bX,\Gamma=1\}\gamma_N(\bX)+\E\{Y(1)\mid\bX,\Gamma=0\}\{1-\gamma_N(\bX)\}\\
&=\E\{Y(1)\mid\bX\}\gamma_N(\bX)+\E\{Y(1)\mid\bX,\Gamma=0\}\{1-\gamma_N(\bX)\}.
\end{align*}
It follows that
\begin{align*}
\E\{Y(1)\mid\bX\}\{1-\gamma_N(\bX)\}=\E\{Y(1)\mid\bX,\Gamma=0\}\{1-\gamma_N(\bX)\}.
\end{align*}
Since $\gamma_N(\bX)=\pi(\bX)p_N(\bX)>0$ almost surely under Assumption \ref{cond:basic}, we have
\begin{align*}
\E\{Y(1)\mid\bX,\Gamma=0\}=\E\{Y(1)\mid\bX\}=\E\{Y(1)\mid\bX,\Gamma=1\}\;\;\mbox{almost surely}.
\end{align*}
That is, $Y(1)\independent\Gamma\mid\bX$, i.e., the missing at random (MAR) condition holds if we consider $\Ytil=Y(1)$ and $\Gamma$ as the outcome variable and labeling indicator. The remaining results hold as long as we apply the results in Theorem 3.2 of \cite{zhang2023double}, with $(Y,R,\bX)$ replaced with the triple $(\Ytil,\Gamma,\bX)$.
\end{proof}

\end{document}

%% file: Definitions.tex
\def\bx{\mathbf{x}}
\def\bX{\mathbf{X}}

\def\bz{\mathbf{z}}
\def\bZ{\mathbf{Z}}

\def\bU{\mathbf{U}}
\def\bv{\mathbf{v}}
\def\bV{\mathbf{V}}

\def\bW{\mathbf{W}}
\def\bs{\mathbf{s}}
\def\br{\mathbf{r}}

\def\K{\mathbb{K}}
\def\E{\mathbb{E}}

\def\P{\mathbb{P}}

\def\balpha{\boldsymbol{\alpha}}
\def\bbeta{\boldsymbol{\beta}}

\def\bSigma{\boldsymbol{\Sigma}}

\def\balphahat{\widehat{\balpha}}
\def\bbetahat{\widehat{\bbeta}}

\def\Sigmahat{\widehat{\Sigma}}

\def\etabold{\boldsymbol{\eta}}

\def\Deltahat{\widehat{\Delta}}

\def\bDeltatil{\widetilde{\bDelta}}
\def\bdeltatil{\widetilde{\bdelta}}
\def\elltil{\widetilde{\ell}}
\def\Bcaltil{\widetilde{\mathcal B}}
\def\Acaltil{\widetilde{\mathcal A}}
\def\Ecaltil{\widetilde{\mathcal E}}
\def\varepsilontil{\widetilde{\varepsilon}}
\def\sbar{\bar{s}}
\def\Sigmatil{\widetilde{\Sigma}}
\def\pbar{\bar{p}}

\def\mhat{\widehat{m}}

\def\fhat{\widehat{f}}

\def\mtil{\widetilde{m}}

\newcommand\ind{\protect\mathpalette{\protect\independenT}{\perp}}
\def\independenT#1#2{\mathrel{\rlap{$#1#2$}\mkern4mu{#1#2}}}

\def\tcr{\textcolor{black}}

\def\Lsc{\mathcal{L}}
\def\Usc{\mathcal{U}}

\def\Xsc{\mathcal{X}}

\def\Var{\mbox{Var}}

\def\bzero{\mathbf{0}}

\def\bS{\mathbf{S}}

\def\ghat{\widehat{g}}

\def\S{\mathbb{S}}

\def\bDelta{\boldsymbol{\Delta}}

\def\bdelta{\boldsymbol{\delta}}

\def\Asc{\mathcal{A}}

\def\R{\mathbb{R}}

\def\S{\mathbb{S}}

\def\Ytil{\widetilde{Y}}

\def\bbetatil{\widetilde{\bbeta}}
\def\balphatil{\widetilde{\balpha}}

\def\Lsc{\mathcal{L}}

\def\bSigma{\boldsymbol{\Sigma}}

\def \hs2{\hspace{2mm}}

\numberwithin{table}{section}
\numberwithin{equation}{section}

\definecolor{jcolor}{RGB}{041,122,000}
\definecolor{darkred}{RGB}{100,000,000}
\definecolor{purple}{RGB}{200,000,200}

\def\boxit#1{\vbox{\hrule\hbox{\vrule\kern6pt  \vbox{\kern6pt#1\kern6pt}\kern6pt\vrule}\hrule}}

\def\muhat{\widehat{\mu}}

\def\muhat{\widehat{\mu}}

\def\muhat{\widehat{\mu}}

\def\be{\mathbf{e}}

\def\bS{\mathbf{S}}

\def\pitil{\widetilde{\pi}}

\def\gammabar{\bar{\gamma}}

\def\psihat{\widehat{\psi}}
\def\pihat{\widehat{\pi}}
\def\mhat{\widehat{m}}

\def\Ytil{\widetilde{Y}}

\def\mtil{\widetilde{m}}

\def\mtil{\widetilde{m}}

\def\pitil{\widetilde{\pi}}

\def\bOmega{\boldsymbol{\Omega}}
\def\bomega{\boldsymbol{\omega}}

\def\thetahat{\widehat{\theta}}
\def\psitil{\widetilde{\psi}}

\def\thetahat{\widehat{\theta}}
\def\psitil{\widetilde{\psi}}

\def\phat{\widehat{p}}
\def\gammahat{\widehat{\gamma}}
\def\gammatil{\widetilde{\gamma}}

%% file: arxiv.bbl
\begin{thebibliography}{69}

\bibitem[\protect\citeauthoryear{Athey, Chetty and
  Imbens}{2020}]{athey2020combining}
\begin{barticle}[author]
\bauthor{\bsnm{Athey},~\bfnm{Susan}\binits{S.}},
  \bauthor{\bsnm{Chetty},~\bfnm{Raj}\binits{R.}} \AND
  \bauthor{\bsnm{Imbens},~\bfnm{Guido}\binits{G.}}
(\byear{2020}).
\btitle{Combining experimental and observational data to estimate treatment
  effects on long term outcomes}.
\bjournal{arXiv preprint arXiv:2006.09676}.
\end{barticle}
\endbibitem

\bibitem[\protect\citeauthoryear{Avagyan and
  Vansteelandt}{2021}]{avagyan2021high}
\begin{barticle}[author]
\bauthor{\bsnm{Avagyan},~\bfnm{Vahe}\binits{V.}} \AND
  \bauthor{\bsnm{Vansteelandt},~\bfnm{Stijn}\binits{S.}}
(\byear{2021}).
\btitle{High-dimensional inference for the average treatment effect under model
  misspecification using penalized bias-reduced double-robust estimation}.
\bjournal{Biostatistics $\&$ Epidemiology}
\bpages{1--18}.
\end{barticle}
\endbibitem

\bibitem[\protect\citeauthoryear{Azriel et~al.}{2022}]{azriel2022semi}
\begin{barticle}[author]
\bauthor{\bsnm{Azriel},~\bfnm{David}\binits{D.}},
  \bauthor{\bsnm{Brown},~\bfnm{Lawrence~D}\binits{L.~D.}},
  \bauthor{\bsnm{Sklar},~\bfnm{Michael}\binits{M.}},
  \bauthor{\bsnm{Berk},~\bfnm{Richard}\binits{R.}},
  \bauthor{\bsnm{Buja},~\bfnm{Andreas}\binits{A.}} \AND
  \bauthor{\bsnm{Zhao},~\bfnm{Linda}\binits{L.}}
(\byear{2022}).
\btitle{Semi-supervised linear regression}.
\bjournal{Journal of the American Statistical Association}
\bvolume{117}
\bpages{2238--2251}.
\end{barticle}
\endbibitem

\bibitem[\protect\citeauthoryear{Bang and Robins}{2005}]{bang2005doubly}
\begin{barticle}[author]
\bauthor{\bsnm{Bang},~\bfnm{Heejung}\binits{H.}} \AND
  \bauthor{\bsnm{Robins},~\bfnm{James~M}\binits{J.~M.}}
(\byear{2005}).
\btitle{Doubly robust estimation in missing data and causal inference models}.
\bjournal{Biometrics}
\bvolume{61}
\bpages{962--973}.
\end{barticle}
\endbibitem

\bibitem[\protect\citeauthoryear{Bhattacharya, Fan and
  Mukherjee}{2024}]{bhattacharya2024deep}
\begin{barticle}[author]
\bauthor{\bsnm{Bhattacharya},~\bfnm{Sohom}\binits{S.}},
  \bauthor{\bsnm{Fan},~\bfnm{Jianqing}\binits{J.}} \AND
  \bauthor{\bsnm{Mukherjee},~\bfnm{Debarghya}\binits{D.}}
(\byear{2024}).
\btitle{Deep neural networks for nonparametric interaction models with
  diverging dimension}.
\bjournal{The Annals of Statistics}
\bvolume{52}
\bpages{2738--2766}.
\end{barticle}
\endbibitem

\bibitem[\protect\citeauthoryear{Bradic, Wager and
  Zhu}{2019}]{bradic2019sparsity}
\begin{barticle}[author]
\bauthor{\bsnm{Bradic},~\bfnm{Jelena}\binits{J.}},
  \bauthor{\bsnm{Wager},~\bfnm{Stefan}\binits{S.}} \AND
  \bauthor{\bsnm{Zhu},~\bfnm{Yinchu}\binits{Y.}}
(\byear{2019}).
\btitle{Sparsity double robust inference of average treatment effects}.
\bjournal{arXiv preprint arXiv:1905.00744}.
\end{barticle}
\endbibitem

\bibitem[\protect\citeauthoryear{Bradic et~al.}{2019}]{bradic2019minimax}
\begin{barticle}[author]
\bauthor{\bsnm{Bradic},~\bfnm{Jelena}\binits{J.}},
  \bauthor{\bsnm{Chernozhukov},~\bfnm{Victor}\binits{V.}},
  \bauthor{\bsnm{Newey},~\bfnm{Whitney~K}\binits{W.~K.}} \AND
  \bauthor{\bsnm{Zhu},~\bfnm{Yinchu}\binits{Y.}}
(\byear{2019}).
\btitle{Minimax semiparametric learning with approximate sparsity}.
\bjournal{arXiv preprint arXiv:1912.12213}.
\end{barticle}
\endbibitem

\bibitem[\protect\citeauthoryear{Cai and Guo}{2020}]{tony2020semisupervised}
\begin{barticle}[author]
\bauthor{\bsnm{Cai},~\bfnm{T~Tony}\binits{T.~T.}} \AND
  \bauthor{\bsnm{Guo},~\bfnm{Zijian}\binits{Z.}}
(\byear{2020}).
\btitle{Semisupervised inference for explained variance in high dimensional
  linear regression and its applications}.
\bjournal{Journal of the Royal Statistical Society: Series B (Statistical
  Methodology)}
\bvolume{82}
\bpages{391--419}.
\end{barticle}
\endbibitem

\bibitem[\protect\citeauthoryear{Chakrabortty and
  Cai}{2018}]{chakrabortty2018efficient}
\begin{barticle}[author]
\bauthor{\bsnm{Chakrabortty},~\bfnm{Abhishek}\binits{A.}} \AND
  \bauthor{\bsnm{Cai},~\bfnm{Tianxi}\binits{T.}}
(\byear{2018}).
\btitle{Efficient and adaptive linear regression in semi-supervised settings}.
\bjournal{The Annals of Statistics}
\bvolume{46}
\bpages{1541--1572}.
\end{barticle}
\endbibitem

\bibitem[\protect\citeauthoryear{Chakrabortty and
  Dai}{2022}]{chakrabortty2022general}
\begin{barticle}[author]
\bauthor{\bsnm{Chakrabortty},~\bfnm{Abhishek}\binits{A.}} \AND
  \bauthor{\bsnm{Dai},~\bfnm{Guorong}\binits{G.}}
(\byear{2022}).
\btitle{A general framework for treatment effect estimation in semi-supervised
  and high dimensional settings}.
\bjournal{arXiv preprint arXiv:2201.00468}.
\end{barticle}
\endbibitem

\bibitem[\protect\citeauthoryear{Chakrabortty, Dai and
  Carroll}{2022}]{chakrabortty2022semi}
\begin{barticle}[author]
\bauthor{\bsnm{Chakrabortty},~\bfnm{Abhishek}\binits{A.}},
  \bauthor{\bsnm{Dai},~\bfnm{Guorong}\binits{G.}} \AND
  \bauthor{\bsnm{Carroll},~\bfnm{Raymond~J}\binits{R.~J.}}
(\byear{2022}).
\btitle{Semi-Supervised Quantile Estimation: Robust and Efficient Inference in
  High Dimensional Settings}.
\bjournal{arXiv preprint arXiv:2201.10208}.
\end{barticle}
\endbibitem

\bibitem[\protect\citeauthoryear{Chakrabortty
  et~al.}{2019}]{chakrabortty2019high}
\begin{barticle}[author]
\bauthor{\bsnm{Chakrabortty},~\bfnm{Abhishek}\binits{A.}},
  \bauthor{\bsnm{Lu},~\bfnm{Jiarui}\binits{J.}},
  \bauthor{\bsnm{Cai},~\bfnm{T~Tony}\binits{T.~T.}} \AND
  \bauthor{\bsnm{Li},~\bfnm{Hongzhe}\binits{H.}}
(\byear{2019}).
\btitle{High Dimensional M-Estimation with Missing Outcomes: A Semi-Parametric
  Framework}.
\bjournal{arXiv preprint arXiv:1911.11345}.
\end{barticle}
\endbibitem

\bibitem[\protect\citeauthoryear{Chan et~al.}{2020}]{chan2020semi}
\begin{barticle}[author]
\bauthor{\bsnm{Chan},~\bfnm{Stephanie~F}\binits{S.~F.}},
  \bauthor{\bsnm{Hejblum},~\bfnm{Boris~P}\binits{B.~P.}},
  \bauthor{\bsnm{Chakrabortty},~\bfnm{Abhishek}\binits{A.}} \AND
  \bauthor{\bsnm{Cai},~\bfnm{Tianxi}\binits{T.}}
(\byear{2020}).
\btitle{Semi-supervised estimation of covariance with application to
  phenome-wide association studies with electronic medical records data}.
\bjournal{Statistical Methods in Medical Research}
\bvolume{29}
\bpages{455--465}.
\end{barticle}
\endbibitem

\bibitem[\protect\citeauthoryear{Chapelle, Sch{\"o}lkopf and
  Zien}{2006}]{chapelle2006semi}
\begin{bbook}[author]
\bauthor{\bsnm{Chapelle},~\bfnm{Olivier}\binits{O.}},
  \bauthor{\bsnm{Sch{\"o}lkopf},~\bfnm{Bernhard}\binits{B.}} \AND
  \bauthor{\bsnm{Zien},~\bfnm{Alexander}\binits{A.}}
(\byear{2006}).
\btitle{Semi-Supervised Learning}.
\bpublisher{MIT Press}.
\end{bbook}
\endbibitem

\bibitem[\protect\citeauthoryear{Cheng, Ananthakrishnan and
  Cai}{2021}]{cheng2021robust}
\begin{barticle}[author]
\bauthor{\bsnm{Cheng},~\bfnm{David}\binits{D.}},
  \bauthor{\bsnm{Ananthakrishnan},~\bfnm{Ashwin~N}\binits{A.~N.}} \AND
  \bauthor{\bsnm{Cai},~\bfnm{Tianxi}\binits{T.}}
(\byear{2021}).
\btitle{Robust and efficient semi-supervised estimation of average treatment
  effects with application to electronic health records data}.
\bjournal{Biometrics}
\bvolume{77}
\bpages{413--423}.
\end{barticle}
\endbibitem

\bibitem[\protect\citeauthoryear{Chernozhukov
  et~al.}{2018}]{chernozhukov2018double}
\begin{barticle}[author]
\bauthor{\bsnm{Chernozhukov},~\bfnm{Victor}\binits{V.}},
  \bauthor{\bsnm{Chetverikov},~\bfnm{Denis}\binits{D.}},
  \bauthor{\bsnm{Demirer},~\bfnm{Mert}\binits{M.}},
  \bauthor{\bsnm{Duflo},~\bfnm{Esther}\binits{E.}},
  \bauthor{\bsnm{Hansen},~\bfnm{Christian}\binits{C.}},
  \bauthor{\bsnm{Newey},~\bfnm{Whitney}\binits{W.}} \AND
  \bauthor{\bsnm{Robins},~\bfnm{James}\binits{J.}}
(\byear{2018}).
\btitle{Double/debiased machine learning for treatment and structural
  parameters}.
\bjournal{The Econometrics Journal}
\bvolume{21}
\bpages{C1--C68}.
\end{barticle}
\endbibitem

\bibitem[\protect\citeauthoryear{Chi et~al.}{2022}]{chi2022asymptotic}
\begin{barticle}[author]
\bauthor{\bsnm{Chi},~\bfnm{Chien-Ming}\binits{C.-M.}},
  \bauthor{\bsnm{Vossler},~\bfnm{Patrick}\binits{P.}},
  \bauthor{\bsnm{Fan},~\bfnm{Yingying}\binits{Y.}} \AND
  \bauthor{\bsnm{Lv},~\bfnm{Jinchi}\binits{J.}}
(\byear{2022}).
\btitle{Asymptotic properties of high-dimensional random forests}.
\bjournal{The Annals of Statistics}
\bvolume{50}
\bpages{3415--3438}.
\end{barticle}
\endbibitem

\bibitem[\protect\citeauthoryear{Colnet et~al.}{2024}]{colnet2024causal}
\begin{barticle}[author]
\bauthor{\bsnm{Colnet},~\bfnm{B{\'e}n{\'e}dicte}\binits{B.}},
  \bauthor{\bsnm{Mayer},~\bfnm{Imke}\binits{I.}},
  \bauthor{\bsnm{Chen},~\bfnm{Guanhua}\binits{G.}},
  \bauthor{\bsnm{Dieng},~\bfnm{Awa}\binits{A.}},
  \bauthor{\bsnm{Li},~\bfnm{Ruohong}\binits{R.}},
  \bauthor{\bsnm{Varoquaux},~\bfnm{Ga{\"e}l}\binits{G.}},
  \bauthor{\bsnm{Vert},~\bfnm{Jean-Philippe}\binits{J.-P.}},
  \bauthor{\bsnm{Josse},~\bfnm{Julie}\binits{J.}} \AND
  \bauthor{\bsnm{Yang},~\bfnm{Shu}\binits{S.}}
(\byear{2024}).
\btitle{Causal inference methods for combining randomized trials and
  observational studies: a review}.
\bjournal{Statistical science}
\bvolume{39}
\bpages{165--191}.
\end{barticle}
\endbibitem

\bibitem[\protect\citeauthoryear{Crump et~al.}{2009}]{crump2009dealing}
\begin{barticle}[author]
\bauthor{\bsnm{Crump},~\bfnm{Richard~K}\binits{R.~K.}},
  \bauthor{\bsnm{Hotz},~\bfnm{V~Joseph}\binits{V.~J.}},
  \bauthor{\bsnm{Imbens},~\bfnm{Guido~W}\binits{G.~W.}} \AND
  \bauthor{\bsnm{Mitnik},~\bfnm{Oscar~A}\binits{O.~A.}}
(\byear{2009}).
\btitle{Dealing with limited overlap in estimation of average treatment
  effects}.
\bjournal{Biometrika}
\bvolume{96}
\bpages{187--199}.
\end{barticle}
\endbibitem

\bibitem[\protect\citeauthoryear{Dahabreh
  et~al.}{2019}]{dahabreh2019generalizing}
\begin{barticle}[author]
\bauthor{\bsnm{Dahabreh},~\bfnm{Issa~J}\binits{I.~J.}},
  \bauthor{\bsnm{Robertson},~\bfnm{Sarah~E}\binits{S.~E.}},
  \bauthor{\bsnm{Tchetgen},~\bfnm{Eric~J}\binits{E.~J.}},
  \bauthor{\bsnm{Stuart},~\bfnm{Elizabeth~A}\binits{E.~A.}} \AND
  \bauthor{\bsnm{Hern{\'a}n},~\bfnm{Miguel~A}\binits{M.~A.}}
(\byear{2019}).
\btitle{Generalizing causal inferences from individuals in randomized trials to
  all trial-eligible individuals}.
\bjournal{Biometrics}
\bvolume{75}
\bpages{685--694}.
\end{barticle}
\endbibitem

\bibitem[\protect\citeauthoryear{D'Amour et~al.}{2021}]{d2021overlap}
\begin{barticle}[author]
\bauthor{\bsnm{D'Amour},~\bfnm{Alexander}\binits{A.}},
  \bauthor{\bsnm{Ding},~\bfnm{Peng}\binits{P.}},
  \bauthor{\bsnm{Feller},~\bfnm{Avi}\binits{A.}},
  \bauthor{\bsnm{Lei},~\bfnm{Lihua}\binits{L.}} \AND
  \bauthor{\bsnm{Sekhon},~\bfnm{Jasjeet}\binits{J.}}
(\byear{2021}).
\btitle{Overlap in observational studies with high-dimensional covariates}.
\bjournal{Journal of Econometrics}
\bvolume{221}
\bpages{644--654}.
\end{barticle}
\endbibitem

\bibitem[\protect\citeauthoryear{Degtiar and Rose}{2023}]{degtiar2023review}
\begin{barticle}[author]
\bauthor{\bsnm{Degtiar},~\bfnm{Irina}\binits{I.}} \AND
  \bauthor{\bsnm{Rose},~\bfnm{Sherri}\binits{S.}}
(\byear{2023}).
\btitle{A review of generalizability and transportability}.
\bjournal{Annual Review of Statistics and Its Application}
\bvolume{10}
\bpages{501--524}.
\end{barticle}
\endbibitem

\bibitem[\protect\citeauthoryear{Dukes, Avagyan and
  Vansteelandt}{2020}]{dukes2020doubly}
\begin{barticle}[author]
\bauthor{\bsnm{Dukes},~\bfnm{Oliver}\binits{O.}},
  \bauthor{\bsnm{Avagyan},~\bfnm{Vahe}\binits{V.}} \AND
  \bauthor{\bsnm{Vansteelandt},~\bfnm{Stijn}\binits{S.}}
(\byear{2020}).
\btitle{Doubly robust tests of exposure effects under high-dimensional
  confounding}.
\bjournal{Biometrics}
\bvolume{76}
\bpages{1190--1200}.
\end{barticle}
\endbibitem

\bibitem[\protect\citeauthoryear{Dukes and
  Vansteelandt}{2021}]{dukes2021inference}
\begin{barticle}[author]
\bauthor{\bsnm{Dukes},~\bfnm{Oliver}\binits{O.}} \AND
  \bauthor{\bsnm{Vansteelandt},~\bfnm{Stijn}\binits{S.}}
(\byear{2021}).
\btitle{Inference for treatment effect parameters in potentially misspecified
  high-dimensional models}.
\bjournal{Biometrika}
\bvolume{108}
\bpages{321--334}.
\end{barticle}
\endbibitem

\bibitem[\protect\citeauthoryear{Farrell}{2015}]{farrell2015robust}
\begin{barticle}[author]
\bauthor{\bsnm{Farrell},~\bfnm{Max~H}\binits{M.~H.}}
(\byear{2015}).
\btitle{Robust inference on average treatment effects with possibly more
  covariates than observations}.
\bjournal{Journal of Econometrics}
\bvolume{189}
\bpages{1--23}.
\end{barticle}
\endbibitem

\bibitem[\protect\citeauthoryear{Graham}{2011}]{graham2011efficiency}
\begin{barticle}[author]
\bauthor{\bsnm{Graham},~\bfnm{Bryan~S}\binits{B.~S.}}
(\byear{2011}).
\btitle{Efficiency bounds for missing data models with semiparametric
  restrictions}.
\bjournal{Econometrica}
\bvolume{79}
\bpages{437--452}.
\end{barticle}
\endbibitem

\bibitem[\protect\citeauthoryear{Hou, Mukherjee and
  Cai}{2021}]{hou2021efficient}
\begin{barticle}[author]
\bauthor{\bsnm{Hou},~\bfnm{Jue}\binits{J.}},
  \bauthor{\bsnm{Mukherjee},~\bfnm{Rajarshi}\binits{R.}} \AND
  \bauthor{\bsnm{Cai},~\bfnm{Tianxi}\binits{T.}}
(\byear{2021}).
\btitle{Efficient and Robust Semi-supervised Estimation of ATE with Partially
  Annotated Treatment and Response}.
\bjournal{arXiv preprint arXiv:2110.12336}.
\end{barticle}
\endbibitem

\bibitem[\protect\citeauthoryear{Imbens and Rubin}{2015}]{imbens2015causal}
\begin{bbook}[author]
\bauthor{\bsnm{Imbens},~\bfnm{Guido~W}\binits{G.~W.}} \AND
  \bauthor{\bsnm{Rubin},~\bfnm{Donald~B}\binits{D.~B.}}
(\byear{2015}).
\btitle{Causal inference in statistics, social, and biomedical sciences}.
\bpublisher{Cambridge University Press}.
\end{bbook}
\endbibitem

\bibitem[\protect\citeauthoryear{Kallus and Mao}{2024}]{kallus2024role}
\begin{barticle}[author]
\bauthor{\bsnm{Kallus},~\bfnm{Nathan}\binits{N.}} \AND
  \bauthor{\bsnm{Mao},~\bfnm{Xiaojie}\binits{X.}}
(\byear{2024}).
\btitle{On the role of surrogates in the efficient estimation of treatment
  effects with limited outcome data}.
\bjournal{Journal of the Royal Statistical Society Series B: Statistical
  Methodology}
\bpages{qkae099}.
\end{barticle}
\endbibitem

\bibitem[\protect\citeauthoryear{Kang and Schafer}{2007}]{kang2007demystifying}
\begin{barticle}[author]
\bauthor{\bsnm{Kang},~\bfnm{Joseph~DY}\binits{J.~D.}} \AND
  \bauthor{\bsnm{Schafer},~\bfnm{Joseph~L}\binits{J.~L.}}
(\byear{2007}).
\btitle{Demystifying double robustness: A comparison of alternative strategies
  for estimating a population mean from incomplete data}.
\bjournal{Statistical Science}
\bvolume{22}
\bpages{523--539}.
\end{barticle}
\endbibitem

\bibitem[\protect\citeauthoryear{Kawakita and
  Kanamori}{2013}]{kawakita2013semi}
\begin{barticle}[author]
\bauthor{\bsnm{Kawakita},~\bfnm{Masanori}\binits{M.}} \AND
  \bauthor{\bsnm{Kanamori},~\bfnm{Takafumi}\binits{T.}}
(\byear{2013}).
\btitle{Semi-supervised learning with density-ratio estimation}.
\bjournal{Machine Learning}
\bvolume{91}
\bpages{189--209}.
\end{barticle}
\endbibitem

\bibitem[\protect\citeauthoryear{Khan and Tamer}{2010}]{khan2010irregular}
\begin{barticle}[author]
\bauthor{\bsnm{Khan},~\bfnm{Shakeeb}\binits{S.}} \AND
  \bauthor{\bsnm{Tamer},~\bfnm{Elie}\binits{E.}}
(\byear{2010}).
\btitle{Irregular identification, support conditions, and inverse weight
  estimation}.
\bjournal{Econometrica}
\bvolume{78}
\bpages{2021--2042}.
\end{barticle}
\endbibitem

\bibitem[\protect\citeauthoryear{Kuchibhotla and
  Chakrabortty}{2022}]{kuchibhotla2022moving}
\begin{barticle}[author]
\bauthor{\bsnm{Kuchibhotla},~\bfnm{Arun~Kumar}\binits{A.~K.}} \AND
  \bauthor{\bsnm{Chakrabortty},~\bfnm{Abhishek}\binits{A.}}
(\byear{2022}).
\btitle{Moving beyond sub-Gaussianity in high-dimensional statistics:
  Applications in covariance estimation and linear regression}.
\bjournal{Information and Inference: A Journal of the IMA}
\bvolume{11}
\bpages{1389--1456}.
\end{barticle}
\endbibitem

\bibitem[\protect\citeauthoryear{Ledoux and
  Talagrand}{2013}]{ledoux2013probability}
\begin{bbook}[author]
\bauthor{\bsnm{Ledoux},~\bfnm{Michel}\binits{M.}} \AND
  \bauthor{\bsnm{Talagrand},~\bfnm{Michel}\binits{M.}}
(\byear{2013}).
\btitle{Probability in Banach Spaces: isoperimetry and processes}.
\bpublisher{Springer Science \& Business Media}.
\end{bbook}
\endbibitem

\bibitem[\protect\citeauthoryear{Lesko et~al.}{2017}]{lesko2017generalizing}
\begin{barticle}[author]
\bauthor{\bsnm{Lesko},~\bfnm{Catherine~R}\binits{C.~R.}},
  \bauthor{\bsnm{Buchanan},~\bfnm{Ashley~L}\binits{A.~L.}},
  \bauthor{\bsnm{Westreich},~\bfnm{Daniel}\binits{D.}},
  \bauthor{\bsnm{Edwards},~\bfnm{Jessie~K}\binits{J.~K.}},
  \bauthor{\bsnm{Hudgens},~\bfnm{Michael~G}\binits{M.~G.}} \AND
  \bauthor{\bsnm{Cole},~\bfnm{Stephen~R}\binits{S.~R.}}
(\byear{2017}).
\btitle{Generalizing study results: a potential outcomes perspective}.
\bjournal{Epidemiology (Cambridge, Mass.)}
\bvolume{28}
\bpages{553}.
\end{barticle}
\endbibitem

\bibitem[\protect\citeauthoryear{Liu, Wang and Wang}{2023}]{liu2023root}
\begin{barticle}[author]
\bauthor{\bsnm{Liu},~\bfnm{Lin}\binits{L.}},
  \bauthor{\bsnm{Wang},~\bfnm{Xinbo}\binits{X.}} \AND
  \bauthor{\bsnm{Wang},~\bfnm{Yuhao}\binits{Y.}}
(\byear{2023}).
\btitle{Root-n consistent semiparametric learning with high-dimensional
  nuisance functions under minimal sparsity}.
\bjournal{arXiv preprint arXiv:2305.04174}.
\end{barticle}
\endbibitem

\bibitem[\protect\citeauthoryear{Manski}{1997}]{manski1997monotone}
\begin{barticle}[author]
\bauthor{\bsnm{Manski},~\bfnm{Charles~F}\binits{C.~F.}}
(\byear{1997}).
\btitle{Monotone Treatment Response}.
\bjournal{Econometrica}
\bvolume{65}
\bpages{1311}.
\end{barticle}
\endbibitem

\bibitem[\protect\citeauthoryear{Manski and Pepper}{2000}]{manski2000monotone}
\begin{barticle}[author]
\bauthor{\bsnm{Manski},~\bfnm{Charles~F}\binits{C.~F.}} \AND
  \bauthor{\bsnm{Pepper},~\bfnm{John~V}\binits{J.~V.}}
(\byear{2000}).
\btitle{Monotone Instrumental Variables: With an Application to the Returns to
  Schooling}.
\bjournal{Econometrica}
\bvolume{68}
\bpages{997--1010}.
\end{barticle}
\endbibitem

\bibitem[\protect\citeauthoryear{Mebane and Poast}{2013}]{mebane2013causal}
\begin{barticle}[author]
\bauthor{\bsnm{Mebane},~\bfnm{Walter~R}\binits{W.~R.}} \AND
  \bauthor{\bsnm{Poast},~\bfnm{Paul}\binits{P.}}
(\byear{2013}).
\btitle{Causal inference without ignorability: Identification with nonrandom
  assignment and missing treatment data}.
\bjournal{Political Analysis}
\bvolume{21}
\bpages{233--251}.
\end{barticle}
\endbibitem

\bibitem[\protect\citeauthoryear{Molinari}{2010}]{molinari2010missing}
\begin{barticle}[author]
\bauthor{\bsnm{Molinari},~\bfnm{Francesca}\binits{F.}}
(\byear{2010}).
\btitle{Missing Treatments}.
\bjournal{Journal of Business \& Economic Statistics}
\bvolume{28}
\bpages{82--95}.
\end{barticle}
\endbibitem

\bibitem[\protect\citeauthoryear{Mou et~al.}{2023}]{mou2023kernel}
\begin{barticle}[author]
\bauthor{\bsnm{Mou},~\bfnm{Wenlong}\binits{W.}},
  \bauthor{\bsnm{Ding},~\bfnm{Peng}\binits{P.}},
  \bauthor{\bsnm{Wainwright},~\bfnm{Martin~J}\binits{M.~J.}} \AND
  \bauthor{\bsnm{Bartlett},~\bfnm{Peter~L}\binits{P.~L.}}
(\byear{2023}).
\btitle{Kernel-based off-policy estimation without overlap: Instance optimality
  beyond semiparametric efficiency}.
\bjournal{arXiv preprint arXiv:2301.06240}.
\end{barticle}
\endbibitem

\bibitem[\protect\citeauthoryear{Negahban et~al.}{2010}]{negahban2010unified}
\begin{barticle}[author]
\bauthor{\bsnm{Negahban},~\bfnm{Sahand~N}\binits{S.~N.}},
  \bauthor{\bsnm{Ravikumar},~\bfnm{Pradeep}\binits{P.}},
  \bauthor{\bsnm{Wainwright},~\bfnm{Martin~J}\binits{M.~J.}} \AND
  \bauthor{\bsnm{Yu},~\bfnm{Bin}\binits{B.}}
(\byear{2010}).
\btitle{A Unified Framework for High-Dimensional Analysis of M-Estimators with
  Decomposable Regularizers}.
\bjournal{arXiv preprint arXiv:1010.2731}.
\end{barticle}
\endbibitem

\bibitem[\protect\citeauthoryear{Ning, Sida and Imai}{2020}]{ning2020robust}
\begin{barticle}[author]
\bauthor{\bsnm{Ning},~\bfnm{Yang}\binits{Y.}},
  \bauthor{\bsnm{Sida},~\bfnm{Peng}\binits{P.}} \AND
  \bauthor{\bsnm{Imai},~\bfnm{Kosuke}\binits{K.}}
(\byear{2020}).
\btitle{Robust estimation of causal effects via a high-dimensional covariate
  balancing propensity score}.
\bjournal{Biometrika}
\bvolume{107}
\bpages{533--554}.
\end{barticle}
\endbibitem

\bibitem[\protect\citeauthoryear{Robins, Rotnitzky and
  Zhao}{1994}]{robins1994estimation}
\begin{barticle}[author]
\bauthor{\bsnm{Robins},~\bfnm{James~M}\binits{J.~M.}},
  \bauthor{\bsnm{Rotnitzky},~\bfnm{Andrea}\binits{A.}} \AND
  \bauthor{\bsnm{Zhao},~\bfnm{Lue~Ping}\binits{L.~P.}}
(\byear{1994}).
\btitle{Estimation of regression coefficients when some regressors are not
  always observed}.
\bjournal{Journal of the American Statistical Association}
\bvolume{89}
\bpages{846--866}.
\end{barticle}
\endbibitem

\bibitem[\protect\citeauthoryear{Robins and
  Rotnitzky}{1995}]{robins1995semiparametric}
\begin{barticle}[author]
\bauthor{\bsnm{Robins},~\bfnm{James~M}\binits{J.~M.}} \AND
  \bauthor{\bsnm{Rotnitzky},~\bfnm{Andrea}\binits{A.}}
(\byear{1995}).
\btitle{Semiparametric efficiency in multivariate regression models with
  missing data}.
\bjournal{Journal of the American Statistical Association}
\bvolume{90}
\bpages{122--129}.
\end{barticle}
\endbibitem

\bibitem[\protect\citeauthoryear{Robins et~al.}{2007}]{robins2007comment}
\begin{barticle}[author]
\bauthor{\bsnm{Robins},~\bfnm{James}\binits{J.}},
  \bauthor{\bsnm{Sued},~\bfnm{Mariela}\binits{M.}},
  \bauthor{\bsnm{Lei-Gomez},~\bfnm{Quanhong}\binits{Q.}} \AND
  \bauthor{\bsnm{Rotnitzky},~\bfnm{Andrea}\binits{A.}}
(\byear{2007}).
\btitle{Comment: Performance of double-robust estimators when" inverse
  probability" weights are highly variable}.
\bjournal{Statistical Science}
\bvolume{22}
\bpages{544--559}.
\end{barticle}
\endbibitem

\bibitem[\protect\citeauthoryear{Rosenbaum and
  Rubin}{1983}]{rosenbaum1983central}
\begin{barticle}[author]
\bauthor{\bsnm{Rosenbaum},~\bfnm{Paul~R}\binits{P.~R.}} \AND
  \bauthor{\bsnm{Rubin},~\bfnm{Donald~B}\binits{D.~B.}}
(\byear{1983}).
\btitle{The central role of the propensity score in observational studies for
  causal effects}.
\bjournal{Biometrika}
\bvolume{70}
\bpages{41--55}.
\end{barticle}
\endbibitem

\bibitem[\protect\citeauthoryear{Rothe}{2017}]{rothe2017robust}
\begin{barticle}[author]
\bauthor{\bsnm{Rothe},~\bfnm{Christoph}\binits{C.}}
(\byear{2017}).
\btitle{Robust confidence intervals for average treatment effects under limited
  overlap}.
\bjournal{Econometrica}
\bvolume{85}
\bpages{645--660}.
\end{barticle}
\endbibitem

\bibitem[\protect\citeauthoryear{Rubin}{1974}]{rubin1974estimating}
\begin{barticle}[author]
\bauthor{\bsnm{Rubin},~\bfnm{Donald~B}\binits{D.~B.}}
(\byear{1974}).
\btitle{Estimating causal effects of treatments in randomized and nonrandomized
  studies.}
\bjournal{Journal of Educational Psychology}
\bvolume{66}
\bpages{688}.
\end{barticle}
\endbibitem

\bibitem[\protect\citeauthoryear{Rubin}{1976}]{rubin1976inference}
\begin{barticle}[author]
\bauthor{\bsnm{Rubin},~\bfnm{Donald~B}\binits{D.~B.}}
(\byear{1976}).
\btitle{Inference and missing data}.
\bjournal{Biometrika}
\bvolume{63}
\bpages{581--592}.
\end{barticle}
\endbibitem

\bibitem[\protect\citeauthoryear{Rudelson and
  Zhou}{2012}]{rudelson2012reconstruction}
\begin{binproceedings}[author]
\bauthor{\bsnm{Rudelson},~\bfnm{Mark}\binits{M.}} \AND
  \bauthor{\bsnm{Zhou},~\bfnm{Shuheng}\binits{S.}}
(\byear{2012}).
\btitle{Reconstruction from Anisotropic Random Measurements}.
In \bbooktitle{Proceedings of the 25th Annual Conference on Learning Theory}
(\beditor{\bfnm{Shie}\binits{S.}~\bsnm{Mannor}},
  \beditor{\bfnm{Nathan}\binits{N.}~\bsnm{Srebro}} \AND
  \beditor{\bfnm{Robert~C.}\binits{R.~C.}~\bsnm{Williamson}}, eds.).
\bseries{Proceedings of Machine Learning Research}
\bvolume{23}
\bpages{10.1--10.24}.
\bpublisher{JMLR Workshop and Conference Proceedings}, \baddress{Edinburgh,
  Scotland}.
\end{binproceedings}
\endbibitem

\bibitem[\protect\citeauthoryear{Scharfstein, Rotnitzky and
  Robins}{1999}]{scharfstein1999adjusting}
\begin{barticle}[author]
\bauthor{\bsnm{Scharfstein},~\bfnm{Daniel~O}\binits{D.~O.}},
  \bauthor{\bsnm{Rotnitzky},~\bfnm{Andrea}\binits{A.}} \AND
  \bauthor{\bsnm{Robins},~\bfnm{James~M}\binits{J.~M.}}
(\byear{1999}).
\btitle{Adjusting for nonignorable drop-out using semiparametric nonresponse
  models}.
\bjournal{Journal of the American Statistical Association}
\bvolume{94}
\bpages{1096--1120}.
\end{barticle}
\endbibitem

\bibitem[\protect\citeauthoryear{Shi, Pan and Miao}{2023}]{shi2023data}
\begin{barticle}[author]
\bauthor{\bsnm{Shi},~\bfnm{Xu}\binits{X.}},
  \bauthor{\bsnm{Pan},~\bfnm{Ziyang}\binits{Z.}} \AND
  \bauthor{\bsnm{Miao},~\bfnm{Wang}\binits{W.}}
(\byear{2023}).
\btitle{Data integration in causal inference}.
\bjournal{Wiley Interdisciplinary Reviews: Computational Statistics}
\bvolume{15}
\bpages{e1581}.
\end{barticle}
\endbibitem

\bibitem[\protect\citeauthoryear{Smucler, Rotnitzky and
  Robins}{2019}]{smucler2019unifying}
\begin{barticle}[author]
\bauthor{\bsnm{Smucler},~\bfnm{Ezequiel}\binits{E.}},
  \bauthor{\bsnm{Rotnitzky},~\bfnm{Andrea}\binits{A.}} \AND
  \bauthor{\bsnm{Robins},~\bfnm{James~M}\binits{J.~M.}}
(\byear{2019}).
\btitle{A unifying approach for doubly-robust $\ell_1$ regularized estimation
  of causal contrasts}.
\bjournal{arXiv preprint arXiv:1904.03737}.
\end{barticle}
\endbibitem

\bibitem[\protect\citeauthoryear{Sun and Tan}{2022}]{sun2022high}
\begin{barticle}[author]
\bauthor{\bsnm{Sun},~\bfnm{Baoluo}\binits{B.}} \AND
  \bauthor{\bsnm{Tan},~\bfnm{Zhiqiang}\binits{Z.}}
(\byear{2022}).
\btitle{High-dimensional model-assisted inference for local average treatment
  effects with instrumental variables}.
\bjournal{Journal of Business \& Economic Statistics}
\bvolume{40}
\bpages{1732--1744}.
\end{barticle}
\endbibitem

\bibitem[\protect\citeauthoryear{Tan}{2020}]{tan2020model}
\begin{barticle}[author]
\bauthor{\bsnm{Tan},~\bfnm{Zhiqiang}\binits{Z.}}
(\byear{2020}).
\btitle{Model-assisted inference for treatment effects using regularized
  calibrated estimation with high-dimensional data}.
\bjournal{The Annals of Statistics}
\bvolume{48}
\bpages{811--837}.
\end{barticle}
\endbibitem

\bibitem[\protect\citeauthoryear{Tsiatis}{2007}]{tsiatis2007semiparametric}
\begin{bbook}[author]
\bauthor{\bsnm{Tsiatis},~\bfnm{Anastasios}\binits{A.}}
(\byear{2007}).
\btitle{Semiparametric theory and missing data}.
\bpublisher{Springer Science \& Business Media}.
\end{bbook}
\endbibitem

\bibitem[\protect\citeauthoryear{van~de Geer and
  Lederer}{2013}]{van2013bernstein}
\begin{barticle}[author]
\bauthor{\bparticle{van~de} \bsnm{Geer},~\bfnm{Sara}\binits{S.}} \AND
  \bauthor{\bsnm{Lederer},~\bfnm{Johannes}\binits{J.}}
(\byear{2013}).
\btitle{The Bernstein--Orlicz norm and deviation inequalities}.
\bjournal{Probability Theory and Related Fields}
\bvolume{157}
\bpages{225--250}.
\end{barticle}
\endbibitem

\bibitem[\protect\citeauthoryear{Van~der Vaart}{2000}]{van2000asymptotic}
\begin{bbook}[author]
\bauthor{\bparticle{Van~der} \bsnm{Vaart},~\bfnm{Aad~W}\binits{A.~W.}}
(\byear{2000}).
\btitle{Asymptotic statistics}
\bvolume{3}.
\bpublisher{Cambridge university press}.
\end{bbook}
\endbibitem

\bibitem[\protect\citeauthoryear{Wainwright}{2019}]{wainwright2019high}
\begin{bbook}[author]
\bauthor{\bsnm{Wainwright},~\bfnm{Martin~J}\binits{M.~J.}}
(\byear{2019}).
\btitle{High-dimensional statistics: A non-asymptotic viewpoint}
\bvolume{48}.
\bpublisher{Cambridge University Press}.
\end{bbook}
\endbibitem

\bibitem[\protect\citeauthoryear{Wang}{2020}]{wang2020logistic}
\begin{binproceedings}[author]
\bauthor{\bsnm{Wang},~\bfnm{HaiYing}\binits{H.}}
(\byear{2020}).
\btitle{Logistic regression for massive data with rare events}.
In \bbooktitle{International Conference on Machine Learning}
\bpages{9829--9836}.
\bpublisher{PMLR}.
\end{binproceedings}
\endbibitem

\bibitem[\protect\citeauthoryear{Wei et~al.}{2022}]{wei2022doubly}
\begin{barticle}[author]
\bauthor{\bsnm{Wei},~\bfnm{Kecheng}\binits{K.}},
  \bauthor{\bsnm{Qin},~\bfnm{Guoyou}\binits{G.}},
  \bauthor{\bsnm{Zhang},~\bfnm{Jiajia}\binits{J.}} \AND
  \bauthor{\bsnm{Sui},~\bfnm{Xuemei}\binits{X.}}
(\byear{2022}).
\btitle{Doubly robust estimation in causal inference with missing outcomes:
  With an application to the Aerobics Center Longitudinal Study}.
\bjournal{Computational Statistics \& Data Analysis}
\bvolume{168}
\bpages{107399}.
\end{barticle}
\endbibitem

\bibitem[\protect\citeauthoryear{Xue, Ma and Li}{2021}]{xue2021semi}
\begin{barticle}[author]
\bauthor{\bsnm{Xue},~\bfnm{Fei}\binits{F.}},
  \bauthor{\bsnm{Ma},~\bfnm{Rong}\binits{R.}} \AND
  \bauthor{\bsnm{Li},~\bfnm{Hongzhe}\binits{H.}}
(\byear{2021}).
\btitle{Semi-Supervised Statistical Inference for High-Dimensional Linear
  Regression with Blockwise Missing Data}.
\bjournal{arXiv preprint arXiv:2106.03344}.
\end{barticle}
\endbibitem

\bibitem[\protect\citeauthoryear{Yang and Ding}{2018}]{yang2018asymptotic}
\begin{barticle}[author]
\bauthor{\bsnm{Yang},~\bfnm{Shu}\binits{S.}} \AND
  \bauthor{\bsnm{Ding},~\bfnm{Peng}\binits{P.}}
(\byear{2018}).
\btitle{Asymptotic inference of causal effects with observational studies
  trimmed by the estimated propensity scores}.
\bjournal{Biometrika}
\bvolume{105}
\bpages{487--493}.
\end{barticle}
\endbibitem

\bibitem[\protect\citeauthoryear{Zhang and Bradic}{2022}]{zhang2022high}
\begin{barticle}[author]
\bauthor{\bsnm{Zhang},~\bfnm{Yuqian}\binits{Y.}} \AND
  \bauthor{\bsnm{Bradic},~\bfnm{Jelena}\binits{J.}}
(\byear{2022}).
\btitle{High-dimensional semi-supervised learning: in search of optimal
  inference of the mean}.
\bjournal{Biometrika}
\bvolume{109}
\bpages{387--403}.
\end{barticle}
\endbibitem

\bibitem[\protect\citeauthoryear{Zhang, Brown and Cai}{2019}]{zhang2019semi}
\begin{barticle}[author]
\bauthor{\bsnm{Zhang},~\bfnm{Anru}\binits{A.}},
  \bauthor{\bsnm{Brown},~\bfnm{Lawrence~D}\binits{L.~D.}} \AND
  \bauthor{\bsnm{Cai},~\bfnm{T~Tony}\binits{T.~T.}}
(\byear{2019}).
\btitle{Semi-supervised inference: General theory and estimation of means}.
\bjournal{The Annals of Statistics}
\bvolume{47}
\bpages{2538--2566}.
\end{barticle}
\endbibitem

\bibitem[\protect\citeauthoryear{Zhang, Chakrabortty and
  Bradic}{2023}]{zhang2023double}
\begin{barticle}[author]
\bauthor{\bsnm{Zhang},~\bfnm{Yuqian}\binits{Y.}},
  \bauthor{\bsnm{Chakrabortty},~\bfnm{Abhishek}\binits{A.}} \AND
  \bauthor{\bsnm{Bradic},~\bfnm{Jelena}\binits{J.}}
(\byear{2023}).
\btitle{Double robust semi-supervised inference for the mean: Selection bias
  under mar labeling with decaying overlap}.
\bjournal{Information and Inference: A Journal of the IMA}
\bvolume{12}
\bpages{2066--2159}.
\end{barticle}
\endbibitem

\bibitem[\protect\citeauthoryear{Zhang et~al.}{2016}]{zhang2016causal}
\begin{barticle}[author]
\bauthor{\bsnm{Zhang},~\bfnm{Zhiwei}\binits{Z.}},
  \bauthor{\bsnm{Liu},~\bfnm{Wei}\binits{W.}},
  \bauthor{\bsnm{Zhang},~\bfnm{Bo}\binits{B.}},
  \bauthor{\bsnm{Tang},~\bfnm{Li}\binits{L.}} \AND
  \bauthor{\bsnm{Zhang},~\bfnm{Jun}\binits{J.}}
(\byear{2016}).
\btitle{Causal inference with missing exposure information: Methods and
  applications to an obstetric study}.
\bjournal{Statistical Methods in Medical Research}
\bvolume{25}
\bpages{2053--2066}.
\end{barticle}
\endbibitem

\bibitem[\protect\citeauthoryear{Zhu}{2005}]{zhu2005semi}
\begin{btechreport}[author]
\bauthor{\bsnm{Zhu},~\bfnm{Xiaojin~Jerry}\binits{X.~J.}}
(\byear{2005}).
\btitle{Semi-supervised learning literature survey}
\btype{Technical Report},
\bpublisher{University of Wisconsin-Madison, Department of Computer Sciences}.
\end{btechreport}
\endbibitem

\end{thebibliography}
